\documentclass[11pt,twoside,a4paper]{book}

\usepackage[T1]{fontenc}
\usepackage[utf8]{inputenc}
\usepackage[english]{babel}
\usepackage[pdftex]{graphicx}           
\usepackage{setspace}                   
\usepackage{indentfirst}                
\usepackage{makeidx}                    
\usepackage[nottoc]{tocbibind}          
\usepackage{courier}                    
\usepackage{type1cm}                    
\usepackage{listings}                   
\usepackage{titletoc}
\usepackage[fixlanguage]{babelbib}
\usepackage[font=small,format=plain,labelfont=bf,up,textfont=it,up]{caption}
\usepackage[usenames,svgnames,dvipsnames]{xcolor}
\usepackage[a4paper,top=2.54cm,bottom=2.0cm,left=2.0cm,right=2.54cm]{geometry} 
\usepackage[pdftex,plainpages=false,pdfpagelabels,pagebackref,colorlinks=true,citecolor=DarkGreen,linkcolor=NavyBlue,urlcolor=DarkRed,filecolor=green,bookmarksopen=true]{hyperref} 
\usepackage[all]{hypcap}                
\usepackage[square,sort,nonamebreak,comma]{natbib}  
\fontsize{60}{62}\usefont{OT1}{cmr}{m}{n}{\selectfont}

\usepackage{ragged2e}

\usepackage{epigraph}
\usepackage{framed}
\usepackage{tikz, tkz-euclide}
\usetikzlibrary{patterns}
\usetikzlibrary{calc}
\usetikzlibrary{decorations.markings}
\usepackage{fancybox}
\usepackage{mdframed}
\usepackage{float}

\usepackage{amsmath,amssymb,amsthm,epsfig,color,graphicx, mathtools, bbm, dsfont, geometry, float, mathrsfs, upgreek, bm, enumitem, comment, hyperref,hyphenat}

\definecolor{vermelho}{RGB}{208,2,27}    
\definecolor{verde}{RGB}{126,211,33} 

\usepackage{tikz, tkz-euclide}
\usetikzlibrary{patterns}
\usetikzlibrary{arrows.meta}
\usetikzlibrary{calc}
\usetikzlibrary{decorations.markings}
\usepackage{fancybox}

\usepackage{graphicx}
\graphicspath{ {./images/} }
\newcommand{\Z}{\mathbb{Z}}
\newcommand{\E}{\mathbb{E}}
\newcommand{\V}{\mathbb{V}}

\newcommand{\calW}{\mathcal{W}}

\newcommand{\lab}{\mathrm{lab}}     
\newcommand{\supp}{\text{supp}\;}

\newcommand{\diam}{\mathrm{diam}}
\newcommand{\dis}{\mathrm{dist}}

\newcommand{\I}{\mathrm{I}}
\newcommand{\rmh}{\mathrm{h}}
\newcommand{\Sp}{\mathrm{sp}}

\newcommand{\R}{\mathcal{R}}
\newcommand{\fint}{\partial_{\mathrm{in}}}
\newcommand{\fext}{\partial_{\mathrm{ex}}}
\newcommand{\C}{\mathscr{C}}

\newcommand{\calP}{\mathcal{P}}
\newcommand{\lambdadec}{\widehat{\bm{\lambda}}}
\newcommand{\IS}{\mathrm{IS}}

\newcommand{\vertiii}[1]{{\left\vert\kern-0.25ex\left\vert\kern-0.25ex\left\vert #1 
    \right\vert\kern-0.25ex\right\vert\kern-0.25ex\right\vert}}
    
\def\H+{{\mathbb{H}^d_+}}

\title{}

\def\supp{\mathop{\textrm{\rm supp}}\nolimits}            
\def\d{\mathop{\textrm{\rm d}}\nolimits}                  
\def\Int{\mathop{\textrm{\rm Int}}\nolimits}                
\def\Ext{\mathop{\textrm{\rm Ext}}\nolimits}                  

\newcommand{\be}{\begin{equation}}
\newcommand{\ee}{\end{equation}}

\numberwithin{equation}{section}

\usepackage{bbm}
\newtheorem*{theorem*}        {Theorem}
\newtheorem*{conjecture*}   {Conjecture}
\newtheorem{theorem}           {Theorem}[section]
\newtheorem{lemma}[theorem]{Lemma}
\newtheorem*{lemma*}          {Lemma}

\newtheorem{definition}[theorem]{Definition}

\newtheorem{corollary}[theorem]{Corollary}
\newtheorem{proposition}[theorem]{Proposition}

\newtheorem{remark}[theorem]{Remark}

\usepackage{fancyhdr}
\graphicspath{{./figuras/}}             
\makeindex                              
\fontsize{60}{62}\usefont{OT1}{cmr}{m}{n}{\selectfont}
\cleardoublepage
\normalsize

\DeclareMathOperator{\h}{\bm{h}}

\DeclareMathOperator{\s}{\mathrm{sp}}

\DeclarePairedDelimiter\ceil{\lceil}{\rceil} 
\DeclarePairedDelimiter\floor{\lfloor}{\rfloor} 

\begin{document}
\frontmatter 
\fancyhead[RO]{{\footnotesize\rightmark}\hspace{2em}\thepage}
\setcounter{tocdepth}{2}
\fancyhead[LE]{\thepage\hspace{2em}\footnotesize{\leftmark}}
\fancyhead[RE,LO]{}
\fancyhead[RO]{{\footnotesize\rightmark}\hspace{2em}\thepage}

\onehalfspacing  

\thispagestyle{empty}
\begin{center}
    \vspace*{2.3cm}
    \textbf{\Large{Phase Transitions in Ising models: the Semi-infinite with decaying field and the Random Field Long-range}}\\
    
    \vspace*{1.2cm}
    \Large{Jo\~{a}o Maia}
    
    \vskip 2cm
    \textsc{
    Thesis presented\\[-0.25cm] 
    to the \\[-0.25cm]
    Institute of Mathematics and Statistics\\[-0.25cm]
    of the\\[-0.25cm]
    University of S\~ao Paulo\\[-0.25cm]
    in partial fulfillment of the requirements\\[-0.25cm]
    for the degree\\[-0.25cm]
    of\\[-0.25cm]
    Doctor of Science}
    
    \vskip 1.5cm
    Program: Applied Mathematics\\
    Advisor: Prof. Dr. Rodrigo Bissacot\\ 

   	\vskip 1cm
    \normalsize{During the development of this work the author was supported by FAPESP grant 
2018/26698-8.}
    
    \vskip 0.5cm
    \normalsize{S\~ao Paulo, March 2024}
\end{center}

\newpage
\thispagestyle{empty}
    \begin{center}
        \vspace*{2.3 cm}
        \textbf{\Large{Phase Transitions in Ising models: the Semi-infinite with decaying field and the Random Field Long-range}}\\
        \vspace*{2 cm}
    \end{center}

    \vskip 2cm

    \begin{FlushRight}
     This is the final version of this thesis and it contains corrections \\
     and changes suggested by the examiner committee during the\\
     defense of the original work realized on February 7th, 2024.\\
     A copy of the original version of this text is available at the \\
     Institute of Mathematics and Statistics of the University of S\~ao Paulo.
    
    \vskip 2cm

   \end{FlushRight}
    \vskip 4.2cm

    \begin{quote}
    \noindent Examination Board:
    
    \begin{itemize}
		\item Prof. Dr. Rodrigo Bissacot - IME-USP (President)
		\item Prof. Dr. Luiz Renato Goncalves Fontes  - IME-USP
		\item Prof. Dr. Christof Kuelske - Ruhr University Bochum
        \item Prof. Dr. Pierre Picco - CNRS
        \item Prof. Dr. Leandro Chiarini Medeiros - Durham University
    \end{itemize}
      
    \end{quote}
\pagebreak

\pagenumbering{roman}     



\chapter*{Acknowledgements}
\epigraph{Someone once asked me what I would do if I had only 30 days left to live. I promptly answer: for 29 days, I would live my life as usual, going to work, and coming back home. On the last day, I would gather all my friends and family in my house and throw them a barbecue. 
}{\textbf{Antônio Maia}}

To understand the meaning of the quotation opening this section, we need to know more about its author. Antônio Maia was a man with humble beginnings. Born in a village in the Northeast part of Brazil, his childhood was summarized by working as a farmer, in a land where nothing grows, under the scorching sun of the semiarid climate of the region. At the time, in the 70's, this whole region of Brazil was very poor. In the city where he was born, for example, there were not any hospitals, schools, not even electric energy. At 18 years old, he moved to São Paulo, an emerging city at the time. Arriving here, he moved to a small house, living together with around 10 other people, relatives of his that had moved to the big city earlier. Living with no money in a city like São Paulo was no easy task, so he had to accept any kind of work available, from being a waiter, a construction worker, salesman. 

This is to say that most of the days of his life were a combination of hard work and tough decisions. Despite that, in the quote, he claimed he would gladly live all of those days as if they were his last. These words taught me to embrace the difficulties of life, to find happiness in them, and to not see them as the means to an end. This mindset resonated with me throughout my Ph.D. As you may imagine, concluding a Ph.D. in mathematics, especially in a country like Brazil, is not an easy task. The frustrations and uncertainties of this path are tremendous, and so are the effort and the dedication needed to accomplish it. This beautiful idea of finding joy in the difficulties of life was what helped me push forward, in my journey from an undergraduate student to a researcher, a scientist. 

The author of the quotation was my father, who tragically passed away a couple of months before my Ph.D. defense. To him, I own who I am.

My family was the pillar that supported me all those years, and to them, I am profoundly grateful. To my mother Ivoneide Teixeira, who teaches me every day, through example, the importance of dedicating and sacrificing yourself to your work, a value I carry deeply in me. She, like my father, does not come from a privileged background and had to struggle and work hard all her life. Her dedication to work and her unbreakable spirit, she never stops fighting for what she wants, are core values that I try to carry with me. She is the person I respect and own the most. I also thank my sisters, Fernanda Maia and Natália Maia, for the support they always gave me, for being an example to me in several ways, and for the joy you always brought to my life. It is a privilege to have you as my siblings. 

The last member of my family is my girlfriend, Jessica Dantas. We have been together since high school, so we have grown beside one another. Through the hardships and ups and downs of a relationship, I loved and admired her every step of the way. Her life was, in many ways, much tougher than mine. That never made her shine away from any challenge thrown at her. And after overcoming them, she still finds the will to pursue her future goals, which are plenty. She is the person I admire the most, and It is a pleasure and a privilege to have beside me such an amazing person. Her love and care are, to me, my most valuable gifts. I am eager to have you by my side for many more years to come.  

I would like also to thank all the friends I made along the way. My undergraduate colleagues Amanda Barros, Aline Alves, and Diego Mucciolo, as well as the friends I met after like Lucas Affonso, Thiago Raszeja, Rodrigo Cabral, Kelvyn Welsch, and João Rodrigues. All of those are tremendous mathematicians and researchers, with whom I had the pleasure to work and share the most interesting discussions. 

A special thanks goes to my advisor, Rodrigo Bissacot. We have spent many nights at the University, together with Lucas Affonso, working on the content of this thesis, discussing interesting mathematical problems, and also not-so-interesting problems regarding doing science in Brazil. The word "advisor" barely scratches the actual work that needs to be put into forming a Ph.D. in a country like Brazil. The scientific advising is just a fraction of the actual work. While helping an undergrad student become a fully formed scientist by learning the theory about a subject, and producing new and exciting results, you also have to serve as a therapist, a family counselor, and a founding agency even. Rodrigo does all that, and much more. More impressively, he does all that for all his students, in equal measure. He does so thanks to his incessant hard work, driven by his love for science and the university. For all those reasons, I am glad to have Rodrigo not only as an advisor but as a close friend. 

I would also like to thank the Examination Board for accepting to come to Brazil to participate in my Ph.D. defense, and also for the improvements suggested to this thesis. 

Finally, I want to thank the National Centre of Competence in Research SwissMAP, for accepting me in the SwissMAP Masterclass and founding my stay in Switzerland. I also thank the Fundação ao Amparo à Pesquisa do Estado de São Paulo (FAPESP) for the financial support during the development of the work presented in this thesis through the grant 2018/26698-8.
\chapter*{Resumo}

    \noindent MAIA, J. \textbf{Transição de fase em modelos de Ising: o semi-infinito com campo decaindo e o longo-alcance com campo aleatório}. 
2023. 
Tese (Doutorado) - Instituto de Matem\'atica e Estat\'istica,
Universidade de S\~ao Paulo, S\~ao Paulo, 2023.
\\

Nesta tese apresentamos resultados referentes ao problema de transição de fase para dois modelos: o modelo de Ising semi-infinito com um campo decaindo e o modelo de Ising de longo-alcance com um campo aleatório. 

No modelo de Ising semi-infinito, o parâmetro relevante na existência de transição de fase é $\lambda$, a interação entre os spins do sistema e a parede que divide o \textit{lattice}. Introduzindo um campo magnético que da forma $h_i = \lambda |i_d|^{-\delta}$ com $\delta>1$, que decai conforme se afasta da parede, conseguimos mostrar que, em baixas temperaturas, o modelo ainda apresenta um ponto de criticalidade $0< \lambda_c(J,\delta)$ satisfazendo: para $0\leq \lambda <\lambda_c(J,\delta)$ existem múltiplos estados de Gibbs e para $\lambda>\lambda_c(J,\delta)$ temos unicidade.  Mostramos ainda que quando $\delta<1$, $\lambda_c(J, \delta)=0$ e portanto temos sempre unicidade. 

No modelo de Ising de longo-alcance com campo aleatório, estendemos um argumento de Ding e Zhuang do modelo de primeiros vizinhos para o modelo com interação de longo-alcance. Combinando uma generalização dos contornos de Fr\"ohlich-Spencer, proposta por Affonso, Bissacot, Endo e Handa, com um procedimento de \textit{coarse-graining} introduzido por Fisher, Fr\"ohlich, and Spencer, conseguimos provar que o modelo de Ising com interação $J_{xy}=|x-y|^{- \alpha}$ com $\alpha > d$ em dimensão $d\geq 3$ apresenta transição de fase. Consideramos um campo aleatório dado por uma coleção i.i.d com distribuição Gaussiana ou Bernoulli. Nossa prova constitui uma prova alternativa que não usa grupos de renormalização (GR), uma vez que Bricmont e Kupiainen afirmaram que seus resultados usando GR funcionam para qualquer modelo que possua um sistema de contornos. \\

\noindent \textbf{Palavras-chave:} Transição de fase, modelo de Ising semi-infinito, energia livre de superfície, campo externo não-homogêneo, modelo de Ising longo-alcance, campo aleatório, contornos, análise multiescala, coarse-graining, mecânica estatística clássica

\chapter*{Abstract}

    \noindent MAIA, J. \textbf{Phase Transitions in Ising models: the Semi-infinite with decaying field and the Random Field Long-range}. 
2023. 
PhD Thesis - Institute of Mathematics and Statistics,
University of S\~ao Paulo, S\~ao Paulo, 2023. \\

In this thesis, we present results on phase transition for two models: the semi-infinite Ising model with a decaying field, and the long-range Ising model with a random field.

We study the semi-infinite Ising model with an external field $h_i = \lambda |i_d|^{-\delta}$, $\lambda$ is the wall influence, and $\delta>0$. This external field decays as it gets further away from the wall. We are able to show that when $\delta>1$ and $\beta > \beta_c(d)$, there exists a critical value $0< \lambda_c:=\lambda_c(\delta,\beta)$ such that, for $\lambda<\lambda_c$ there is phase transition and for $\lambda>\lambda_c$ we have uniqueness of the Gibbs state. In addition, when $\delta<1$ we have only one Gibbs state for any positive $\beta$ and $\lambda$.

For the model with a random field, we extend the recent argument by Ding and Zhuang from nearest-neighbor to long-range interactions and prove the phase transition in the class of ferromagnetic random field Ising models. Our proof combines a generalization of Fr\"ohlich-Spencer contours to the multidimensional setting proposed by Affonso, Bissacot, Endo and Handa, with the coarse-graining procedure introduced by Fisher, Fr\"ohlich, and Spencer. Our result shows that the Ding-Zhuang strategy is also useful for interactions $J_{xy}=|x-y|^{- \alpha}$ when $\alpha > d$ in dimension $d\geq 3$ if we have a suitable system of contours, yielding an alternative proof that does not use the Renormalization Group Method (RGM), since Bricmont and Kupiainen claimed that the RGM should also work on this generality. We can consider i.i.d. random fields with Gaussian or Bernoulli distributions. \\

\noindent \textbf{Keywords:} Phase transition, semi-infinite Ising model, surface free energy, inhomogeneous external field, Long-range random field Ising model, random field, contours, multiscale analysis, coarse-graining, classical statistical mechanics.

\tableofcontents    



\listoffigures            

\mainmatter

\singlespacing              

\chapter*{Introduction}
\addcontentsline{toc}{chapter}{Introduction}
\markboth{INTRODUCTION}{}
The problem of the presence or absence of phase transitions is key in statistical mechanics. The Ising model is one of the most studied ones on this matter. It was introduced by Lenz and studied by Ising \cite{Ising_25}. In the Ising model, we represent the positions of the particles by the points of the lattice $\Z^d$. Each site is associated with a spin $\sigma_i$, taking values $+1$ or $-1$. Configurations are elements of $\Omega\coloneqq\{-1,+1\}^{\mathbb{Z}^d}$. The parameters of importance are the interaction $\bm{J}=(J_{ij})_{i,j\in\mathbb{Z}^d}$, the external field $\bm{h}=(h_i)_{i\in\mathbb{Z}^d}$, and the inverse temperature $\beta$. The family $\bm{J}$ represents that the energy of the interaction of two sites $i$ and $j$ is $-J_{ij}$ if the spins are $\sigma_i$ and $\sigma_j$ aligned, i.e. if having the same sign, and is $J_{ij}$ otherwise. The most studied Ising model is the one with nearest neighbor ferromagnetic interaction, for which $J_{ij}= J$ whenever $i$ and $j$ are neighbors on the graph and are zero otherwise. The external field $\bm{h}$ represents a quantity that favors each spin to align in the same direction as the field, and so the formal Hamiltonian of the model is given by
\begin{equation*}
H_{J, \bm{h}}(\sigma) = -\sum_{\substack{i,j\in\mathbb{Z}^d \\ |i-j|=1}}J\sigma_i\sigma_j - \sum_{i\in\mathbb{Z}^d} h_i\sigma_i,
\end{equation*}
where $J>0$ and $|i-j|=1$ means that the distance in the graph is one. The role of the temperature is expressed in the formal Gibbs measure, given by
\begin{equation*}
\mu_{\bm{J}, \bm{h}}^{\beta}(\sigma)=\frac{e^{-\beta H_{J,\bm{h}}(\sigma)}}{\mathcal{Z}^{\beta}_{\bm{J},\bm{h}}}     
\end{equation*}
where $\mathcal{Z}^{\beta}_{\bm{J},\bm{h}}$ is the partition function
\begin{equation*}
    \mathcal{Z}^{\beta}_{\bm{J},\bm{h}} = \sum_{\sigma\in\Omega} {e^{-\beta H_{J,\bm{h}}(\sigma)}}.
\end{equation*}

In this thesis, we will study the phase diagram of two variations of this model. First, we will study the \textit{semi-infiinite Ising model with non-homogeneous external fields}. After that, we focus on proving phase transition for the \textit{long-range random field Ising model} for $d\geq 3$.  \\

The semi-infinite Ising model is a variation of the Ising model where, instead of $\Z^d$ ($d\geq 2$), the lattice is $\H+=\mathbb{Z}^{d-1}\times\mathbb{N}$ and the configurations space is $\Omega \coloneqq\{-1,+1\}^{\H+}$. In the semi-infinite model, the sites in the wall $\mathcal{W}=\mathbb{Z}^{d-1}\times \{1\}$ are in contact with a substrate favoring one of the spins. This influence is represented by an external field, with intensity $\lambda\in\mathbb{R}$, acting only on spins at the wall $\calW$. 
The other parameters of the model are the interaction $\bm{J}=(J_{ij})_{i,j\in\H+}$, the external field $\bm{h}=(h_i)_{i\in\H+}$ and the inverse temperature $\beta$. We will always consider nearest neighbor ferromagnetic interaction, hence $J_{ij}= J > 0$ whenever $|i-j|=1$ and are zero otherwise. The distance here is taken concerning the $\ell_1$-norm. The interaction $\bm{J}$ and the external field $\bm{h}$ play the same role in the energy as in the standard Ising model, so the formal Hamiltonian is 
\begin{equation*}
H_{J,\lambda, \bm{h}}(\sigma) = -\sum_{\substack{i,j\in\H+ \\ |i-j|=1}}J\sigma_i\sigma_j - \sum_{i\in\H+} h_i\sigma_i - \sum_{i\in\mathcal{W}}\lambda\sigma_i.
\end{equation*}
The role of the temperature is expressed in the formal Gibbs measure, given by
\begin{equation*}
\mu_{\bm{J},\lambda, \bm{h}}^{\beta}(\sigma)=\frac{e^{-\beta H_{J,\lambda,\bm{h}}(\sigma)}}{\mathcal{Z}^{\beta}_{\bm{J}, \lambda, \bm{h}}}     
\end{equation*}
where $\mathcal{Z}^{\beta}_{\bm{J}, \lambda, \bm{h}}$ is the partition function, a normalizing weight. Throughout this thesis, we will assume $\lambda\geq0$ for simplicity. The extension of our statements for $\lambda\leq 0$ will follow from spin-flip symmetry.

The semi-infinite Ising model was extensively studied by Fr{\"o}hlich and Pfister in \cite{FP-I, FP-II}, where they presented a wide range of results. Regarding the macroscopic behavior of the system, it was shown that, for $\beta\leq \beta_c$, the system behaves exactly as the Ising model and the spins align independently, where $\beta_c$ denotes the critical inverse temperature of the Ising model. Below the critical temperature, when there is no external field, there exists a critical value $\lambda_c>0$ that determines the behavior of the spins near the wall. When the influence of the wall is bigger than $\lambda_c$, the influence of the substrate "penetrates" the model and we see a thick layer of spins near the wall aligned with the substrate phase:

\begin{figure}[htbp]
    \centering
    \includegraphics[scale=0.4]{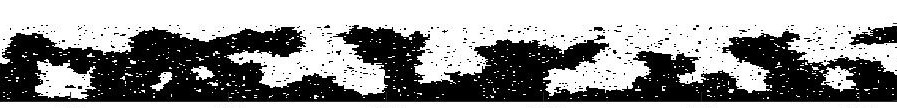}
    \caption{Complete wetting: $\lambda = 1 > \lambda_c$, $\beta=0.5$, $J=1$. The $+$ spins are black and the $-$ are white.}
\end{figure}

This regime is then called \textit{complete wetting}. When $0\leq \lambda<\lambda_c$, the influence of the substrate is only capable of creating disconnected clusters on the wall, so we say there is \textit{partial wetting}:

\begin{figure}[htbp]
    \centering
    \includegraphics[scale=0.5]{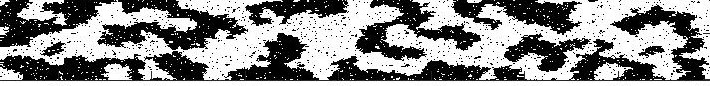}
    \caption{Partial wetting: $\lambda = 0.03 <\lambda_c$, $\beta=0.5$, $J=1$. The $+$ spins are black and the $-$ are white.}
\end{figure}

 In \cite{FP-II}, they also showed that this critical value is related to the existence of multiple Gibbs states, in the sense that, for $0\leq\lambda<\lambda_c$ there are multiple Gibbs states, and for $\lambda>\lambda_c$ we have uniqueness. The existence of this critical value $\lambda_c$ is proved using a notion of wall-free energy, defined formally as 
\begin{equation*}
    \tau_w(J,\lambda, \bm{h}) = \lim_{\Lambda\to\Z^d} \frac{1}{2|\mathcal{W}\cap\Lambda|}\ln\left[ \frac{(\mathcal{Z}^{-, J}_{\Lambda\cap\H+; \lambda,\bm{h}})^2}{Q^{-, J}_{\Lambda; \bm{h}}} \right] - \lim_{\Lambda\to\Z^d}\frac{1}{2|\mathcal{W}\cap \Lambda|}\ln\left[ \frac{(\mathcal{Z}^{+, J}_{\Lambda\cap \H+; \lambda,\bm{h}})^2}{Q^{+, J}_{\Lambda; \bm{h}}} \right],
\end{equation*}
where $Q^{\pm, J}_{\Lambda; \bm{h}}$ is the partition function of the Ising model (on $\Z^d$) with $\pm$ boundary condition on $\Lambda$, a finite set such that $\mathcal{W}\cap\Lambda \neq \emptyset$ and $\mathcal{Z}^{\pm, J}_{\Lambda\cap\H+; \lambda,\bm{h}}$ is the partition function of the semi-infinite Ising model with $\pm$ boundary condition on $\Lambda\cap\H+$. This quantity is suitable for measuring the wall influence when there is no external field since the partition functions of the Ising model cancel out, and the remaining terms can be written as an integral of the difference of the magnetization with respect to the wall influence $\lambda$, see Proposition \ref{Prop: Tilde_tau_as_integral}. 

\begin{figure}[H]
\centering
\includegraphics[scale=0.2]{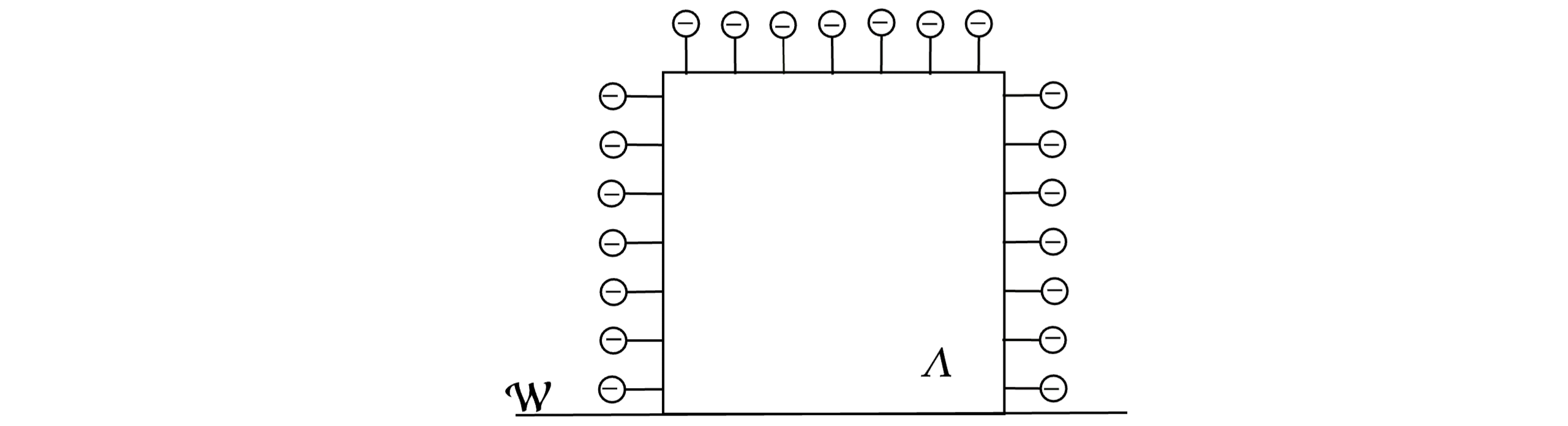}
\caption{The \textcolor{green}{-} boundary condition of the semi-infinite model.}
\end{figure}

In the Ising model, adding a nonnull constant external field disrupts the phase transition at every temperature, as a consequence of Lee-Yang theorem \cite{FV-Book, Lee-Yang.II.1952}. However, it was shown in \cite{Bissacot_Cioletti_10} that we can add an external field that decays as it goes to infinity and still preserves phase transition. This paper started a streak of new results on models with decaying fields \cite{Affonso.2021, Bissacot_Cass_Cio_Pres_15, Bissacot_Endo_Enter_2017, Bissacot.Endo.18}.

One particular result \cite{Bissacot_Cass_Cio_Pres_15}, states that we can consider an intermediate external field $\bm{h}^* = (h_i^*)_{i\in\Z^d}$ given by
\begin{equation*}
    h_i^* = \begin{cases}
            h^* &\text{ if }i=0,\\
            \frac{h^*}{|i|^\delta} &\text{ otherwise}.\\
            \end{cases}
\end{equation*}
that preserves phase transition for low temperatures when $\delta>1$ and induces uniqueness at low temperatures when $\delta<1$. In the critical value $\delta=1$, there is phase transition for $h^*$ small enough. The proof of uniqueness when $\delta<1$ was extended to all temperatures in \cite{Cioletti_Vila_2016}. The argument in \cite{Bissacot_Cass_Cio_Pres_15} involves contour arguments and Peierls' bounds techniques for low temperatures, while \cite{Cioletti_Vila_2016} uses a generalization of the Edwards-Sokal representation. Both techniques are fairly distinct and complement each other, which makes the complete proof of uniqueness involved. There is no standard strategy to prove uniqueness, with each model requiring particular techniques.

For the semi-infinite Ising model, a more natural choice of the external field is one decaying as it gets further from the wall, that is, $h_i \leq h_j$ whenever $j_d\leq i_d$. Given $h\in \mathbb{R}$, one such external field is
$\widehat{\bm{h}}=(h_i)_{i\in\H+}$ with 
\begin{equation*}
    h_i=\frac{h}{i_d^\delta}
\end{equation*}

for all $i\in\H+$.  Figures \ref{EF.on.Lambda} and \ref{graph.EF} shows how this external field behaves. 
   \begin{figure}[h]
            \centering
            \includegraphics[scale=0.12]{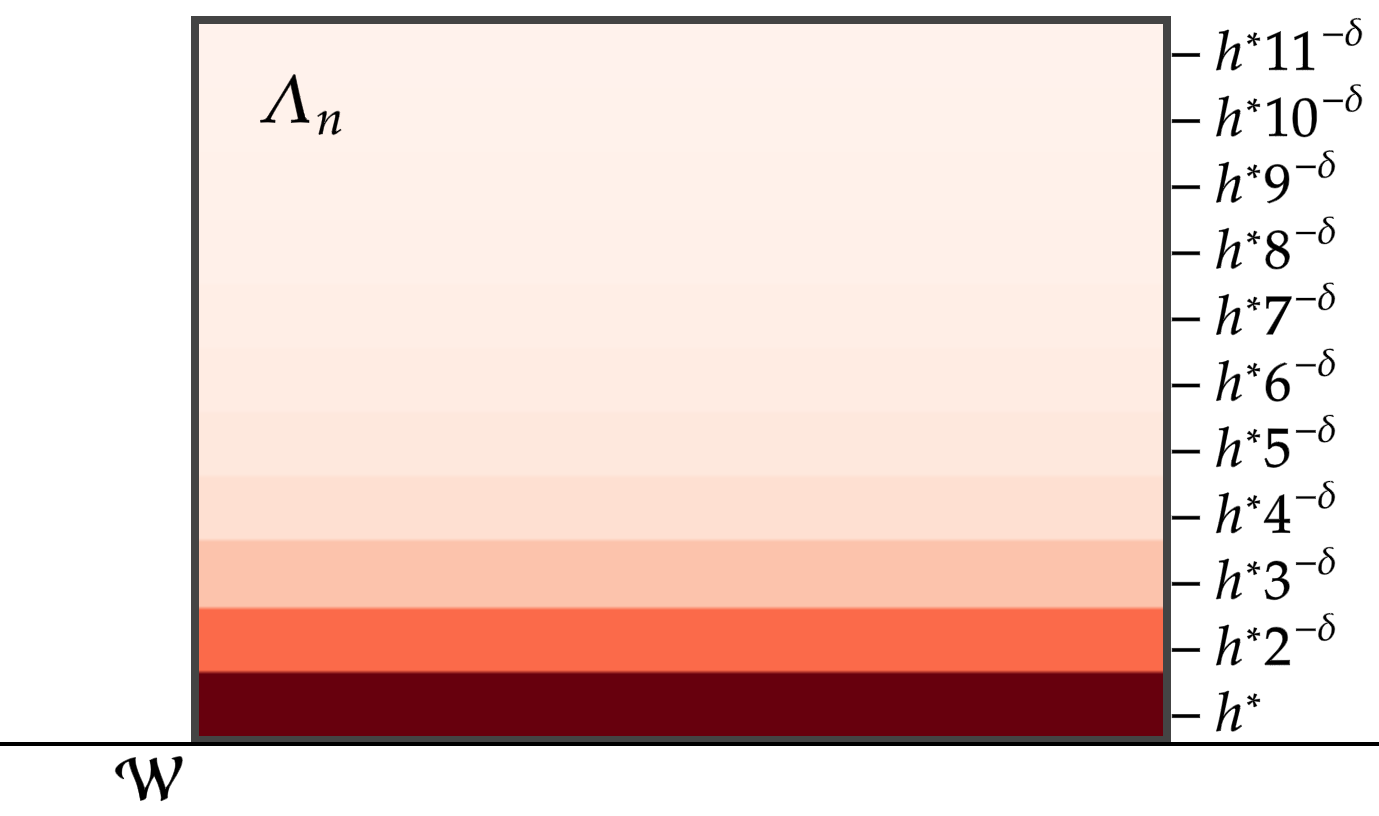}
            \caption{The influence of the external field in a box $\Lambda_n$.}
            \label{EF.on.Lambda}
    \end{figure}  
    \begin{figure}[h]
            \centering
            \includegraphics[scale=0.12]{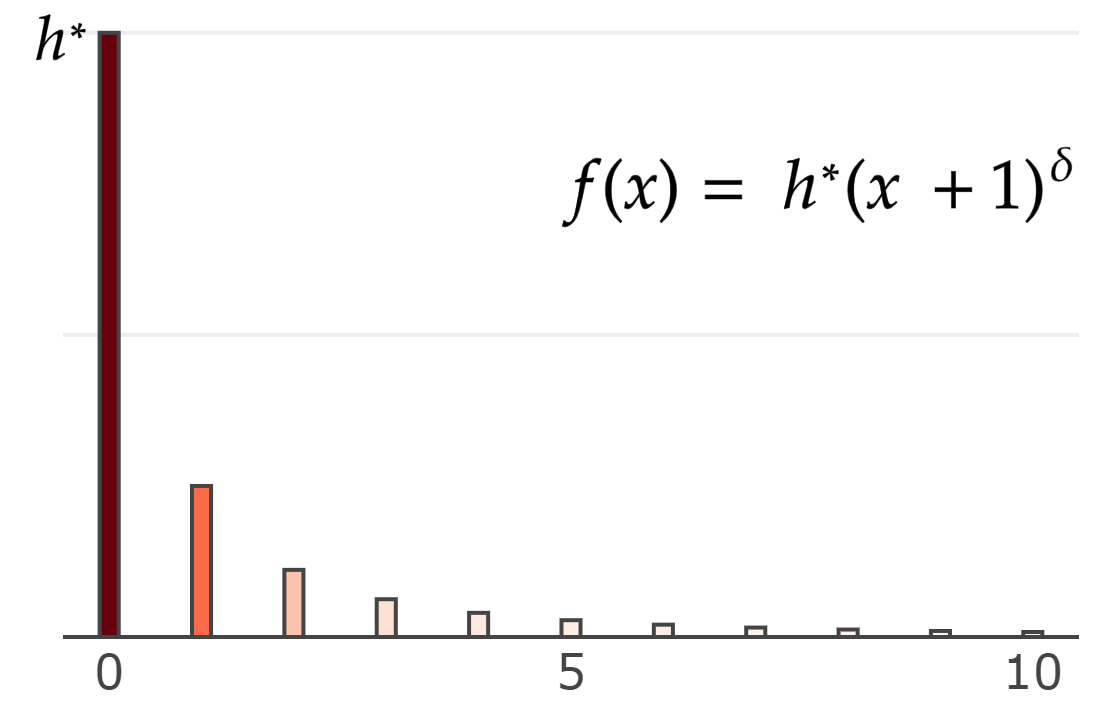}
            \caption{The external field w.r.t. the distance of a spin to the wall.}
            \label{graph.EF}
    \end{figure}
A particularly interesting choice of $h$ is $h=\lambda$, so we link the wall influence and the field. This particular case will be denoted  $\widehat{\bm{\lambda}}=(\lambda_i)_i\in\H+$, with 
\begin{equation*}
    \lambda_i=\frac{\lambda}{i_d^\delta}.
\end{equation*}
We can prove that, when $\delta>1$, the semi-infinite model with external field $\lambdadec$ behaves as the model with no field, so if we fix $\beta =1$, there exists a critical value $\overline{\lambda}_c(J)$ such that there are multiple Gibbs states when $0\leq\lambda<\overline{\lambda}_c(J)$, and we have uniqueness when $\overline{\lambda}_c(J)<\lambda$. At last, we show that when $\delta<1$, the semi-infinity Ising model with this choice of external field presents only one Gibbs state for any $J>0$. To simplify the notation, we choose to fix $\beta=1$ and let $J$ vary, as in the previous papers about the semi-infinite Ising model. Our results are summarized in the following theorem. 

\begin{theorem*}
    Let $d\geq 2$, and let $J_c$ be the critical value of the Ising model in $\Z^d$ at $\beta=1$. Given any $\delta>0$, there exists a critical value $\overline{\lambda}_c=\overline{\lambda}_c(J,\delta)\geq 0$ such that the semi-infinite Ising model with external field $\lambdadec$ presents phase transition for all $0\leq \lambda < \overline{\lambda}_c$ and uniqueness for $\overline{\lambda}_c <\lambda$. Moreover, for $\delta>1$ and $J>J_c$ we have $0<\overline{\lambda}_c$. When $\delta<1$, $\overline{\lambda}_c=0$ and there is uniqueness for all $J>0$. 
\end{theorem*}

In the second half of this thesis, we study the phase transition in the long-range random field Ising model at dimensions $d\geq 3$. The problem of the presence or absence of phase transition is central in statistical mechanics. To prove the existence of phase transition, the standard idea is to define a notion of contour and use \textit{Peierls' argument} \cite{Peierls.1936}. In the Ising model \cite{Ising_25}, particles of the system interact only with their nearest neighbors. On ferromagnetic long-range Ising models \cite{Anderson_Yuval_69}, there is interaction between each pair of spins in the lattice. The Hamiltonian of the model is given formally by
\begin{equation*}
    H(\sigma) = - \sum_{x,y\in \Z^d}J_{xy}\sigma_x\sigma_y,
\end{equation*}
where $J_{xy}=J|x-y|^{-\alpha}$, $J>0$, $\alpha > d$. It is well-known phase transition in dimension 2 for Ising models with nearest-neighbors implies phase transition for long-range interactions when $d\geq 2$, as a consequence of correlation inequalities. For the one-dimensional lattice, it is known that short-range models do not present phase transition \cite{georgii.gibbs.measures}. In the long-range case, a different behavior was conjectured depending on the exponent $\alpha$ (see \cite{Kac_Thompson_69}), but the problem was challenging.

In dimension $d=1$, phase transition was proved first in 1969 by Dyson \cite{Dyson.69}, for $\alpha \in (1,2)$, by proving phase transition in an auxiliary model, known nowadays as the \emph{Dyson model} or hierarchical model. Dyson's approach fails exactly on the critical exponent $\alpha=2$. It was already known that for $\alpha>2$ uniqueness holds \cite{georgii.gibbs.measures}. In 1982, Fr{\"o}hlich and Spencer \cite{Frohlich.Spencer.82} introduced a notion of one-dimensional contours and then applied Peierls' argument to show phase transition for the critical value $\alpha = 2$. These contours were inspired by the multiscale techniques previously introduced to study the Berezinskii-Kosterlitz-Thouless transition in two-dimensional continuous spin systems \cite{FS81}. Later, Cassandro, Ferrari, Merola and Presutti  \cite{Cassandro.05} extended the contour argument previously available for $\alpha=2$ to exponents $\alpha\in (3-\frac{\ln 3}{\ln 2}, 2]$, with the additional restriction that the nearest-neighbor interaction is strong, i.e.,  ${J(1)\gg 1}$; this restriction was removed for a subclass of interactions in \cite{Bissacot.Endo.18}. Further results were obtained using contour arguments, such as the decay of correlations, cluster expansions, phase transition with random interactions, etc; some references with these results are \cite{Cassandro.Merola.Picco.17, Cassandro.Merola.Picco.Rozikov.14, Imbrie.82, Imbrie.Newman.88, Johansson.91}. 

In the multidimensional setting ($d\geq 2$), Ginibre, Grossmann, and Ruelle, in \cite{Ginibre.Grossmann.Ruelle.66}, proved the phase transition for $\alpha > d+1$, using an enhanced version of Peierls' argument and the usual contours. Park used a different notion of contour for long-range systems in \cite{Park.88.I, Park.88.II}, extending the Pirogov-Sinai theory available for short-range interactions assuming $\alpha > 3d+1$, although he can also consider Potts models and models without symmetry with his methods. Some results in the literature suggest that truly long-range effects appear only when $d < \alpha \leq d+1$, see \cite{Biskup_Chayes_Kivelson_07}. Recently, Affonso, Bissacot, Endo, and Handa \cite{Affonso.2021}, inspired by the ideas from Fr{\"o}hlich and Spencer in \cite{FS81, Frohlich.Spencer.82}, introduced a version of multiscale multidimensional contour and proved phase transition by a contour argument in the whole region $\alpha > d$. They can consider long-range Ising models with deterministic decaying fields, first introduced in the context of nearest-neighbor interactions in \cite{Bissacot_Cioletti_10}. For such models, the lack of analyticity of the free energy does not imply phase transition since these models have the same free energy as the models with zero field. It is expected that slowly decaying fields imply uniqueness. In this setting, a contour argument is useful for proofs of phase transitions as well as for uniqueness, some papers with models with deterministic decaying fields are \cite{Aoun_Ott_Velenik_23, Bissacot_Cass_Cio_Pres_15, Bissacot.Endo.18, Cioletti_Vila_2016}.

The Random Field Ising model (RFIM) \cite{Imry.Ma.75} is the nearest-neighbor Ising model with an additional external field acting on each site $(h_x)_{x\in\Z^d}$ that is a family of i.i.d. Gaussian random variable with mean 0 and variance 1. Formally, the Hamiltonian of the model is given by
\begin{equation*}
    H(\sigma) = - \sum_{\substack{x,y\in \Z^d \\|x-y|=1}}J\sigma_x\sigma_y  - \varepsilon\sum_{x\in\Z^d}h_x\sigma_x,
\end{equation*}
where $J>0$, $\varepsilon>0$, and $d \geq 1$. A detailed account of the history of the phase transition problem for this model, as well as detailed proofs, was given in \cite{Bovier.06}. Here we present a brief overview.

During the 1980s, the question of the specific dimension where phase transition for the RFIM should happen attracted much attention and was a topic of heated debate. Two convincing arguments divided the physics community. One of them, due to Imry and Ma \cite{Imry.Ma.75}, was a non-rigorous application of the Peierls' argument together with the use of the isoperimetric inequality. The key idea of Peierls' argument is to define a notion of contour and calculate the energy cost of "erasing" each contour, i.e., the energy cost of flipping all spins inside the contour. When there is no external field, the energy necessary to flip the spins in a region $A\subset \Z^d$ is of the order of the boundary $|\partial A|$. When we add an external field, we get an extra cost depending on this field. Imry and Ma argued that this cost should be approximately $\sqrt{|A|}$. By the isoperimetric inequelity, $\sqrt{|A|}\leq |\partial A|^{\frac{d}{2(d-1)}}$, which is strictly smaller than $|\partial A|$ for all regions only when $d\geq 3$, so this should be the region where phase transition occurs. The other argument, due to Parisi and Sourlas \cite{Parisi.Sourlas.79}, based on dimensional reduction \cite{Aharony_Imry_Ma_76} and supersymmetry arguments, predicted that the $d$-dimensional RFIM would behave like the $d-2$-dimensional nearest-neighbor Ising model, therefore presenting phase transition only when $d\geq 4$. 

The question was settled by two celebrated papers showing that Imry and Ma's prediction was correct. First, in 1988, Bricmont and Kupiainen \cite{Bricmont.Kupiainen.88} showed that there is phase transition almost surely in $d\geq3$, for low temperatures and $\varepsilon$ small enough. Their proof uses a rigorous renormalization group analysis and it is considered involved. Still, they claimed that the result works for any model with a suitable contour representation and centered sub-gaussian external field. Later on, Aizenman and Wehr \cite{Aizenman.Wehr.90} proved uniqueness for $d\leq 2$. For detailed proofs of these results, we refer the reader to \cite{Bovier.06} (see also \cite{Berretti.85, Camia.18, Frohlich.Imbre.84,  Klein.Masooman.97} for more uniqueness results). 

Recently, Ding and Zhuang \cite{Ding2021}, provided a simpler proof of the phase transition, not using RGM. In addition, Ding, Liu, and Xia \cite{Ding.Liu.Xia.22} proved that if $\beta_c(d)$ is the critical inverse of the temperature of the Ising model with no field, for all $\beta>\beta_c(d)$ there exists a critical value $\varepsilon_0(d, \beta)$ such that the RFIM with $\varepsilon \leq \varepsilon_0$ presents phase transition. 

In the present paper, we are considering a long-range Ising model with a random field, whose Hamiltonian is given formally by
\begin{equation*}
    H(\sigma) = - \sum_{x,y\in \Z^d}J_{xy}\sigma_x\sigma_y - \varepsilon\sum_{x\in\Z^d}h_x\sigma_x,
\end{equation*}
where $J_{xy}=J|x-y|^{-\alpha}$, $J, \varepsilon>0$, $\alpha > d$, $d\geq 3$, and $(h_x)_{x\in\Z^d}$ that is a family of i.i.d. Gaussian random variable with mean 0 and variance 1. The only rigorous result on phase transition in the long-range setting is for the one-dimensional long-range Ising model with a random field, by Cassandro, Orlandi, and Picco \cite{Cassandro.Picco.09}. They used the contours of \cite{Cassandro.05} to show the phase transition for the model when $\alpha\in (3-\frac{\ln 3}{\ln 2}, \frac{3}{2})$, under the assumption $J(1) \gg 1$. We stress that, as remarked by Aizenman, Greenblatt, and Lebowitz \cite{Aizenman_Greenblatt_Lebowitz_2012}, although their argument does not work for the whole region of the exponent $\alpha$, the phase transition holds for values close to the critical value $\alpha=3/2$, since by the Aizenman-Wehr theorem we know that there is uniqueness for $\alpha\geq 3/2$.

The argument from Ding and Zhuang in \cite{Ding2021}, for $d\geq3$, involves controlling the probability of a bad event, which is related to controlling the quantity $$\sup_{\substack{0\in A\subset\Z^d \\ A \text{ connected }}}\frac{\sum_{x\in A}h_x}{|\partial A|},$$ known as the greedy animal lattice normalized by the boundary. The greedy animal lattice normalized by the size, instead of the boundary, was extensively studied for general distributions of $(h_x)_{x\in\Z^d}$, see \cite{Cox_Gandolfi_Griffin_Kesten_93, Gandolfi_Kesten_94, Hammond_06, Martin_02}. When we normalize by the boundary, an argument by Fisher, Fr\"{o}hlich and Spencer \cite{FFS84} shows that the expected value of the greedy animal lattice is finite. In dimension $d=2$, the expected value is not finite, see \cite{Ding.Wirth.20}. The supremum is taken over connected regions containing the origin since the interiors of the usual Peierls contours are of this form.

For the long-range model, the interior of contours is not necessarily connected. In fact, long-range contours may have considerably large diameters with respect to their size, so their interiors can be very sparse. Our definition of the contours is strongly inspired by the $(M,a,r)$-partition in \cite{Affonso.2021}, that are constructed using a multiscaled procedure that assures that the contours have no cluster with small density. With them, we generalize the arguments by Fisher-Fr\"{o}hlich-Spencer \cite{FFS84}, and prove that the expected value of the greedy animal lattice is finite, even considering regions not necessarily connected. Then, we prove the phase transition for $d\geq 3$. Our main result can be stated as
\begin{theorem*}Given $d\geq 3$, $\alpha>d$, there exists $\beta_c\coloneqq\beta(d, \alpha)$ and $\varepsilon_c\coloneqq\varepsilon(d, \alpha)$ such that, for $\beta> \beta_c$ and $\varepsilon\leq \varepsilon_c$, the extremal Gibbs measures $\mu_{\beta, \varepsilon}^+$ and $\mu_{\beta, \varepsilon}^-$ are distinct, that is, $\mu_{\beta, \varepsilon}^+ \neq \mu_{\beta, \varepsilon}^-$ $\mathbb{P}$-almost surely. Therefore the long-range random field Ising model presents phase transition.
\end{theorem*}

\chapter{Semi-infinite Ising Model}

\fancyhead[RE,LO]{\thesection}

In this chapter, we study the semi-infinite Ising model following \cite{FP-I, FP-II} closely. The model represents the usual Ising model but now with a constraining wall that absorbs some state, that is, it has a preference for aligning with the $+$ or $-$ spin. We first prove phase transition when there is no external field. Then we show that the macroscopic picture described in \cite{FP-II} does not change even if we add a non-homogeneous external field, that depends only on the distance to the wall and is small in a sense.

\section{The model and important definitions}   To represent a wall in the usual Ising model we change the graph where the model takes place. Instead of $\mathbb{Z}^d$, we consider the graph
$$\mathbb{H}^d_+ \coloneqq \{ i=(i_1, i_2, \dots, i_d)\in \mathbb{Z}^d : i_d\geq 0 \},$$
and the wall is denoted by $\mathcal{W}\coloneqq \{ i\in \mathbb{H}^d_+ : i_d = 0\}$. To represent the influence of the wall in its neighbors, we introduce a parameter $\lambda$ and, for the finite sets $\Lambda \Subset \mathbb{H}^d_+$, we define the local Hamiltonian

\begin{equation}\label{SI.Hamiltonian}
    \mathcal{H}_{\Lambda; \lambda, \bm{h}}^{\bm{J}}(\sigma) \coloneqq -\sum_{\substack{i \sim j \\ \{i,j \} \cap \Lambda \neq \emptyset}} J_{i,j}\sigma_i\sigma_j - \sum_{i\in\Lambda}h_i\sigma_i - \sum_{i\in \Lambda\cap \mathcal{W}} \lambda\sigma_i.
\end{equation}
Here, ${\bm{h}=(h_i)_{i\in \mathbb{H}^d_+}}$ is the \textit{external field} and the interaction ${\bm{J}=(J_{i,j})_{i,j\in\mathbb{H}^d_+}}$ is non-negative for all $i,j\in\mathbb{Z}^d$, so we say the model is \textit{ferromagnetic}. The configurations in $\Lambda\Subset\mathbb{Z}^d$ with boundary condition $\eta\in\Omega$ are the elements of $\Omega_\Lambda^\eta \coloneqq \{ \omega\in \Omega : \omega_i=\eta_i \text{ for all } i \notin \Lambda \}$. The \textit{finite Gibbs measure in $\Lambda$ with $\eta$-boundary condition} is  
\begin{equation}\label{eq:def.local.gibbs.measure}
    \mu_{\Lambda; \lambda, \bm{h}}^\eta (\sigma) \coloneqq \mathbbm{1}_{\{  \sigma\in\Omega_{\Lambda}^\eta\}}\frac{e^{-\beta  \mathcal{H}_{\Lambda; \lambda, \bm{h}}^{\bm{J}}(\sigma)}}{\mathcal{Z}^{\eta, \bm{J}}_{\Lambda; \lambda,\bm{h}}},
\end{equation}
 where $\mathcal{Z}^{\eta, \bm{J}}_{\Lambda; \lambda,\bm{h}}\coloneqq \sum_{\sigma \in\Omega_\Lambda^\eta}e^{-\beta  H_{\Lambda; \bm{h}}^{\eta}(\sigma)}$ is the usual \textit{partition function}. This measure is in the $\sigma$-algebra generated by the cylinder set, which coincides with the Borel $\sigma$-algebra when we consider the product topology on $\Omega$, a compact space. Then, the set of probability measures defined over the Borel sets is a weak* compact set. 
 
To construct the infinite measures we consider sequences of finite subsets $(\Lambda_n)_{n\in\mathbb{N}}$  such that, for any subset $\Lambda\subset\H+$, there exists $N=N(\Lambda)>0$ such that $\Lambda\subset\Lambda_n$ for every $n>N$. We say such sequences \textit{invades} $\mathbb{H}^d_+$ and we denote it by $\Lambda_n\nearrow\H+$. A particularly important sequence that invades $\H+$ is the finite boxes 
$$\Lambda_{n, m} \coloneqq \{i\in \mathbb{H}^d_+ : i_d\leq m, -n \leq i_k \leq n \text{ for } k=1,\dots,d-1 \},$$ with $n,m\geq 0$. We define also $\mathcal{W}_{n,m} \coloneqq \mathcal{W}\cap \Lambda_{n, m}$ the restriction of the wall for this boxes. The set of \textit{Gibbs measures} is the closed convex hull of all the weak* limits obtained by sequences invading $\H+$:
\begin{equation}
\mathcal{G}_{\bm{J}} \coloneqq \overline{\text{conv}}\{\mu_{\bm{J}}: \mu_{\bm{J}} = w^*\text{-}\lim_{\Lambda^\prime \nearrow \H+}\mu_{\Lambda^\prime; \lambda, \bm{h}}^{\omega}\}.
\end{equation}
To simplify the notation, we are omitting the dependency on $\beta$, $\bm{h}$ and $\lambda$ in the definition of $\mathcal{G}_{\bm{J}}$. Moreover, throughout this chapter and the next, we fix $\beta=1$ and let $J$ vary, to simplify the notation. Most of this chapter will be dedicated to investigating whether $|\mathcal{G}_{\bm{J}}|>1$, hence there is a phase transition or $|\mathcal{G}_{\bm{J}}|=1$, so there is uniqueness. 

Replacing the semi-infinite lattice $\H+$ by the whole lattice $\Z^d$ in all definitions above, we can define the set of \textit{Ising Gibbs measures} with ferromagnetic interaction $\bm{J}=(J_{i,j})_{i,j\in \Z^d}$ and external field $\bm{h}=(h_i)_{i\in\Z^d}$ as
\begin{equation*}
    \mathcal{G}_{\bm{J}}^{IS} \coloneqq \overline{\text{conv}}\{\mu_{\bm{J}}: \mu_{\bm{J}} = w^*\text{-}\lim_{\Lambda^\prime \nearrow \Z^d}\mu_{\Lambda^\prime; 0, \bm{h}}^{\omega}\}.
\end{equation*}
As an alternative for working with Gibbs measures, we can use the Gibbs states. These are linear functional defined on the space of local functions. A function $f:\Omega\mapsto\mathbb{R}$ is said \textit{local} when there exists $\Lambda\Subset\H+$ such that, $\forall\sigma,\omega \in\Omega$, $f(\sigma) = f(\omega)$ whenever $\sigma_i = \omega_i$ for all $i\in\Lambda^c$. The smallest such $\Lambda$, with respect to the inclusion, is called the \textit{support} of $f$. Moreover, $f$ is \textit{quasilocal} if there is a sequence $(g_n)_{n}$ of local functions such that $\lim_{n\to\infty}\| g_n - f \|_\infty = 0$. To each finite Gibbs measure, we associate the \textit{local Gibbs state in $\Lambda$ with $\eta$-boundary condition}, defined by 

\begin{equation*}
    \langle f\rangle^{\eta}_{\Lambda, \lambda, \bm{h}} \coloneqq \sum_{\sigma\in\Omega} f(\sigma)\mu_{\Lambda; \lambda, \bm{h}}^\eta (\sigma).
\end{equation*}
Moreover, the Gibbs state with $\eta$ boundary condition is
\begin{equation*}
     \langle f\rangle^{\eta}_{ \lambda, \bm{h}} \coloneqq \lim_{\Lambda^\prime \nearrow \H+} \langle f\rangle^{\eta}_{\Lambda^\prime; \lambda, \bm{h}},
\end{equation*}
when the limit exists. The semi-infinite model inherits several properties from the usual Ising model in $\mathbb{Z}^d$ once the Hamiltonian (\ref{SI.Hamiltonian}) can be seen as a particular case of the Ising model. The usual Ising Hamiltonian in $\Lambda\Subset \mathbb{Z}^d$ is 
\begin{equation}\label{Hamiltonian.Ising}
    \mathcal{H}_{\Lambda; \bm{h}}^{\bm{J}}(\sigma) = -\sum_{\substack{i \sim j \\ \{i,j \} \cap \Lambda \neq \emptyset}} J_{i,j}\sigma_i\sigma_j - \sum_{i\in\Lambda}h_i\sigma_i
\end{equation}
where ${\bm{J}=(J_{i,j})_{i,j\in\mathbb{Z}^d}}$ a family of non-negative real number and  ${\bm{h}=(h_i)_{i\in \mathbb{Z}^d}}$ is the external field with $h_i\in\mathbb{R}$ for all $i\in\mathbb{Z}^d$.
The \textit{Ising local state with boundary condition $\eta\in\Omega$} is, for any local function $f$, 
\begin{equation*}
    \langle f \rangle_{\Lambda; \bm{h}}^\eta \coloneqq (Q^{\eta, \bm{J}}_{\Lambda; \bm{h}})^{-1} \sum_{\sigma\in\Omega_\Lambda^\eta} f(\sigma)e^{-\mathcal{H}_{\Lambda; \bm{h}}^{\bm{J}}(\sigma)}, 
\end{equation*}
where $Q^{\eta, \bm{J}}_{\Lambda; \bm{h}} \coloneqq \sum_{\sigma\in\Omega_\Lambda^\eta} e^{-\mathcal{H}_{\Lambda; \bm{h}}^{\bm{J}}(\sigma)}$ is the usual partition function. So, for $\Lambda\Subset\mathbb{H}^d_+$, the state $\langle f\rangle^{\eta}_{\Lambda; \lambda, \bm{h}}$ is the Ising state with interaction $(J^\lambda_{i,j})_{i,j\in\mathbb{Z}^d}$ given by
\begin{equation}\label{J.for.the.semi.infinite}
J^\lambda_{i,j}=\begin{cases} J_{i,j} &\text{ if } \{i,j\}\subset\mathbb{H}^d_+, \\ \lambda &\text{ otherwise,} \end{cases}
\end{equation}
 and boundary condition $\eta^+$ given by, for all $i\in\mathbb{Z}^d$,
 \begin{equation*}\label{eta^+}
     \eta^+_i = \begin{cases} \eta &\text{ if } i\in\mathbb{H}^d_+, \\ +1 &\text{ otherwise.} \end{cases}
 \end{equation*}

We are interested in phase transition results for a uniform interaction $\bm{J} \equiv J>0$, $\lambda>0$ and $h_i\geq 0$ for all $i\in\H+$. For the Ising model with no external field, the existence of two distinct states translates to the fact that, on a macroscopic scale, the spins will align in the same direction. If we consider the plus state, when we look at huge boxes, we will see a sea of pluses and some rare occurrences of minus.
        
        In the semi-infinite Ising model, again with no field, the macroscopic consequence of the phase transition is different, it has to do with the existence of a layered phase separating the wall from the bulk. Writing the semi-infinite model interaction as in (\ref{J.for.the.semi.infinite}), the minus state in a box is given by the boundary condition like 
       in  Figure \ref{minus.b.c},
\begin{figure}[ht]
\centering
\includegraphics[scale=0.5]{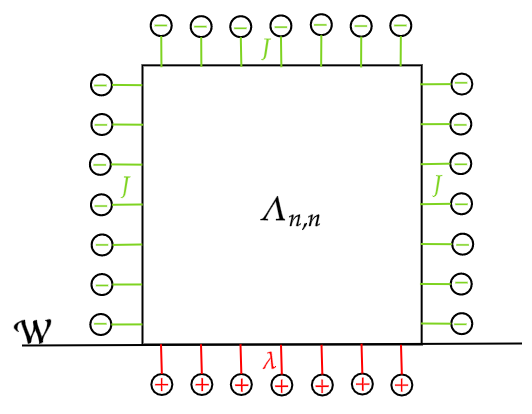}
\caption{The $\textcolor{red}{+}$ spins in the wall compete with the $\textcolor{green}{-}$ on the boundary.}
\label{minus.b.c}
\end{figure}
        so if $\lambda$ is big enough we have the phenomenon of complete wetting, where the wall forces to spin to align in the plus direction.

\begin{sloppypar}
To better understand this type of phenomenon, we recommend the survey paper \cite{IV18}. The surface free energy of the wall is a quantity that tries to identify whether or not we have complete wetting and is defined as follows. Consider the sequence that invades $\H+$ given by $\Lambda_n \coloneqq \Lambda_{n,n^\alpha}$ for some $0<\alpha<1$. Take  $\Lambda^\prime_n$ as the reflection of $\Lambda_n$ with respect to the line ${\mathcal{L} \coloneqq \{ (i_1,\dots,i_d)\in \mathbb{Z}^d : i_d = -\frac{1}{2}\}}$. Similarly, define the reflection of the walls $\mathcal{W}_n$ as $\mathcal{W}^\prime_n \coloneqq [-n,n]\times\{-1\}$. Denoting $\Delta_n \coloneqq \Lambda_n \cup \Lambda^\prime_n$ and extending $\bm{h}$ to $\mathbb{Z}^d$ by choosing $h_i = h_{(i_1,\dots, -i_d -1)}$ (the reflection of $i$ through line $\mathcal{L}$), the partition function of the usual Ising model in $\Delta_n$ with $\eta$-boundary condition is
\end{sloppypar}

\begin{equation*}
    Q^{\eta, J}_{\Delta_n; \bm{h}} = \sum_{\sigma\in \Omega^{\eta}_{\Delta_n}}\exp{ \{ \sum_{\substack{i \sim j \\ \{i,j \} \cap \Delta_n \neq\emptyset}} J\sigma_i\sigma_j + \sum_{i\in\Delta_n} h_i\sigma_i \} }
\end{equation*}

and the \textit{free surface energy} for the $+$  b.c. and $-$ b.c. are, respectively,

\begin{equation}\label{Def.F+}
    F^{+}( J, \lambda, \bm{h}) \coloneqq \lim_{n\to \infty} -\frac{1}{2|\mathcal{W}_n|}\ln\left[ \frac{(\mathcal{Z}^{+, J}_{n; \lambda,\bm{h}})^2}{Q^{+, J}_{\Delta_n; \bm{h}}} \right] 
\end{equation}
and
\begin{equation}\label{Def.F-}
    F^{-}( J, \lambda, \bm{h}) \coloneqq \lim_{n\to \infty} -\frac{1}{2|\mathcal{W}_n|}\ln\left[ \frac{(\mathcal{Z}^{-, J}_{n; \lambda,\bm{h}})^2}{Q^{-, J}_{\Delta_n; \bm{h}}} \right].\\
\end{equation}
Here, $\mathcal{Z}^{+, J}_{n; \lambda,\bm{h}}$ (respectively $\mathcal{Z}^{-, J}_{n; \lambda,\bm{h}}$) is the partition function of the semi-infinite model in the box $\Lambda_n$ with plus boundary condition. (respectively minus boundary condition).
\\

All results in this chapter follows \cite{FP-I, FP-II} closely, with minor generalizations. We first prove that these limits exists.  As we do not fix a constant external field, this shows that the proof in \cite{FP-I} extends trivially for space dependent external fields. After that, we will study the wetting transition when $\bm{h}=0$. Using the \textit{surface free energy} between the $+$ and $-$ b.c., defined as 
\begin{equation*}
\tau_w(J, \lambda, \bm{h}) \coloneqq F^{-}(J, \lambda, \bm{h}) - F^{+}(J, \lambda, \bm{h})
\end{equation*}
we characterize the presence or absence of phase transition. Moreover, we will compare this quantity to the \textit{interface free energy for the Ising model}, defined as
\begin{equation}\label{tau(J)}
    \tau(J) = \lim_{n,m\to\infty} -\frac{1}{|\mathcal{W}_n|} \ln\left[ \frac{Q^{\mp, J}_{\Delta_{m,n}; 0}}{Q^{+, J}_{\Delta_{m,n}; 0}} \right],
\end{equation}
where the $\mp$-boundary condition denotes the configuration $\sigma\in\Omega$ defined by 
\begin{equation*}
    \sigma_i=\begin{cases}
                -1, & \text{ if } i\in \mathbb{H}^d_+, \\
                +1 & \text{ if } i\in \mathbb{H}_d^-
                \end{cases}
\end{equation*}
and 
\begin{equation*}
    \Delta_{m,n} = \left[-m,m\right]^{d-1}\times\left[-n,n\right].
\end{equation*}
At least, we present some characterizations of the macroscopic picture in both the uniqueness and phase transition regime. In this last part, we consider the model with an external field that depend only on the distance to the wall. We are able to show that the characterization presented in \cite{FP-II} is preserved as long as the sum of the external field in a line perpendicular to the wall is small enough. 

\section{Correlation inequalities and the limiting states}  Before we introduce the correlation inequalities, we define the \textit{free boundary condition} and the state associated with it. This is an important state, for which all of the correlation inequalities apply to. For $\Lambda\Subset\mathbb{Z}^d$ and a configuration $\sigma\in\Omega$,
\begin{equation}\label{Hamiltonian.Ising.free.bc}
\mathcal{H}_{\Lambda; \bm{h}}^{\bm{J}, \mathrm{f}}(\sigma) \coloneqq -\sum_{\substack{i \sim j \\ \{i,j \} \subset \Lambda}} J_{i,j}\sigma_i\sigma_j - \sum_{i\in\Lambda}h_i\sigma_i,
\end{equation}
where, again, $\bm{h}=(h_i)_{i\in\mathbb{Z}^d}$ is a family of real numbers and $\bm{J}=(J_{i,j})_{i,j\in\mathbb{Z}^d}$ are all positive real numbers. The difference between this Hamiltonian and the usual one defined in (\ref{Hamiltonian.Ising}) is that there is no interaction between the interior and the exterior of $\Lambda$. The state defined by such Hamiltonian is 

\begin{equation}\label{state.free.bc}
\langle f \rangle_{\Lambda; \bm{h}}^{\mathrm{f}} \coloneqq (\mathcal{Z}_{\Lambda; \bm{h}}^{\mathrm{f}})^{-1}\sum_{\sigma\in\{-1,+1\}^\Lambda} f(\sigma)e^{\mathcal{H}_{\Lambda; \bm{h}}^{\bm{J}, \mathrm{f}}(\sigma)},
\end{equation}
where $f$ is any local function and $\mathcal{Z}_{\Lambda; \bm{h}}^{\mathrm{f}}$ is the usual partition function
\begin{align*}
    \mathcal{Z}_{\Lambda; \bm{h}}^{\mathrm{f}}\coloneqq \sum_{\sigma\in\{-1,+1\}^\Lambda} e^{\mathcal{H}_{\Lambda; \bm{h}}^{\bm{J},\mathrm{f}}(\sigma)}.
\end{align*}

The correlation inequalities we will need are the GKS, FKG and duplicated variables inequalities. They are used to prove the existence and some essential properties of the limiting states. A proof for the FKG and GKS inequalities can be found in \cite{FV-Book}. Since both inequalities are classical results, the proof of them will not be presented here. The duplicated variable inequalities were proven in \cite{lebowitz1974ghs}. One important observation is that all of these inequalities are stated for the finite-volume Ising states in $\mathbb{Z}^d$, but are easily translated to the semi-infinite states due to the particularization (\ref{J.for.the.semi.infinite}). We start with the GKS inequalities:

\begin{proposition}[GKS inequalities]Let $\bm{J} = (J_{i,j})_{i,j\in\mathbb{Z}^d}$ and $\bm{h}=(h_i)_{i\in\mathbb{H}_d^+}$ be two collections of non-negative real numbers and $\Lambda \Subset \mathcal{Z}^d$. Then for any $A,B\Subset \Lambda$ we have
\begin{align}
    \langle \sigma_A \rangle^{+}_{\Lambda;\bm{h}} &\geq 0\\
    \langle \sigma_A\sigma_B \rangle^{+}_{\Lambda;\bm{h}} &\geq \langle \sigma_A \rangle^{+}_{\Lambda;\bm{h}}\langle \sigma_B \rangle^{+}_{\Lambda;\bm{h}}.
\end{align}
Both inequalities also hold for the free boundary condition. 
\end{proposition}

This inequality, named after Griffiths, Kelly, and Sherman, was first proved in \cite{G, Kelly_Sherman_68}. The next inequality is the FKG, one of the most important in Statistical Mechanics and it is related to the notion of non-decreasing function. Given two configurations $\omega,\sigma\in \Omega$, we write $\sigma\leq \omega$ if $\sigma_i\leq \omega_i$ for all $i\in\mathbb{H}_d^+$. We say that a local function $f$ is \textit{non-decreasing} if, for all $\sigma\leq \omega$, $f(\sigma)\leq f(\omega).$ The FKG inequality, named after Fortuin–Kasteleyn–Ginibre \cite{FKG}, is  

\begin{theorem}[FKG Inequality]
Let $\bm{J} = (J_{i,j})_{i,j\in\mathbb{Z}^d}$ be a collection of non-negative real numbers and $\bm{h}=(h_i)_{i\in\mathbb{H}_d^+}$ be a collection of arbitrary real numbers. Then for any $\Lambda\Subset \mathbb{Z}^d$ and any non-decreasing functions $f$ and $g$ we have
\begin{equation}\label{Eq: FKG_Inequality_Spins}
    \langle fg \rangle^{\eta}_{\Lambda;\bm{h}}\geq \langle f \rangle^{\eta}_{\Lambda; \bm{h}}\langle g \rangle^{\eta}_{\Lambda; \bm{h}}
\end{equation}
for an arbitrary boundary condition $\eta$, including the free boundary condition.
\end{theorem}

Now we use some consequences of this inequality to characterize the limiting states and get equivalences of phase transition. The first one is:
\begin{lemma}\label{extremality.of.+.and.-.bc}
Let $f$ be a non-decreasing function, $\beta\geq0$, $\bm{J}=(J_{x,y})_{i,j\in\mathbb{Z}^d}$ a family of positive real numbers and $\bm{h}=(h_i)_{i\in\mathbb{H}_d^+}$ be any external field. Then, for any boundary condition $\eta\in\Omega$ and $\Lambda\Subset\mathbb{Z}^d$,
\begin{equation*}
   \langle f \rangle^{-}_{\Lambda; \bm{h}} \leq \langle f \rangle^{\eta}_{\Lambda; \bm{h}} \leq \langle f \rangle^{+}_{\Lambda; \bm{h}}.
\end{equation*}

Moreover, if $\omega\in\Omega$ is such that $\eta\leq\omega$, then
\begin{equation*}
    \langle f \rangle^{\eta}_{\Lambda; \bm{h}}\leq \langle f \rangle^{\omega}_{\Lambda; \bm{h}}.
\end{equation*}

If $f$ is also local satisfying $supp(f)\subset\Lambda$, then
\begin{equation*}
     \langle f \rangle^{-}_{\Lambda; \bm{h}} \leq \langle f \rangle^{\mathrm{f}}_{\Lambda; \bm{h}} \leq \langle f \rangle^{+}_{\Lambda; \bm{h}}.
\end{equation*}
\end{lemma}

Another important consequence of FKG is that it allow us to define precisely the limiting states.

 \begin{lemma}\label{decreasing.states} For a non-decreasing local function $f$, $\beta>0$, $\bm{J}=(J_{i,j})_{i,j\in\mathbb{Z}^d}$ a family of positive real numbers and $\bm{h}=(h_i)_{i\in\mathbb{H}_d^+}$ any external field, if $\Lambda_1\subset\Lambda_2\Subset \mathbb{Z}^d$ we have
\begin{equation}
    \langle f \rangle_{\Lambda_2; \bm{h}}^{+} \leq \langle f \rangle_{\Lambda_1; \bm{h}}^{+}.
\end{equation}

The same inequality holds for the $-$b.c when $f$ is non-increasing.
\end{lemma}
Similarly, we can use the GKS inequalities to prove:

\begin{lemma}\label{correlation.behaviour.wrt.Lambda}Let $\bm{J}$ and $\bm{h}$ be as in the hypothesis of the GKS Inequalities, and $\Lambda_1\subset\Lambda_2\Subset\mathbb{Z}^d$. Then, for any $A\subset\Lambda_1$ 
\begin{equation}
    \langle \sigma_A \rangle_{\Lambda_2; \bm{h}}^{+} \leq  \langle \sigma_A \rangle_{\Lambda_1; \bm{h}}^{+}
\end{equation}
and 
\begin{equation}
    \langle \sigma_A \rangle_{\Lambda_2; \bm{h}}^{\mathrm{f}} \geq  \langle \sigma_A \rangle_{\Lambda_1; \bm{h}}^{\mathrm{f}}.
\end{equation}

\end{lemma}\label{decreasing.states.GKS}
Both these lemmas are important to prove some fundamental properties of the limit states in both the usual and the semi-infinite Ising models, including its existence. Below we enunciate some of these properties for the semi-infinite model.

\begin{proposition}\label{properties.from.fkg} Let $\beta\geq0$, $\bm{J}=(J_{x,y})_{i,j\in\mathbb{Z}^d}$ a family of positive real numbers, $\bm{h}=(h_i)_{i\in\mathbb{H}_d^+}$ any external field and $\lambda\in \mathbb{R}$ be the wall influence. Then

\begin{enumerate}
    \item{Extremality:} The states $\langle \cdot \rangle_{\lambda, \bm{h}}^{+}$ is extremal, in the sense that it is not a convex combination of other states.  
    
    \item{Tranlation invariance:} For any local function $f$ and $a\in \mathcal{W}$ we have 
        \begin{equation*}
            \langle f\circ\Theta_a \rangle_{\lambda, \bm{h}}^{+} = \langle f \rangle_{\lambda, \bm{h}}^{+},
        \end{equation*}
        where $\Theta_a$ is the translation of configurations defined by
        \begin{align}
            (\Theta_a(\omega))_i = \omega_{i-a}, \hspace{1cm} \forall i\in \mathbb{H}_d^+.
        \end{align}
        
    \item{Short-range correlations:} Given two local functions $f$ and $g$ we have
        \begin{equation*}
            \lim_{\mid a \mid\to\infty} \langle (f\circ \Theta_a)g \rangle_{\lambda, \bm{h}}^{+}=\langle f \rangle_{\lambda, \bm{h}}^{+}\langle g \rangle_{\lambda, \bm{h}}^{+}.
        \end{equation*}
\end{enumerate}
Moreover, all of these statements are also true for the minus state.
\end{proposition}


At last, we discuss the so-called duplicate variable inequalities. This set of inequalities was proved in \cite{lebowitz1974ghs} as a consequence of the GKS and a more general form of the FKG inequalities. In a subset $\Lambda\Subset\mathbb{Z}^d$, we consider a Hamiltonian of two independent systems with free boundary condition
\begin{equation}\label{hamiltonian.two.states.free.bc}
    \mathcal{H}_{\Lambda; \bm{h}, \bm{h}^\prime}^{\bm{J}}(\sigma,\sigma^\prime) 
                \coloneqq \mathcal{H}_{\Lambda; \bm{h}}^{\bm{J}, \mathrm{f}}(\sigma) +  \mathcal{H}_{\Lambda; \bm{h^\prime}}^{\bm{J}, \mathrm{f}}(\sigma^\prime) \\[2ex]
                = -\sum_{\substack{i\sim j\\ \{ i,j \}\subset \Lambda}}  J_{i,j}\sigma_i\sigma_j - \sum_{\substack{i\sim j\\ \{ i,j \}\subset \Lambda}}  J_{i,j}\sigma_i^\prime\sigma_j^\prime - \sum_{i\in\Lambda} h_i\sigma_i - \sum_{i\in\Lambda} h_i^\prime\sigma_i^\prime.
\end{equation}

With this we can define the state
\begin{equation}
    \langle f \rangle_{\Lambda; \bm{h}, \bm{h}^\prime}^{\mathrm{f}, \mathrm{f}} \coloneqq (\mathcal{Z}^{\mathrm{f},\mathrm{f}}_{\Lambda;\bm{h},\bm{h^\prime}})^{-1}\sum_{\sigma,\sigma^\prime\in\{-1,+1\}^\Lambda} f(\sigma,\sigma^\prime)\exp{-\beta\mathcal{H}_{\Lambda; \bm{h}, \bm{h}^\prime}^{J}(\sigma,\sigma^\prime)},
\end{equation}
where $f$ is any local function and
\begin{equation*}
    \mathcal{Z}^{\mathrm{f},\mathrm{f}}_{\Lambda;\bm{h},\bm{h^\prime}} \coloneqq \sum_{\sigma,\sigma^\prime\in\{-1,+1\}^\Lambda} \exp{-\beta\mathcal{H}_{\Lambda; \bm{h}, \bm{h}^\prime}^{J}(\sigma,\sigma^\prime)}.
\end{equation*}

With this definition, we see that if $f(\sigma,\sigma^\prime)=f(\sigma)$ and $g(\sigma,\sigma^\prime)=g(\sigma^\prime)$, that is, $f$ depends only of the first variable and $g$ depends only on the second, then
\begin{align*}
    \langle f \rangle_{\Lambda; \bm{h}, \bm{h}^\prime} = \langle f \rangle_{\Lambda; \bm{h}}^\mathrm{f} \hspace{1cm}\text{and}\hspace{1cm}  \langle g \rangle_{\Lambda; \bm{h}, \bm{h}^\prime}=\langle g \rangle_{\Lambda; \bm{h}^\prime}^\mathrm{f},
\end{align*}
so we say that the marginal distributions of $\langle \cdot \rangle_{\Lambda; \bm{h}, \bm{h}^\prime}$ are $\langle \cdot \rangle_{\Lambda; \bm{h}}^\mathrm{f}$ and $\langle  \rangle_{\Lambda; \bm{h}^\prime}^\mathrm{f}$. 

Introducing the random variables $s_i=\sigma_i + \sigma^\prime_i$ and $t_i=\sigma_i - \sigma^\prime_i$, for all $i\in\mathbb{Z}^d$, as well as 
\begin{align*}
    s_A=\prod_{j\in A}s_j \hspace{1cm}\text{and}\hspace{1cm} t_A=\prod_{j\in A} t_j
\end{align*}
for all $A\Subset\mathbb{Z}^d$, the duplicate variable inequalities are:
\begin{theorem}[Duplicate variable inequalities] 
Let $\bm{J} = (J_{i,j})_{i,j\in\mathbb{Z}^d}$ be a collection of non-negative real numbers, $\bm{h}=(h_i)_{i\in\mathbb{H}_d^+}$ and $\bm{h^\prime}=(h_i)^\prime_{i\in\mathbb{H}_d^+}$ be two collection of arbitrary real numbers satisfying $h_i \pm h^\prime_i\geq 0$ for all $i\in\mathbb{Z}^d$. Then, for any $A,B\Subset \mathbb{Z}^d$, we have

\begin{align}
    & 0\leq \langle t_As_B \rangle_{\Lambda; \bm{h}, \bm{h}^\prime}^{\mathrm{f},\mathrm{f}}\leq \langle t_A \rangle_{\Lambda; \bm{h}, \bm{h}^\prime}^{\mathrm{f},\mathrm{f}}\langle s_B \rangle_{\Lambda; \bm{h}, \bm{h}^\prime}^{\mathrm{f},\mathrm{f}}, \label{DVI.1} \\ 
    &\langle t_At_B\rangle_{\Lambda; \bm{h}, \bm{h}^\prime}^{\mathrm{f},\mathrm{f}}\geq \langle t_A\rangle_{\Lambda; \bm{h}, \bm{h}^\prime}^{\mathrm{f},\mathrm{f}}\langle t_B \rangle_{\Lambda; \bm{h}, \bm{h}^\prime}^{\mathrm{f},\mathrm{f}}, \label{DVI.2}\\
    &\langle s_As_B \rangle_{\Lambda; \bm{h}, \bm{h}^\prime}^{\mathrm{f},\mathrm{f}} \geq \langle s_A \rangle_{\Lambda; \bm{h}, \bm{h}^\prime}^{\mathrm{f},\mathrm{f}}\langle s_B \rangle_{\Lambda; \bm{h}, \bm{h}^\prime}^{\mathrm{f},\mathrm{f}}. \label{DVI.3}\\ \notag
\end{align}

\end{theorem}

One important remark is that we can make a change in the external field so that the marginal distributions of $\langle \cdot \rangle_{\Lambda; \bm{h}, \bm{h}^\prime}^{\mathrm{f},\mathrm{f}}$ became $\langle \cdot \rangle_{\Lambda; \bm{h}}^\eta$ for some $\eta\in\Omega$. Indeed, defining an altered external field $\bm{h}^\eta$ as
\begin{equation*}
    h_i^\eta = h_i + \sum_{\substack{j\sim i\\ j\in\mathbb{Z}^d\setminus \Lambda}} J_{ij}\eta_j,
\end{equation*}
for any configuration $\eta\in\Omega$, we have
\begin{equation*}
    \langle f \rangle_{\Lambda; \bm{h}^\eta, \bm{h}^{\prime\omega}} = (\mathcal{Z}_{\Lambda; \bm{h},\bm{h}^\prime}^{\eta,\omega})^{-1} \sum_{\substack{\sigma\in\Omega_\Lambda^\eta \\ \sigma^\prime\in\Omega_\Lambda^\omega}} f(\sigma,\sigma^\prime)e^{\mathcal{H}_{\Lambda;\bm{h}}^{\bm{J}}(\sigma) + \mathcal{H}_{\Lambda;\bm{h}^\prime}^{\bm{J}}(\sigma^\prime)}.
    \end{equation*}
With this, it is straightforward that, for functions $f(\sigma,\sigma^\prime)=f(\sigma)$, ${g(\sigma,\sigma^\prime) = g(\sigma^\prime)}$ and $\eta,\omega\in\Omega\cup\{\mathrm{f}\}$,  
  \begin{align*}
    \langle f \rangle_{\Lambda; \bm{h}^\eta, \bm{h}^{\prime\omega}} = \langle f \rangle_{\Lambda; \bm{h}}^\mathrm{\eta} \hspace{1cm}\text{and}\hspace{1cm}  \langle g \rangle_{\Lambda; \bm{h}^\eta, \bm{h}^{\prime\omega}}=\langle g \rangle_{\Lambda; \bm{h}^\prime}^\mathrm{\omega}.
\end{align*}
For consistency, we are defining $\bm{h}^\mathrm{f} = \bm{h}$. To stress this properties we define the state, for $\eta,\omega\in\Omega\cup\{\mathrm{f}\}$ and $\bm{h} = \bm{h}^\prime$,
\begin{align}
    \langle \cdot \rangle_{\Lambda;\bm{h},\bm{h}}^{\eta,\omega}\coloneqq \langle \cdot \rangle_{\Lambda; \bm{h}^\eta, \bm{h}^{\omega}},
\end{align}
and we have the following:
\begin{corollary}\label{DVI.for.+.bc}
Let $\bm{J} = (J_{i,j})_{i,j\in\mathbb{Z}^d}$ and $\bm{h}=(h_i)_{i\in\mathbb{Z}^d}$ be a collection of non-negative real numbers.  Then, for any $A,B\Subset \mathbb{Z}^d$ and any $\eta\in\Omega\cup\{\mathrm{f}\}$, we have
\begin{align}
    & 0\leq \langle t_As_B \rangle_{\Lambda; \bm{h}}^{+,\eta}\leq \langle t_A \rangle_{\Lambda; \bm{h}}^{+,\eta}\langle s_B \rangle_{\Lambda; \bm{h}}^{+,\eta}, \label{DVI.1.+} \\ 
    &\langle t_At_B\rangle_{\Lambda; \bm{h}}^{+,\eta}\geq \langle t_A\rangle_{\Lambda; \bm{h}}^{+,\eta}\langle t_B \rangle_{\Lambda; \bm{h}}^{+,\eta}, \label{DVI.2.+}\\
    &\langle s_As_B \rangle_{\Lambda; \bm{h}}^{+,\eta} \geq \langle s_A \rangle_{\Lambda; \bm{h}}^{+,\eta}\langle s_B \rangle_{\Lambda; \bm{h}}^{+,\eta}. \label{DVI.3.+}\\ \notag
\end{align}
\end{corollary}
The most important consequence of the duplicated variables inequalities is
\begin{proposition}\label{basta.comparar.magnetizacoes}
Let $\bm{J}=(J_{ij})_{i,j\in\mathbb{Z}^d}$ be a non-negative interaction satisfying $J_{ij}>0$ if $|i-j|=1$, $\bm{h}=(h_i)_{i\in\mathbb{Z}^d}$ be a non-negative external field and $\beta>0$. Then, if $\langle \sigma_i \rangle^{-}_{\beta, \bm{h}} = \langle \sigma_i \rangle^{+}_{\beta, \bm{h}}$ for some $i\in\mathbb{Z}^d$, there is a unique Gibbs state.  
\end{proposition}

\begin{proof}
Fix $i,j\in\mathbb{Z}^d$ with $|i-j|=1$, and $\Lambda\Subset\mathbb{Z}^d$ containing $i$ and $j$. For the duplicated variable system with plus and minus boundary conditions, we will prove that 
\begin{equation}
    \langle t_it_j\rangle_{\Lambda; \bm{h}}^{+,-} \geq \langle t_it_j\rangle_{\{i,j\}; \bm{h}}^{+,+}>0.
\end{equation}
Start by rewriting the Hamiltonian of this duplicated system as 
\begin{align*}
    \mathcal{H}_{\Lambda; \bm{h}}^{+,-}(\sigma,\sigma^\prime) &= \sum_{\{i,j\}\subset\Lambda}J_{ij}(\sigma_i\sigma_j + \sigma_i^\prime\sigma_j^\prime) + \sum_{i\in\Lambda}h_i(\sigma_i+\sigma^\prime_i) + \sum_{i\in\Lambda,j\in\Lambda^c}J_{ij}(\sigma_i-\sigma^\prime_i).
\end{align*}
For any $u\in\Lambda$, we can differentiate w.r.t. $h_u$,
\begin{equation*}
   \frac{d}{dh_u}(\langle t_it_j\rangle_{\Lambda; \bm{h}}^{+,-}) = \beta \langle t_it_js_u\rangle_{\Lambda; \bm{h}}^{+,-} - \langle t_it_j\rangle_{\Lambda; \bm{h}}^{+,-}\langle s_u\rangle_{\Lambda; \bm{h}}^{+,-} 
\end{equation*}
which is negative by (\ref{DVI.1.+}), hence $\langle t_it_j\rangle_{\Lambda; \bm{h}}^{+,-}$ is decreasing in $h_k$ for all $k\in\Lambda$. Given $\mu\geq 0$, define the external field $\bm{h}+\bm{\mu}_{i,j}$ as $(\bm{h}+\bm{\mu}_{i,j})_k\coloneqq h_i+\mu\mathbbm{1}_{k\in\Lambda\setminus\{i,j\}}$, for all $k\in\Lambda$. Then,
\begin{equation}
    \langle t_it_j\rangle_{\Lambda; \bm{h}}^{+,-}\geq \langle t_it_j\rangle_{\Lambda; \bm{h} + \bm{\mu}_{i,j}}^{+,-}.
\end{equation}
To take the limit as $\mu$ goes to infinity, write the Hamiltonian as
\begin{align*}
    \langle t_it_j\rangle_{\Lambda; \bm{h}+ \bm{\mu}_{i,j}}^{+,-} &= \frac{\sum_{\sigma,\sigma^\prime\in\Omega_\Lambda}t_it_j\exp\left\{\beta \mathcal{H}_{\Lambda; \bm{h}, \bm{h}^\prime}^{+,-}(\sigma,\sigma^\prime) + \beta\sum_{u\in\Lambda\setminus\{i,j\}}\mu(\sigma_u + \sigma_u^\prime)\right\}}{\sum_{\sigma,\sigma^\prime\in\Omega_\Lambda}\exp\left\{\beta \mathcal{H}_{\Lambda; \bm{h}, \bm{h}^\prime}^{+,-}(\sigma,\sigma^\prime) + \beta\sum_{u\in\Lambda\setminus\{i,j\}}\mu(\sigma_u + \sigma_u^\prime)\right\}}\\
    &=\frac{\sum_{\sigma,\sigma^\prime\in\Omega_\Lambda}t_it_j\exp\left\{\beta \mathcal{H}_{\Lambda; \bm{h}, \bm{h}^\prime}^{+,-}(\sigma,\sigma^\prime) - \beta\sum_{u\in\Lambda\setminus\{i,j\}}\mu[2 - (\sigma_u + \sigma_u^\prime)]\right\}}{\sum_{\sigma,\sigma^\prime\in\Omega_\Lambda}\exp\left\{\beta \mathcal{H}_{\Lambda; \bm{h}, \bm{h}^\prime}^{+,-}(\sigma,\sigma^\prime) - \beta\sum_{u\in\Lambda\setminus\{i,j\}}\mu[2-(\sigma_u + \sigma_u^\prime)]\right\}}.
\end{align*}
As $e^{- \beta\mu[2-(\sigma_u + \sigma_u^\prime)]}$ does not converge to zero if and only if $\sigma_u=\sigma_u^\prime=1$, the limit is
\begin{equation*}
    \lim_{\mu\to+\infty}\langle t_it_j\rangle_{\Lambda; \bm{h}+ \bm{\mu}_{i,j}}^{+,-} = \langle t_it_j\rangle_{\{i,j\}; \bm{h}}^{+,+}.
\end{equation*}
The positivity of this limit comes from writing its Hamiltonian in terms of the $s$ and $t$ variables
\begin{equation*}
    \mathcal{H}_{\{i,j\}; \bm{h}}^{+,+}(\sigma,\sigma^\prime) = -\frac{1}{2}J_{ij}(s_is_j + t_it_j) - (h_i +2d-1)s_i - (h_j + 2d-1)s_j.
\end{equation*}
Uniqueness follows from the simple calculation
\begin{align*}
    \langle t_i\rangle_{\Lambda; \bm{h}}^{+,-} &= \frac{1}{4}\langle t_i(t_j^2 + s_j^2)\rangle_{\Lambda; \bm{h}}^{+,-}\\
    &\geq \frac{1}{4}\langle t_it_j^2\rangle_{\Lambda; \bm{h}}^{+,-}\geq \frac{1}{4}\langle t_it_j\rangle_{\Lambda; \bm{h}}^{+,-}\langle t_j\rangle_{\Lambda; \bm{h}}^{+,-},
\end{align*}
where in the first inequality we use (\ref{DVI.1.+}) and in the second we use (\ref{DVI.2.+}). Taking the limit as $\Lambda\nearrow\mathbb{Z}^d$ we conclude that if $\langle t_i\rangle_{\Lambda; \bm{h}}^{+,-}=0$, then $\langle t_j\rangle_{\Lambda; \bm{h}}^{+,-}=0$ for all $j$ neighbour of $i$. As the graph is connected and $i,j$ are arbitrary, this yields $\langle t_j\rangle_{\Lambda; \bm{h}}^{+,-}$ for all sites $j\in\mathbb{Z}^d$. 
\end{proof}
The last result needed is
\begin{proposition}\label{Consequence.of.DVI}
Let $\lambda\geq0$ and $\bm{h}=(h_{i})_{i\in\mathbb{H}_d^+}$ be a family of non-negative real numbers. Then, for all $\Lambda\Subset \mathbb{H}_d^+$ and all $i,j\in\mathbb{H}_d^+$
\begin{equation}
    \langle \sigma_i\sigma_j \rangle_{\Lambda; \lambda,\bm{h}}^-\leq  \langle \sigma_i\sigma_j \rangle_{\Lambda; \lambda,\bm{h}}^+ \label{increasing.correlations}
\end{equation}
and
\begin{equation}
     \langle \sigma_i\sigma_j \rangle_{\Lambda; \lambda,\bm{h}}^- -  \langle \sigma_i \rangle_{\Lambda; \lambda,\bm{h}}^- \langle \sigma_j \rangle_{\Lambda; \lambda,\bm{h}}^- \geq \langle \sigma_i\sigma_j \rangle_{\Lambda; \lambda,\bm{h}}^+ -  \langle \sigma_i \rangle_{\Lambda; \lambda,\bm{h}}^+ \langle \sigma_j \rangle_{\Lambda; \lambda,\bm{h}}^+
\end{equation}
\end{proposition}

\begin{proof}
Consider the duplicated state $\langle \cdot \rangle_{\Lambda; \bm{h}_\lambda}^{+,-}$ where $$(h_\lambda)_i=h_i +\lambda\mathbbm{1}_{\{i\in\mathcal{W}\}}.$$ 
The marginal distributions of this state are $\langle \cdot \rangle_{\Lambda;\lambda,\bm{h}}^{+}$ and $\langle \cdot \rangle_{\Lambda;\lambda,\bm{h}}^{-}$. So, 
\begin{align*}
   \frac{1}{2} \Big( \langle s_it_j  & \rangle_{\Lambda; \bm{h}_\lambda}^{+,-} +\langle t_is_j \rangle_{\Lambda; \bm{h}_\lambda}^{+,-}\Big) = \\
   &= \frac{1}{2}\left( \langle (\sigma_i + \sigma_i^\prime)(\sigma_j - \sigma_j^\prime) \rangle_{\Lambda; \bm{h}_\lambda}^{+,-} +\langle (\sigma_i - \sigma_i^\prime)(\sigma_j + \sigma_j^\prime) \rangle_{\Lambda; \bm{h}_\lambda}^{+,-}\right) \\
   &=\langle \sigma_i\sigma_j\rangle_{\Lambda; \bm{h}_\lambda}^{+,-} - \langle \sigma_i^\prime\sigma_j^\prime \rangle_{\Lambda; \bm{h}_\lambda}^{+,-} = \langle \sigma_i\sigma_j\rangle_{\Lambda; \lambda, \bm{h}}^{+} - \langle \sigma_i\sigma_j \rangle_{\Lambda; \lambda, \bm{h}}^{-}.
\end{align*}
The positivity of the RHS of this equation comes from (\ref{DVI.1.+}), and from the other part of the inequality we get
\begin{align*}
    \langle \sigma_i\sigma_j\rangle_{\Lambda; \lambda, \bm{h}}^{+} - \langle \sigma_i  \sigma_j \rangle_{\Lambda; \lambda, \bm{h}}^{-} &\leq \frac{1}{2} \left( \langle s_i\rangle_{\Lambda; \bm{h}_\lambda}^{+,-} \langle t_j   \rangle_{\Lambda; \bm{h}_\lambda}^{+,-} +\langle t_i \rangle_{\Lambda; \bm{h}_\lambda}^{+,-} \langle s_j \rangle_{\Lambda;\bm{h}_\lambda}^{+,-}\right)\\
    &= \langle \sigma_i \rangle_{\Lambda; \bm{h}_\lambda}^{+,-}\langle \sigma_j \rangle_{\Lambda; \bm{h}_\lambda}^{+,-} - \langle \sigma_i^\prime \rangle_{\Lambda; \bm{h}_\lambda}^{+,-}\langle \sigma_j^\prime \rangle_{\Lambda; \bm{h}_\lambda}^{+,-}\\
    &=\langle \sigma_i \rangle_{\Lambda; \lambda, \bm{h}}^{+}\langle \sigma_j \rangle_{\Lambda; \lambda, \bm{h}}^{+} - \langle \sigma_i \rangle_{\Lambda; \lambda, \bm{h}}^{-}\langle \sigma_j \rangle_{\Lambda; \lambda, \bm{h}}^{-},
\end{align*}
what concludes the proof.
\end{proof}

We finish this section by proving the existence of the limits (\ref{Def.F+}) and (\ref{Def.F-}). From now on we always assume that the external field depends only on the distance to the wall, and when taking a field $\mathrm{h}=(\mathrm{h}_k)_{k=0}^\infty$ we hope it is clear that, in the Hamiltonian, $h_i = \mathrm{h}_{i_d}$ for all $i\in\mathbb{H}_d^+$. Moreover, we will only consider the uniform interaction $J_{i,j}=J>0$ for all $i,j\in\mathbb{H}_d^+$.

\begin{theorem}
Given $J>0$, $\lambda \in \mathbb{R}$, an external field $\bm{h}$ induced by $\mathrm{h}=(\mathrm{h}_i)_{i=1}^\infty$ and $\Lambda_n\coloneqq\Lambda_{n,n^\alpha}$, with $0<\alpha<1$, the limits 

\begin{equation}
    \lim_{n\to \infty} \frac{1}{2|\mathcal{W}_n|}\ln\left[ \frac{(\mathcal{Z}^{+, J}_{n; \lambda,\bm{h}})^2}{Q^{+, J}_{\Delta_n; \bm{h}}} \right] 
\end{equation}

and

\begin{equation}
    \lim_{n\to \infty} \frac{1}{2|\mathcal{W}_n|}\ln\left[ \frac{(\mathcal{Z}^{-, J}_{n; \lambda,\bm{h}})^2}{Q^{-, J}_{\Delta_n; \bm{h}}} \right]
\end{equation}
exists for all $0<\alpha<1$.
\end{theorem}

\begin{proof}
As the parameters $J, \lambda$, and $\bm{h}$ are fixed, we omit them from the notation. Also, in all the sums we are going to omit $i\sim j$ since it is always the case. Start by noticing that 
\begin{equation*}
    (\mathcal{Z}^{+}_{n})^2 = \sum_{\sigma\in\Omega^+_{\Delta_n}} \exp{ \{ \sum_{\substack{i,j\in \mathbb{H}_d^+  \{i,j\} \cap \Lambda_n \neq \emptyset}} J\sigma_i\sigma_j + \sum_{\substack{i,j\in \mathbb{H}_d^-  \{i,j\} \cap \Lambda_n^\prime \neq \emptyset}} J\sigma_i\sigma_j +  \sum_{i \in \Delta_n} h_i\sigma_i  + \sum_{i\in \mathcal{W}_n \cup \mathcal{W}^\prime_n} \lambda \sigma_i \} }
\end{equation*}

Defining 

$$\Tilde{\mathcal{H}}_n(\sigma) \coloneqq \sum_{i\in \mathcal{W}_n, j\in \mathcal{W}_n^\prime} J\sigma_i\sigma_j -\sum_{i\in \mathcal{W}_n\cup \mathcal{W}_n^\prime} \lambda \sigma_i$$ 

we have that 
\begin{equation*}
     (\mathcal{Z}^{+}_{n})^2 = \sum_{\sigma\in\Omega^+_{\Delta_n}} \exp{- \{ H_{Is, n}(\sigma) + \Tilde{\mathcal{H}}_n (\sigma)\}}
\end{equation*}
where $H_{Is, n}$ is the usual Ising Hamiltonian such that

\begin{equation*}
    Q^{+, J}_{\Delta_n; \bm{h}} = \sum_{\omega \in \Omega^+_{\Delta_n}}\exp{\{- H_{Is, n}\}}. 
\end{equation*}

Now, defining 
\begin{equation*}
    \Xi_n(t) = \sum_{\sigma\in\Omega^+_{\Delta_n}} \exp{- \{ H_{Is, n}(\sigma) + t\Tilde{\mathcal{H}}_n(\sigma)\}} 
\end{equation*}
we get
\begin{align*}
    \ln\left[ \frac{(\mathcal{Z}^{+, J}_{n; \lambda,\bm{h}})^2}{Q^{+, J}_{\Delta_n; \bm{h}}} \right] &= \ln\left[ \frac{\Xi_n(1)}{\Xi_n(0)}\right] = \int_0^1 (\frac{d}{dt} \ln\left[ \Xi_n(t) \right]) dt
\end{align*}

As
\begin{align*}
    \frac{d}{dt} \ln\left[ \Xi_n(t) \right] &= \frac{1}{\Xi_n(t)}\sum_{\sigma\in\Omega^+_{\Delta_n}} \Tilde{\mathcal{H}}_n \exp{- \{ H_{Is, n}(\sigma) + t\Tilde{\mathcal{H}}_n\}}\\
    &= \sum_{i\in \mathcal{W}_n, j\in \mathcal{W}_n^\prime} J\langle\sigma_i\sigma_j\rangle^+_n(t) -\sum_{i\in \mathcal{W}_n\cup \mathcal{W}_n^\prime} \lambda \langle\sigma_i\rangle^+_n(t),
\end{align*}
where $\langle\cdot\rangle^+_n(t)$ are the Gibbs states with Hamiltonian $ H_{Is, n}(\sigma) + t\Tilde{\mathcal{H}}_n$ and $+$ boundary condition. Taking now the limiting states ${\langle\cdot\rangle^+(t) = \lim_{n \to \infty} \langle\cdot\rangle^+_n(t)}$, these are invariant under translations parallel to the wall since the local states are, and therefore

\begin{equation*}
    \lim_{n\to\infty} \frac{1}{2|\mathcal{W}_n|}\sum_{i\in \mathcal{W}_n, j\in \mathcal{W}_n^\prime} J\langle\sigma_i\sigma_j\rangle^t_n -\sum_{i\in \mathcal{W}_n\cup \mathcal{W}_n^\prime} \lambda \langle\sigma_i\rangle^t_n = \frac{J}{2}\langle\sigma_i\sigma_j\rangle^t + \lambda \langle\sigma_i\rangle^t 
\end{equation*}
for any $i,j\in \mathcal{W}_n$. The factor $\frac{1}{2}$ vanishes on the second term since the states are also invariant under reflection through the line $\mathcal{L}$.  Now we just use the dominated convergence theorem to get
\begin{equation*}
    \lim_{n\to \infty} \frac{1}{2|\mathcal{W}_n|}\ln\left[ \frac{(\mathcal{Z}^{+, J}_{n; \lambda,\bm{h}})^2}{Q^{+, J}_{\Delta_n; \bm{h}}} \right] = \int_0^1 \frac{J}{2}\langle\sigma_i\sigma_j\rangle^t + \lambda \langle\sigma_i\rangle^t dt
\end{equation*}

\end{proof}

\section{The wetting transition with no field} In this section, we characterize the picture of the wetting transition with no external field. As $\bm{h}\equiv 0$, we are omitting $\bm{h}$ from the notation through this whole subsection. We also assume $\lambda\geq 0$, since the other case is equivalent by spin-flip symmetry. Also, as the external field plays no role, we opt to emphasize the interaction $J$ in the limiting states, so we write
\begin{equation*}
    \langle f\rangle^{\eta}_{J, \lambda} \coloneqq \lim_{n\to\infty} \langle f\rangle^{\eta}_{\Lambda_n; \lambda, 0}.
\end{equation*}

With the previously defined wall free energy, $\tau_w(J, \lambda, \bm{h}) \coloneqq F^{-}(J, \lambda, \bm{h}) - F^{+}(J, \lambda, \bm{h})$ and the interface free energy for the Ising model
\begin{equation*}
    \tau(J) = \lim_{n,m\to\infty} -\frac{1}{|\mathcal{W}_n|} \ln\left[ \frac{Q^{\mp, J}_{\Delta_{m,n}; 0}}{Q^{+, J}_{\Delta_{m,n}; 0}} \right],
\end{equation*}
our goal will be to show the following results, proved in \cite{FP-II}:

\begin{proposition}\label{PF.1}
For $\bm{h}\equiv 0$, the wall free energy $\tau_w(J,\lambda)$ can be written as 

\begin{equation}\label{tau_w.as.integral}
    \tau_w(J,\lambda)=\int_0^\lambda \langle \sigma_0 \rangle^+_{J,s} - \langle \sigma_0 \rangle^-_{J,s} ds.
\end{equation}
\end{proposition}
and
\begin{theorem}\label{PF.2}
When $\bm{h}\equiv 0$ and $J>0$ we have
\begin{itemize}
    \item[(a)] $0 \leq \tau_w(J,\lambda) \leq \tau(J)$ for all $\lambda\geq 0$. Also, $\tau_w(J,0) = 0$;
    \item[(b)] $\tau_w$ is an monotonic non-decreasing function of $J$ and $\lambda \geq 0$;
    \item[(c)] $\tau_w(J,\lambda)$ is a concave function of $\lambda\geq 0$; 
    \item[(d)] If $\lambda\geq J$ then $\tau_w(J,\lambda)=\tau(J)$.\\
\end{itemize}
\end{theorem}
With these results, we see that 
\begin{equation*}\label{lambda_c.by.free.enegies}
    \lambda_c \coloneqq \inf \{\lambda\geq 0 : \tau_w(J,\lambda,0)=\tau(J) \}
\end{equation*}
 is finite. Moreover, from (\ref{tau_w.as.integral}) we get
 
 \begin{equation*}
     \tau(J)=\int_0^{\lambda_c} \langle \sigma_0 \rangle^+_{J,s} - \langle \sigma_0 \rangle^-_{J,s} ds.
 \end{equation*}
 
Let $J_c$ be the critical value for the phase transition for the Ising model. It was proved in \cite{BLP.1980} that $\tau(J)=0$ for all $J<J_c$ and in \cite{Lebowitz_Pfister_81} that $\tau(J)>0$ for all $J>J_c$. This, together with \mbox{Theorem \ref{PF.2}(a)}, shows that the semi-infinite Ising model has the same critical temperature as the usual Ising model, independent of the wall influence $\lambda$.\\

The most import consequence of this results is that get a criteria for uniqueness of the states once $\langle \sigma_0 \rangle^+_{J,\lambda} = \langle \sigma_0 \rangle^-_{J,\lambda}$ for $\lambda > \lambda_c$ and for $\lambda<\lambda_c$ this equality does not hold. This will be proved at the end of the section, completing the phase transition picture. We now proceed to the proof of Proposition \ref{PF.1}.\\

\begin{proof}[Proof of Proposition \ref{PF.1}]
We start by noting that, as we don't have an external field, $Q_{\Delta_n}^+ = Q_{\Delta_n}^-$.This simplifies the surface tension to 
\begin{equation}\label{tau_w.with.no.field}
    \tau_w(J,\lambda) = F^-(J,\lambda) - F^+(J,\lambda) = \lim_{n \to \infty} -\frac{1}{|\mathcal{W}_n|}\ln\left[ \frac{\mathcal{Z}^{-}_{n; \lambda}}{\mathcal{Z}^{+}_{n; \lambda}} \right].
\end{equation}

Differentiating each term in the limit w.r.t. $\lambda$ we get
\begin{align*}
    -\partial_\lambda \left( \ln\left[ \frac{Z_{n;\lambda}^{-, J}}{Z_{n;\lambda}^{+, J}} \right] \right) 
        &= \partial_\lambda \left( \ln Z_{n;\lambda}^{+, J} - \ln{Z_{n;\lambda}^{-, J}} \right) \\
        &= \frac{1}{Z_{n;\lambda}^{+, J}} \partial_\lambda \left( Z_{n;\lambda}^{+, J} \right) - \frac{1}{Z_{n;\lambda}^{-, J}} \partial_\lambda \left( Z_{n;\lambda}^{-, J} \right).
\end{align*}
As 
\begin{equation*}
    \partial_\lambda\left(Z_{n;\lambda}^{+, J}\right)=\sum_{i\in \mathcal{W}_n}\sum_{\sigma\in\Sigma_{\Lambda_n}^{+}} \sigma_i e^{-\mathcal{H}_{\Lambda_n; \lambda, 0}^J(\sigma)},
\end{equation*}
we conclude that
\begin{equation}\label{d.lambda.of.ln[Z+/Z-]}
     -\partial_\lambda \left( \ln\left[ \frac{Z_{n;\lambda}^{-, J}}{Z_{n;\lambda}^{+, J}} \right] \right) = \sum_{i\in \mathcal{W}_n} \langle \sigma_i \rangle^{+,J}_{n;\lambda} - \langle \sigma_i \rangle^{-,J}_{n;\lambda}.
\end{equation}

All of the above functions are continuous and bounded since they are the logarithm of positive polynomials. Therefore we can apply the fundamental theorem of calculus to get 
\begin{equation}\label{tau_w.as.lim.of.integral}
    \tau_w(J,\lambda) = \lim_{n\to \infty} \frac{1}{|\mathcal{W}_n|}\sum_{i\in \mathcal{W}_n}\int_0^\lambda \langle \sigma_i \rangle^{+,J}_{n;s} - \langle \sigma_i \rangle^{-,J}_{n;s}ds,
\end{equation}

so the proof will be finished once we prove that, for any fixed $s>0$,
\begin{equation}\label{lim.to.magnetization.plus}
    \lim_{n\to \infty} \frac{1}{|\mathcal{W}_n|}\sum_{i\in \mathcal{W}_n} \langle \sigma_i \rangle^{+,J}_{n;s} = \langle \sigma_0 \rangle^{+,J}_{s}
\end{equation}

and

\begin{equation}\label{lim.to.magnatization.minus}
    \lim_{n\to \infty} \frac{1}{|\mathcal{W}_n|}\sum_{i\in \mathcal{W}_n} \langle \sigma_i \rangle^{-,J}_{n;s} = \langle \sigma_0 \rangle^{-,J}_{s},
\end{equation}
then, by the dominated convergence theorem we conclude (\ref{tau_w.as.integral}).\\

We will prove only (\ref{lim.to.magnetization.plus}) since the proof for the other limit is analogous. Start by noticing that, for any $i\in \mathcal{W}_n$, $\langle \sigma_i \rangle^{+,J}_{n;s} \xrightarrow{n \uparrow \infty} \langle \sigma_0 \rangle^{+,J}_{s}$ by the translation invariance of the limit state (Proposition \ref{properties.from.fkg}). 

The rest of the proof consists of bounding from above and below the terms in the limit (\ref{lim.to.magnetization.plus}). For the upper bound, fix an $m\in\mathbb{N}$. For any $n\geq m$
\begin{align*}
    \frac{1}{|\mathcal{W}_n|}\sum_{i\in \mathcal{W}_n} \langle \sigma_i \rangle^{+,J}_{n;s} &= \frac{1}{|\mathcal{W}_n|}\sum_{i\in W_{n-m}} \langle \sigma_i \rangle^{+,J}_{n;s} + \frac{1}{|\mathcal{W}_n|}\sum_{i\in \mathcal{W}_n\setminus W_{n-m}} \langle \sigma_i \rangle^{+,J}_{n;s}.
\end{align*}

If $i\in W_{n-m}$ we have that $i+\Lambda_m \subset \Lambda_n$ and by Lemma \ref{decreasing.states} $\langle \sigma_i \rangle^{+,J}_{n;s}\leq \langle \sigma_i \rangle^{+,J}_{\Lambda_m + i;s}= \langle \sigma_0 \rangle^{+,J}_{m;s}$. Therefore
\begin{equation}
    \frac{1}{|\mathcal{W}_n|}\sum_{i\in W_{n-m}} \langle \sigma_i \rangle^{+,J}_{n;s} \leq \frac{1}{|W_{n-m}|}\sum_{i\in W_{n-m}} \langle \sigma_0 \rangle^{+,J}_{m;s} = \langle \sigma_0 \rangle^{+,J}_{m;s}.
\end{equation}

If $i\in \mathcal{W}_n\setminus W_{n-m}$, then $i + \Lambda_m\not\subset \Lambda_n$ and this set intersects the boundary of the wall $$\partial\mathcal{W}_n\coloneqq \{ i\in\mathcal{W}_n : \exists j\in \mathcal{W}_{n+1}\setminus \mathcal{W}_{n} \text{ s.t. }i\sim j  \},$$ so we can bound the number of such vertex by $|\Lambda_m||\partial\Lambda_n|$. Since $|\langle \sigma_i \rangle^{+,J}_{n;s}|\leq 1$, we have
\begin{equation*}
    \frac{1}{|\mathcal{W}_n|}\sum_{i\in \mathcal{W}_n\setminus W_{n-m}} \langle \sigma_i \rangle^{+,J}_{n;s} \leq \frac{2|\Lambda_m||\partial\mathcal{W}_n|}{|\mathcal{W}_n|}
\end{equation*}
which goes to zero as $n$ increases. Putting both bounds together we get 
\begin{equation*}
    \limsup_{n}  \frac{1}{|\mathcal{W}_n|}\sum_{i\in W_{n}} \langle \sigma_i \rangle^{+,J}_{n;s} \leq \langle \sigma_0 \rangle^{+,J}_{m;s}.
\end{equation*}
As $m$ is arbitrary, we can take the limit to get the upper bound in (\ref{lim.to.magnetization.plus}). The lower bound is a direct consequence of the translation invariance and Lemma \ref{decreasing.states} since

\begin{equation*}
    \langle \sigma_0 \rangle^{+,J}_{s} = \frac{1}{|\mathcal{W}_n|}\sum_{i\in \mathcal{W}_n} \langle \sigma_i \rangle^{+,J}_{s} \leq \frac{1}{|\mathcal{W}_n|}\sum_{i\in \mathcal{W}_n} \langle \sigma_i \rangle^{+,J}_{n;s},
\end{equation*}
therefore $\liminf_{n}\frac{1}{|\mathcal{W}_n|}\sum_{i\in \mathcal{W}_n} \langle \sigma_i \rangle^{+,J}_{n;s} \geq \langle \sigma_0 \rangle^{+,J}_{s}$.

\end{proof}

\begin{remark}
One fundamental difference when we have a non zero external field is that the integral in (\ref{tau_w.as.integral})  from $0$ to $\lambda$ leaves one extra term, that does not vanish when the external field is not zero.
\end{remark}

\begin{proof}[Proof of Theorem \ref{PF.2}] \text{  }\\

\textit{Proof of (a):} We define the set $\Omega_n^\mp$ of configurations with $\mp$-boundary condition as the configurations $\omega$ such that, for all $i\in\mathbb{H}^d_+$,  
\begin{equation*}
    \omega_i=\begin{cases}
                +1, & \text{ if } i_d\leq n^\alpha/2 \\
                -1 & \text{ if } i_d > n^\alpha/2.
                \end{cases}
\end{equation*}

Observe that 
\begin{equation}\label{excahnge.for.mp}
    \lim_{n\to\infty} \frac{1}{|\mathcal{W}_n|}\ln\left[ \frac{\mathcal{Z}^{-}_{n; \lambda,\bm{h}}}{\mathcal{Z}^{\mp}_{n; \lambda,\bm{h}}} \right] = 0
\end{equation}
since (omitting $i\sim j$ in the sums)
\begin{align*}
    \mathcal{Z}^{-}_{n; \lambda,\bm{h}} &= \sum_{\sigma\in\Omega^-_n} \exp{ \{ \sum_{\substack{i,j\in \Lambda_n }} J\sigma_i\sigma_j +  \sum_{i \in \Lambda_n} h_i\sigma_i  + \sum_{i\in \mathcal{W}_n } \lambda \sigma_i - \sum_{\substack{i\in \Lambda_n \\ j\in \mathbb{H}^d_+ \cap \Lambda_n^c}} J\sigma_i  \}} \\
        &= \sum_{\sigma\in\Omega^\mp_n} \exp{ \{ - \mathcal{H}_{n; \lambda, \bm{h}}^{J}(\sigma) - 2\sum_{\substack{i\in \Lambda_n \\ j\in \mathbb{H}^d_+ \cap \Lambda_n^c \\ i_d\leq n^\alpha}} J\sigma_i  \}}\\
        &\leq  \mathcal{Z}^{\mp}_{n; \lambda,\bm{h}}\exp{ \{ 2\sum_{\substack{i\in \Lambda_n j\in \mathbb{H}^d_+ \cap \Lambda_n \\ i_d\leq n^\alpha}} J  \} } = \mathcal{Z}^{\mp}_{n; \lambda,\bm{h}} \exp{ \{ 2Jdn^\alpha (2n+1)^{d-2} \}}.
\end{align*}
If in the last inequality we take $\sigma_i = +1$ instead, we get an similar lower bound, thus
\begin{equation*}
    \mid \ln\left[ \frac{1}{|\mathcal{W}_n|} \frac{\mathcal{Z}^{-}_{n; \lambda,\bm{h}}}{\mathcal{Z}^{\mp}_{n; \lambda,\bm{h}}} \right] \mid \leq  2Jdn^\alpha \frac{(2n+1)^{d-2}}{|\mathcal{W}_n|},
\end{equation*}
 from which follows (\ref{excahnge.for.mp}), since $|\mathcal{W}_n| = \mathcal{O}((2n)^{d-1})$ and we choose $0<\alpha< 1$.\\
 
 This gives us
 \begin{equation}\label{taw_w.with.mp-b.c.}
     \tau_w(J,\lambda,\h)=  \lim_{n \to \infty} -\frac{1}{|\mathcal{W}_n|}\ln\left[ \frac{\mathcal{Z}^{\mp}_{n; \lambda,\bm{h}}}{\mathcal{Z}^{+}_{n; \lambda,\bm{h}}} \right] - \frac{1}{2|\mathcal{W}_n|}\ln\left[ \frac{Q^{-, J}_{\Delta_n; \bm{h}}}{Q^{+, J}_{\Delta_n; \bm{h}}} \right].
 \end{equation}

For a fixed $n<\infty$, if we take the derivative we get
\begin{equation}
   \partial_\lambda(\tau_w(J,\lambda,\h)) = -\partial_\lambda\left(\frac{1}{|\mathcal{W}_n|}\ln\left[ \frac{\mathcal{Z}^{\mp}_{n; \lambda,\bm{h}}}{\mathcal{Z}^{+}_{n; \lambda,\bm{h}}} \right] \right) = |\mathcal{W}_n|^{-1} \sum_{i\in \mathcal{W}_n}  \langle \sigma_i\rangle^{+}_{n; \lambda, \bm{h}} - \langle \sigma_i\rangle^{\mp}_{n; \lambda, \bm{h}}
\end{equation}
which is positive by the second part of Lemma \ref{extremality.of.+.and.-.bc}. Therefore, we bound each term of the sequence by its limit when $\lambda$ goes to infinity, getting
\begin{align*}
    -\frac{1}{|\mathcal{W}_n|}\ln\left[ \frac{\mathcal{Z}^{\mp}_{n; \lambda,\bm{h}}}{\mathcal{Z}^{+}_{n; \lambda,\bm{h}}} \right] 
        &\leq \lim_{\lambda\to\infty} -\frac{1}{|\mathcal{W}_n|}\ln{\left[ \dfrac{\sum_{\sigma\in\Omega^\mp_n} \exp{ \{ \sum_{\substack{i,j\in \Lambda_n}} J\sigma_i\sigma_j + \lambda (\sum_{i\in \mathcal{W}_n } \sigma_i - |\mathcal{W}_n|)\} } }{\sum_{\sigma\in\Omega^+_n} \exp{ \{ \sum_{\substack{i,j\in \Lambda_n}} J\sigma_i\sigma_j + \lambda (\sum_{i\in \mathcal{W}_n } \sigma_i - |\mathcal{W}_n|)\} }} \right]}\\
        &=-\frac{1}{|\mathcal{W}_n|}\ln{\left[ \dfrac{\sum_{\sigma\in\Omega^\mp_n} \exp{ \{ \sum_{\substack{i,j\in \Lambda_n}} J\sigma_i\sigma_j \}\mathbbm{1}_{\{\sigma_i= +1,\text{ } \forall i\in \mathcal{W}_n \}}} }{\sum_{\sigma\in\Omega^+_n} \exp{ \{ \sum_{\substack{i,j\in \Lambda_n}} J\sigma_i\sigma_j \}\mathbbm{1}_{\{\sigma_i= +1,\text{ } \forall i\in \mathcal{W}_n \}}}} \right]}\\
        &=-\frac{1}{|\mathcal{W}_n|}\ln{\left[ \dfrac{\sum_{\sigma\in\Omega^\mp_{\Lambda_n\setminus \mathcal{W}_n}} \exp{ \{ \sum_{\substack{i,j\in \Lambda_n\setminus \mathcal{W}_n}} J\sigma_i\sigma_j \}} }{\sum_{\sigma\in\Omega^+_{\Lambda_n\setminus \mathcal{W}_n}} \exp{ \{ \sum_{\substack{i,j\in \Lambda_n\setminus \mathcal{W}_n}} J\sigma_i\sigma_j \}}} \right]}.
\end{align*}
Notice that, when $\bm{h}\equiv 0$, this last term is equal to 
\begin{equation*}
    -\frac{1}{|\mathcal{W}_n|}\ln{\left[\frac{Q^{\mp, J}_{[-n,n]^{d-1}\times[-n^\alpha/2,n^\alpha/2] , 0}}{Q^{+, J}_{[-n,n]^{d-1}\times[-n^\alpha/2,n^\alpha/2]  , 0}} \right]}
\end{equation*}
that coincides with the definition (\ref{tau(J)}).\\

\textit{Proof of (b): } Having in mind the simplification (\ref{tau_w.with.no.field}), the non-decreasing property of $\tau_w$ w.r.t. $\lambda$ comes directly from the positivity of (\ref{d.lambda.of.ln[Z+/Z-]}), that is a consequence of Lemma \ref{extremality.of.+.and.-.bc}. Analogously, if we differentiate the RHS of \eqref{tau_w.with.no.field} w.r.t. $J$ we get
\begin{equation}
    -\partial_J\left(\frac{1}{|\mathcal{W}_n|}\ln\left[ \frac{\mathcal{Z}^{-}_{n; \lambda,\bm{h}}}{\mathcal{Z}^{+}_{n; \lambda,\bm{h}}} \right] \right) = |\mathcal{W}_n|^{-1} \sum_{\substack{i\sim J \\ \{ i,j \} \cap \Lambda_n \neq \emptyset}}   \langle \sigma_i\sigma_j \rangle_{n; \lambda}^+ -  \langle \sigma_i\sigma_j \rangle_{n; \lambda}^-,
\end{equation}
that is positive by Proposition \ref{Consequence.of.DVI}, a consequence of the duplicate variables inequalities. \\

\textit{Proof of (c): }To see that $\tau_w$ is a concave function of $\lambda$, we use (\ref{d.lambda.of.ln[Z+/Z-]}) to get
\begin{multline}\label{d2.lambda.of.ln[Z+/Z-]}
     -\partial^2_\lambda \left( \ln\left[ \frac{Z_{n;\lambda}^{-, J}}{Z_{n;\lambda}^{+, J}} \right] \right) =  \sum_{i,j\in \mathcal{W}_n} \langle \sigma_i\sigma_j \rangle^{+,J}_{n;\lambda} - \langle \sigma_i \rangle^{+,J}_{n;\lambda}\langle \sigma_j \rangle^{+,J}_{n;\lambda}
            - \langle \sigma_i\sigma_j \rangle^{-,J}_{n;\lambda} + \langle \sigma_i \rangle^{-,J}_{n;\lambda}\langle \sigma_j \rangle^{-,J}_{n;\lambda},
\end{multline}
which is negative by Proposition \ref{Consequence.of.DVI}. So $\tau_w$ is the limit of concave functions, therefore it is concave.\\

\textit{Proof of (d): } For $\lambda\geq J$, since $\tau_w(J, \lambda)$ is non-decreasing in $\lambda$, we have that 
\begin{equation*}
    \tau_w(J, \lambda) \geq \tau_w(J, J).
\end{equation*}
But, by definition, 
\begin{equation*}
    \tau_w(J, J) = \lim_{n \to \infty} -\frac{1}{|\mathcal{W}_n|}\ln\left[ \frac{\mathcal{Z}^{\mp}_{n; J}}{\mathcal{Z}^{+}_{n; J,}} \right] = \lim_{n \to \infty} -\frac{1}{|\mathcal{W}_n|}\ln\left[ \frac{\mathcal{Q}^{\mp}_{[-n,n]^{d-1}\times[-n^\alpha/2, n^\alpha/2]; 0}}{\mathcal{Q}^{+}_{[-n,n]^{d-1}\times[-n^\alpha/2, n^\alpha/2];0}} \right],
\end{equation*}
where in the last equality we just used the representation of the semi-infinite model as the usual one and the translation invariance of the latter. Going back to the definition of $\tau(J)$, the last term in the sequence above is just the sub-sequence $m=n/2$, which concludes the proof. 
\end{proof}

We proceed to prove that such critical value  coincides with the critical value for non-uniqueness of the states, that is, $\lambda_c$ defined as in (\ref{lambda_c.by.free.enegies}) satisfies
\begin{equation}\label{lambda_c.by.states}
    \lambda_c = \inf \{\lambda\geq 0 : \langle \sigma_0 \rangle^{+}_{\lambda} = \langle \sigma_0 \rangle^{-}_{\lambda} \}.
\end{equation}

\begin{proposition}[Uniqueness with $\bm{h}\equiv 0$]\label{Uniqueness.for.h.zero}
For $\lambda>\lambda_c$, we have  ${\langle \sigma_0 \rangle^{+}_{\lambda} = \langle \sigma_0 \rangle^{-}_{\lambda}}$.
\end{proposition}

\begin{proof}
Indeed, since $\tau_w$ is non-decreasing in $\lambda$ and bounded by $\tau(J)$, we get that $\tau_w(J,\lambda)$ is constant equal to $\tau(J)$ for all $\lambda>\lambda_c$. Therefore, for any given $\lambda>\lambda_c$, $\tau_w$ is differentiable in $\lambda$ and $\partial_\lambda \tau_w(J,\lambda) = 0$. \\

So, $\tau_w$ is differentiable in $\lambda$ and is the point-wise limit of the sequence $|\calW_n|^{-1}\partial_\lambda \left( \ln Z_{n;\lambda}^{+, J} - \ln{Z_{n;\lambda}^{-, J}} \right)$, which is concave. We can then use a known theorem for convex functions, see for example \cite[Theorem B.12 ]{FV-Book}, to conclude that
\begin{multline*}
    0 = \partial_\lambda\tau_w(J,\lambda) = \lim_{n\to\infty}\frac{1}{|\calW_n|} \partial_\lambda \left( \ln Z_{n;\lambda}^{+, J} - \ln{Z_{n;\lambda}^{-, J}} \right) = \lim_{n\to\infty} \frac{1}{|\mathcal{W}_n|}\sum_{i\in \mathcal{W}_n}\langle \sigma_i \rangle^{+}_{n;\lambda} - \langle \sigma_i \rangle^{-}_{n;\lambda} = \langle \sigma_0 \rangle^{+}_{\lambda} - \langle \sigma_0 \rangle^{-}_{\lambda},
\end{multline*}
which is what we wanted to prove. The second equality is just (\ref{d.lambda.of.ln[Z+/Z-]}) and the last was proved during the demonstration of Proposition \ref{PF.1}.
\end{proof}

\begin{proposition}[Non-uniqueness with $\bm{h}\equiv 0$] 
For $\lambda<\lambda_c$, $\langle \sigma_0 \rangle^{+}_{\lambda} > \langle \sigma_0 \rangle^{-}_{\lambda}$.
\end{proposition}

\begin{proof}
By equation (\ref{d2.lambda.of.ln[Z+/Z-]}), we see that $\langle \sigma_0 \rangle^{+}_{\lambda} - \langle \sigma_0 \rangle^{-}_{\lambda}$ is a decreasing function of $\lambda$, since the finite states also are. Then 
\begin{equation}
    \langle \sigma_0 \rangle^{+}_{\lambda^\prime} - \langle \sigma_0 \rangle^{-}_{\lambda^\prime} \leq \langle \sigma_0 \rangle^{+}_{\lambda} - \langle \sigma_0 \rangle^{-}_{\lambda} \hspace{1cm} \forall\lambda^\prime\geq \lambda.
\end{equation}
Suppose that $\langle \sigma_0\rangle^{+}_{\lambda} = \langle \sigma_0 \rangle^{-}_{\lambda}$. Then, the RHS of the equation above is zero and so is the LHS. We then have 
\begin{equation*}
    \tau_w(J,\lambda) = \int_0^\lambda \langle \sigma_0 \rangle^{+}_{s} - \langle \sigma_0 \rangle^{-}_{s}ds = \int_0^{\lambda_c} \langle \sigma_0 \rangle^{+}_{s} - \langle \sigma_0 \rangle^{-}_{s}ds = \tau(J),
\end{equation*}
so $\lambda\geq\lambda_c$ which is a contradiction.  
\end{proof} Lastly, we have some bounds for $\lambda_c$. A lower bound comes easily from (\ref{tau_w.as.integral}), just by bounding $\langle \sigma_0 \rangle^+_{J,s}$ and  $\langle \sigma_0 \rangle^-_{J,s}$ we get

\begin{equation}\label{Eq: Lower_bound_lambda_c}
    \lambda_c\geq \frac{\tau(J)}{2}.
\end{equation}
In particular, $\lambda_c(J)>0$ for all $J>J_c$, since  $\tau(J)>0$ for $J>J_c$, see \cite{Lebowitz_Pfister_81}. An upper bound is a direct consequence of Theorem \ref{PF.2}(d):
\begin{equation*}
    \lambda_c \leq J.
\end{equation*}

\section{The macroscopic phenomenon of phase transition}    To define precisely what is the layer described in the first subsection, we need to use contours. We start this section by defining the \textit{low-temperature representation} and the \textit{Peierls contours}, and then we show how the existence or absence of multiple states determines the wetting transition.

\subsection{Low-temperature representation and Peierls contours}

Looking back at the definition of the Ising model, we see that a low temperature ($\beta>>0$) favors the configurations with spins aligned, so we rewrite the Hamiltonian trying to emphasize the non-aligned spins. Remember that we are considering a uniform interaction $\bm{J}\equiv J$. Again, we can see the semi-infinite model as the Ising model with interaction \emph{\textbf{\~{J}}} as in (\ref{J.for.the.semi.infinite}). 

Using the graph structure of $\mathbb{Z}^d$, we define $\mathcal{E}_\Lambda= {\{ \{x,y\}\in\mathbb{H}^+_d : x\sim y, x\in\Lambda \}}$, the set of edges with at least one vertex in $\Lambda\subset \mathbb{H}^+_d$ and no vertices in the wall, and $\mathcal{E}^{\mathcal{W}}_\Lambda = {\{ \{x,y\}\in\mathbb{Z}^d : x\sim y, x\in\Lambda\cap \mathcal{W}\}}$ so we have that

\begin{align*}
    -\sum_{\substack{i\sim j \\ \{i,j\}\cap \Lambda \neq \emptyset}} \tilde{J}_{i,j} \sigma_i\sigma_j &= -\sum_{\{i,j\}\in \mathcal{E}_\Lambda} J\sigma_i\sigma_j - \sum_{\{i,j\}\in \mathcal{E}^\mathcal{W}_\Lambda} \lambda\sigma_i\sigma_j \\
    &=-J|\mathcal{E}_\Lambda| - \lambda|\mathcal{E}^{\mathcal{W}}_\Lambda| + \sum_{\{i,j\}\in \mathcal{E}_\Lambda} J(1-\sigma_i\sigma_j) + \sum_{\{i,j\}\in \mathcal{E}^\mathcal{W}_\Lambda} \lambda(1-\sigma_i\sigma_j).\\
\end{align*}
and the \textit{low temperature representation} of the Hamiltonian is 
\begin{align}\label{loe.temp.representation}
     \mathcal{H}_{\Lambda; \lambda, \bm{h}}^{\bm{J}}(\sigma) = -J|\mathcal{E}_\Lambda| - \lambda&|\mathcal{E}^{\mathcal{W}}_\Lambda| + 2J |\{\{ i,j\} \in\mathcal{E}_\Lambda : \sigma_i\neq \sigma_j\}|\\ &+ 2\lambda |\{\{ i,j\} \in\mathcal{E}^\mathcal{W}_\Lambda : \sigma_i\neq \sigma_j\}| - \sum_{i\in\Lambda} h_i\sigma_i. \notag
\end{align}

The Peierls contours are defined in $\mathbb{Z}^d_*$, the dual graph of $\mathbb{Z}^d$. Such graph is constructed in the following way: for each $x\in\mathbb{Z}^d$, the closed unit cube with the center in $x$ is $C_x\subset\mathbb{R}^d$, and $\mathbb{Z}^d_*$ is the union of all faces $C_x\cap C_y$, for $x$ and $y$ nearest neighbors in $\mathbb{Z}^d$. With this we define the \textit{interface} of a configuration $\omega \in\Omega$ as
\begin{equation}
\Gamma(\omega)= \bigcup_{\substack{x\sim y\\ \omega_x\neq \omega_y}}C_x\cap C_y.    
\end{equation}
Each maximal connected component of $\Gamma(\omega)$ is called a \textit{contour}, which are usually denoted by $\gamma$. Each one of the contours $\gamma\in\Gamma(\omega)$ separates the vertices of $\mathbb{Z}^d$ into two subsets, the interior and the exterior of $\gamma$. The interior of gamma, denoted $Int(\gamma)$, are the vertices that are connected to infinity only by paths that cross $\gamma$, and the exterior is just $Ext(\gamma)\coloneqq\mathbb{Z}^d\setminus Int(\gamma)$. With these definitions, we see that given a configuration $\omega\in\Omega$ there is a one-to-one correspondence between non-aligned spins and the faces of $\Gamma(\omega)$.\\

\begin{figure}[ht]
    \centering
    \includegraphics[scale=0.4]{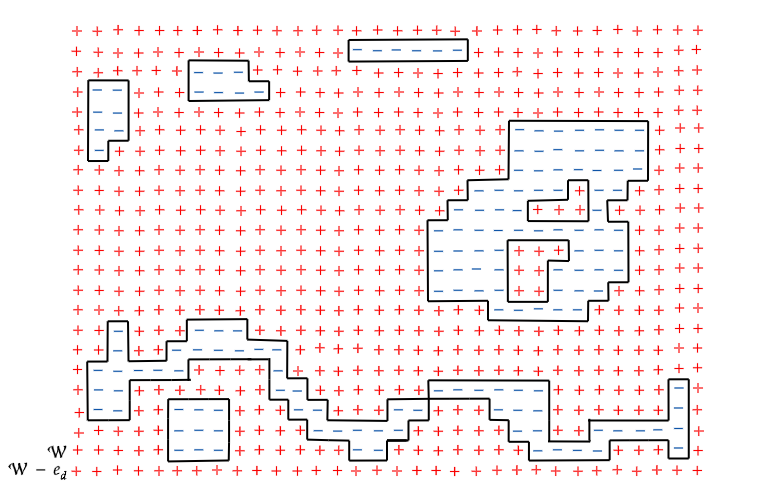}
    \caption{An example of a configuration with plus boundary condition and its contours.}
    \label{Conf+d2}
\end{figure}
The last observation is that, as the interaction between vertices in the bulk and at the wall differs, it is natural to differentiate faces separating the wall and the layer below it, so we  consider  ${\mathcal{W}^* \coloneqq \{ C_x \cap C_y: x\in \mathcal{W}, y\in \mathbb{H}_d^-\}}$ the wall on the dual lattice and we rewrite the measure $\mu^{\eta}_{\Lambda; \lambda, \bm{h}}$ as
\begin{equation}
    \mu^{\eta}_{\Lambda; \lambda, \bm{h}}(\sigma)=\frac{\mathbbm{1}_{\{\sigma\in\Omega_{\Lambda}^\eta\}} \exp{\left\{ -2\beta \sum_{\gamma\in\Gamma(\sigma)}\left( J|\gamma\setminus\mathcal{W}^*|  + \lambda|\gamma \cap \mathcal{W}^*|\right)  + \sum_{i\in\Lambda} h_i\sigma_i\right\}}}{\sum\limits_{\omega\in\Omega_\Lambda^\eta}\exp{ \left\{ -2\beta \sum_{\gamma\in\Gamma(\omega)}\left( J|\gamma\setminus\mathcal{W}^*|  + \lambda|\gamma \cap \mathcal{W}^*|\right) + \sum_{i\in\Lambda} h_i\omega_i \right\}}}.
\end{equation}
At least, if we fix a contour $\gamma^*$, the event that this contour occurs for some configuration has probability 
\begin{equation}\label{eq:low.temp.rep.}
    \mu^{\eta}_{\Lambda; \lambda, \bm{h}}(\gamma^*) = \frac{ \sum_{\substack{\omega\in\Omega_\Lambda^\eta\\ \gamma^*\in\Gamma(\omega)}} \exp{\left\{ -2\beta \sum_{\gamma\in\Gamma(\omega)}\left( J|\gamma\setminus\mathcal{W}^*|  + \lambda|\gamma \cap \mathcal{W}^*|\right)  + \sum_{i\in\Lambda}h_i\omega_i\right\}}}{\sum\limits_{\omega\in\Omega_\Lambda^\eta}\exp{ \left\{ -2\beta \sum_{\gamma\in\Gamma(\omega)}\left( J|\gamma\setminus\mathcal{W}^*|  + \lambda|\gamma \cap \mathcal{W}^*|\right) + \sum_{i\in\Lambda} h_i\omega_i \right\}}}.
\end{equation}

This is the basic setup of the famous Peierls' argument, one of the most important tools in the study of phase transition in lower temperatures. One application will be seen in the next subsection.

\subsection{The wetting transition in terms of contours}
    To define precisely what it means to appear a thick layer of pluses in the wall we define the $\omega^L$ boundary condition, where $L\in\mathbb{N}$ and, for $i\in\mathbb{Z}^d$, 
    \begin{equation*}
        (\omega^L)_i = \begin{cases}
                            +1, \text{ if }i\in \mathcal{W}_L - e_d,\\
                            -1, \text{ otherwise,}
                        \end{cases}
    \end{equation*}
where $e_d = (0,\dots,0,1)$. For a fixed $L\in\mathbb{N}$ and $\mathcal{W}_L\subset \Lambda \Subset\mathbb{Z}^d$, if we pick any configuration $\omega\in\Omega_\Lambda^{\omega_L}$, there will be one open contour $\gamma_L\in\Gamma(\omega)$ that is induced by the defect in the layer right below the wall, that is, $\gamma_L$ is the contour that separates $\mathcal{W}_L - e_d$ from $\mathcal{W}_{L+1}\setminus \mathcal{W}_L - e_d$, as in Figure \ref{fig:gamma_L}.\\

\begin{figure}[htbp]
    \centering
    \includegraphics[scale=0.2]{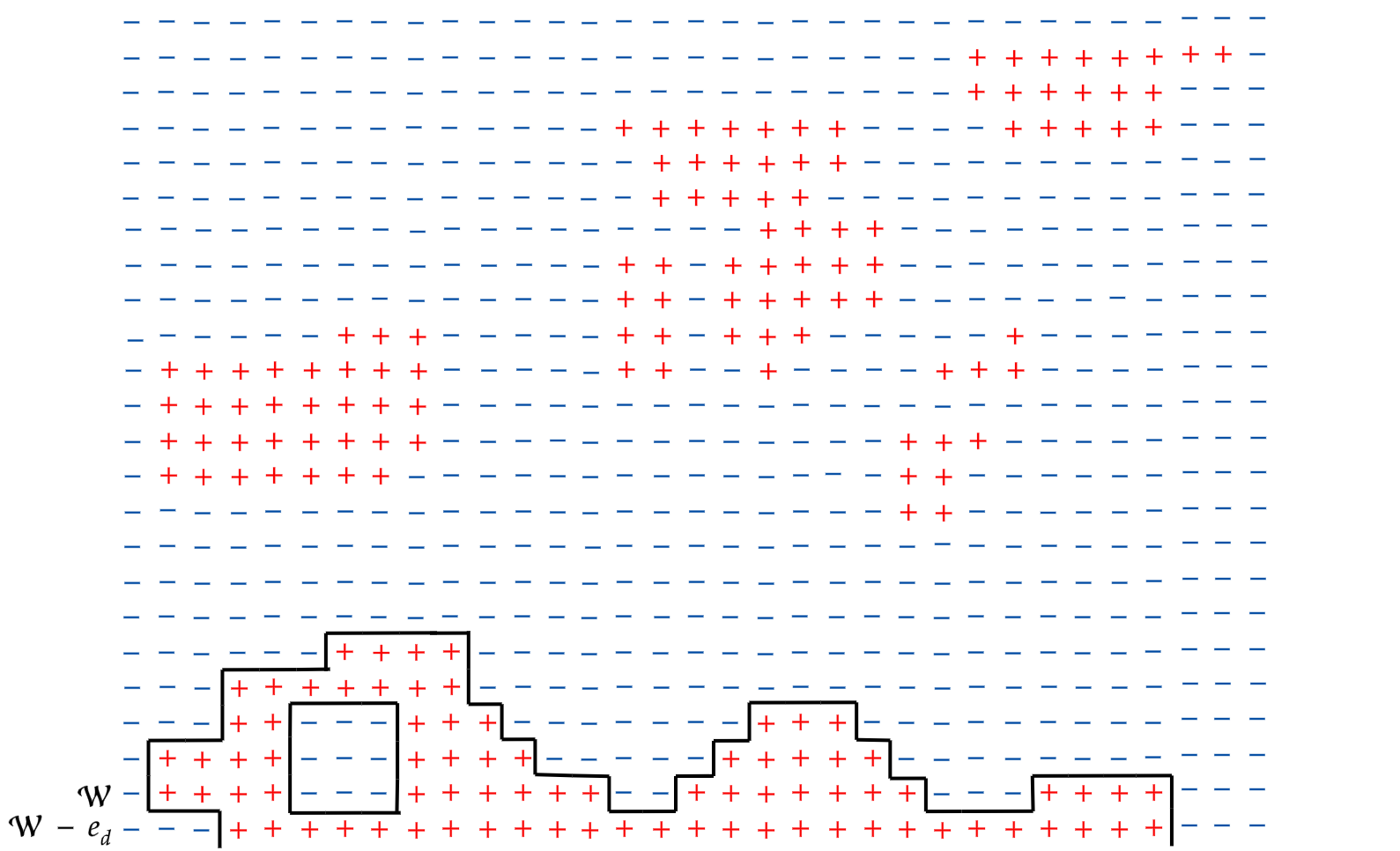}
    \caption{The $\gamma_L$ contour of a configuration in $\Lambda_{L+3}$.}
    \label{fig:gamma_L}
\end{figure}
This boundary condition induces the limiting state 
\begin{equation}\label{eq:infinite.L.state}
    \langle \cdot \rangle^L_{\lambda,\bm{h}} \coloneqq \lim_{\Lambda\nearrow\H+} \langle \cdot \rangle^{\omega^L}_{\Lambda;\lambda,\bm{h}},
\end{equation}
which exists since, for any $\mathcal{W}_L\subset \Lambda_1\subset\Lambda_2\Subset\H+$ and any local non-decreasing function $f$, $\langle f \rangle^{\omega^L}_{\Lambda_1;\lambda,\bm{h}} \leq \langle f \rangle^{\omega^L}_{\Lambda_2;\lambda,\bm{h}}$. The proof of this is identical to the proof of Lemma \ref{decreasing.states}. As expected, the state \eqref{eq:infinite.L.state} converges, as L diverges, to the minus state. Indeed, by Lemma \ref{extremality.of.+.and.-.bc}, if $L_1\leq L_2$ and $\mathcal{W}_{L_2}\subset \Lambda$,
\begin{equation*}
    \langle f \rangle^{\omega^{L_1}}_{\Lambda;\lambda,\bm{h}} \leq \langle f \rangle^{\omega^{L_2}}_{\Lambda;\lambda,\bm{h}}
\end{equation*}
for any local non-decreasing function $f$. Therefore, $ \langle \cdot \rangle^{L_1}_{\lambda,\bm{h}}\leq  \langle \cdot \rangle^{L_2}_{\lambda,\bm{h}}$ and we can then take the limit as $L\to\infty$, which converges to the minus state since $ \langle \cdot \rangle^-_{\lambda,\bm{h}}=\lim_{L\to\infty} \langle \cdot \rangle^{\omega_L}_{\Lambda_L;\lambda,\bm{h}}$. Hence
\begin{equation*}
    \lim_{L\to\infty}  \langle \cdot \rangle^L_{\lambda,\bm{h}} =  \langle \cdot \rangle^-_{\lambda,\bm{h}}.
\end{equation*}

From now on, we go back to the notation of Gibbs measures \eqref{eq:def.local.gibbs.measure} since it is more intuitive to use probabilities to deal with contours. Two key events are the configurations for which $\gamma_L$ separates $0$ and $-e_d$, denoted by $\{\gamma_L\in\underbar{0}\}$ and its complementary, denoted $\{\gamma_L\in\overline{0}\}$. For example, the configuration in Figure \ref{fig:gamma_L} belongs to $\{\gamma_L\in\overline{0}\}$. 

\begin{proposition} Consider the nearest neighbour semi-infinite Ising model with interaction $J>0$, wall influence $\lambda>0$ and external field $\bm{h}=(h_i)_{i\in\H+}$ induced by a non-negative summable sequence $\mathrm{h}=(\mathrm{h}_k)_{k=0}^\infty$, that is, $h_i=\mathrm{h}_{i_d}$ for all $i\in\H+$. Then, if there is phase transition, $\mu_{\lambda,\bm{h}}^{\omega_L}({\gamma_L\in\underbar{0}})>0$ for every $L$ large enough. Conversely, when $\mu_{\lambda,\bm{h}}^{+} = \mu_{\lambda,\bm{h}}^{-}$,  $||\mathrm{h}||_1<2\min\{J,\lambda\}$ and $\beta$ is large enough, we have that $\lim_{L\to\infty} \mu_{\lambda,\bm{h}}^{L}({\gamma_L\in\underbar{0}})=0$.
\end{proposition}

\begin{remark}
The condition $||\mathrm{h}||_1<2\min\{J,\lambda\}$ may seem very restrictive at first, but it is not so much. Given $\mathrm{h}=(\mathrm{h}_k)_{k=0}^\infty$, consider the sequence $\mathrm{h}^\prime$ that is equal to $\mathrm{h}$, except at $0$, where $\mathrm{h}^\prime_0=0$. Then, for all $\eta\in\Omega$, $\mu_{\lambda,\bm{h}}^{\eta} = \mu_{\lambda + \mathrm{h}_1,\bm{h}^\prime}^{\eta}$, where $\bm{h}^\prime$ is the field induced by $\mathrm{h}^\prime$. The condition above, hence, becomes $||\mathrm{h}^\prime||_1 < 2\min\{\lambda+\mathrm{h}_1, J\}$. 
\end{remark}

\begin{proof}Notice that
\begin{align}\label{eq:mag.as.contours.prob}
    \mu_{\Lambda_n;\lambda,\bm{h}}^{\omega_L}(\sigma_0) &=  \mu_{\Lambda_n;\lambda,\bm{h}}^{\omega_L}(\sigma_0\mathbbm{1}_{\{\gamma_L\in\underbar{0}\}}) + \mu_{\Lambda_n;\lambda,\bm{h}}^{\omega_L}(\sigma_0\mathbbm{1}_{\{\gamma_L\in\overline{0}\}}) \nonumber \\
    &= \mu_{\Lambda_n;\lambda,\bm{h}}^{\omega_L}({\{\gamma_L\in\underbar{0}\}}) + \sum_{\gamma\in\overline{0}}\mu_{\Lambda_n;\lambda,\bm{h}}^{\omega_L}(\sigma_0\mathbbm{1}_{\{\gamma=\gamma_L\}}) \nonumber \\
    &=\mu_{\Lambda_n;\lambda,\bm{h}}^{\omega_L}({\{\gamma_L\in\underbar{0}\}}) + \sum_{\gamma\in\overline{0}}\mu_{\Lambda_n;\lambda,\bm{h}}^{\omega_L}(\sigma_0|{\gamma=\gamma_L})\mu_{\Lambda_n;\lambda,\bm{h}}^{\omega_L}({\gamma=\gamma_L}).
\end{align}

If $\gamma\in\overline{0}$, it means that the origin is surrounded by a set with plus boundary condition. More precisely, since the model is short range, for $\gamma\in\overline{0}$,
\begin{equation*}
    \mu_{\Lambda_n;\lambda,\bm{h}}^{\omega_L}(\sigma_0|{\gamma=\gamma_L}) = \mu_{Int(\gamma);\lambda,\bm{h}}^{+}(\sigma_0),
\end{equation*}
and by Lemma \ref{decreasing.states}, $\mu_{\Lambda_n;\lambda,\bm{h}}^{\omega_L}(\sigma_0|{\gamma=\gamma_L})\geq \mu_{\lambda,\bm{h}}^{+}(\sigma_0)$. Together with \eqref{eq:mag.as.contours.prob}, this implies that

\begin{equation*}
    \mu_{\Lambda_n;\lambda,\bm{h}}^{\omega_L}({\gamma_L\in\underbar{0}}) \geq \mu_{\lambda,\bm{h}}^{+}(\sigma_0) (1- \mu_{\Lambda_n;\lambda,\bm{h}}^{\omega_L}({\gamma_L\in\underbar{0}})) - \mu_{\Lambda_n;\lambda,\bm{h}}^{\omega_L}(\sigma_0).
\end{equation*}
Finally, taking the limit as $n\to\infty$, 
\begin{equation*}
      \mu_{\lambda,\bm{h}}^{L}({\gamma_L\in\underbar{0}})\geq \frac{1}{2}(\mu_{\lambda,\bm{h}}^{+}(\sigma_0) - \mu_{\lambda,\bm{h}}^{L}(\sigma_0)) = \frac{1}{2}(\mu_{\lambda,\bm{h}}^{+}(\sigma_0) - \mu_{\lambda,\bm{h}}^{-}(\sigma_0)).
\end{equation*}

This shows that, as soon as we have a phase transition, there is a positive probability of not seeing a layer of pluses on the wall. To get a converse, we express the magnetization as $\mu_{\Lambda_n;\lambda,\bm{h}}^{\omega_L}(\sigma_0) = 1 -2\mu_{\Lambda_n;\lambda,\bm{h}}^{\omega_L}(\sigma_0=-1)$. When $\sigma_0=-1$, either $\gamma_L\in\underbar{0}$ or there exists a contour surrounding $0$, hence
\begin{equation}\label{eq:prob.of.minus.one}
    \mu_{\Lambda_n;\lambda,\bm{h}}^{\omega_L}(\sigma_0=-1) =  \mu_{\Lambda_n;\lambda,\bm{h}}^{\omega_L}({\gamma_L\in\underbar{0}}) + \sum_{o\in Int(\gamma)} \mu_{\Lambda_n;\lambda,\bm{h}}^{\omega_L}({\gamma}).
\end{equation}
We proceed to prove that we can take the limit as $n\to\infty$ in the equation above, or equivalently, that $\sum_{0\in Int(\gamma)}\mu_{\lambda,\bm{h}}^{L}({\gamma})<+\infty$. We will do so by using a Peierls-type argument. For $n>L$ and $\gamma^*$ a contour in $\Lambda_n$ with $0\in Int(\gamma)$,  the low temperature representation \eqref{eq:low.temp.rep.} yields
\begin{align}\label{eq:peierls}
     \mu_{\Lambda_n;\lambda,\bm{h}}^{\omega_L}(\gamma^*) &= \frac{1}{\mathcal{Z}^{\omega_L, \bm{J}}_{n; \lambda,\bm{h}}}\sum_{\substack{\omega\in\Omega_{\omega_L} \\ \gamma^*\in\Gamma(\omega)}} \prod_{\gamma\in\Gamma(\omega)}\exp{\{-2K(\gamma) + \sum_{i\in\Lambda_n}\beta h_i\omega_i\}} \nonumber \\
     &\leq \frac{e^{\{-2K(\gamma^*) + 2\sum_{i\in Int(\gamma^*)}h_i\}}}{\mathcal{Z}^{\omega_L, \bm{J}}_{n; \lambda,\bm{h}}}\sum_{\substack{\omega\in\Omega_{\omega_L} \\ \gamma^*\in\Gamma(\omega)}} \prod_{\gamma\in\Gamma(\omega)\setminus \{\gamma^*\} }\exp{\{-2K(\gamma) + \sum_{i\in\Lambda_n}\beta h_i\omega_i(1 - 2\mathbbm{1}_{\{i\in Int(\gamma^*)\}} )\}} \nonumber \\
     &\leq \exp{\{-2\beta (J\wedge\lambda)|\gamma^*| + 2\sum_{i\in Int(\gamma^*)}h_i\}},
\end{align}
where $K(\gamma) \coloneqq \beta J|\gamma\setminus\mathcal{W}^*|  + \beta\lambda|\gamma \cap \mathcal{W}^*|$ is a bounded value that depends on $\gamma$ and $x\wedge y\coloneqq\min\{x,y\}$. To relate the influence of the field in the interior of a contour with its size, we define the \textit{layers of} $\gamma$, $layer_k(\gamma) = \{i\in Int(\gamma): i_d=k\}$. The largest layer is $layer_{max}(\gamma)$, that is, $|layer_{max}(\gamma)| = \max_{k=1,\dots, n^\alpha}|layer_k(\gamma)|$. Hence
\begin{equation*}
    \sum_{i\in Int(\gamma)}h_i \leq \sum_{k=0}^{n^\alpha}\mathrm{h}_k|layer_{max}(\gamma)|.
\end{equation*}
Finally, for each vertex $v=(v_1,\dots,v_d)\in layer_{max}(\gamma)$, there are two distinct vertices $a_v,b_v\in\{(v_1,\dots,v_{d-1}, s)\in\H+: 0\leq s \leq n^\alpha +1\}$ for which $C_v\cap C_{a_v},C_v\cap C_{b_v}\in\gamma$. This is because $v$ is surrounded by $\gamma$, so there is one plaquette of $\gamma$ above $v$ and one below it. Hence $|layer_{max}(\gamma)|\leq |\gamma|/2$. This, together with \eqref{eq:peierls} yields
\begin{equation*}
    \mu_{\Lambda_n;\lambda,\bm{h}}^{\omega_L}(\gamma^*) \leq \exp{\{(-2\beta(J\wedge\lambda) + \beta ||\mathrm{h}||_1})|\gamma^*|\},
\end{equation*}
Therefore $\sum_{0\in Int(\gamma)}\mu_{\lambda,\bm{h}}^{L}({\gamma})<+\infty$ whenever $||\mathrm{h}||_1<2(J\wedge \lambda)$ and $\beta$ is large enough. Hence, taking $n\to\infty$, \eqref{eq:prob.of.minus.one} yields
\begin{equation*}
    \mu_{\lambda,\bm{h}}^{L}(\sigma_0=-1) =  \mu_{\lambda,\bm{h}}^{L}({\gamma_L\in\underbar{0}}) + \sum_{o\in Int(\gamma)} \mu_{\lambda,\bm{h}}^{L}({\gamma}).
\end{equation*}
Since for the plus state $\mu_{\lambda,\bm{h}}^{+}(\sigma_0=-1) = \sum_{o\in Int(\gamma)} \mu_{\lambda,\bm{h}}^{+}({\gamma})$, we conclude that when we have uniqueness, that is, $ \mu_{\lambda,\bm{h}}^{+} = \mu_{\lambda,\bm{h}}^{-}$, then
\begin{align*}
    \lim_{L\to\infty} \mu_{\lambda,\bm{h}}^{L}({\gamma_L\in\underbar{0}}) &= \lim_{L\to \infty} ( \mu_{\lambda,\bm{h}}^{L}(\sigma_0=-1) -  \sum_{o\in Int(\gamma)} \mu_{\lambda,\bm{h}}^{L}({\gamma}))
    = \mu_{\lambda,\bm{h}}^{-}(\sigma_0=-1) -  \sum_{o\in Int(\gamma)} \mu_{\lambda,\bm{h}}^{-}({\gamma}))\\
    & =  \mu_{\lambda,\bm{h}}^{-}(\sigma_0=-1) -  \sum_{o\in Int(\gamma)} \mu_{\lambda,\bm{h}}^{+}({\gamma})) =  \mu_{\lambda,\bm{h}}^{-}(\sigma_0=-1) -  \mu_{\lambda,\bm{h}}^{+}(\sigma_0=-1) = 0.
\end{align*}
\end{proof}  


\chapter{Semi-infinite Ising model with inhomogeneous external fields}
In this chapter, we consider some choices of external fields that vary according to the distance to the wall. For the nearest neighbor Ising model with external field $\bm{h}^* = (h_i^*)_{i\in\Z^d}$ given by
\begin{equation}\label{particular.external.field}
    h_i^* = \begin{cases}
            h^* &\text{ if }i=0,\\
            \frac{h^*}{|i|^\delta} &\text{ otherwise},\\
            \end{cases}
\end{equation}
it was proved in \cite{Bissacot_Cass_Cio_Pres_15} that when $\delta<1$ we have uniqueness for temperatures below a critical one. In \cite{Cioletti_Vila_2016}, it was shown that this critical temperature must be zero. The proof in \cite{Bissacot_Cass_Cio_Pres_15} involves contour arguments and the one in \cite{Cioletti_Vila_2016} uses a generalization of the Edwards–Sokal representation. 

For the semi-infinite Ising model, we first consider external fields $\bm{\rmh}$ induced by a summable sequence $\rmh = (\rmh_i)_{i\in\H+}$, that is, $\bm{\rmh}_i=\rmh_{i_d}$ for all $i\in\H+$. We show in Section 1 that such an external field preserves the phase transition, as long as the $\ell_1$-norm of $\rmh$ is small compared to $\lambda$.   

A more natural choice of the external field is one decaying as it gets further from the wall, that is, $h_i \leq h_j$ whenever $j_d\leq i_d$. Hence, we will consider the external field given by $h_i=\lambda i^{-\delta}$. In Section 2 we prove that, when $\delta>1$, the model behaves as the model with no field, so there is a critical value $\overline{\lambda}_c(J, \delta)$ such that there are multiple Gibbs states when $0\leq\lambda<\overline{\lambda}_c(J, \delta)$, and there is uniqueness otherwise. We are also able to show that $0<\overline{\lambda}_c(J, \delta)\leq\lambda_c$ whenever $J<J_c$. At last, we show that when $\delta<1$, the semi-infinity Ising model with this choice of external field presents only one Gibbs state.

\section{External field decaying with \texorpdfstring{$\lambda$}{l}}   
In order to simplify the notation, we make a slight change to the lattice defined previously. We consider now the model takes place in $\H+ = \mathbb{Z}^{d-1}\times\mathbb{N}$, with the natural numbers starting at $1$. All the definitions made previously can be easily adapted to this lattice. In particular, given $h\in \mathbb{R}$, the external field we are interested in is
$\widehat{\bm{h}}=(h_i)_{i\in\H+}$ with 
\begin{equation}
    h_i=\frac{h}{i_d^\delta},
\end{equation}
for all $i\in\H+$.  Figures \ref{EF.on.Lambda} and \ref{graph.EF} shows how this external field behaves. A particularly interesting choice of $h$ is $h=\lambda$. This particular case will be denoted  $\widehat{\bm{\lambda}}=(\lambda_i)_{i\in\H+}$, hence 
\begin{equation}
    \lambda_i=\frac{\lambda}{i_d^\delta}.
\end{equation}
We will always assume $\lambda \geq0$. 

\subsection{Critical behavior when  \texorpdfstring{$\delta>1$}{d>1}}
In the same steps of the case with no external field, we prove that there exists a critical value $\overline{\lambda}_c$ such that, for $\lambda>\overline{\lambda}_c$ there is a unique state, and for $\lambda<\overline{\lambda}_c$, there is a phase transition. To do so, we need to use a modified notion of wall-free energy. 

As we shifted the lattice, we reintroduce some previously defined regions. Consider the sequence that invades $\H+$ given by $\Lambda_n \coloneqq \Lambda_{n,n} = [-n,n]^{d-1}\times [1,n]$. Take  $\Lambda^\prime_n$ as the reflection of $\Lambda_n$ with respect to the line ${\mathcal{L} \coloneqq \{ (i_1,\dots,i_d)\in \mathbb{Z}^d : i_d = \frac{1}{2}\}}$. Similarly, define the walls as $\calW_n\coloneqq [-n,n]^{d-1}\times\{1\}$ and reflection of the walls $\mathcal{W}_n$ as $\mathcal{W}^\prime_n \coloneqq [-n,n]\times\{0\}$. Moreover, we denote $\Delta_n \coloneqq \Lambda_n \cup \Lambda^\prime_n$. For any summable sequence of positive real numbers $\mathrm{h}=(\mathrm{h}_\ell)_{\ell\geq 1}$, let $\bm{\mathrm{h}} = \{\mathrm{h}_{i}\}_{i\in\H+^d}$ be the external field induced by $\mathrm{h}$, that is, $\mathrm{h}_i\coloneqq\mathrm{h}_{i_d}$, with $i_d$ being the last coordinate of $i\in\Z^d$. We also denote $\overline{\bm{\rmh}} = \{\overline{\rmh}_i\}_{i\in\Z^d}$ the natural extension on $\bm{\rmh}$ to $\Z^d$, defined by
\begin{equation*}
    \overline{\rmh}_i = \begin{cases}
                            \rmh_i & \text{ if }i\in\H+ \\
                            \rmh_{-i + e_d} & \text{ if }i\in\Z^d\setminus\H+,
                        \end{cases}
\end{equation*}
where $e_d=(0,\dots,0,1)$ is a canonical base vector. Given any $J>0$ and summable  $\mathrm{h}=(\mathrm{h}_\ell)_{\ell\geq 0}$, the \textit{free surface energy} for the $+$-boundary condition and $-$-boundary condition are, respectively, 
\begin{equation*}
    \Tilde{F}^+(J,\bm{\mathrm{h}}) \coloneqq \lim_{n\to\infty} -\frac{1}{2|\calW_n|} \ln{\left[\frac{\left(\mathcal{Z}_{n;\bm{\rmh}}^+\right)^2}{Q_{\Delta_n; J}^+}\right]}
\end{equation*}
and 
\begin{equation*}
    \Tilde{F}^-(J,\bm{\mathrm{h}}) \coloneqq \lim_{n\to\infty} -\frac{1}{2|\calW_n|} \ln{\left[\frac{\left(\mathcal{Z}_{n;\bm{\rmh}}^-\right)^2}{Q_{\Delta_n; J}^-}\right]}.
\end{equation*}
The difference between this definition and the one introduced previously, and in \cite{FP-II}, is that we are also erasing the external field in the partition functions of the Ising model. We first prove that this limits are well defined. 

\begin{proposition}
    For any $J\geq 0$ and any summable sequence of positive real numbers $\mathrm{h}=(\mathrm{h}_\ell)_{\ell\geq 1}$, let $\bm{\rmh}$ be the external field induced by $\mathrm{h}$. The limits $ \Tilde{F}^+(J,\bm{\mathrm{h}})$ and $ \Tilde{F}^-(J,\bm{\mathrm{h}})$ are well defined.
\end{proposition}
\begin{proof}
    As the parameters $J$ and $\bm{\rmh}$ are fixed, we will omit then from the notation. Also, in the sums, we omit $|i- j|=1$, since this is always the case. Start by noticing that 
    \begin{align*}
        \left(\mathcal{Z}^+_{n;\bm{h}}\right)^2 &= \sum_{\sigma\in \Omega^+_{\Delta_n}}\exp{\left\{ \sum_{i,j\in\Delta_n}J\sigma_i\sigma_j - \sum_{i\in\calW_n} J\sigma_i\sigma_{i-e_d} + \sum_{i\in \Delta_n}\overline{\rmh}_i\sigma_i + \sum_{\substack{i\in \Delta_n \\ j\notin \Delta_n}}J\sigma_i\right\}} \\ 
        & = \sum_{\sigma\in \Omega^+_{\Delta_n}}\exp{\left\{ \sum_{i,j\in\Delta_n}J\sigma_i\sigma_j + \sum_{\substack{i\in \Delta_n \\ j\notin \Delta_n}}J\sigma_i - \Tilde{H}_n(\sigma)\right\}}, 
    \end{align*}
    with $\Tilde{H}_n(\sigma) = \sum_{i\in\calW_n} J\sigma_i\sigma_{i-e_d} - \sum_{i\in\Delta_n}\overline{\rmh}_i\sigma_i$. Take 
    \begin{equation*}
         \Xi_n(t) \coloneqq  \sum_{\sigma\in \Omega^+_{\Delta_n}}\exp{\left\{ \sum_{i,j\in\Delta_n}J\sigma_i\sigma_j + \sum_{\substack{i\in \Delta_n \\ j\notin \Delta_n}}J\sigma_i - t\Tilde{H}_n(\sigma)\right\}},
    \end{equation*}
    and let $\langle\cdot \rangle^{+}_{\Delta_n}(t)$ be the state with $+$-boundary condition given by the Hamiltonian $H_n(t)(\sigma) = \sum_{i,j\in\Delta_n}J\sigma_i\sigma_j + \sum_{\substack{i\in \Delta_n \\ j\notin \Delta_n}}J\sigma_i - t\Tilde{H}_n(\sigma)$. We can write 
    \begin{align}\label{Eq: F_as_integral}
        \ln{\left[\frac{\left(\mathcal{Z}_{n;\bm{\rmh}}^+\right)^2}{Q_{\Delta_n; J}^+}\right]} &= \ln{\left[\frac{\Xi_n(1)}{\Xi_n(0)}\right]} = \int_{0}^1 \frac{\d}{\d s}\left(\ln\Xi_n(s)\right)ds \nonumber \\
        &= - \int_0^1\sum_{i\in\calW_n} J \langle \sigma_i\sigma_{i-e_d} \rangle_{\Delta_n}^+(s) ds + \int_0^1 \sum_{i\in\Delta_n}\overline{\rmh_i}\langle\sigma_i \rangle_{\Delta_n}^+(s)ds.
    \end{align}
    Considering $\langle \cdot \rangle^+(t)\coloneqq\lim_{n\to\infty} \langle \cdot \rangle^+_{\Delta_n}(t)$ the limiting state, we will show that 
    \begin{equation}\label{Eq: Lim_correlation_on_wall}
        \lim_{n\to\infty} \frac{1}{|\calW_n|}\sum_{i\in\calW_n}  \langle \sigma_i\sigma_{i-e_d} \rangle_{\Delta_n}^+(s) = \langle \sigma_0\sigma_{-e_d} \rangle^+(s).
    \end{equation}
    Defining, for all $i\in\Z^d$, $\eta_i = \frac{(\sigma_i + 1)}{2}$, $\eta_i\eta_j$ is an increasing function and $\sigma_i\sigma_j = 4\eta_i\eta_j - \sigma_i - \sigma_j - 1$. To show \eqref{Eq: Lim_correlation_on_wall}, we first prove that
    \begin{equation*}
        \lim_{n\to\infty} \frac{1}{|\calW_n|}\sum_{i\in\calW_n}  \langle \eta_i\eta_{i-e_d} \rangle_{\Delta_n}^+(s) = \langle \eta_0\eta_{-e_d} \rangle^+(s).
    \end{equation*}
    Fix $m\in\mathbb{N}$. For any $n\geq m$
\begin{align*}
    \frac{1}{|\mathcal{W}_n|}\sum_{i\in \mathcal{W}_n} \langle \eta_i\eta_{i-e_d} \rangle_{\Delta_n}^+(s) &= \frac{1}{|\mathcal{W}_n|}\sum_{i\in W_{n-m}} \langle \eta_i\eta_{i-e_d} \rangle_{\Delta_n}^+(s) + \frac{1}{|\mathcal{W}_n|}\sum_{i\in \mathcal{W}_n\setminus W_{n-m}} \langle \eta_i\eta_{i-e_d} \rangle_{\Delta_n}^+(s).
\end{align*}

If $i\in W_{n-m}$ we have that $i+\Lambda_m \subset \Lambda_n$ and by Lemma \ref{decreasing.states} $\langle \eta_i\eta_{i-e_d} \rangle_{\Delta_n}^+(s)\leq \langle \eta_i\eta_{i-e_d} \rangle_{\Delta_m + i}^+(s)= \langle \eta_0\eta_{-e_d} \rangle_{\Delta_m}^+(s)$. Therefore
\begin{equation}
    \frac{1}{|\mathcal{W}_n|}\sum_{i\in W_{n-m}} \langle \eta_i\eta_{i-e_d} \rangle_{\Delta_n}^+(s)  \leq \frac{1}{|W_{n-m}|}\sum_{i\in W_{n-m}} \langle \eta_0\eta_{-e_d} \rangle_{\Delta_m}^+(s) = \langle \eta_0\eta_{-e_d} \rangle_{\Delta_m}^+(s).
\end{equation}

If $i\in \mathcal{W}_n\setminus W_{n-m}$, then $i + \Lambda_m\not\subset \Lambda_n$ and this set intersects the boundary of the wall $$\partial\mathcal{W}_n\coloneqq \{ i\in\mathcal{W}_n : \exists j\in \mathcal{W}_{n+1}\setminus \mathcal{W}_{n} \text{ s.t. }i\sim j  \}.$$ We can bound the number of such vertex by $|\Lambda_m||\partial\Lambda_n|$. Since $|\langle \eta_i\eta_{i-e_d} \rangle_{\Delta_n}^+(s) |\leq 1$, we have
\begin{equation*}
    \frac{1}{|\mathcal{W}_n|}\sum_{i\in \mathcal{W}_n\setminus W_{n-m}}\langle \eta_i\eta_{i-e_d} \rangle_{\Delta_n}^+(s)  \leq \frac{2|\Lambda_m||\partial\mathcal{W}_n|}{|\mathcal{W}_n|}
\end{equation*}
which goes to zero as $n$ increases. Putting both bounds together we get 
\begin{equation*}
    \limsup_{n}  \frac{1}{|\mathcal{W}_n|}\sum_{i\in W_{n}} \langle \eta_i\eta_{i-e_d} \rangle_{\Delta_n}^+(s)  \leq \langle \eta_0\eta_{-e_d} \rangle_{\Delta_m}^+(s) .
\end{equation*}
As $m$ is arbitrary, we can take the limit to get the upper bound in (\ref{lim.to.magnetization.plus}). The lower bound is a direct consequence of the translation invariance and Lemma \ref{decreasing.states} since

\begin{equation*}
    \langle \eta_0\eta_{-e_d} \rangle^+(s)  = \frac{1}{|\mathcal{W}_n|}\sum_{i\in \mathcal{W}_n} \langle \eta_i\eta_{i-e_d} \rangle^+(s)  \leq \frac{1}{|\mathcal{W}_n|}\sum_{i\in \mathcal{W}_n}\langle \eta_i\eta_{i-e_d} \rangle_{\Delta_n}^+(s) ,
\end{equation*}
therefore $\liminf_{n}\frac{1}{|\mathcal{W}_n|}\sum_{i\in \mathcal{W}_n} \langle \eta_i\eta_{i-e_d} \rangle_{\Delta_n}^+(s)  \geq \langle \eta_0\eta_{-e_d} \rangle^+(s)$.
In a completely analogous way, we prove that 
\begin{equation*}
    \lim_{n\to\infty} \frac{1}{|\calW_n|}\sum_{i\in\calW_n}  \langle \sigma_i \rangle_{\Delta_n}^+(s) =
    \langle \sigma_0\rangle^+(s),
\end{equation*}
    what shows \eqref{Eq: Lim_correlation_on_wall}. By the same argument, we can show that 
    \begin{equation}\label{Eq: Lim_h_i_sigma_i}
        \lim_{n\to\infty} \frac{1}{|\calW_n|}\sum_{i\in\Delta_n}  \overline{\rmh}_i\langle \sigma_{i} \rangle_{\Delta_n}^+(s) = \sum_{\ell=1}^{\infty} \rmh_\ell \left(\langle \sigma_{\ell e_d} \rangle^+(s) + \langle \sigma_{-(\ell -1)e_d} \rangle^+(s)\right).
    \end{equation}
    Notice that, since the sequence $ \left(\langle \sigma_{\ell e_d} \rangle^+(s) + \langle \sigma_{-(\ell -1)e_d} \rangle^+(s)\right)_{n\geq 0}$ is bounded by 2, the series in the right hand side of equation above is well defined. Start by noticing that, by our choice of external field, 
    \begin{equation*}
        \sum_{i\in\Delta_n} \overline{\rmh}_i\langle \sigma_{i} \rangle_{\Delta_n}^+(s) = \sum_{i\in\calW_n^\prime}\sum_{\ell=1}^{n}\rmh_\ell \left(\langle \sigma_{i + \ell e_d}  \rangle_{\Delta_n}^+(s) + \langle \sigma_{i-(\ell-1) e_d} \rangle_{\Delta_n}^+(s)\right).
    \end{equation*}
    Fixed $m\in \mathbb{N}$, for any $n\geq m$, we split 
    \begin{equation*}
        \sum_{i\in\calW_n^\prime}\sum_{\ell=1}^{n}\rmh_\ell \langle \sigma_{i + \ell e_d}  \rangle_{\Delta_n}^+(s) = \sum_{i\in\calW_{n-m}^\prime}\sum_{\ell=1}^{n}\rmh_\ell \langle \sigma_{i + \ell e_d}  \rangle_{\Delta_n}^+(s) + \sum_{i\in\calW_n^\prime \setminus \calW_{n-m}^\prime}\sum_{\ell=1}^{n}\rmh_\ell\langle \sigma_{i + \ell e_d}  \rangle_{\Delta_n}^+(s)
    \end{equation*}
    If $i\in W_{n-m}$ we have that $i+\Lambda_m \subset \Lambda_n$ and by Lemma \ref{decreasing.states} $\langle \sigma_{i + \ell e_d}  \rangle_{\Delta_n}^+(s)\leq \langle \sigma_{i + \ell e_d}  \rangle_{\Delta_m + i}^+(s)= \langle \sigma_{\ell e_d}  \rangle_{\Delta_m}^+(s)$. Therefore
    \begin{equation}
    \frac{1}{|\mathcal{W}_n|}\sum_{i\in W_{n-m}^\prime} \sum_{\ell=1}^{n} \rmh_\ell \langle \sigma_{i + \ell e_d} \rangle_{\Delta_n}^+(s)  \leq \frac{1}{|W_{n-m}|}\sum_{i\in W_{n-m}} \sum_{\ell=1}^{n}\rmh_\ell \langle \sigma_{\ell e_d}  \rangle_{\Delta_m}^+(s) = \sum_{\ell=1}^{n}\rmh_\ell \langle \sigma_{\ell e_d}  \rangle_{\Delta_m}^+(s) \leq \sum_{\ell=1}^{\infty}\rmh_\ell \langle \sigma_{\ell e_d}  \rangle_{\Delta_m}^+(s) .
\end{equation}
In the last equation, we used that $\rmh_\ell\geq0$ and $\langle \sigma_{\ell e_d}  \rangle_{\Delta_m}^+(s)$, for all $\ell\geq 1$. If $i\in \mathcal{W}_n\setminus W_{n-m}$, then $i + \Lambda_m\not\subset \Lambda_n$ and this set intersects the boundary of the wall $\partial\mathcal{W}_n.$ We can bound the number of such vertex by $|\Lambda_m||\partial\Lambda_n|$. Since $|\langle \sigma_{i + \ell e_d} \rangle_{\Delta_n}^+(s) |\leq 1$, we have
\begin{equation*}
    \frac{1}{|\mathcal{W}_n|}\sum_{i\in \mathcal{W}_n\setminus W_{n-m}}\sum_{\ell=1}^{n}\langle\sigma_{i + \ell e_d} \rangle_{\Delta_n}^+(s)  \leq \frac{2|\Lambda_m||\partial\mathcal{W}_n|}{|\mathcal{W}_n|}\sum_{\ell=1}^{\infty}h_\ell.
\end{equation*}
which goes to zero as $n$ increases. Putting both bounds together we get 
\begin{equation*}
    \limsup_{n}  \frac{1}{|\mathcal{W}_n|} \sum_{i\in\calW_n^\prime}\sum_{\ell=1}^{n}\rmh_\ell \langle \sigma_{i + \ell e_d}  \rangle_{\Delta_n}^+(s)  \leq  \sum_{\ell=1}^{\infty}\rmh_\ell \langle \sigma_{ \ell e_d}  \rangle_{\Delta_m}^+(s).
\end{equation*}
As $m$ is arbitrary, we can take the limit to get the upper bound in the sum. The lower bound is a direct consequence of the translation invariance and Lemma \ref{decreasing.states} since
\begin{equation*}
   \sum_{\ell=1}^{n}\rmh_\ell \langle \sigma_{ \ell e_d}  \rangle^+(s)  = \frac{1}{|\mathcal{W}_n|}\sum_{i\in \mathcal{W}_n^\prime} \sum_{\ell=1}^{n}\rmh_\ell \langle \sigma_{i + \ell e_d}  \rangle^+(s)  \leq \frac{1}{|\mathcal{W}_n|}\sum_{i\in \mathcal{W}_n^\prime}\sum_{\ell=1}^{n}\rmh_\ell \langle \sigma_{i + \ell e_d}  \rangle_{\Delta_n}^+(s),
\end{equation*}
therefore $\liminf_{n}\frac{1}{|\mathcal{W}_n|}\sum_{i\in\calW_n^\prime}\sum_{\ell=1}^{n}\rmh_\ell \langle \sigma_{i + \ell e_d}  \rangle_{\Delta_n}^+(s)  \geq \sum_{\ell=1}^{\infty}\rmh_\ell \langle \sigma_{ \ell e_d}  \rangle^+(s)$. This proves that 
\begin{equation}\label{Eq: Lim_perpendicular_magnatization}
    \lim_{n\to\infty}  \frac{1}{|\mathcal{W}_n|} \sum_{i\in\calW_n^\prime}\sum_{\ell=1}^{n}\rmh_\ell \langle \sigma_{i + \ell e_d}  \rangle_{\Delta_n}^+(s) =  \sum_{\ell=1}^{\infty}\rmh_\ell \langle \sigma_{ \ell e_d}  \rangle^+(s). 
\end{equation}
By the exact same argument, we can show that 
\begin{equation*}
    \lim_{n\to\infty}\frac{1}{|\calW_n|}\sum_{i\in\calW_n^\prime}\sum_{\ell=1}^{n}\rmh_\ell \ \langle \sigma_{i-(\ell-1) e_d} \rangle_{\Delta_n}^+(s) = \sum_{\ell=1}^{\infty} \rmh_\ell \langle \sigma_{-(\ell -1)e_d} \rangle^+(s),
\end{equation*}
and therefore we have \eqref{Eq: Lim_h_i_sigma_i}. Equations \eqref{Eq: F_as_integral}, \eqref{Eq: Lim_correlation_on_wall} and \eqref{Eq: Lim_h_i_sigma_i}, together with the dominated convergence theorem, yields
\begin{align*}
    F^+(J,\bm{\rmh}) &= \lim_{n\to\infty}-\frac{1}{|\calW_n|} \int_0^1\sum_{i\in\calW_n} J \langle \sigma_i\sigma_{i-e_d} \rangle_{\Delta_n}^+(s) ds +\frac{1}{|\calW_n|}  \int_0^1 \sum_{i\in\Delta_n}\overline{\rmh_i}\langle\sigma_i \rangle_{\Delta_n}^+(s)ds \\
    &= \int_0^1 J\langle \sigma_0\sigma_{-e_d} \rangle^+(s) ds - \sum_{\ell=1}^{\infty} \int_0^1  \rmh_\ell \left(\langle \sigma_{\ell e_d} \rangle^+(s) + \langle \sigma_{-(\ell -1)e_d} \rangle^+(s)\right) ds,
\end{align*}
so $\Tilde{F}^+(J,\bm{\rmh})$ is well defined. The proof that the limit $\Tilde{F}^-(J,\bm{\rmh})$ exists is analogous. 
\end{proof} 

To characterize the phase transition, we proceed as in \cite{FP-II} and use the \textit{wall free energy}, defined as
\begin{equation}
    \Tilde{\tau}_w(J,\bm{\rmh}) \coloneqq \Tilde{F}^-(J,\bm{\rmh}) - \Tilde{F}^+(J,\bm{\rmh}).
\end{equation}
Notice that, when we do not have an external field, $Q_{\Delta_n}^+ = Q_{\Delta_n}^-$. This simplifies the surface tension to 
\begin{equation}\label{tau_tilde_w_wo_ising_partition}
    \Tilde{\tau}_w(J,\bm{\rmh}) = \lim_{n \to \infty} -\frac{1}{|\mathcal{W}_n|}\ln\left[ \frac{\mathcal{Z}^{-}_{n; \bm{\rmh}}}{\mathcal{Z}^{+}_{n; \bm{\rmh}}} \right].
\end{equation}
First we prove that, similarly to \eqref{tau_w.as.integral}, for the external field $\widehat{\bm{\lambda}}$, we can write $ \Tilde{\tau}_w(J,\widehat{\bm{\lambda}})$ in terms of differences of the magnetization. 

\begin{proposition}\label{Prop: Tilde_tau_as_integral}
For $J>0$ and $\lambda\geq0$, the wall free energy can be written as 
\begin{equation}
     \Tilde{\tau}_w(J,\widehat{\bm{\lambda}}) = \int_0^\lambda \sum_{\ell=1}^\infty \frac{1}{\ell^{\delta}}\left( \langle \sigma_{\ell e_d} \rangle^+_{J, \widehat{\bm{s}}} - \langle \sigma_{\ell e_d} \rangle^-_{J, \widehat{\bm{s}}} \right) ds.
\end{equation}
\end{proposition}

\begin{proof}
    Using \eqref{tau_tilde_w_wo_ising_partition}, the wall free energy simplifies to
\begin{equation}
    \tau_w(J,\lambdadec) = \Tilde{F}^-(J,\lambdadec) - \Tilde{F}^+(J,\lambdadec) = \lim_{n \to \infty} -\frac{1}{|\mathcal{W}_n|}\ln\left[ \frac{\mathcal{Z}^{-}_{n; \lambdadec}}{\mathcal{Z}^{+}_{n; \lambdadec}} \right].
\end{equation}
Differentiating each term in the limit w.r.t. $\lambda$ we get
\begin{align*}
    -\partial_\lambda \left( \ln\left[ \frac{\mathcal{Z}_{n;\lambdadec}^{-, J}}{\mathcal{Z}_{n;\lambdadec}^{+, J}} \right] \right) 
        &= \partial_\lambda \left( \ln \mathcal{Z}_{n;\lambdadec}^{+, J} - \ln{\mathcal{Z}_{n;\lambdadec}^{-, J}} \right) \\
        &= \frac{1}{\mathcal{Z}_{n;\lambdadec}^{+, J}} \partial_\lambda \left( \mathcal{Z}_{n;\lambdadec}^{+, J} \right) - \frac{1}{\mathcal{Z}_{n;\lambda}^{-, J}} \partial_\lambda \left( \mathcal{Z}_{n;\lambdadec}^{-, J} \right).
\end{align*}
As 
\begin{equation*}
    \partial_\lambda\left(\mathcal{Z}_{n;\lambdadec}^{+, J}\right)=\sum_{i\in \Lambda_n}\sum_{\sigma\in\Sigma_{\Lambda_n}^{+}} \frac{1}{i_d^\delta}\sigma_i e^{-\mathcal{H}_{\Lambda_n; \lambdadec}^J(\sigma)} = \sum_{i\in \calW_n^\prime}\sum_{\ell=0}^{n} \frac{1}{\ell^\delta}\sigma_{i+\ell e_d} e^{-\mathcal{H}_{\Lambda_n; \lambdadec}^J(\sigma)} ,
\end{equation*}
we conclude that
\begin{equation}\label{d.lambda.of.ln[Z+/Z-].decaying.field}
     -\partial_\lambda \left( \ln\left[ \frac{\mathcal{Z}_{n;\lambdadec}^{-, J}}{\mathcal{Z}_{n;\lambdadec}^{+, J}} \right] \right) = \sum_{i\in \mathcal{W}_n^\prime} \sum_{\ell=1}^{n} \frac{1}{\ell^\delta}\left(\langle \sigma_{i+\ell e_d} \rangle^{+,J}_{n;\lambdadec} - \langle \sigma_{i+\ell e_d} \rangle^{-,J}_{n;\lambdadec}\right).
\end{equation}
All of the above functions are continuous and bounded since they are the logarithm of positive polynomials. Moreover, $\mathcal{Z}^{+}_{n; 0} = \mathcal{Z}^{-}_{n; 0}$. Hence we can write
\begin{equation*}
     \Tilde{\tau}_w(J,\lambdadec) = \lim_{n\to \infty} \frac{1}{|\mathcal{W}_n|}\sum_{i\in \mathcal{W}_n^\prime}\sum_{\ell=1}^{n}\int_0^\lambda \frac{1}{\ell^\delta}\left(\langle \sigma_{i+\ell e_d} \rangle^{+,J}_{n;\widehat{\bm{s}}} - \langle \sigma_{i+\ell e_d} \rangle^{-,J}_{n;\widehat{\bm{s}}}\right)ds.
\end{equation*}
The result follows from the dominated convergence theorem once we note that, for any $0< s$,
\begin{equation*}
    \lim_{n\to \infty} \frac{1}{|\mathcal{W}_n|}\sum_{i\in \mathcal{W}_n^\prime}\sum_{\ell=1}^{n} \frac{1}{\ell^\delta}\langle \sigma_{i+\ell e_d} \rangle^{+,J}_{n;\widehat{\bm{s}}} = \sum_{\ell=1}^\infty  \frac{1}{\ell^\delta}\langle \sigma_{\ell e_d} \rangle^{+,J}_{\widehat{\bm{s}}} 
\end{equation*}
and 
\begin{equation*}
    \lim_{n\to \infty} \frac{1}{|\mathcal{W}_n|}\sum_{i\in \mathcal{W}_n^\prime}\sum_{\ell=1}^{n} \frac{1}{\ell^\delta}\langle \sigma_{i+\ell e_d} \rangle^{-,J}_{n;\widehat{\bm{s}}}= \sum_{\ell=1}^\infty  \frac{1}{\ell^\delta}\langle \sigma_{\ell e_d} \rangle^{-,J}_{\widehat{\bm{s}}}. 
\end{equation*}
These limits are proved in the same steps we proved \eqref{Eq: Lim_perpendicular_magnatization}. 
\end{proof}

This new wall free energy also presents the monotonicity and convexity properties of the previous one. Such properties are described in the next proposition. 

\begin{proposition}\label{Prop: Monotonicity_and_conv_of_tilde_tau}
    For $J>0$, and an external field $\bm{\rmh}$ induced by a positive, summable sequence $\rmh = (\rmh_\ell)_{\ell=1}^\infty$, and $\lambda>0$, we have
    \begin{itemize}
        \item[(a)] $\Tilde{\tau}_w(J,\bm{\rmh})$ is non-decreasing in $J$ and $\rmh_\ell$, for all $\ell\geq 1$; 
        \item[(b)]$\Tilde{\tau}_w(J, \lambdadec)$ is a concave function of $\lambda >0$. 
    \end{itemize}
\end{proposition}

\begin{proof}
    Item \textit{(a)} follows from the representation \eqref{tau_tilde_w_wo_ising_partition} after we differentiate the liming term with respect to the appropriate variable. Differentiating the term in the limit with respect to $J$ we get
\begin{equation*}
    -\partial_J\left(\frac{1}{|\mathcal{W}_n|}\ln\left[ \frac{\mathcal{Z}^{-}_{n; \lambda,\bm{\rmh}}}{\mathcal{Z}^{+}_{n; \lambda,\bm{\rmh}}} \right] \right) = |\mathcal{W}_n|^{-1} \sum_{\substack{i\sim J \\ \{ i,j \} \cap \Lambda_n \neq \emptyset}}   \langle \sigma_i\sigma_j \rangle_{n; \bm{\rmh}}^+ -  \langle \sigma_i\sigma_j \rangle_{n; \bm{\rmh}}^-,
\end{equation*}
that is positive by Proposition \ref{Consequence.of.DVI}, a consequence of the duplicate variables inequalities. Differentiating the same term we respect to $\rmh_\ell$ for a fixed $\ell\geq 1$ we have
\begin{equation*}
    -\partial_{\rmh_\ell}\left(\frac{1}{|\mathcal{W}_n|}\ln\left[ \frac{\mathcal{Z}^{-}_{n; \lambda,\bm{\rmh}}}{\mathcal{Z}^{+}_{n; \lambda,\bm{\rmh}}} \right] \right) = |\mathcal{W}_n|^{-1} \sum_{i\in\calW_n^\prime}   \langle \sigma_{i+\ell e_d} \rangle_{n; \bm{\rmh}}^+ -  \langle \sigma_{i+ \ell e_d} \rangle_{n; \bm{\rmh}}^-,
\end{equation*}
that is positive by Lemma \ref{extremality.of.+.and.-.bc}. 
To prove claim \textit{(b)}, we use a similar reasoning. By equation \eqref{d.lambda.of.ln[Z+/Z-].decaying.field}, we have 
\begin{multline*}
    -\partial^2_\lambda \left( \ln\left[ \frac{Z_{n;\lambda}^{-, J}}{Z_{n;\lambda}^{+, J}} \right] \right) = \sum_{i\in \mathcal{W}_n^\prime} \sum_{\ell=1}^{n} \frac{1}{\ell^\delta}\partial_{\lambda}\left(\langle \sigma_{i+\ell e_d} \rangle^{+,J}_{n;\lambdadec} - \langle \sigma_{i+\ell e_d} \rangle^{-,J}_{n;\lambdadec}\right) \\
    = \sum_{i,j\in \mathcal{W}_n^\prime} \sum_{\ell,k=1}^{n} \frac{1}{\ell ^\delta}\frac{1}{k^\delta}\left(\langle \sigma_{i+\ell e_d} \sigma_{j+k e_d}\rangle^{+,J}_{n;\lambdadec} - \langle \sigma_{i+\ell e_d} \rangle^{+,J}_{n;\lambdadec}\langle \sigma_{j+k e_d}\rangle^{+,J}_{n;\lambdadec}   - \langle \sigma_{i+\ell e_d} \rangle^{-,J}_{n;\lambdadec} + \langle \sigma_{i+\ell e_d} \rangle^{-,J}_{n;\lambdadec}\langle \sigma_{j+k e_d}\rangle^{-,J}_{n;\lambdadec} \right),
\end{multline*}
that is smaller or equal to zero by Proposition \ref{Consequence.of.DVI}. So $\Tilde{\tau}_w$ is the limit of concave functions, and therefore it is concave.
\end{proof}

To relate the wall free energy and the phase-transition or uniqueness, we introduce the critical quantity
\begin{equation*}
    \overline{\lambda}_{c}(J)\coloneqq \inf\{\lambda : \Tilde{\tau}_w(J,\lambdadec) = \max_{s\geq 0}\Tilde{\tau}_w(J,\widehat{\bm{s}})\}.
\end{equation*}
Using Proposition \ref{Prop: Monotonicity_and_conv_of_tilde_tau} we can show that the wall free energy reaches a maximum and therefore $\overline{\lambda}_c$ is finite.  
\begin{lemma}\label{Lemma: Comparing_taus}
    For $J>0$, $\overline{\lambda}_c$ is finite and $\overline{\lambda}_c\leq\lambda_c$. Moreover, $\overline{\lambda}_c>0$ whenever  $J>J_c$. 
\end{lemma}
\begin{proof}
  To prove that $\overline{\lambda}_c\leq\lambda_c$, it is enough to show that for $J>0$ and $\lambda\geq \lambda_c$, 
    \begin{equation*}
    \Tilde{\tau}_w(J,\lambdadec) = \Tilde{\tau}_w(J,\widehat{\bm{\lambda_c}}).    
    \end{equation*}
     Indeed, for all $i,j\in \H+$, $n\in\mathbb{N}$ and positive external field $\bm{h}$, $\langle \sigma_i\rangle_{n;\bm{h}}^+ - \langle \sigma_i\rangle_{n;\bm{h}}^-$ is decreasing in $h_j$, since 
    \begin{equation}\label{Eq: Diff_mag_is_decreasing}
       \partial_{h_j}(\langle \sigma_i\rangle_{n;\bm{h}}^+ - \langle \sigma_i\rangle_{n;\bm{h}}^-) = \langle \sigma_i\sigma_j\rangle_{n;\bm{h}}^+  -\langle \sigma_i\rangle_{n;\bm{h}}^+\langle \sigma_j\rangle_{n;\bm{h}}^+ - \langle \sigma_i\sigma_j\rangle_{n;\bm{h}}^- + \langle \sigma_i\rangle_{n;\bm{h}}^-\langle \sigma_j\rangle_{n;\bm{h}}^- \leq 0
    \end{equation}
    by Proposition \ref{Consequence.of.DVI}. In particular, for any $\lambda\geq 0$, $\langle \sigma_i\rangle_{n;\lambdadec}^+ - \langle \sigma_i\rangle_{n;\lambdadec}^- \leq \langle \sigma_i\rangle_{n;\lambda}^+ - \langle \sigma_i\rangle_{n;\lambda}^-$, and the same inequality holds for the limit states. As $\langle \sigma_i\rangle_{\lambda}^+ - \langle \sigma_i\rangle_{\lambda}^- = 0$ for all $\lambda > \lambda_c$ and $i\in\H+$, using Proposition \ref{Prop: Tilde_tau_as_integral} we conclude that 
    \begin{equation*}
        \Tilde{\tau}_w(J,\lambdadec) = \int_0^\lambda \sum_{\ell=1}^\infty \frac{1}{\ell^{\delta}}\left( \langle \sigma_{\ell e_d} \rangle^+_{J, \widehat{\bm{s}}} - \langle \sigma_{\ell e_d} \rangle^-_{J, \widehat{\bm{s}}} \right) ds = \int_0^{\lambda_c} \sum_{\ell=1}^\infty \frac{1}{\ell^{\delta}}\left( \langle \sigma_{\ell e_d} \rangle^+_{J, \widehat{\bm{s}}} - \langle \sigma_{\ell e_d} \rangle^-_{J, \widehat{\bm{s}}} \right) ds = \Tilde{\tau}_w(J,\widehat{\bm{\lambda_c}}).
    \end{equation*}
  By the monotonicity on the external field, given $\lambda\geq 0$ and taking $\bm{\lambda}_0 \coloneqq \{\lambda \mathbbm{1}_{\{i\in\calW\}}\}$, the external field that is zero outside of $\calW$, we have
     \begin{equation}\label{Eq: tau_smaller_than_tilde_tau}
        \tau_w(J,\lambda) = \Tilde{\tau}_w(J,\bm{\lambda}_0)\leq \Tilde{\tau}_w(J,\lambdadec)
    \end{equation}   
     For $J> J_c$, it was shown in \cite{Lebowitz_Pfister_81} that $\tau(J)>0$. The lower bound \eqref{Eq: Lower_bound_lambda_c} yields $\lambda_c>0$, hence $\tau_w(J,\lambda_c)>0$. Then, inequality \eqref{Eq: tau_smaller_than_tilde_tau} implies that 
\begin{equation*}
    0<\tau_w(J,\lambda_c)\leq \Tilde{\tau}_w(J,\widehat{\bm{\lambda_c}}),
\end{equation*}
and therefore $\overline{\lambda}_c>0$ whenever $J>J_c$.
\end{proof}

We end this section by proving that $\overline{\lambda}_c$ is the critical value for phase transition. 
\begin{proposition}
    For any $0\leq \lambda<\overline{\lambda}_c$, $\langle \cdot \rangle_{J,\lambdadec}^+ \neq  \langle \cdot \rangle_{J,\lambdadec}^-$. And for $\lambda>\overline{\lambda}_c$,  $\langle \cdot \rangle_{J,\lambdadec}^+ =  \langle \cdot \rangle_{J,\lambdadec}^-$.
\end{proposition}
\begin{proof}
    Fixed $0\leq \lambda\leq \overline{\lambda}_c$, lets assume by contradiction that $\langle \cdot \rangle_{J,\lambdadec}^+ =  \langle \cdot \rangle_{J,\lambdadec}^-$. As we argued before, inequality \eqref{Eq: Diff_mag_is_decreasing} shows that the difference $\langle \sigma_i \rangle_{n;J,\lambdadec}^+ - \langle \sigma_i \rangle_{n;J,\lambdadec}^-$ is decreasing in $\lambda$ for all $i\in\H+$. Hence, for every $\lambda^\prime>\lambda$, 
    \begin{equation*}
        \langle \sigma_i \rangle_{n;J,\widehat{\bm{\lambda^\prime}}}^+ - \langle \sigma_i \rangle_{n;J,\widehat{\bm{\lambda^\prime}}}^- \leq \langle \sigma_i \rangle_{n;J,\lambdadec}^+ - \langle \sigma_i \rangle_{n;J,\lambdadec}^- = 0.
    \end{equation*}
By Proposition \ref{Prop: Tilde_tau_as_integral}, this implies that 
\begin{equation*}
     \Tilde{\tau}_w(J,\lambdadec) = \int_0^\lambda \sum_{\ell=1}^\infty \frac{1}{\ell^{\delta}}\left( \langle \sigma_{\ell e_d} \rangle^+_{J, \widehat{\bm{s}}} - \langle \sigma_{\ell e_d} \rangle^-_{J, \widehat{\bm{s}}} \right) ds = \int_0^{\overline{\lambda}_c} \sum_{\ell=1}^\infty \frac{1}{\ell^{\delta}}\left( \langle \sigma_{\ell e_d} \rangle^+_{J, \widehat{\bm{s}}} - \langle \sigma_{\ell e_d} \rangle^-_{J, \widehat{\bm{s}}} \right) ds = \max_{s\geq 0}\Tilde{\tau}_w(J,\widehat{\bm{s}}).
\end{equation*}
In the last equation, we are using that the maximum is reached at $\overline{\lambda}_c$, since all concave functions are continuous. This shows that $\lambda\geq \overline{\lambda}_c$. 

Since $\Tilde{\tau}_w$ is non-decreasing in $\lambda$, for all $\lambda>\overline{\lambda}_c$, $\Tilde{\tau}_w(J,\lambdadec) = \max_{s\geq 0}\Tilde{\tau}_w(J,\widehat{\bm{s}})$. Moreover, it is differentiable in $\lambda$ and is the point-wise limit of the sequence $|\mathcal{W}_n|^{-1}\left(\ln{\mathcal{Z}^{+}_{n; \lambdadec}} - \ln {\mathcal{Z}^{-}_{n; \lambdadec}}\right)$, which is concave. We can then use a known theorem for convex functions, see for example \cite[Theorem B.12 ]{FV-Book}, to conclude that
\begin{align*}
    0 = \partial_\lambda\Tilde{\tau}_w(J,\lambdadec) = \lim_{n\to\infty} \frac{1}{\left|\calW_n\right|}\partial_\lambda\left(\ln{\mathcal{Z}^{+}_{n; \bm{\rmh}}} - \ln {\mathcal{Z}^{-}_{n; \bm{\rmh}}}\right) &= \lim_{n\to\infty} \frac{1}{|\calW_n|}\sum_{i\in \mathcal{W}_n^\prime} \sum_{\ell=1}^{n} \frac{1}{\ell^\delta}\left(\langle \sigma_{i+\ell e_d} \rangle^{+,J}_{n;\lambdadec} - \langle \sigma_{i+\ell e_d} \rangle^{-,J}_{n;\lambdadec}\right) \\
    &= \sum_{\ell=1}^\infty \frac{1}{\ell^{\delta}}\left( \langle \sigma_{\ell e_d} \rangle^+_{J, \lambdadec} - \langle \sigma_{\ell e_d} \rangle^-_{J, \lambdadec} \right).
\end{align*}
By Lemma \ref{extremality.of.+.and.-.bc}, all the terms in the sum are non-negative, therefore $\langle \sigma_{\ell e_d} \rangle^+_{J, \lambdadec}= \langle \sigma_{\ell e_d} \rangle^-_{J, \lambdadec}$ for all $\ell\geq 1$. By translation invariance, we conclude that $\langle \sigma_{i} \rangle^+_{J, \lambdadec} = \langle \sigma_{i} \rangle^-_{J, \lambdadec}$ for all $i\in\H+$, what show uniqueness by Proposition \ref{basta.comparar.magnetizacoes}.
\end{proof}

\subsection{Uniqueness for  \texorpdfstring{$\delta<1$}{d<1}}
In this section we will prove uniqueness for the semi-infinite Ising model with external field $\bm{\lambda}$, for any inverse temperature $\beta>0$ and ferromagnetic interaction. To do this, we will first prove uniqueness for Ising model in $\Z^d$ with external field $\bm{h}^*$ given by \eqref{particular.external.field} and interaction $\bm{J}_{\lambda}=(J_{ij})_{i,j\in\Z^d}$ given by
\begin{equation}\label{Eq: Definition_of_J_Lambda}
    (\bm{J}_{\lambda})_{i,j} = \begin{cases}
                                            \frac{\lambda}{2} &\text{ if }i\in L_0 \text{ and }  j\in L_{-1}\cup L_{1},\\
                                            J &\text{ otherwise.}
                                        \end{cases}
\end{equation}
for $|i-j|=1$. As we are always considering short-range interactions, $J_{ij}=0$ whenever $|i-j|\neq 1$. We will then show how uniqueness for this model implies uniqueness for our model of interest. 


The proof of uniqueness given by \cite{Bissacot_Cass_Cio_Pres_15} together with \cite{Cioletti_Vila_2016} only considers constant interactions. The extension to the interaction $\bm{J_\lambda}$ is a direct consequence of the monotonicity properties of the Random cluster representation, proved first by \cite{Biskup_Borgs_Chayes_Kotecky_00} for constant external fields and extended by \cite{Cioletti_Vila_2016} to more general models. 

\subsubsection{Random Cluster Representation and Edward-Sokal coupling}
    
In this section, following \cite{Biskup_Borgs_Chayes_Kotecky_00} and \cite{Cioletti_Vila_2016}, we define the Random Cluster model (RC), then we introduce the Edward-Sokal (ES) coupling between the RC model and the Ising model. Next present a result showing that uniquiness for the ES model implies uniquiness for the Ising model. We conclude the section introducing some monotonicity properties of the RC model and proving that there is only one RC measure.

In \cite{Biskup_Borgs_Chayes_Kotecky_00} and \cite{Cioletti_Vila_2016}, they consider the Potts model and the General Random Cluster model, so their setting is more general.  We will restrict the results presented here to a particular case of interest.  

\subsubsection*{The Random Cluster model}
    Given $\E = \{ \{i,j\}\subset\Z^d : \ |i-j|=1\}$, $(\Z^d,\E)$ defines a graph. The configuration space of the RC model is $\{\}^\E$. A general configuration will be denoted $\omega$ and called an \textit{edge configuration}. An edge $e\in\E$ is \textit{open} (in a configuration $\omega$) if $\omega_e=1$, and it is \textit{closed} otherwise. A path $(e_0,e_1,\dots,e_n)$ is an \textit{open path} if $\omega_{e_k} = 1$ for all $k=0,\dots,n$. Vertices $i,j\in\Z^d$ are \textit{connected in} $\omega$ if there is an open path $(e_0,\dots,e_n)$ connecting $i$ and $j$, that is, $i\in e_0$ and $j\in e_n$. We denote $x\longleftrightarrow y$ when $x$ and $y$ are connected in $\omega$. The open connected component of $x\in\Z^d$ is $C_x(\omega) = \{\{i,j\}\in\E : x \longleftrightarrow i\}\cup\{x\}$. An arbitrary connected component of $\omega$ is denoted $C(\omega)$. Moreover, for any $E\subset \E$ $\V(E) = \{x\in\Z^d :  x\in e \textit{ for some }e\in E\}$ is the set of vertices touched by $E$. 

    Given $\Lambda\Subset \Z^d$, consider $E(\Lambda) = \{\{i,j\}\in\E : x\in\Lambda \}$ the edges with at least one endpoint in $\Lambda$. Consider also $E_0(\Lambda) = \{e\in \E : e\subset \Lambda\}$, the edges with both endpoints in $\Lambda$. For any $G=(V,E)$ finite sub-graph of $(\Z^d,\E)$, the probability measure of the Random Cluster model in $E\Subset \E$ with ferromagnetic interaction $J = (J_{ij})_{i,j}$, external field $\bm{h}=(h)_{x\in\Z^d}$ and boundary condition $\omega_{E^c}\in\{0,1\}^{E^c}$ is 
    \begin{equation}
    \phi_{E;\bm{J}, \bm{h}}(\omega_E|\omega_{E^c}) 
    = \frac{1}{Z^{RC, \omega_{E^c}}_{E; \bm{J},\bm{H}}}B_{\bm{J}}(\omega_E)\prod_{\substack{C(\omega): C(\omega)\cap \V(E) \neq \emptyset}}(1+ e^{-2\beta\sum_{i\in C(\omega)}h_i}) ,     
    \end{equation}
    where the product is taken over connected open clusters only, with the convention that $e^{-\infty}=0$. The term in the denominator is the usual partition function
    \begin{equation*}
        Z^{RC, \omega_{E^c}}_{E; \bm{J},\bm{h}} = \sum_{\omega_E\in\{0,1\}^E} B_{\bm{J}}(\omega_E)\prod_{\substack{C(\omega): C(\omega)\cap \V(E) \neq \emptyset}}(1+ e^{-2\beta\sum_{i\in C(\omega)}h_i}).
    \end{equation*}
and $B_{\bm{J}}$ is the Bernoulli like factor
    \begin{equation*}
        B_{\bm{J}}(\omega) \coloneqq \prod_{e:\omega_e=1}(e^{2\beta J_e} -1).
    \end{equation*}
This is not a Bernoulli factor since the weights can be bigger than one. Moreover, the interaction of an edge $e=\{i,j\}$ is, as expected, $J_e\coloneqq J_{ij}$. For $i=0,1$,  let $\omega_{E}^{(i)}$ be the configuration satisfying $\omega_e^{(i)}=i$ for all $e\in E^c$. Two particularly important measures are the \textit{RC model with free boundary condition} in $\Lambda\Subset\Z^d$, given by
\begin{equation*}
    \phi_{\Lambda; \bm{J}, \bm{h}}^{0} \coloneqq \phi_{E_0(\Lambda); \beta;\bm{J}, \bm{h}}(\omega_{E_0(\Lambda)}|\omega_{E_0(\Lambda)}^{(0)}),
\end{equation*}
and the \textit{RC model with wired boundary condition} in $\Lambda\Subset\Z^d$, given by
\begin{equation*}
     \phi_{\Lambda; \bm{J}, \bm{h}}^{1} \coloneqq \phi_{E_0(\Lambda); \beta;\bm{J}, \bm{h}}(\omega_{E_0(\Lambda)}|\omega_{E_0(\Lambda)}^{(1)}).
\end{equation*}
The RC model is related to the Ising model through the Edwards-Sokal coupling, introduced next.
\subsubsection{The Edwards-Sokal model}

Given $\Lambda\Subset\Z^d$ and $E\Subset\E$, two configurations $\sigma\in\Omega$, $\omega\in\{0,1\}^\E$ and weights
\begin{equation*}
    \mathcal{W}(\sigma_\Lambda, \omega_E | \sigma_{\Lambda^c}, \omega_{E^c}) = \prod_{\substack{\{i,j\}\in E: \\ \omega_{i,j}=1}} \delta_{\sigma_i,\sigma_j}(e^{2\beta J_{ij}}-1)\prod_{ i\in \Lambda}e^{\beta h_i\sigma_i},
\end{equation*}
the Edwards-Sokal (ES) measure in $\Lambda\Subset\Z^d$ and $E\Subset \E$ is given by
\begin{equation*}
    \phi^{ES}_{\Lambda, E; \bm{J}, \bm{h}}(\sigma_\Lambda, \omega_E | \sigma_{\Lambda^c}, \omega_{E^c}) \coloneqq \frac{\mathcal{W}(\sigma_\Lambda, \omega_E | \sigma_{\Lambda^c}, \omega_{E^c})}{Z^{ES}_{\Lambda, E; \bm{J}, \bm{h}}(\sigma_{\Lambda^c}, \omega_{E^c})},
\end{equation*}
with 
\begin{equation*}
    Z^{ES}_{\Lambda, E; \bm{J}, \bm{h}}(\sigma_{\Lambda^c}, \omega_{E^c}) \coloneqq \sum_{\substack{\eta_\Lambda\in\Omega_\Lambda \\ \xi_E\in \{0,1\}^E}}\mathcal{W}(\eta_\Lambda, \xi_E | \sigma_{\Lambda^c}, \omega_{E^c}).
\end{equation*}
If $Z^{ES}_{\Lambda, E; \bm{J}, \bm{h}}(\sigma_{\Lambda^c}, \omega_{E^c}) = 0$, we simply take $\phi^{ES}_{\Lambda, E; \bm{J}, \bm{h}}(\cdot | \sigma_{\Lambda^c}, \omega_{E^c}) = 0$. To simplify the notation, as we are considering arbitrary interactions and external fields, we will omit them from the notation. We also highlight two particularly important ES-measures, the \textit{ES-measure with free boundary condition} in $\Lambda\Subset \Z^d$, given by
\begin{equation*}
    \phi^{ES, 0}_{\Lambda; \beta, \bm{J}, \bm{h}}(\cdot) \coloneqq  \phi^{ES}_{\Lambda, E_0(\Lambda)}(\cdot | \sigma_{\Lambda^c}, \omega^{(0)}_{E_0(\Lambda)}),
\end{equation*}
 and the \textit{ES-measure with wired boundary condition} in $\Lambda\Subset \Z^d$, given by
 \begin{equation*}
     \phi^{ES, 1}_{\Lambda; \beta, \bm{J}, \bm{h}}(\cdot) \coloneqq  \phi^{ES}_{\Lambda, E(\Lambda)}(\cdot | \sigma_{\Lambda^c}^+, \omega^{(1)}_{E(\Lambda)}).
 \end{equation*}
 
\begin{remark}
    The measure $\phi^{ES, 0}_{\Lambda; \beta, \bm{J}, \bm{h}}(\cdot)$ does not depend on the choice of $\sigma_{\Lambda^c}$. Moreover, the measures $\phi^{ES}_{\Lambda, E(\Lambda)}(\cdot|\sigma_{\Lambda^c}, \omega_{E^{c}(\Lambda)})$ do not depend on the choice of configuration $\omega_{E^{c}(\Lambda)}$. We choose to keep it  in the notation since, later on, we will want to see $\phi^{ES}_{\Lambda, E(\Lambda)}$ as an specification. 
\end{remark}

\begin{remark}
    To simplify the notation, we will omit the dependency of the RC and ES measures on $\beta$, $\bm{J}$ and $\bm{h}$. They will appear again only on results concerning a specific interaction or external field. 
\end{remark}
The following two lemmas guarantee that the ES model is indeed a coupling between the Ising and the RC model. 
\begin{lemma}[Spin Marginals]
    Given $\Lambda \Subset \Z^d$ and $f:\Omega \longrightarrow \mathbb{R}$ with $\supp(f)\subset \Lambda$, 
    \begin{equation*}
        \phi^{ES}_{\Lambda, E(\Lambda)}(f | \sigma_{\Lambda^c}, \omega_{E(\Lambda)^c}) = \mu_{\Lambda; \bm{J}, \bm{h}}^{\sigma}(f).  
    \end{equation*}
\end{lemma}
\begin{proof}
    We first write the Boltzmann factor as 
    \begin{align*}
        e^{-\beta \mathcal{H}_{\Lambda; \bm{h}}^{\bm{J}}(\sigma) } &= \prod_{\{i,j\}\in E(\Lambda)}e^{\beta J_{i,j}\sigma_i\sigma_j}\prod_{i\in\Lambda}e^{\beta h_i\sigma_i}\\
        &=\prod_{\{i,j\}\in E(\Lambda)}e^{\beta J_{i,j}(2\delta_{\sigma_i,\sigma_j} - 1)}\prod_{i\in\Lambda}e^{\beta h_i\sigma_i}\\
        &= \prod_{\{i,j\}\in E(\Lambda)} e^{-\beta J_{i,j}}\prod_{\{i,j\}\in E(\Lambda)}e^{2\beta J_{i,j}\delta_{\sigma_i,\sigma_j}}\prod_{i\in\Lambda}e^{\beta h_i\sigma_i}.
    \end{align*}
    Moreover, writing $e^{2\beta J_{i,j}\delta_{\sigma_i,\sigma_j}} = 1 + (e^{\beta J_{i,j}} - 1)\delta_{\sigma_i,\sigma_j}$, we have
    \begin{align*}
        e^{-\beta \mathcal{H}_{\Lambda; \bm{h}}^{\bm{J}}(\sigma) } &= \prod_{\{i,j\}\in E(\Lambda)} e^{-\beta J_{i,j}}\prod_{\{i,j\}\in E(\Lambda)}\left(  1 + (e^{\beta J_{i,j}} - 1)\delta_{\sigma_i,\sigma_j}\right)\prod_{i\in\Lambda}e^{\beta h_i\sigma_i}\\
        &= \prod_{\{i,j\}\in E(\Lambda)} e^{-\beta J_{i,j}}\sum_{\omega\in \{0,1\}^{E(\Lambda)}}\prod_{\substack{\{i,j\}\in E(\Lambda): \\ \omega_{i,j}=1}}(e^{\beta J_{i,j}} - 1)\delta_{\sigma_i,\sigma_j}\prod_{i\in\Lambda}e^{\beta h_i\sigma_i}\\
        &= \prod_{\{i,j\}\in E(\Lambda)} e^{-\beta J_{i,j}}\sum_{\omega\in \{0,1\}^{E(\Lambda)}}\mathcal{W}(\sigma_\Lambda, \omega_E | \sigma_{\Lambda^c}, \omega_{E^c}).
    \end{align*}
    Multiplying by the normalizing factors, we get
    \begin{equation*}
        \mu_{\Lambda; \bm{J}, \bm{h}}^{\sigma}(f) =  \frac{Z^{ES}_{\Lambda, E; \bm{J}, \bm{h}}(\sigma_{\Lambda^c}, \omega_{E^c})}{\mathcal{Z}^{\eta, \bm{J}}_{\Lambda; \lambda,\bm{h}}} \prod_{\{i,j\}\in E(\Lambda)} e^{-\beta J_{i,j}}\phi^{ES}_{\Lambda, E(\Lambda)}(f | \sigma_{\Lambda^c}, \omega_{E(\Lambda)^c}). 
    \end{equation*}
    Applying this for $f\equiv 1$, we conclude that $\frac{Z^{ES}_{\Lambda, E; \bm{J}, \bm{h}}(\sigma_{\Lambda^c}, \omega_{E^c})}{\mathcal{Z}^{\eta, \bm{J}}_{\Lambda; \lambda,\bm{h}}} \prod_{\{i,j\}\in E(\Lambda)} =1$, and the lemma follows. 
\end{proof}

\begin{lemma}[RC Marginals]
    Given $f,g : \{0,1\}^{\E}\longrightarrow \mathbb{R}$ and $\Lambda\Subset\Z^d$ with $\supp(f)\subset E_0(\Lambda)$ and $\supp(g)\subset E(\Lambda)$,
    \begin{equation}\label{Eq: ES-RC marginal - Free b.c.}
        \phi^{ES, 0}_{\Lambda}(f) = \phi_{\Lambda}^0(f)
    \end{equation}
    and 
    \begin{equation}\label{Eq: ES-RC marginal - Wired b.c.}
        \phi^{ES, 1}_{\Lambda}(g) = \phi_{\Lambda}^1(g)
    \end{equation}
\end{lemma}
\begin{proof}
    Given any $\omega\in\{0,1\}^E$, let $\Omega_\Lambda(\omega) \coloneqq \{\sigma\in\Omega_\Lambda : \forall C(\omega) \text{ open cluster and }i,j\in C(\omega)\cap \Lambda, \sigma_i=\sigma_j \}$ be the configurations that are constant in the open clusters. Then, 
    \begin{align*}
         Z^{ES}_{\Lambda, E}(\sigma_{\Lambda^c}, \omega^{(0)}_{E})\phi^{ES, 0}_{\Lambda}(f) &= \sum_{\omega\in \{0,1\}^{E_0(\Lambda)}}\sum_{\sigma\in\Omega_\Lambda}f(\omega)\prod_{\substack{\{i,j\}\in E_0(\Lambda): \\ \omega_{i,j}=1}} \delta_{\sigma_i,\sigma_j}(e^{2\beta J_{ij}}-1)\prod_{ i\in \Lambda}e^{\beta h_i\sigma_i}\\
         &= \sum_{\omega\in \{0,1\}^{E_0(\Lambda)}}\sum_{\sigma\in\Omega_\Lambda(\omega)}f(\omega)\prod_{\substack{\{i,j\}\in E_0(\Lambda): \\ \omega_{i,j}=1}} (e^{2\beta J_{ij}}-1)\prod_{ i\in \Lambda}e^{\beta h_i\sigma_i} \\
         &= \sum_{\omega\in \{0,1\}^{E_0(\Lambda)}}f(\omega)B_{\bm{J}}(\omega_{E_0(\Lambda)}) \sum_{\sigma\in\Omega_\Lambda(\omega)}\prod_{ i\in \Lambda}e^{\beta h_i\sigma_i}. 
    \end{align*}
As the sum above is only over configurations that are constant in the open clusters, we have
\begin{align*}
    \sum_{\sigma\in\Omega_\Lambda(\omega)}\prod_{ i\in \Lambda}e^{\beta h_i\sigma_i} &= \sum_{\sigma\in\Omega_\Lambda(\omega)}\prod_{C(\omega)\subset \Lambda}e^{\beta \sum_{i\in C(\omega)}h_i\sigma_i} \\
    &=e^{\beta \sum_{i\in \Lambda}h_i}\prod_{C(\omega)\subset \Lambda}(e^{\beta \sum_{i\in C(\omega)}h_i} + e^{ - \beta \sum_{i\in C(\omega)}h_i}) = \prod_{C(\omega)\subset \Lambda}(1 + e^{ - 2\beta \sum_{i\in C(\omega)}h_i}).
\end{align*}
This shows that 
\begin{align*}
     Z^{ES}_{\Lambda, E}(\sigma_{\Lambda^c}, \omega^{(0)}_{E})\phi^{ES, 0}_{\Lambda}(f) &= e^{\beta \sum_{i\in \Lambda}h_i} \sum_{\omega\in \{0,1\}^{E_0(\Lambda)}}f(\omega)B_{\bm{J}}(\omega_{E_0(\Lambda)}) \prod_{C(\omega)\subset \Lambda}(1 + e^{ - 2\beta \sum_{i\in C(\omega)}h_i}) \\
     &= e^{\beta \sum_{i\in \Lambda}h_i}Z^{RC, \omega_{E_0(\Lambda)}^{(0)}}_{E_0(\Lambda)}\phi_\Lambda^0(f).  
\end{align*}
We get equation \eqref{Eq: ES-RC marginal - Free b.c.} by noticing that $\phi^{ES, 0}_{\Lambda}(1) = \phi_\Lambda^0(1) = 1$, hence $Z^{RC, \omega_{E_0(\Lambda)}^{(0)}}_{E_0(\Lambda)}e^{\beta \sum_{i\in \Lambda}h_i}\left( Z^{ES}_{\Lambda, E}(\sigma_{\Lambda^c}, \omega^{(0)}_{E})\right)^{-1} = 1$. For the other equation, we take $\Omega_\Lambda^+(\omega) \coloneqq \Omega_\Lambda(\omega) \cap \Omega_\Lambda^+$. This is the set of configurations with constant configurations in the clusters, with the restriction that  clusters connecting $\Lambda$ and $\Lambda^c$ must have sign $+$. Proceeding in the same steps as before, we can write
\begin{equation*}
       Z^{ES}_{\Lambda, E}(\sigma_{\Lambda^c}^+, \omega^{(1)}_{E})\phi^{ES, 1}_{\Lambda}(g) =  \sum_{\omega\in \{0,1\}^{E(\Lambda)}}g(\omega)B_{\bm{J}}(\omega_{E(\Lambda)}) \sum_{\sigma\in\Omega_\Lambda^+(\omega)}\prod_{ i\in \Lambda}e^{\beta h_i\sigma_i}. 
\end{equation*}
As we are considering wired boundary conditions, we can write
\begin{align*}
    \sum_{\sigma\in\Omega_\Lambda^+(\omega)}\prod_{ i\in \Lambda}e^{\beta h_i\sigma_i} &= \sum_{\sigma\in\Omega_\Lambda^+(\omega)}\prod_{ C(\omega): C(\omega)\subset \Lambda}e^{\beta \sum_{i\in C(\omega)}h_i\sigma_i}\prod_{ C(\omega): C(\omega)\cap \Lambda^c\neq \emptyset}e^{\beta \sum_{i\in C(\omega)\cap \Lambda}h_i} \\
    &= \prod_{ C(\omega): C(\omega)\subset \Lambda}\left( e^{\beta \sum_{i\in C(\omega)}h_i} + e^{-\beta \sum_{i\in C(\omega)}h_i}\right)\prod_{ C(\omega): C(\omega)\cap \Lambda^c\neq \emptyset}e^{\beta \sum_{i\in C(\omega)\cap \Lambda}h_i} \\
    &= e^{\beta \sum_{i\in\Lambda}h_i}\prod_{ C(\omega): C(\omega)\subset \Lambda}\left( 1 + e^{-2\beta \sum_{i\in C(\omega)}h_i}\right)\\
    &= e^{\beta \sum_{i\in\Lambda}h_i}\prod_{ C(\omega): C(\omega)\cap \Lambda\neq \emptyset}\left( 1 + e^{-2\beta \sum_{i\in C(\omega)}h_i}\right).
\end{align*}
This shows that $ Z^{ES}_{\Lambda, E}(\sigma_{\Lambda^c}^+, \omega^{(1)}_{E})\phi^{ES, 1}_{\Lambda}(f) =  e^{\beta \sum_{i\in\Lambda}h_i} Z^{RC, \omega_{E(\Lambda)}^{(1)}}_{E(\Lambda)}\phi_\Lambda^1(g)$. Again, this proves equation \eqref{Eq: ES-RC marginal - Wired b.c.} once we take $g=1$ to conclude that $\left(Z^{ES}_{\Lambda, E}(\sigma_{\Lambda^c}^+, \omega^{(1)}_{E})\right)^{-1} e^{\beta \sum_{i\in\Lambda}h_i} Z^{RC, \omega_{E(\Lambda)}^{(1)}}_{E(\Lambda)} = 1$.
\end{proof}

To define the infinity volume measures, we use the DLR equations. Let $\mathcal{F}_1$ be the cylinders $\sigma$-algebra of $\{0,1\}^\E$, and $\mathcal{F}_2$ the cylinders $\sigma$-algebra of $\Omega\times \{0,1\}^\E$. We take $\mathcal{P}(\{0,1\}^\E)$ the set of probability measures in $(\{0,1\}^\E, \mathcal{F}_1)$ and $\mathcal{P}(\Omega\times\{0,1\}^\E)$ the set of probability measures in $(\Omega\times\{0,1\}^\E, \mathcal{F}_2)$. The set of RC measures is 
\begin{equation*}
    \mathcal{G}^{RC}_{\beta, \bm{J},\bm{h}} \coloneqq \left\{\phi\in \mathcal{P}(\{0,1\}^\E) : \phi(f) = \int \phi_E(f|\omega_E^c)\}d\phi(\omega), \text{ whenever } \supp(f) \subset E \text{ and } E\Subset \E\right\}. 
\end{equation*}
Analogously, the set of ES measures is 
\begin{multline}
    \mathcal{G}^{ES}_{\beta, \bm{J}, \bm{h}} \coloneqq \left\{ \nu\in \mathcal{P}(\Omega\times\{0,1\}^\E) :  \nu(f) = \int \phi^{ES}_{\Lambda, E(\Lambda)}(f|\sigma_\Lambda^c, \omega_{E(\Lambda)^c})d\nu(\sigma, \omega), \right. \\ \left. \text{ whenever } \supp(f) \subset \Lambda\times E(\Lambda) \text{ and } \Lambda\Subset \Z^d \right\}.
\end{multline}
We will often omit the parameters $\beta$ and $\bm{h}$ in statements that hold for an arbitrary choice of them. So $ \mathcal{G}^{RC}_{\bm{J}}$ denotes $ \mathcal{G}^{RC}_{\beta, \bm{J}, \bm{h}}$ and $ \mathcal{G}^{ES}_{\bm{J}}$ denotes $\mathcal{G}^{ES}_{\beta, \bm{J}, \bm{h}}$. 
\begin{remark}
    Since the families $\{\phi_E\}_{E\Subset\E}$ and $\{\phi^{ES}_{\Lambda, E(\Lambda)}\}_{\Lambda\Subset\Z^d}$ are specifications, the sets $\mathcal{G}^{RC}_{\bm{J}}$ and $\mathcal{G}^{ES}_{\bm{J}}$ are the usual set of DLR Gibbs measures. 
\end{remark}

At first, it is unclear if the spin marginal of an infinity ES measure is a spin Gibbs measure in $\mathcal{G}_{\bm{J}}$. In fact, an even stronger statement holds. The following theorem was proved in \cite{Biskup_Borgs_Chayes_Kotecky_00} and extended to general external fields in \cite{Cioletti_Vila_2016}.

\begin{theorem}\label{Theo: ES_to_Ising_isomosphism}
Let $\Pi_S: \mathcal{G}^{ES}_{\bm{J}}: \longrightarrow \mathcal{G}^{IS}_{\bm{J}}$ be the application that takes an ES - measure to its spin marginal, that is, for any $\nu\in\mathcal{G}^{ES}_{\bm{J}}$ and $f:\Omega\longrightarrow \mathbb{R}$ with $\supp(f)\Subset \Z^d$, 
\begin{equation*}
    \Pi_S(\nu)(f) \coloneqq \int f(\sigma) d\nu(\sigma,\omega).
\end{equation*}
Then, $\Pi_S$ is a linear isomorphism. In particular, $|\mathcal{G}^{ES}_{\bm{J}}| = 1$ if and only if $|\mathcal{G}_{\bm{J}}^{IS}|=1$. 
\end{theorem}

This shows that uniqueness for the ES model implies uniqueness for the Ising model. It is left to relate the uniqueness of the RC model with the uniqueness of the ES model. To do so, we use the FKG property, and some consequences of it, of the RC and ES models. The main contribution of \cite{Cioletti_Vila_2016} was the extension of these properties from the models with constant external fields to models with non-constant external fields. These results are described next.

As we did for the configuration space, we can consider a partial order on $\{0,1\}^\E$ defining $\omega\leq \omega^\prime$ when $\omega_e\leq\omega^\prime_e$ for all $e\in\E$. The first key property of the RC model is that it satisfies the so-called \textit{strong FKG}.

\begin{theorem}[Strong FKG]
    Given $E \subset E^\prime \subset \E$ and $\xi\in\{0,1\}^\E$, take $\Upsilon_{E^\prime\setminus E}^{\xi} \coloneqq \{\omega\in\{0,1\}^\E : \omega_e=\xi_e \ \forall e\in E^\prime\setminus E\}$. Then, for any $\Lambda\Subset \Z^d$ and non-decreasing functions $f$ and $g$,
    \begin{equation*}
        \phi_\Lambda^0(f.g|\Upsilon_{E^\prime\setminus E}^{\xi}) \geq  \phi_\Lambda^0(f|\Upsilon_{E^\prime\setminus E}^{\xi})  \phi_\Lambda^0(g|\Upsilon_{E^\prime\setminus E}^{\xi}) 
    \end{equation*}
    whenever $ \phi_\Lambda^0(\Upsilon_{E^\prime\setminus E}^{\xi}) >0$. The same result holds for $\phi_\Lambda^1$.
\end{theorem}

\begin{remark}
    Choosing $E=E_0(\Lambda)$, $E^\prime = E_0(\Lambda)^c$ and $\xi = \omega^{(0)}$ in the definition above, we get $  \phi_\Lambda^0(f.g) \geq  \phi_\Lambda^0(f)  \phi_\Lambda^0(g)$ for any non-decreasing functions $f$ and $g$. Similarly, taking $E=E(\Lambda)$, $E^\prime = E(\Lambda)^c$ and $\xi = \omega^{(1)}$ we conclude that $  \phi_\Lambda^1(f.g) \geq  \phi_\Lambda^1(f)  \phi_\Lambda^1(g)$. This resembles the usual FKG property for spin systems \eqref{Eq: FKG_Inequality_Spins}. 
\end{remark}
Two consequences of the FKG property are particularly important for us. One of them is the existence and extremality of the limit measures with free and wired boundary conditions. 

\begin{theorem}\label{Theo: Limmiting_states_RC_and_ES}
    Let $\beta\geq 0$, $\bm{J}=\{J_{i,j}\}_{i,j\in\Z^d}$ be any ferromagnetic nearest neighbor interaction, and $\bm{h} = \{h_i\}_{i\in\Z^d}$ be a non-negative external field. Then, for any $f$ and $g$ quasi-local function, 
    \begin{itemize}
        \item[(I)] The limits
                \begin{align*}
                      \phi^0(f) \coloneqq   \lim_{\Lambda\nearrow\Z^d}\phi_\Lambda^0(f) \qquad \text{and} \qquad  \phi^1(f) \coloneqq   \lim_{\Lambda\nearrow\Z^d}\phi_\Lambda^1(f)
                \end{align*}
                exists.
        \item[(II)] The limits
                \begin{align*}
                     \phi^{ES, 0}(g)\coloneqq \lim_{\Lambda\nearrow\Z^d} \phi^{ES, 0}_{\Lambda}(g) \qquad \text{and} \qquad   \phi^{ES, 1}(g)\coloneqq \lim_{\Lambda\nearrow\Z^d} \phi^{ES, 1}_{\Lambda}(g)
                \end{align*}
                exists.
        \item[(III)]For any $\phi\in\mathcal{G}^{RC}$, if  $f$ is non-decreasing then
                \begin{equation}\label{Eq: Extremality_of_wired_and_free}
                    \phi^0(f) \leq \phi(f)\leq \phi^1(f).
                \end{equation}
    \end{itemize}
    All limits are taken over sequences invading $\Z^d$.
\end{theorem}
 The next result allows us to compare the models with interaction $\bm{J}_\lambda$ and constant interaction $\bm{J}\equiv J$. The proof is a straightforward adaptation of \cite[Theorem 7]{Cioletti_Vila_2016}.

\begin{proposition}\label{Prop: RC_is_increasing_in_J}
   Let $\bm{J}=\{J_{i,j}\}_{i,j\in \Z^d}$ and $\bm{J}^\prime=\{J^\prime_{i,j}\}_{i,j\in \Z^d}$ be nearest-neighbor interactions with $0\leq J_{i,j}\leq J^\prime_{i,j}$, for all $i,j\in\Z^d$. Then, for any $\Lambda\Subset\Z^d$ and $f$ local non-decreasing function,
   \begin{equation*}
       \phi^0_{\Lambda; \bm{J}}(f)\leq \phi^0_{\Lambda; \bm{J}^\prime}(f) \qquad \text{ and } \qquad   \phi^1_{\Lambda; \bm{J}}(f)\leq \phi^1_{\Lambda; \bm{J}^\prime}(f).
   \end{equation*}
\end{proposition}
\begin{proof}
    Consider a function $g:\{0,1\}^{E(\Lambda)}\longrightarrow \mathbb{R}$ given by
    \begin{equation*}
    g(\omega)= \prod_{e\in E(\Lambda)}\left(\frac{e^{2\beta J_e}-1}{e^{2\beta J^\prime_e} - 1}\right)^{\omega_e}.
    \end{equation*}
    By the restriction on $\bm{J}$ and $\bm{J}^\prime$, the all the fractions above are at most $1$, so $g$ is non-increasing. Given a non-decreasing local function $f$, 
    \begin{equation*}
        \phi^0_{\Lambda;\bm{J}}(f) = \frac{1}{Z^{0}_{\Lambda;\bm{J}}}\sum_{\omega\in\{0,1\}^{E_0(\Lambda)}}f(\omega)g(\omega)\prod_{e:\omega_e=1}\left(e^{2\beta J_e^\prime}-1\right)\prod_{C(\omega)}\left(1+e^{-2\beta\sum_{i\in C(\omega)}h_i}\right) =  \frac{Z^{0}_{\Lambda;\bm{J}^\prime}}{Z^{0}_{\Lambda;\bm{J}}}\phi^0_{\Lambda;\bm{J}^\prime}(f.g).
    \end{equation*}
    Taking, in particular, $f\equiv 1$, we get $\phi^0_{\Lambda;\bm{J}^\prime}(g) = \frac{Z^{0}_{\Lambda;\bm{J}}}{Z^{0}_{\Lambda;\bm{J}^\prime}}$. Using the FKG property, we conclude that 
    \begin{equation*}
        \phi^0_{\Lambda;\bm{J}}(f)  = \frac{\phi^0_{\Lambda;\bm{J}^\prime}(f.g)}{\phi^0_{\Lambda;\bm{J}^\prime}(g)}\leq \phi^0_{\Lambda;\bm{J}^\prime}(f).
    \end{equation*}
This exact same argument can be done for the wired boundary condition, what concludes the proof.
\end{proof}

To guarantee uniqueness for the RC model, we can use the quantity
\begin{equation*}
    P_\infty(\beta, \bm{J}, \bm{h}) \coloneqq \sup_{x\in\Z^d}\sup_{\phi\in \mathcal{G}^{RC}}\phi\left(|C_x| = +\infty \right). 
\end{equation*}
This next theorem was proved in \cite{Cioletti_Vila_2016}.
\begin{theorem}\label{Theo: Uniqueness_with_P_infinity}
    For any $\beta\geq 0$, ferromagnetic nearest-neighbor interacion $\bm{J}=\{J_{i,j}\}_{i,j\in\Z^d}$ and non-negative external field $\bm{h}=(h_{i})_{i\in\Z^d}$, if $P_\infty(\beta, \bm{J}, \bm{h})=0$, then $\left|\mathcal{G}^{ES}\right|=\left|\mathcal{G}^{RC}\right|=1$
\end{theorem}

\subsubsection{Uniqueness for $\delta$<1}
To prove uniqueness for the semi-infinite Ising model, we first prove uniqueness for the usual Ising model with interaction given by \eqref{Eq: Definition_of_J_Lambda}. To do so, we use the RC and ES models presented in the previous section. 
\begin{theorem}\label{Theo: Uniquiness_Ising_J_lambda}
    The Ising model in $\Z^d$ with interaction $\bm{J}_\lambda$ defined in \eqref{Eq: Definition_of_J_Lambda}, $\lambda>0$ and external field $\bm{h}^*=(h_i^*)_{i\in\Z^d}$ given by \eqref{particular.external.field} has a unique state, for any inverse temperature $\beta>0$. 
\end{theorem}
\begin{proof}
    Fixed $\beta>0$, by Theorem \ref{Theo: ES_to_Ising_isomosphism}, it is enough to show that $\left| \mathcal{G}^{ES}_{\beta, \bm{J}_\lambda, \bm{h}^*}\right|=1$. It was shown in \cite{Cioletti_Vila_2016} that, for any constant nearest-neighbor ferromagnetic interaction $\bm{J}\equiv J$, $P_\infty(\beta, \bm{J}, \bm{h}^*)=0$ and, in particular, $\phi^1_{\bm{J}, \bm{h}^*}\left( |C_x| = +\infty \right) = 0$ for all $x\in \Z^d$. 
    For any $x\in\Z^d$, the function $\mathbbm{1}_{\{|C_x| = +\infty\}}$ is increasing. Then, Proposition \ref{Prop: RC_is_increasing_in_J} yields
    \begin{equation*}
        \phi^1_{\Lambda; \bm{J}_\lambda, \bm{h}^*}\left( |C_x| = +\infty \right) \leq    \phi^1_{\Lambda; \bm{J}, \bm{h}^*}\left( |C_x| = +\infty \right)
    \end{equation*}
    for any $\Lambda\Subset\Z^d$ and $x\in\Z^d$. By Theorem \ref{Theo: Limmiting_states_RC_and_ES}, we can take the limit $\Lambda\nearrow\Z^d$ to get $ \phi^1_{\bm{J}_\lambda, \bm{h}^*}\left( |C_x| = +\infty \right) =0$ for all $x\in\Z^d$. Since $\phi^1_{\bm{J}_\lambda, \bm{h}^*}$ is extremal, in the sense of \eqref{Eq: Extremality_of_wired_and_free}, we conclude that $P_\infty(\beta, \bm{J}_\lambda, \bm{h}^*)=0$, and therefore we have $\left| \mathcal{G}^{ES}_{\beta, \bm{J}_\lambda, \bm{h}^*}\right|=1$ by Theorem \ref{Theo: Uniqueness_with_P_infinity}. 
\end{proof}

Now we prove the main result of this section. We prove uniqueness for the semi-infinite Ising with external field given by \eqref{particular.external.field} and $\delta<1$ at any temperature, by comparing it with the model of Theorem \ref{Theo: Uniquiness_Ising_J_lambda}.
\begin{theorem}
    The semi-infinite Ising model with interaction $J>0$ and external field $\bm{\lambda}=(\lambda_i)_i\in\H+$ with $\lambda_i=\frac{\lambda}{i_d^\delta}$
for all $i\in\H+$ and $\delta<1$ has a unique Gibbs state. 
\end{theorem}

\begin{proof}
        Proposition \ref{basta.comparar.magnetizacoes} guarantees that it is enough to prove $\langle \sigma_{e_d} \rangle^{+}_{\bm{\lambda}} = \langle \sigma_{e_d} \rangle^{-}_{\bm{\lambda}}$, where $e_d=(0,\dots,0,1)$ is a base vector of $\H+$. By spin-flip symmetry, we can assume without loss of generality that $\lambda\geq 0$. Split $\mathbb{Z}^d$ in layers $L_k \coloneqq \mathbb{Z}^{d-1}\times\{k\}$, with $k\in\mathbb{Z}$. For any $\Lambda_n=[-n,n]^{d-1}\times[1, n]$, we rewrite the semi-infinite model as the usual Ising model but now with interaction $\bm{J}_{\lambda}$ and external field $\bm{\lambda}_0$ given by 
    $$(\bm{\lambda}_0)_i = \begin{cases} 
                        \lambda/2 &\text{ if } i\in L_1,\\
                        \lambda\mid i_d \mid^{-\delta}, &\text{ if } i\in L_k, k>1,\\
                        0 &\text{ otherwise}, \\
                        \end{cases}$$ 
    for any $i,j\in\mathbb{Z}^d$. Hence, 
    \begin{align}
        &\langle \sigma_{e_d} \rangle^{+}_{\Lambda_n; \bm{\lambda}} = \langle \sigma_{e_d} \rangle^{+, \bm{J}_{\lambda}}_{\Lambda_n; \bm{\lambda}_0} &\text{ and } &&\langle \sigma_{e_d} \rangle^{-}_{\Lambda_n; \bm{\lambda}} = \langle \sigma_{e_d} \rangle^{\mp, \bm{J}_{\lambda}}_{\Lambda_n; \bm{\lambda}_0}
    \end{align}
where $(\mp)_i = \mathbbm{1}_{\{i\in\mathbb{Z}^d\setminus \H+\}} -  \mathbbm{1}_{\{i\in\H+\}}$. Consider $\bm{\lambda}_0^\prime$ an extension of $\bm{\lambda}_0$ to $\Z^d$ given by $(\lambda_0^\prime)_{i} \coloneqq (\lambda_0)_{i}$ when $i\in \H+$ and $(\lambda_0^\prime)_{i} \coloneqq (\lambda_0)_{i^\prime}$ when $i\in \mathbb{Z}^d\setminus\H+$, where $i^\prime = (i_1,\dots,i_{d-1}, -i_d)$. Since the boxes $\Lambda_n$ and $\Lambda_n^\prime \coloneqq [-n,n]^{d-1}\times [-n, -1]$ are not connected, taking $\Delta_n^\prime\coloneqq \Lambda_n \cup \Lambda_n^\prime$ we have 

\begin{align}
    &\langle \sigma_{e_d} \rangle^{+, \bm{J}_{\lambda}}_{\Lambda_n; \bm{\lambda}_0} = \langle \sigma_{e_d} \rangle^{+, \bm{J}_{\lambda}}_{\Delta_n^\prime; \bm{\lambda}^\prime_0} &\text{ and } && \langle \sigma_{e_d} \rangle^{\mp, \bm{J}_{\lambda}}_{\Lambda_n; \bm{\lambda}_0} = \langle \sigma_{e_d} \rangle^{\mp, \bm{J}_{\lambda}}_{\Delta_n^\prime; \bm{\lambda}^\prime_0}.
\end{align}
For every $h\geq0$, let $\bm{h}^w$ be an external field acting only on $L_0$, that is  $h^w_i=h\mathbbm{1}_{\{i\in L_0\}}$ for all $i\in\Z^d$. Then, denoting $\Delta_n=\Delta_n^*\cup L_0$,

\begin{align}
    &\langle \sigma_{e_d} \rangle^{+, \bm{J}_{\lambda}}_{\Delta_n^\prime; \bm{\lambda}^\prime_0} = \lim_{h\to\infty}\langle \sigma_{e_d} \rangle^{+, \bm{J}_{\lambda}}_{\Delta_n; \bm{\lambda}^\prime_0 + \bm{h}^w} &\text{ and } &&  \langle \sigma_{e_d} \rangle^{\mp, \bm{J}_{\lambda}}_{\Delta_n^\prime; \bm{\lambda}^\prime_0} = \lim_{h\to \infty}  \langle \sigma_{e_d} \rangle^{\mp, \bm{J}_{\lambda}}_{\Delta_n; \bm{\lambda}^\prime_0 + \bm{h}^w}.
\end{align}
Differentiating in $h$ and using Proposition  \ref{Consequence.of.DVI}, we see that the difference $\langle \sigma_{e_d} \rangle^{+, \bm{J}_{\lambda}}_{\Delta_n; \bm{\lambda}^\prime_0 + \bm{h}^0} - \langle \sigma_{e_d} \rangle^{\mp, \bm{J}_{\lambda}}_{\Delta_n; \bm{\lambda}^\prime_0 + \bm{h}^0}$ is decreasing in $h$ for $h\geq 0$. Proposition \ref{Consequence.of.DVI} was stated for the semi-infinite states, but it also holds for Ising states since the DVI are in this generality. We can then bound the difference of the states by choosing the particular case $h=\lambda$. Denoting $\bm{\lambda}_l = \bm{\lambda}^\prime_0 + \bm{\lambda}^0$, we conclude that 
\begin{equation}\label{eq: uniqueness}
    \langle \sigma_{e_d} \rangle^{+}_{\Lambda_n; \bm{\lambda}}- \langle \sigma_{e_d} \rangle^{-}_{\Lambda_n; \bm{\lambda}}\leq \langle \sigma_{e_d} \rangle^{+, \bm{J}_{\lambda}}_{\Delta_n; \bm{\lambda}_l} - \langle \sigma_{e_d} \rangle^{\mp, \bm{J}_{\lambda}}_{\Delta_n; \bm{\lambda}_l}.
\end{equation}
Again by Proposition \ref{Consequence.of.DVI}, the RHS of equation \eqref{eq: uniqueness} in non-increasing in $h_i$, the external field on the site $i$, for any $i\in\mathbb{Z}^d$. So, denoting $\bm{\lambda}^{Is}$ the external field \eqref{particular.external.field} considered in \cite{Bissacot_Cass_Cio_Pres_15} with $h^*=\lambda/2$, we have that, for any $i\in\mathbb{Z}^d$, $h_i^{Is}=\frac{\lambda}{2}|i|^{-\delta}\leq \frac{\lambda}{2}|i_d|^{\delta}\leq (\bm{\lambda}_l)_i$, hence 
\begin{equation}
    \langle \sigma_{e_d} \rangle^{+}_{\Lambda_n; \bm{\lambda}}- \langle \sigma_{e_d} \rangle^{-}_{\Lambda_n; \bm{\lambda}}\leq \langle \sigma_{e_d} \rangle^{+, \bm{J}_{\lambda}}_{\Delta_n; \bm{\lambda}^{Is}} - \langle \sigma_{e_d} \rangle^{\mp, \bm{J}_{\lambda}}_{\Delta_n; \bm{\lambda}^{Is}}.
\end{equation}
Taking the limit in $n$, the RHS of equation above goes to $\langle \sigma_{e_d} \rangle^{+, \bm{J}_{\lambda}}_{\bm{\lambda}^{Is}} - \langle \sigma_{e_d} \rangle^{\mp, \bm{J}_{\lambda}}_{\bm{\lambda}^{Is}}$, that is equal to $0$ by Theorem \ref{Theo: Uniquiness_Ising_J_lambda}. We conclude that $\langle \sigma_{e_d} \rangle^{+}_{\bm{\lambda}} - \langle \sigma_{e_d} \rangle^{-}_{\bm{\lambda}}=0$, so there is only one Gibbs state. When $\lambda > 2J$, we replace $\bm{J}_\lambda$ by $\bm{J}\equiv J$ in the argument above and the same proof holds with minor adjustments. 
\end{proof}

\chapter{Random Field Ising Model}
This chapter follows the argument of \cite{Ding2021} and proves phase transition for the nearest-neighbor Ising model with a random field. In Section 1 we present the model and the overall strategy of the Peierls' argument. In Section 2, we present the Ding and Zhuang approach to prove phase transition and define the bad event. In Section 3, we follow the work of \cite{FFS84} and use a coarse-graining argument to upper bound the probability of the bad event and complete the proof of phase transition for the RFIM.

\section{The model}
The random field Ising model (RFIM) consists of the usual Ising model, previously introduced in Chapter 1, but with an external field that is random. The \textit{local Hamiltonian of the random field Ising model} in $\Lambda\Subset\Z^d$ with $\eta$-boundary condition is $H_{\Lambda, \varepsilon h}^{\eta, \IS}:  \Omega_\Lambda^\eta \to \mathbb{R}$, given by 
\begin{equation}
    H_{\Lambda; \varepsilon h}^{\eta, \IS}(\sigma)\coloneqq -\sum_{\substack{x,y\in\Lambda \\ |x-y|=1}} J\sigma_x\sigma_y - \sum_{\substack{x\in \Lambda, y\in\Lambda^c \\ |x-y|=1}} J\sigma_x\eta_y - \sum_{x\in\Lambda} \varepsilon h_x\sigma_x,
\end{equation}
where the external field is a family $\{h_x\}_{x\in\Z^d}$ of i.i.d. random variables in $(\widetilde{\Omega}, \mathcal{A}, \mathbb{P})$, and every $h_x$ has a standard normal distribution\footnote{ Our results also hold for more general distributions of $h_x$, see Remarks \ref{Rmk: Bernoulli_external_field} and \ref{Rmk: More.general.h_x}. }. The parameter $\varepsilon >0$ controls the variance of the external field. Given $\Lambda\Subset\Z^d$, consider $\mathscr{F}_\Lambda$ the $\sigma$-algebra generated by the cylinders sets supported in $\Lambda$ and $\mathscr{F}$ the $\sigma$-algebra generated by finite union of cylinders. One of the main objects of study in classical statistical mechanics is the \textit{finite volume Gibbs measures}, which are probability measures in $(\Omega, \mathscr{F})$, given by 
    \begin{equation}
        \mu_{\Lambda;\beta, \varepsilon h}^{\eta, \IS}(\sigma) \coloneqq \mathbbm{1}_{\Omega_\Lambda^\eta}(\sigma)\frac{e^{-\beta H_{\Lambda, \varepsilon h}^{\eta, \IS}(\sigma)}}{Z_{\Lambda; \beta, \varepsilon}^{\eta, \IS}(h)},
    \end{equation}
where $\beta>0$ is the inverse temperature and $Z_{\Lambda; \beta, \varepsilon}^{\eta, \IS}$ is called \textit{partition function}, defined as 

\begin{equation}
    Z_{\Lambda; \beta, \varepsilon}^{\eta, \IS}(h)\coloneqq \sum_{\sigma\in\Omega_\Lambda^\eta} e^{-\beta H_{\Lambda, \varepsilon h}^{\eta, \IS}(\sigma)}.
\end{equation}
One important remark is that, since the external field is random, the Gibbs measures are random variables. To explicit the dependence of $\mu_{\Lambda;\beta, \varepsilon h}^{\eta, \IS}$ on $\widetilde{\Omega}$, we write $\mu_{\Lambda;\beta, \varepsilon h}^{\eta, \IS}[\omega]$, with $\omega$ being a general element of $\widetilde{\Omega}$. Two particularly important boundary conditions are given by the configurations $\eta_{+} \equiv +1$ and $\eta_{-} \equiv -1$, and are called $+$ and $-$ boundary conditions, respectively. For these boundary conditions, we can $\mathbb{P}$-almost surely define the infinite volume measures by taking the weak*-limit
\begin{equation}
    \mu_{\beta,\varepsilon h}^{\pm, \IS}[\omega] \coloneqq \lim_{n\to\infty} \mu_{\Lambda_n;\beta, \varepsilon h}^{\pm, \IS}[\omega],
\end{equation}
where $(\Lambda_n)_{n\in\mathbb{N}}$ is any sequence invading $\Z^d$, that is, for any subset $\Lambda\Subset\mathbb{Z}^d$, there exists $N=N(\Lambda)>0$ such that $\Lambda\subset\Lambda_n$ for every $n>N$. By Lemma \ref{extremality.of.+.and.-.bc}, for any fixed external field, the measures $\mu_{\Lambda_n;\beta, \varepsilon h}^{\pm, \IS}[\omega]$ are monotone, which guarantees the existence of the limits over sequences invading $\Z^d$. 
To have more than one Gibbs measure, it is enough to show that $\mu_{\beta,\varepsilon h}^{+, \IS}[\omega]\neq  \mu_{\beta,\varepsilon h}^{-, \IS}[\omega]$, with $\mathbb{P}$-probability 1, see \cite[Theorem 7.2.2]{Bovier.06}.

The standard strategy to prove phase transition in the Ising model is to use the Peierls' argument, which based on the idea of erasing contours. Contours are geometric objects in the dual lattice $\mathbb{Z}^d_*$ defined as: denoting $C_x$ the closed unit cube in $\mathbb{R}^d$ centered in $x$, $\mathbb{Z}^d_*$ is the union of all faces $C_x\cap C_y$ with $|x-y| = 1$. Given a configuration, its contours are the maximal connected components of the union of the faces $C_x\cap C_y$ satisfying $\sigma_x \neq \sigma_y$. The set of contours of $\sigma$ is denoted by $\Gamma(\sigma)$, and $\gamma$ denotes a generic element of $\Gamma(\sigma)$. Moreover, $\Gamma(\Omega)$ denotes all family of contours that can be associated to a configuration, so $\Gamma(\Omega)\coloneqq \cup_{\sigma\in\Omega}\Gamma(\sigma)$. The \textit{interior} of a contour $\gamma$, denoted $\I(\gamma)$, is the set of points connected to $\infty$ only by paths crossing $\gamma$. Given $n\in\mathbb{N}$, take
\begin{equation*}
    \Gamma_0(n) \coloneqq \{\gamma\in\Gamma(\Omega) \ : \ 0\in\I(\gamma), \  |\gamma| = n\}
\end{equation*}
and $\Gamma_0 = \cup_{n\geq 1}\Gamma_0(n)$. The operation $\tau_{\gamma}$ used to remove a contour $\gamma\in\Gamma(\sigma)$ can be written as a particular case of the following one: given $A\subset\Z^d$, take $\tau_A:\mathbb{R}^{\Z^d} \xrightarrow{} \mathbb{R}^{\Z^d}$ as 
\begin{equation}\label{Def: tau_A}
    (\tau_A(\sigma))_i \coloneqq \begin{cases}
                        -\sigma_i &\text{if }i\in A,\\
                        \sigma_i   &\text{otherwise},
                      \end{cases}
\end{equation}
for every $i\in\Z^d$. The transformation that erases a contour $\gamma$ is $\tau_\gamma(\sigma) \coloneqq \tau_{\I(\gamma)}(\sigma)$. The key property of this contour system is that we can bound the difference in the Hamiltonian after erasing a contours, when there is no external field.
\begin{proposition}
    There is a constant $c_1(d)>0$ such that, for any $\sigma\in\Omega^+$ and $\gamma\in\Gamma(\sigma)$,
    \begin{equation}\label{Eq: Energy_cost_erasing_contour_SR}
        H_{\Lambda, 0}^{+, \IS}(\tau_{\gamma}(\sigma)) - H_{\Lambda, 0}^{+, \IS}(\sigma) \leq - J c_1(d) |\gamma|.
    \end{equation}
\end{proposition}

This bound on the energy cost of erasing a contour is the first ingredient of a Peierls' argument. The second key ingredient in to bound the number of contours with a fixed size. It is well-known that $|\Gamma_0(n)|\leq e^{c_2(d)n}$ for a suitable constant $c_2(d)>0$. The best bound for this constant is due to Balister and Bollob{\'a}s, \cite{Balister_Bollobas_07}.

\section{Ding and Zhuang approach}\label{Sec: Ding and Zhuang approach}
The main idea used in Ding and Zhuang's proof of phase transition in \cite{Ding2021} is to make the Peierls' argument on the joint space of the configurations and the external field, and when erasing a contour, perform in the external field the same flips you do in the configuration. Doing this, the part on the Hamiltonian that depends on the external field does not change, but the partition function does. The complication of this method is to control such differences. \\

Given $\Lambda\subset\Z^d$, define the local joint measure for $(\sigma, h)$ as
\begin{equation*}
    \mathbb{Q}_{\Lambda; \beta, \varepsilon}^{+, \IS}(\sigma \in A, h\in B) = \int_{B} \mu_{\Lambda;\beta, \varepsilon h}^{+, \IS}(A) d\mathbb{P}(h),
\end{equation*}
for $A\subset\Omega$ measurable and $B\subset \mathbb{R}^{\Lambda}$ borelian. Since $\beta, \varepsilon $ and $\Lambda$ are fixed, we will omit then from the notation. 
This measure $\mathbb{Q}$ has density
\begin{equation*}
    g_{\Lambda; \beta, \varepsilon}^{+, \IS}(\sigma, h) = \prod_{u\in\Lambda}\frac{1}{\sqrt{2\pi}}e^{-\frac{1}{2}h_u^2} \times \mu_{\Lambda;\beta, \varepsilon h}^{+, \IS}(\sigma).
\end{equation*}

The main idea used in the proof of phase transition in \cite{Ding2021} is to make the Peierls' argument on the measure $\mathbb{Q}$, and perform in the external field the same flips you do in the configuration when erasing a contour. Formally, in \cite{Ding2021} they compare the density $g_{\Lambda; \beta, \varepsilon}^{+, \IS}(\sigma, h)$ with the density after erasing a contour $\gamma\in\Gamma(\sigma)$, and performing the same flips on the external field, getting

\begin{align}\label{Eq: quotient.of.gs_SR}
    \frac{g_{\Lambda; \beta, \varepsilon}^{+, \IS}(\sigma, h)}{g_{\Lambda; \beta, \varepsilon}^{+, \IS}(\tau_{\gamma}(\sigma),\tau_{\gamma}(h))} \nonumber
    &= \exp{\{\beta H_{\Lambda, 0}^{+, \IS}(\tau_{\gamma}(\sigma)) - \beta H_{\Lambda, 0}^{+, \IS}(\sigma)\}}\frac{Z_{\Lambda; \beta, \varepsilon}^{+, \IS}(\tau_{\gamma}(h))}{Z_{\Lambda; \beta, \varepsilon}^{+, \IS}(h)}.  \nonumber \\ 
\end{align}

For some realizations of the external field, the quotient of the partition functions can be bigger than the exponential term. Denoting
\begin{equation}\label{Def: Delta_A_SR}
\Delta_A(h) \coloneqq -\frac{1}{\beta}\ln{\frac{Z_{\Lambda; \beta, \varepsilon}^{+, \IS}(h)}{Z_{\Lambda; \beta, \varepsilon}^{+, \IS}(\tau_{A}(h))}}
\end{equation}
 for every $A\subset \Z^d$, the bad event is
\begin{equation}\label{Def: Bad_event_SR}
\mathcal{E}^c\coloneqq \left\{\sup_{\substack{\gamma\in\Gamma_0}} \frac{|\Delta_{\I(\gamma)}(h)|}{c_1|\gamma|} > \frac{1}{4}\right\}.    
\end{equation}
To control the probability of this bad event, we need a concentration result for Gaussian random variables. The following one is due to M. Ledoux and M. Talagrand, and a proof can be found in \cite{Ledoux.Talagrand.91}.

\begin{theorem}\label{Theo: Gaussian.concentration}
    Let $f:\mathbb{R}^M \xrightarrow[]{} \mathbb{R}$ be a uniform Lipschitz continuous function with constant $C_{Lip}$, that is, for any $X,Y\in\mathbb{R}^M$, $$|f(X) - f(Y)| \leq C_{Lip} || X - Y ||_2 .$$ 
    
    Then, if $X_1,\dots, X_M$ are i.i.d. Gaussian random variables with variance 1,
    \begin{equation}\label{Eq: Tail.concentration}
        \mathbb{P}\left(|f(X_1,\dots, X_M) - \mathbb{E}(f(X_1, \dots, X_M))|\geq z\right) \leq 2\exp{\left\{\frac{-z^2}{2C_{Lip}^2}\right\}}.
    \end{equation}
\end{theorem}

\begin{remark}\label{Rmk: MVT.Lipschitz}
    If $f$ is differentiable and $||\nabla f(\cdot)||_2$ is bounded, the mean value theorem guarantees that $\sup_{Z\in\mathbb{R}^M}||\nabla f(Z)||_2$ is a uniform Lipschitz constant for $f$. 
\end{remark}
\begin{remark}\label{Rmk: Bernoulli_external_field}
    If $f$ has a compact support and convex level sets, an equation similar to \eqref{Eq: Tail.concentration} holds, with some adjustments on the constants and replacing the mean by the median, see \cite[Theorem 7.1.3]{Bovier.06}. Therefore, our results hold when $h_i$ has a Bernoulli distribution $\mathbb{P}(h_i=+1) =\mathbb{P}(h_i=-1)= \frac{1}{2}$. 
\end{remark}

Given $A\subset\Z^d$, $h_A\coloneqq (h_x)_{x\in A}$ denotes the restriction of the external field to the subset $A$. The next Lemma was proved in \cite{Ding2021} and is a direct consequence of Theorem \ref{Theo: Gaussian.concentration}. 

\begin{lemma}\label{Lemma: Concentration.for.Delta.General}
    For any $A, A^\prime \Subset \mathbb{Z}^d$ and $\lambda>0$, we have 
\begin{equation}\label{Eq: Tail.of.Delta_A}
    \mathbb{P}\left(|\Delta_A(h)| \geq \lambda \vert h_{A^c}\right) \leq2e^{\frac{-\lambda^2}{8\varepsilon^2 |A|}},
\end{equation}
and 
\begin{equation}\label{Eq: Tail.of.the.diff.of.Deltas}
     \mathbb{P}(|\Delta_{A}(h) - \Delta_{A^\prime}(h)|>\lambda|h_{{A \cup A^\prime}^c}) \leq  2e^{-\frac{{\lambda^2}}{{8\varepsilon^2|A \Delta A^\prime|}}},
\end{equation}
where $A\Delta A^\prime$ is the symmetric difference. 
\end{lemma}

\begin{proof}
        Start by noticing that $$ \mathbb{E}(\Delta_A(h) \vert h_{A^c}) = 0,$$ since $h=^d \tau_A(h)$ by the symmetry of the Gaussian distribution. Now, for any $v\in A$, 
    \begin{equation*}
        \left|\frac{d}{dh_v}(\Delta_A(\h))\right| = \varepsilon\left|\mu_{\Lambda;\beta, \varepsilon \tau_A(h)}^{+, \IS}(\sigma_v) + \mu_{\Lambda;\beta, \varepsilon h}^{+, \IS}(\sigma_v)\right| \leq 2\varepsilon.
    \end{equation*}
Hence, for any $z\in\mathbb{R}^{A}$, $||\nabla f(z)||_2 < 2e|A|^{\frac{1}{2}}$, which, together with Theorem \ref{Theo: Gaussian.concentration} and Remark \ref{Rmk: MVT.Lipschitz}, concludes the proof of equation \eqref{Eq: Tail.of.Delta_A}. 

For the second equation, notice that $\Delta_{A}(h) - \Delta_{A^\prime}(h) = \frac{1}{\beta} Z_{\Lambda; \beta, \varepsilon}^{+, \IS}(\tau_{A}(h)) - \frac{1}{\beta} Z_{\Lambda; \beta, \varepsilon}^{+, \IS}(\tau_{A^\prime(h)})$. By the symmetry of the Gaussian distribution, $(\tau_{A}(h), \tau_{A^\prime}(h))$ and ${(\tau_A\circ\tau_{A}(h), \tau_A\circ\tau_{A^\prime}(h))}$ have the same distribution. Moreover, $\tau_A\circ\tau_{A^\prime} = \tau_{A\Delta A^\prime}$, hence 
$$\Delta_{A}(h) - \Delta_{A^\prime}(h) = \frac{1}{\beta} Z_{\Lambda; \beta, \varepsilon}^{+, \IS}(\tau_{A \Delta A^\prime}(h)) - \frac{1}{\beta} Z_{\Lambda; \beta, \varepsilon}^{+, \IS}(h) =^d \Delta_{A \Delta A^\prime}(h)$$
and equation \eqref{Eq: Tail.of.the.diff.of.Deltas} follows from \eqref{Eq: Tail.of.Delta_A} applied to $\Delta_{A \Delta A^\prime}(h)$. 
\end{proof}

\begin{remark}\label{Rmk: More.general.h_x}
    This lemma holds whenever $h=(h_x)_{x\in\Z^d}$ satisfy equation \eqref{Eq: Tail.concentration}.As a consequence, our results can be stated for more general external fields. 
\end{remark}

\section{Controlling \texorpdfstring{$\mathbb{P}\left(\mathcal{E}^c\right)$}{P(Ec)}}
    
    To control the probability of $\mathcal{E}^c$ we use a multi-scale analysis method presented in \cite{FFS84}. This section is dedicated to prove
    
    \begin{proposition}\label{Prop: Bound.bad.event_SR}       
        There exists $C_1\coloneqq C_1(\alpha, d)$ such that $\mathbb{P}(\mathcal{E}^c)\leq e^{-\frac{C_1}{\varepsilon^2}}$. 
    \end{proposition}
    
    As pointed out by \cite{Ding2021}, the proof presented in \cite{FFS84}, despite being self-contained, is an indirect application of Dudley's entropy bound. Here we adapt the proof presented in \cite{FFS84} using this entropy bound. For the detailed argument of the original proof, see \cite{Bovier.06}. First, we need to introduce some probability tools, then we introduce the coarse-graining procedure and prove Proposition \ref{Prop: Bound.bad.event_SR}.
    
    \subsection{Probability Results}\label{Sec: Probability results}
    
To control the probability of $\mathcal{E}^c$, we use some results on majorizing measures. For an extensive overview, we refer to \cite{Talagrand_14}. Consider $(T,\d)$ a finite metric space and a process $(X_t)_{t\in T}$ such that, for every $\lambda>0$ and $t,s\in T$,
\begin{equation}\label{Eq: Sub_gaussian_def}
    \mathbb{P}\left( |X_t - X_s| \geq \lambda \right) \leq 2\exp{\frac{-\lambda^2}{2\d(s,t)^2}}.
\end{equation}
Assume also that $\mathbb{E}\left(X_t\right) = 0$ for every $t\in T$. One example of such process is $( |\Delta_{\I(\gamma)}|)_{\gamma\in\Gamma_0(n)}$, with the distance $\d_2(\gamma,\gamma^\prime) = 2\varepsilon |\I(\gamma)\Delta \I(\gamma^\prime)|^{\frac{1}{2}}$ over the set $\Gamma_0(n)$. For $n\in \mathbb{N}$, consider the quantities $N_n = 2^{2^n}$ and $N_0=1$. 

\begin{definition}
    Given a set $T$, a sequence $(\mathcal{A}_n)_{n\geq 0}$ of partitions of $T$ is \textit{admissible} when $|\mathcal{A}_n|\leq N_n$ and $\mathcal{A}_{n+1}\preceq \mathcal{A}_n$ for all $n\geq 0$.
\end{definition}

Given $t\in T$ and an admissible sequence $(\mathcal{A}_n)_{n\geq 0}$, $A_n(t)$ denotes the element of $\mathcal{A}_n$ that contains $t$. 

\begin{definition}
    Given $\theta > 0$ and a metric space $(T,d)$, we define
    \begin{equation*}
        \gamma_\theta(T,\d) \coloneqq \inf_{(\mathcal{A}_n)_{n\geq 0}}\sup_{t\in T}\sum_{n\geq 0}2^{\frac{n}{\theta}}\diam(A_n(t)),
    \end{equation*}
where the infimum is taken over all admissible sequences. 
\end{definition}

\begin{theorem}[Majorizing Measure Theorem \cite{Talagrand_87}]\label{MMT} There is a universal constant $L>0$ such that
\begin{equation*}
    \frac{1}{L}\gamma_2(T,\d) \leq \mathbb{E}\left( \sup_{t\in T} X_t \right) \leq L\gamma_2(T,\d).
\end{equation*}
\end{theorem}

    Given $\epsilon>0$, let $N(T,\d, \epsilon)$ be the minimal number of balls with radius $\epsilon$ necessary to cover $T$, using the distance $\d$. 
\begin{proposition}[Dudley's Entropy Bound \cite{Dudley67}]
    Let $(X_t)_{t\in T}$ be a family of centered random variables satisfying \eqref{Eq: Sub_gaussian_def} for some distance $\d$. Then there exists a constant $L>0$ such that 
    \begin{equation*}
        \mathbb{E}\left[\sup_{t\in T}X_t\right]\leq L\int_{0}^\infty \sqrt{\log N(T,\d,\epsilon)}d\epsilon.
    \end{equation*}
\end{proposition}

Dudley's entropy bound together with the Majorizing Measure Theorem yields that there is a constant $L^\prime>0$ such that,
\begin{equation}\label{Eq: gamma_2_bounded_by_Dudley_integral}   
\gamma_2(T,\d)\leq L^\prime\int_{0}^\infty \sqrt{\log N(T,\d,\epsilon)}d\epsilon.
\end{equation}
We also need the following result. 
\begin{theorem}\label{Theo: Theo_2.2.27_Talagrand} Given a metric space $(T,\d)$ and a family $(X_t)_{t\in T}$ of centered random variables satisfying \eqref{Eq: Sub_gaussian_def}, there is a universal constant $L>0$ such that, for any $u>0$,
\begin{equation*}
\mathbb{P}\left( \sup_{t\in T}X_t > L(\gamma_2(T,\d) + u\diam(T)) \right)\leq e^{-{u^2}},
\end{equation*}
where the $\diam(T)$ is the diameter taken with respect to the distance $\d$
\end{theorem} 
A proof can be found in  \cite[Theorem 2.2.27]{Talagrand_14}. Using these results, the bound on the bad event $\mathcal{E}^c$ follows from the next proposition.
\begin{proposition}\label{Prop: Bound.gamma_2_SR}
    Given $n\geq 0$ and $d\geq 3$ there is a constant $L_1 \coloneqq L_1(d,\alpha)>0$  such that $$\gamma_2(\Gamma_0(n),\d_2) \leq \varepsilon L_1 n.$$
\end{proposition}
As a direct consequence of this Proposition, we can control the probability of the bad event $\mathcal{E}^c$.

\begin{proof}[Proof of Proposition \ref{Prop: Bound.bad.event_SR}]   
   To apply Theorem \ref{Theo: Theo_2.2.27_Talagrand}, notice that, by the isoperimetric inequality $\diam(\Gamma_0(n)) = \sup_{\gamma_1,\gamma_2\in\Gamma_0(n)} \leq 4\varepsilon\sqrt{|\I(n)|}\leq 4\varepsilon n^{\frac{1}{2}(1 + \frac{1}{d-1})}$. Hence, for $\varepsilon$ small enough
\begin{align*}
    \mathbb{P}\left( \sup_{\gamma}\Delta_{\I(\gamma)}(h) \geq \frac{c_2}{2}n \right) \leq \mathbb{P}\left( \sup_{\gamma}\Delta_{\I(\gamma)}(h) \geq L(b_5 \varepsilon n + n^{\frac{1}{2}(1 - \frac{1}{d-1})}n^{\frac{1}{2}(1 + \frac{1}{d-1})}) \right)\leq e^{-\frac{n^{(1 - \frac{1}{d-1})}}{\varepsilon^2}}.
\end{align*}
The union bound yields $\mathbb{P}(\mathcal{E}^c)\leq e^{-\frac{C}{\varepsilon^2}}$ for $C>0$ large enough.
\end{proof}

The next subsections are dedicated to proving Proposition \ref{Prop: Bound.gamma_2_SR}.

    \subsection{Coarse-graining Procedure}\label{Sec: Coarse-graining_SR}

We will apply these results for the family $(|\Delta_{\I_-(\gamma)}|)_{\gamma\in\Gamma_0(n)}$. To construct the covering by balls in Dudley's entropy bound, we use the coarse-graining idea introduced in \cite{FFS84}. For each $0<\ell$ and each contour ${\gamma\in\Gamma_0}$, we will associate a region $B_\ell(\gamma)$ that approximates the interior $\I(\gamma)$ in a scaled lattice, with the scale growing with $\ell$. This is done in a way that two contours that have the same representation are in a ball with fixed radius, depending on $\ell$.

For any $x\in\Z^d$ and $m\geq 0$,
\begin{equation}
    C_{m}(x) \coloneqq \left(\prod_{i=1}^d{\left[2^{m}x_i , \ 2^{m}(x_i+1) \right)}\right)\cap \Z^d,
\end{equation}
is the cube of $\mathbb{Z}^d$ centered at $2^{m}x + 2^{m-1} - \frac{1}{2}$ with side length $2^{m} -1$. Any such cube is called an $m$-cube. As all cubes in this paper are of this form, with centers $2^{m}x + 2^{m-1} - \frac{1}{2}$ and $x \in \mathbb{Z}^d$, we will often omit the point $x$ in what follows, writing $C_m$ for an $m$-cube instead of $C_m(x)$. An arbitrary collection of $m$-cubes will be denoted $\mathscr{C}_m$ and $B_{\mathscr{C}_m}\coloneqq \cup_{C\in\mathscr{C}_m}C$ is the region covered by $\mathscr{C}_m$. 
We denote by $\mathscr{C}_m(\Lambda)$ the covering of $\Lambda\Subset\Z^d$ with the smallest possible number of $m$-cubes.

\input{Figures/Figura.0}

Fix $n \in \mathbb{N}$, $\gamma\in \Gamma_0(n)$, and $\ell\in\{0, 1, \dots, k\}$. An $\ell$-cube $C_{\ell}$ is \textit{admissible} if it more than a  half of its points are inside $\I(\gamma)$. Thus, the set of admissible cubes is
\begin{equation*}
    \mathfrak{C}_\ell(\gamma) \coloneqq \{C_{\ell} : |C_{\ell}\cap \I(\gamma)| \geq \frac{1}{2}|C_{\ell}|\}.
\end{equation*}
We choose $B_\ell(\gamma) \coloneqq B_{\mathfrak{C}_{\ell}(\gamma)}$, the region covered by the admissible cubes.
Notice that $B_\ell(\gamma)$ is uniquely determined by $\partial B_\ell(\gamma)$. Moreover, $\partial B_\ell(\gamma)$ is uniquely determined by
$$
\partial \mathfrak{C}_\ell(\gamma) \coloneqq \{ \{C_{\ell}, C^\prime_{\ell}\} : C_{\ell} \in \mathfrak{C}_\ell(\gamma), \ C_\ell^\prime \notin \mathfrak{C}_\ell, \  C_\ell^\prime \text{ shares a face with }C_\ell\}.
$$ 
We will now control the number of cubes in $\mathfrak{C}_\ell(\gamma)$ by proving a proposition similar to \cite[Proposition 2]{FFS84}. This proposition was written for $d=3$ and $\gamma$ simply connected, but it can clearly be extended to $d\geq 2$ with no restriction in $\gamma$, see \cite{Bovier.06}. As we could not find a detailed proof anywhere, we provide one here.

Given a rectangle $\mathcal{R} = [1,r_1]\times[1,r_2]\times\dots\times[1,r_d]$, consider $\R_i\coloneqq\{x\in \R : x_i=1\}$ the face of $\R$ that is perpendicular to the direction $e_i$, for $i=1,\dots,d$. The line that connects a point $x\in \R_i$ to a point in the opposite face of $\R_i$ is $\ell_x^i \coloneqq \{ x + ke_i : 1\leq k\leq r_i\}$. Given $A\subset \Z^d$, the projection of $A\cap \R$ into the face $\R_i$ is
\begin{equation*}
    \calP_i(A\cap\R) \coloneqq \{x\in\R_i : \ell_x^i \cap A \neq \emptyset\}.
\end{equation*}
\input{Figures/Figura.4.1}
In many situations, we will split the projections into \textit{good} and \textit{bad} points. The set of good points is $\calP_i(A\cap\R)^{G} \coloneqq \{x\in \calP_i(A\cap \R) : \ell_x^i \cap (\R\setminus A) \neq \emptyset\}$, that is, there exist a point in $\ell_x^i\cap \R$ that is not in $A$.  The bad points are defined as $\calP^{B}_i(A\cap\R) \coloneqq \calP_i(A\cap\R)\setminus \calP_i^G(A\cap\R)$.

\begin{figure}[H] 
	\centering

\tikzset{every picture/.style={line width=0.75pt}} 

\begin{tikzpicture}[x=0.75pt,y=0.75pt,yscale=-0.9,xscale=0.9]

\draw   (308.2,253.9) -- (168.2,253.9) -- (168.2,31.4) -- (308.2,31.4) -- cycle ;
\draw  [fill={rgb, 255:red, 155; green, 155; blue, 155 }  ,fill opacity=1 ] (307.8,192.3) .. controls (307,237.2) and (277.09,194.1) .. (259.2,186.4) .. controls (241.31,178.7) and (169.4,225) .. (168.6,201.4) .. controls (167.8,177.8) and (168.2,156.2) .. (168.2,134.6) .. controls (168.2,113) and (201.4,128.6) .. (209.2,132.4) .. controls (217,136.2) and (277.2,152.4) .. (257.2,122.4) .. controls (237.2,92.4) and (210.67,64.15) .. (237.2,66.4) .. controls (263.73,68.65) and (249.33,76.09) .. (263.8,83.4) .. controls (278.27,90.71) and (285.66,94.94) .. (287.4,101) .. controls (289.14,107.06) and (307.55,98.48) .. (307.8,104.2) .. controls (308.05,109.92) and (308.6,147.4) .. (307.8,192.3) -- cycle ;
\draw  [dash pattern={on 0.84pt off 2.51pt}]  (167.8,88.2) -- (307.8,88.2) ;
\draw [shift={(167.8,88.2)}, rotate = 0] [color={rgb, 255:red, 0; green, 0; blue, 0 }  ][fill={rgb, 255:red, 0; green, 0; blue, 0 }  ][line width=0.75]      (0, 0) circle [x radius= 3.35, y radius= 3.35]   ;
\draw  [dash pattern={on 0.84pt off 2.51pt}]  (168.6,156.2) -- (308.6,156.2) ;
\draw [shift={(168.6,156.2)}, rotate = 0] [color={rgb, 255:red, 0; green, 0; blue, 0 }  ][fill={rgb, 255:red, 0; green, 0; blue, 0 }  ][line width=0.75]      (0, 0) circle [x radius= 3.35, y radius= 3.35]   ;

\draw (255,168) node [anchor=north west][inner sep=0.75pt]    {$A\cap\R$};
\draw (151.2,79.2) node [anchor=north west][inner sep=0.75pt]  [font=\small]  {$p$};
\draw (150.4,145.6) node [anchor=north west][inner sep=0.75pt]  [font=\small]  {$p^{\prime }$};

\end{tikzpicture}

	\caption{Considering $A\cap\R$ the gray region, both points $p,p^\prime \in \mathcal{P}_1(A\cap \R)$ are in the projection, but $p$ is a good point and $p^\prime$ is a bad point. The doted lines represent $\ell_p^1$ and $\ell_{p^\prime}^1$.} \label{fig: Figura5}
\end{figure}
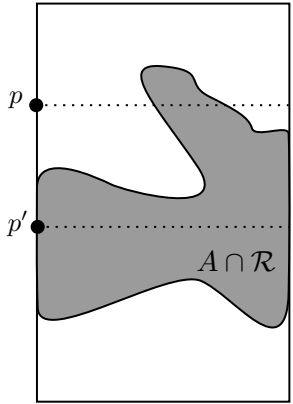
Given $x\in \calP_i(A\cap\R)^{G}$, by definition of the projection, there exists a point in $\ell_x^i\cap A$. Therefore, there exists a point $p\in \ell_x^i$ such that $p\in\fext A \cap \R$. As all lines are disjoint, we conclude that 
\begin{equation}\label{Eq: upper.bound.good.points}
     |\calP_i^{G}(A\cap\R)|\leq |\fext A \cap \R|.
\end{equation}
 We now prove two auxiliary lemmas.
 
\begin{lemma}\label{Lemma: Geo.discreta.1}
    Given $d\geq 2$, for any family of positive integers $\bm{r}=(r_i)_{i=1}^d$ with $R\leq r_i \leq 2R$ for some $R\geq 2$, $0<\lambda < 1$ and $A\subset\Z^d$, there exists a constant $c\coloneqq c(d, \lambda)$ such that, if 
    \begin{equation}\label{Eq: hypothesis.lemma.1}
         |\calP_i(A\cap \R)| \leq \lambda|\R_i|
    \end{equation}
    for all $i= 1,\dots, d$, then 
    \begin{equation*}
        \sum_{i=1}^d |\calP_{i}(A\cap \R)|\leq c|\fext A\cap \R|,
    \end{equation*}
    where $\R=[1,r_1]\times\dots\times [1,r_d]$.
\end{lemma}

\begin{proof}
The proof will be done by induction on the dimension. For $d=2$, take a rectangle ${\R=[1,r_1]\times[1,r_2]}$. If there is no bad points in $\calP_1(A\cap\R)$, then 
\begin{align}\label{Eq: Bound.1.on.P.1}
    |\calP_1(A\cap\R)| = |\calP_1^G(A\cap \R)| \leq |\fext A \cap \R|.
\end{align}

If there is a bad point $p=(1,p_2)\in \calP_1^B(A\cap\R)$, $\ell_p^1\subset A\cap \R$  by definition of bad point. As $|\calP_1(A\cap \R)| \leq \lambda|\R_1| < |\R_1|$, there is a point $p^\prime = (1,p_2^\prime)\in \R_1\setminus \calP_1(A\cap \R)$ that is in the face $\R_1$ but not in the projection. By definition of the projection, $\ell_{p^\prime}^1\in A^c\cap \R$. Therefore, for any $1\leq k\leq r_1$, $(k,p_2)\in  A\cap \R$ and $(k,p^\prime_2)\in  A^c\cap \R$, we can find a point $p^k=(k, p^k_2) \in \fext A \cap \R$. Since $p^{k_1}\neq p^{k_2}$ for every $k_1\neq k_2$, we have $r_1 \leq |\fext A \cap \R|$, hence
\begin{equation}\label{Eq: Bound.2.on.P.1}
   |P_1(A\cap \R)| \leq  |\R_1| = {r_2}\leq  2R \leq 2r_1 \leq  2|\fext A \cap \R|.
\end{equation}

A completely analogous argument can be done to bound $|P_2(A\cap \R)|$, and we conclude that
\begin{equation*}
    \sum_{i=1}^2|\calP_i(A\cap \R)|\leq 4|\fext A \cap \R|,
\end{equation*}
and take $c(2,\lambda)=4$. Suppose the lemma holds for $d-1$ and fix a rectangle $\R=[1,r_1]\times\dots\times[1,r_d]$. We split $\R$ into layers $L_k = \{x\in\Z^d : x_d = k\}$, for $k=1,\dots, r_d$. We can then partition the projection and write
\begin{equation*}
|\calP_i(A\cap \R)| = \sum_{k=1}^{r_d} |\calP_i(A\cap \R)\cap L_k|,    
\end{equation*}
for any $i\in\{1,\dots, d-1\}$. This yields
\begin{align}\label{Eq: Partition.sum.proj.}
    \sum_{i=1}^d|\calP_i(A\cap \R)| &= \sum_{i=1}^{d-1}\sum_{k=1}^{r_d}|\calP_i(A\cap \R)\cap L_k| + |\calP_d(A\cap \R)| \nonumber \\
    &=  \sum_{k=1}^{r_d}\sum_{i=1}^{d-1}|\calP_i(A\cap \R)\cap L_k| + |\calP_d(A\cap \R)|.
\end{align}

Notice now that $\calP_i(A\cap \R)\cap L_k = \calP_i(A\cap (\R\cap L_k))$. Defining the rectangle $\R^k \coloneqq \R\cap L_k$, for every point $p\in \calP_j^B(A\cap \R^k)$, $\ell_p^j \subset A\cap \R^k$. Moreover, we can associate every point $x\in \ell_p^j$ in the line with a point $x^\prime\in \calP_d(A\cap\R)$ by taking $x_m^\prime = x_m$ for $m \leq d-1$ and $x_d^\prime = 1$, therefore

\begin{equation*}
    r_j|\calP_j^B(A\cap \R^k)| = \sum_{p\in \calP_j^B(A\cap \R^k)}|\ell_p^j| \leq |\calP_d(A\cap\R)|.
\end{equation*}

\begin{figure}[H] 
	\centering
	\includegraphics[scale=0.15]{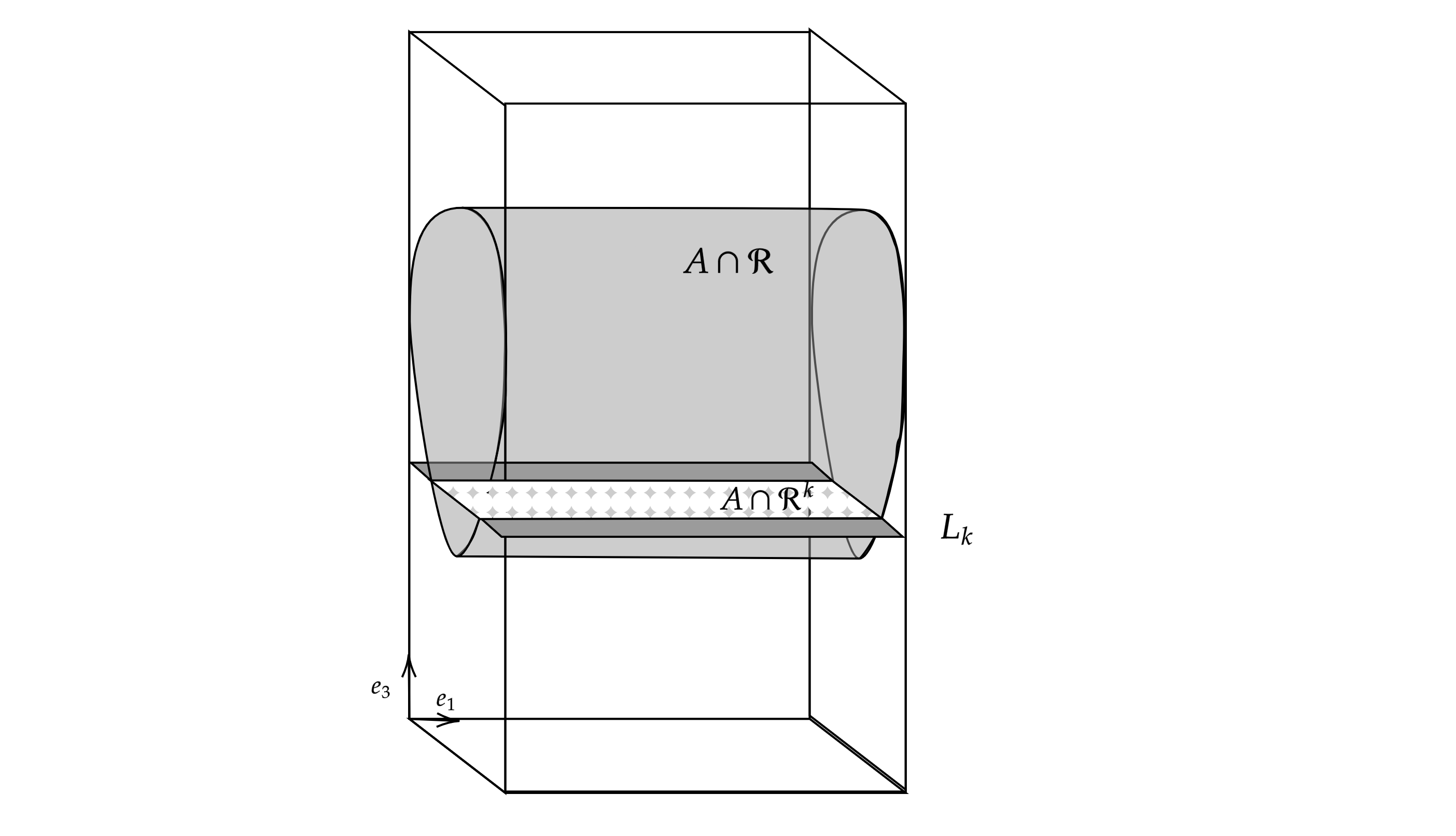}
	\caption{Considering $A\cap\R$ the light gray region, the dark gray region is a layer $L_k$  and the dotted region is $A\cap\R^k$, the restriction of $A$ to the layer $L_k$.} \label{fig: Figura6}
\end{figure}

Using the hypothesis \eqref{Eq: hypothesis.lemma.1} we conclude that

\begin{equation}\label{Eq: upper.bound.projection.i.bad.points}
     |\calP_j^B(A\cap \R^k)| \leq \lambda\frac{|\R_d|}{r_j} = \lambda \frac{ \prod_{q\neq d}r_q}{r_j} =  \lambda \prod_{q\neq j,d}r_q = \lambda |(\R^k)_j|.
\end{equation}
We consider know two cases:
    \begin{itemize}
        \item[(a)] If $|\calP_i(A\cap \R^k)| \leq \frac{\lambda +1}{2}|(\R^k)_i|$, for all $i\leq d-1$, then we are in the hypothesis of the lemma in $d-1$ and therefore
\begin{equation}\label{Eq: Primeiro.bound.soma.projecoes}
    \sum_{i=1}^{d-1} |\calP_i(A\cap \R^k)| \leq c\left(d-1, \frac{\lambda + 1}{2}\right)|\fext A\cap \R^k|.
\end{equation}
    \item [(b)] If there exists $j\in\{1,\dots,d-1\}$ satisfying $|\calP_j(A\cap \R^k)| > \frac{\lambda +1}{2}|(\R^k)_j|$, by \eqref{Eq: upper.bound.projection.i.bad.points} we have $|\calP_j^G(A\cap \R^k)| = |\calP_j(A\cap \R^k)| - |\calP_j^B(A\cap \R^k)| \geq \frac{1-\lambda}{2}|(\R^k)_j|$, hence
\begin{equation*}
    |(\R^k)_j| \leq  \frac{2}{1 - \lambda}|\fext A \cap \R^k|.
\end{equation*}
    Using that $|(\R^k)_i| \leq (2R)^{d-2} \leq 2^{d-2}|(\R^k)_j|$ for every $i\in\{1,\dots,d\}$, we conclude that 
\begin{equation}\label{Eq: Segundo.bound.soma.projecoes}
    \sum_{i=1}^{d-1} |P_i(A\cap\R^k)| \leq \sum_{i=1}^{d-1} |(\R^k)_i| \leq (d-1)2^{d-2}|(\R^k)_j| \leq \frac{(d-1)2^{d-1}}{1 - \lambda}|\fext A \cap \R^k|.
\end{equation}
\end{itemize}

In both cases we were able to bound the sum of projections by a constant times the size of the boundary of $A$ in $\R^k$. Applying \eqref{Eq: Primeiro.bound.soma.projecoes} and \eqref{Eq: Segundo.bound.soma.projecoes} back in \eqref{Eq: Partition.sum.proj.} we get
\begin{align*}
    \sum_{i=1}^d|\calP_i(A\cap \R)| &\leq
    \sum_{k=1}^{r_d}\left[c\left(d-1, \frac{\lambda + 1}{2}\right)+ \frac{(d-1)2^{d-1}}{1 - \lambda}\right]|\fext A \cap \R \cap L_k| + |\calP_d(A\cap \R)|\\
    &=\left[c\left(d-1, \frac{\lambda + 1}{2}\right)+ \frac{(d-1)2^{d-1}}{1 - \lambda}\right]|\fext A \cap \R| + |\calP_d(A\cap \R)|.
\end{align*}
We finish the proof by noticing that we can repeat this exact same argument but now splitting $\R$ into layers $L_k = \{x\in \R : x_j = k\}$. By doing so, we have that  
\begin{equation*}
    \sum_{i=1}^d|\calP_i(A\cap \R)| \leq
     \left[c\left(d-1, \frac{\lambda + 1}{2}\right)+ \frac{(d-1)2^{d-1}}{1 - \lambda}\right]|\fext A \cap \R| + |\calP_j(A\cap \R)|
\end{equation*}
for any $j\in\{1,\dots, d\}$. Summing both sides in $j$ we conclude

\begin{equation}
    \sum_{i=1}^d|\calP_i(A\cap \R)| \leq \frac{d}{d-1}\left[c\left(d-1, \frac{\lambda + 1}{2}\right)+ \frac{(d-1)2^{d-1}}{1 - \lambda}\right]|\fext A \cap \R|,
\end{equation}
which proves our claim if we take $c(d,\lambda) \coloneqq \frac{d}{d-1}\left[c(d-1, \frac{\lambda + 1}{2})+ \frac{(d-1)2^{d-1}}{1 - \lambda}\right] = 2d + \frac{(d-2)d2^{d-1}}{1-\lambda}$.
\end{proof}

\begin{remark}
This lemma can be proved when $R\leq r_i \leq \kappa R$ for any $\kappa>1$. When applying the lemma, we will choose $\lambda=\frac{7}{8}$ to simplify the notation. All the proofs work as long as we choose $\lambda> \frac{3}{4}$.
\end{remark}

\begin{lemma}\label{Lemma: Proposicao1.Aux1}
    Given $A\subset \Z^d$, $\ell\geq 0$ and $U= C_{\ell}\cup C_{\ell}^\prime$ with $C_{\ell}$ and $C_{\ell}^\prime$ being two $\ell$-cubes sharing a face, there exists a constant $b\coloneqq b(d)$ such that, if 
    
\begin{align}\label{Eq. U.condition}
    \frac{1}{2}|C_{\ell}| \leq |C_{\ell}\cap A| \qquad \text{and} \qquad |C_{\ell}^\prime\cap A|< \frac{1}{2}|C_{\ell}^\prime|
\end{align}
then $2^{\ell(d-1)}\leq b|\fext A\cap U|$.
\end{lemma}

\begin{proof}
For $\ell=0$, \eqref{Eq. U.condition} guarantees that $C_{\ell} = \{x\} \subset A$ and $C_{\ell}^\prime = \{y\}\subset A^c$, hence $|\fext A\cap \{x,y\}| = 1$ and it is enough to take $b\geq 1$. For $\ell \geq 1$, \eqref{Eq. U.condition} yields
\begin{equation}\label{Eq. A.cap.U.volume}
    \frac{1}{2}2^{\ell d} \leq |A\cap U| \leq \frac{3}{2}2^{\ell d}.
\end{equation}

To simplify the notation, we can assume wlog that ${U=[1,2^{\ell}]^{d-1}\times [1, 2^{\ell+1}]}$. As discussed before, for each point $p\in\calP^B_j(A\cap U)$ in the projection, $\ell_p^j\subset A\cap U$ and the lines are disjoint. Moreover, the size of the lines  is constant $r_j\coloneqq |\ell_p^j|$, hence $|\calP_{j}^B(A\cap U)|r_j = \sum_{p\in\calP_{j}^B(A\cap U)} |\ell_p^j| \leq |A\cap U|$. Together with the upper bound \eqref{Eq. A.cap.U.volume}, this yields 
\begin{equation}\label{Eq: Upper.bound.bad.points}
    |\calP_{j}^B(A\cap U)| \leq \frac{3}{2}2^{\ell d}r_j^{-1}.
\end{equation}
Using the isometric inequality, the lower bound on \eqref{Eq. A.cap.U.volume} yields $d2^{\frac{1}{d}}2^{\ell(d-1)}\leq |\fext (A\cap U)|$. As 
\begin{align*}
\frac{1}{2d}|\fext (A\cap U)| &\leq |\fint (A\cap U)| = |\fint(A\cap U) \cap \fint U| + |\fint(A\cap U) \cap(U\setminus \fint U)|\\
&\leq 2\sum_{i=1}^d |\calP_{i}(A\cap U)| + |\fint A\cap U| \leq 2\sum_{i=1}^d |\calP_{i}(A\cap U)| + |\fext A\cap U|,
\end{align*}
we get
\begin{equation}\label{Eq: Lemma.geo.discreta.3}
    2^{\frac{1}{d}-1}2^{\ell(d-1)}\leq 2\sum_{i=1}^d |\calP_{i}(A\cap U)| + |\fext A\cap U|
\end{equation}

We again consider two cases:
\begin{itemize}
    \item[(a)] If $|\calP_{j}(A\cap U)|> \frac{7}{8}|U_j|$ for some $j=1,\dots, d$, by \eqref{Eq: Upper.bound.bad.points} and \eqref{Eq: upper.bound.good.points} we get
    \begin{align*}
    \frac{7}{8}|U_j| < |\calP_{j}(A\cap U)| \leq |\fext A \cap U| + \frac{3}{2}2^{\ell d}r_j^{-1}.
    \end{align*}
    A simple calculation shows that $\frac{1}{8}2^{\ell(d-1)}\leq \frac{7}{8}|U_j| - \frac{3}{2}2^{\ell d}r_j^{-1}$, therefore 
    \begin{equation}\label{Eq: upper.bound.big.projections.1}
        \frac{1}{8}2^{\ell(d-1)} \leq |\fext A \cap U|.
    \end{equation}

    \item[(b)] If $|\calP_{i}(A\cap U)|\leq \frac{7}{8} |U_i|$ for all $i$, by Lemma \ref{Lemma: Geo.discreta.1}, there is a constant $c= c(d)$ such that
    \begin{equation}\label{Eq: Lemma.geo.discreta.2}
        \sum_{i=1}^d |\calP_{i}(A\cap U)|\leq c|\fext A\cap U|.
    \end{equation}
    Together with \eqref{Eq: Lemma.geo.discreta.3}, this yields
    \begin{equation}\label{Eq: upper.bound.big.projections.2}
        2^{\ell(d-1)}\leq \frac{2c+1}{2^{\frac{1}{d}-1}}|\fext A\cap U|.
    \end{equation}
\end{itemize}

Equations \eqref{Eq: upper.bound.big.projections.1} and \eqref{Eq: upper.bound.big.projections.2} shows the desired results taking $b\coloneqq \max \{8, {(2c+1)}{2^{1-\frac{1}{d}}}\}$.
\end{proof}

\begin{proposition}\label{Proposition1_SR}For the functions $B_0,\dots,B_k$ defined above, there exists constants $b_1,b_2$ depending only on $d$ and $r$ such that 
\begin{equation}\label{Eq: Prop.1.FFS.i_SR}
    |\partial\mathfrak{C}_\ell(\gamma)| \leq b_1\frac{|\fext \I(\gamma)|}{2^{\ell(d-1)}} \leq b_1 \frac{|\gamma|}{2^{\ell(d-1)}}
\end{equation}
    and 
\begin{equation}\label{Eq: Prop.1.FFS.ii_SR}
    |B_\ell(\gamma)\Delta B_{\ell+1}(\gamma)| \leq b_2 2^{\ell} |\gamma|
\end{equation}
for every $\ell\in\{0,\dots,k\}$ and $\gamma\in\Gamma_0(n)$.
\end{proposition}

    \begin{proof} Fix $\ell\in\{0,\dots,k\}$. Given a pair $(C_{\ell}, C_\ell^\prime)$, we will write $C_{\ell} \sim C_{\ell}^\prime$ when $(C_{\ell}, C_{\ell}^\prime) \in  \partial\mathfrak{C}_\ell(\gamma)$, and the union we be denoted by $U = C_{\ell} \cup C_{\ell}^\prime$. Then, 
   defining $\overline{\mathscr{C}}_{\ell} = \{C_{\ell} \in \partial \mathfrak{C}_\ell(\gamma): C_{\ell} \sim C_{\ell}^\prime \text{ for some } C_{\ell}^prime \notin  \mathfrak{C}_\ell(\gamma)\}$, for any $A\Subset\Z^d$, 
    
    \begin{align*}
        \sum_{\substack{(C_{\ell}, C^\prime_\ell)\in \partial \mathfrak{C}_\ell(\gamma)}}|A \cap \{C_{\ell} \cup C_{\ell}^\prime\}| &\leq \sum_{\substack{C_{\ell}\in \overline{\mathscr{C}}_{\ell}}}\sum_{\substack{C_{\ell}^\prime\notin  \mathfrak{C}_\ell(\gamma)\\ C_{\ell} \sim C_{\ell}^\prime}}  \left(|A \cap C_{\ell}| + |A \cap C_{\ell}^\prime|\right)\\
        &\leq \sum_{\substack{C_{\ell}\in\overline{\mathscr{C}}_{\ell}}} 2d |A \cap C_{\ell}| + \sum_{C_{\ell}^\prime\notin \mathfrak{C}_\ell(\gamma)} 2d|A \cap C_{\ell}^\prime| \leq  2d |A|\\
    \end{align*}

 Any pair of cubes $C_{\ell}\sim C_{\ell}^\prime$ are in the hypothesis of Lemma \ref{Lemma: Proposicao1.Aux1}, hence $ b 2^{\ell(d-1)} \leq |\fext \I(\gamma) \cap U|$. Applying equation above for $A=\fext\I(\gamma)$ we get that
    \begin{equation*}
       \frac{b}{2d} 2^{\ell(d-1)}| \partial \mathfrak{C}_\ell(\gamma)| \leq \frac{1}{2d} \sum_{\substack{(C_{\ell}, C^\prime_\ell)\in \partial \mathfrak{C}_\ell(\gamma)}} |\fext \I(\gamma) \cap \{C_{\ell} \cup C_{\ell}^\prime\}| \leq |\fext \I(\gamma)|,
    \end{equation*}
that concludes \eqref{Eq: Prop.1.FFS.i_SR} for $b_1\coloneqq 2d/b$.

Given $C_{(\ell+1)}\in \mathscr{C}_{(\ell+1)}(B_{\ell+1}(\gamma)\setminus B_{\ell}(\gamma))$, there is a $\ell$-cube $C_{\ell}^\prime\subset C_{r(\ell + 1)}$ with $C_{\ell}^\prime\notin \mathfrak{C}_\ell(\gamma)$, otherwise  $(B_{\ell+1}(\gamma)\setminus B_{\ell}(\gamma))\cap C_{(\ell+1)} = \emptyset$. There is also a $\ell$-cube $C_{\ell}\subset C_{r(\ell + 1)}$ with $C_{\ell}\in \mathfrak{C}_\ell(\gamma)$, otherwise we would have 
\begin{align*}
    |\I(\gamma)\cap C_{(\ell+1)}| &= \sum_{C_{\ell}\subset C_{(\ell+1)}} |\I(\gamma)\cap C_{\ell}| \leq \frac{1}{2} |C_{(\ell+1)}|.
\end{align*}

Moreover, we can assume that $C_{\ell}$ and $C_{\ell}^\prime$ share a face. Again, we use Lemma \ref{Lemma: Proposicao1.Aux1} to get,
\begin{align}\label{Eq: bound.on.c.bar}
    |B_{\ell+1}(\gamma)\setminus B_\ell(\gamma)\cap C_{(\ell+1)}| &\leq |C_{(\ell+1)}| =2^{d}2^{\ell}2^{\ell(d-1)} \nonumber\\
                &\leq 2^{d}2^{\ell} b|\fext \I(\gamma) \cap \{C_{\ell} \cup C_{\ell}^\prime\}| \nonumber\\
                &\leq 2^{d}2^{\ell}b|\fext \I(\gamma) \cap C_{(\ell+1)}|.
\end{align}
Therefore, 
\begin{align*}
    |B_{\ell+1}(\gamma)\setminus B_\ell(\gamma)| &= \sum_{C_{(\ell+1)}\in \mathscr{C}_{(\ell+1)}(B_{\ell+1}(\gamma)\setminus B_\ell(\gamma))} |B_{\ell+1}(\gamma)\setminus B_\ell(\gamma)\cap C_{(\ell+1)}| \\
    &\leq \sum_{C_{(\ell+1)}\in \mathscr{C}_{(\ell+1)}(B_{\ell+1}(\gamma)\setminus B_\ell(\gamma))} 2^{d}2^{\ell}b|\fext \I(\gamma) \cap C_{(\ell+1)}| \leq  \frac{b_2}{2}2^{\ell}|\fext \I(\gamma)|.
\end{align*}
with $b_2=b2^{d+1}$. To get the same bound for $|B_{\ell}(\gamma)\setminus B_{\ell+1}(\gamma)|$ we repeat a similar argument, covering $B_{\ell}(\gamma)\setminus B_{\ell+1}(\gamma)$ with $(\ell+1)$-cubes. 
    \end{proof}

\begin{corollary}
    For any $\ell>0$ and any two contours $\gamma_1,\gamma_2 \in \Gamma_0(n)$ such that $B_\ell(\gamma_1)=B_{\ell}(\gamma_2)$, there exists a constant $b_3>0$ such that 
    \begin{equation*}
        \d_2(\gamma_1,\gamma_2)\leq 4 \varepsilon b_3 2^{\frac{\ell}{2}} n^{\frac{1}{2}}. 
    \end{equation*} 
\end{corollary}

\begin{proof}
    This is a simple application of the triangular inequality, since $d_2(\gamma_1,\gamma_2) \leq d_2(\gamma_1,B_\ell(\gamma_1)) + d_2(\gamma_2,B_\ell(\gamma_2))$ and 
    \begin{align*}
        d_2(\gamma_1,B_\ell(\gamma_1)) &\leq \sum_{i=1}^\ell d_2(B_i(\gamma_1),B_{i-1}(\gamma_1)) = \sum_{i=1}^\ell 2\varepsilon\sqrt{B_i(\gamma_1)\Delta B_{i-1}(\gamma_1)} \\
        & \leq \sum_{i=1}^\ell 2\varepsilon\sqrt{b_2} 2^{\frac{i}{2}} \sqrt{n}  \leq 2\varepsilon\sqrt{b_2}(\sqrt{2}+1)2^{\frac{\ell}{2}} \sqrt{n} 
    \end{align*}
    where in the second to last equation used \eqref{Eq: Prop.1.FFS.ii_SR}. As the same bound holds for $d_2(\gamma_2,B_\ell(\gamma_2))$, the corollary is proved by taking ${b_3 = \sqrt{b_2}(\sqrt{2}+1)}$
\end{proof}

\begin{remark}\label{Rmk: Bounding_N_by_B_ell_SR}
    This corollary shows that we can create a covering of $\Gamma_0(n)$, indexed by $B_\ell(\Gamma_0(n))$, of ball with radius $4 \varepsilon b_3 2^{\frac{\ell}{2}} n^{\frac{1}{2}}$. Therefore $N(\Gamma_0(n), \d_2, 4\varepsilon b_3 2^{\frac{\ell}{2}} n^{\frac{1}{2}}) \leq |B_\ell(\Gamma_0(n))|$. 
\end{remark}
In the next proposition we bound $|B_\ell(\Gamma_0(n))|$, again following \cite{FFS84}.

\begin{proposition}\label{Prop: Proposition_2_FFS}
    There exists a constant $b_4\coloneqq b_4(d)$ such that, for any $n\in\mathbb{N}$,
    \begin{equation}\label{Eq: Proposition_2_FFS}
        |B_\ell(\Gamma_0(n))|\leq \exp{\left\{\frac{b_4\ell n}{2^{\ell(d-1)}}\right\}},
    \end{equation}
    that is, the number of coarse-grained contours in $B_\ell(\Gamma_0(n))$ is bounded above by an exponential term. 
\end{proposition}

\begin{proof}
Start by noticing that $|B_\ell(\Gamma_0(n))| = |\partial B_\ell(\Gamma_0(n))|$, and to each $B_\ell(\gamma)$ we can associate a contour $\xi_\ell(\gamma)$ with $\I(\xi_\ell) = B_{\ell}(\gamma)$. Given $\xi_\ell\in\xi_\ell(\Gamma_0(n))$ for $\ell\in\{1,\dots,k\}$, let $\{\xi_\ell^{(1)}, \xi_\ell^{(2)},\dots,\xi_\ell^{(m)}\}$ be the connected components of $\xi_\ell$. To connected component $\xi_\ell^{(i)}$ we can uniquely associate a pair $(C_\ell,C^\prime_\ell)\in \partial\mathfrak{C}_\ell(\gamma)$ with $B_{C_\ell}\subset\I(\xi_\ell^{(i)})$. By Lemma \ref{Lemma: Proposicao1.Aux1}, there are at least $b^{-1}2^{\ell (d-1)}$ points of $\fext\I(\gamma)$ in $C_\ell\cup C_\ell^\prime$. Hence $\xi_\ell$ has at most $\frac{b n}{2^{\ell(d-1)}}$ connected components and we take $M_n \coloneqq \frac{b n}{2^{\ell(d-1)}}$.
Moreover, by Lemma \ref{Proposition1_SR}, $ |\xi_\ell| = \sum_{i=1}^{M_n} |\xi_\ell^{(i)}| = 2^{\ell(d-1)}|\partial\mathfrak{C}_\ell(\gamma)| \leq b_1 n$.

Fixed $x_1,x_2,\dots, x_{M_n}\in 2^\ell\Z^d$, and $s^1,s^2,\dots, s^{M_n}\in 2^{\ell (d-1)}\mathbb{N}$, if $\Gamma(\{x_i\}_{i=1}^{M_n},\{s_i\}_{i=1}^{M_n})$ is the number of coarse-grained contours with $x_i\in \I(\xi_\ell^{(i)})$ and $|\xi_\ell^{(i)}| = s^i$ for every $i$, then
\begin{align}\label{Eq: Proposition_2_FFS_Aux_1}
        \Gamma(\{x_i\}_{i=1}^{M_n},\{s_i\}_{i=1}^{M_n}) \leq \exp{\left( \ln{(4d)} \sum_{i=1}^{M_n} \frac{s^i}{2^{\ell (d-1)}}\right)} \leq \exp{\left( b_1\ln{(4d)}   \frac{n}{2^{\ell (d-1)}}\right)}.
\end{align}
The number of choices of $s^1,s^2,\dots, s^1\in 2^{\ell d}\mathbb{N}$ with $\sum_{i=1}^q s^i \leq b_1 n$ is less then $2^{\frac{b_1 n}{2^{\ell(d-1)}}}$. This is a simple bound on the number of ways of putting up to $\frac{b_1 n}{2^{\ell(d-1)}}$ balls on $M_n$ spaces. 

It remains know to bound the number of choices for $({x_i})_{i=1}^{M_n}$. Set $d_1 = |x_1|$ and $d_i = |x_i - x_{i-1}|$ for $i=2,\dots, M_n$. For every $i=1,\dots, M_n$, we choose $y_1,\dots,y_{M_n}\in \I({\gamma})$ such that $|x_i - y_i| = d(x_i, \I({\gamma}))$. With this choice we get $d(x_i,y_i) \leq d2^{\ell}$. Since all the cubes are not connected, $d(y_i,y_{i-1})>2^\ell$, hence
\begin{equation*}
    d(x_i,y_i) \leq d |y_i - y_{i-1}|.
\end{equation*}
As all $y_i$ are in $\I(\gamma)$ and $y_0=x_0=0$, we can reorder the terms to minimize the sum of distances, getting 
\begin{equation*}
    \sum_{i=1}^{M_n} d(y_i, y_{i-1}) \leq 2d |\fext\I(\gamma)| = 2dn.
\end{equation*}
This yields 
\begin{equation}\label{Eq: bound.sum.of.d_i}
    \sum_{i=1}^{M_n} d_i \leq 2\sum_{i=1}^{M_n} d(x_i,y_i) + \sum_{i=1}^{M_n} d(y_i,y_{i-1}) \leq (2d + 1)\sum_{i=1}^{M_n} d(y_i, y_{i-1}) \leq (2d + 1)^2 n
\end{equation}
 Fixing $d_1, d_2,\dots, d_q$, the number of ways of choosing $x_1,\dots,x_q$ is bounded by $\prod_{i=1}^{M_n} (2d_i)^{d}$. The maximum of this quantity is reached when all the distances are the same. Assuming $d_1=\dots =d_q=d^*$, equation \eqref{Eq: bound.sum.of.d_i} yields 
 \begin{equation*}
     d^* \leq \frac{(2d + 1)^2 n}{M_n} = \frac{(2d+1)^2 2^{\ell(d-1)}}{b},
 \end{equation*}
so we have at most 
\begin{equation}\label{Eq: Proposition_2_FFS_Aux_2}
    (\frac{(2d+1)^2 2^{\ell(d-1)}}{b})^{\frac{d b n}{2^{\ell(d-1)}}} \leq \exp\left\{(2d(d-1)\ln{(2)}b\ln{(\frac{(2d+1)}{b})}) \frac{\ell n}{2^{\ell(d-1)}}\right\}
\end{equation}
 ways of choose $x_1,\dots,x_{M_n}$ given $d_1,\dots,d_{M_n}$. The number of solutions $(d_1,\dots,d_{M_n})$ to $\sum_{i=1}^{M_n} d_i =N $ is $\binom{N-1}{M_n}$. As $\binom{N-1}{M_n} < \frac{N^{M_n}}{M_n!}\leq (\frac{eN}{M_n})^{M_n}$, the number of solutions of \eqref{Eq: bound.sum.of.d_i} is bounded by 
 \begin{align}\label{Eq: Proposition_2_FFS_Aux_3}
     \sum_{N=1}^{(2d+1)^2n}\left(\frac{eN}{M_n}\right)^{M_n} &\leq\int_{0}^{(2d+1)^2n +1} \left(\frac{ex}{M_n}\right)^{M_n}dx =  e^{M_n}\frac{(2d+1)^2n +1}{M_n + 1}\left(\frac{(2d+1)^2n +1}{M_n} \right)^{M_n} \nonumber \\
     &\leq \left(\frac{2e(2d+1)^2 2^{\ell(d-1)}}{b} \right)^{\frac{b n}{2^{\ell(d-1)}}} \leq \exp{\left\{4eb^{-1}\ln(2)(2d+1)^2 \frac{\ell n}{2^{\ell(d-1)}} \right\}}.
 \end{align}
 Taking $b_4=\max{\{ b_1\ln{(4d)}, 2d(d-1)\ln{(2)}b\ln{(\frac{(2d+1)}{b})}, eb^{-1}\ln(2)(2d+1)^2 \}}$, equations \eqref{Eq: Proposition_2_FFS_Aux_1}, \eqref{Eq: Proposition_2_FFS_Aux_2} and \eqref{Eq: Proposition_2_FFS_Aux_3} proves the proposition. 
\end{proof}

We are ready to prove the main proposition.

\textit{Proof of Proposition \ref{Prop: Bound.gamma_2_SR}:} As $N(\Gamma_0(n), \d_2, \epsilon)$ is decreasing in $\epsilon$, we can use Dudley's integral bound to bound
    \begin{align}\label{Eq: Prop_bound_Ec_1}
        {\mathbb{E}\left[\sup_{\gamma\in\Gamma_0(n)}{\Delta_{\I(\gamma)}(h)}\right]} &\leq \int_{0}^\infty \sqrt{\log N(\Gamma_0(n), \d_2, \epsilon)}d\epsilon \nonumber\\
        &\leq 2\varepsilon b_3 n^{\frac{1}{2}}\sum_{\ell=1}^\infty (2^{\frac{\ell}{2}} - 2^{\frac{\ell-1}{2}})\sqrt{\log N(\Gamma_0(n), \d_2,\varepsilon b_3 2^{\frac{\ell}{2}}n^{\frac{1}{2}})}.
    \end{align}

Remember that, as discussed in Remark \ref{Rmk: Bounding_N_by_B_ell_SR}, $N(\Gamma_0(n), \d_2,\varepsilon b_3 2^{\frac{\ell}{2}}n^{\frac{1}{2}})\leq |B_{\ell}(\Gamma_0(n))|$. Therefore, by Proposition \ref{Prop: Proposition_2_FFS},
\begin{align}\label{Eq: Prop_bound_Ec_2}
    \sum_{\ell=1}^\infty (2^{\frac{\ell}{2}} - 2^{\frac{\ell-1}{2}})\sqrt{\log N(\Gamma_0(n), \d_2,\varepsilon b_3 2^{\frac{\ell}{2}}n^{\frac{1}{2}})} &\leq \sum_{\ell=1}^{\infty} (2^{\frac{\ell}{2}} - 2^{\frac{\ell-1}{2}}) \sqrt{\frac{b_4\ell n}{2^{\ell(d-1)}}} \nonumber \\
    &\leq \sqrt{b_4}n^{\frac{1}{2}}(1 - \frac{\sqrt{2}}{2})\sum_{\ell=1}^{\infty}\sqrt{\frac{\ell}{2^{\ell(d-2)}}}.
\end{align}
Denoting $\tau(d) = \sum_{\ell=1}^{\infty}\sqrt{\frac{\ell}{2^{\ell(d-2)}}}$, and $b_5 = 2\tau(d)b_3\sqrt{b_4}(1 - \frac{\sqrt{2}}{2})$, equation \eqref{Eq: Prop_bound_Ec_1} and \eqref{Eq: Prop_bound_Ec_2} yields
\begin{equation*}
    {\mathbb{E}\left[\sup_{\gamma\in\Gamma_0(n)}{\Delta_{\I(\gamma)}(h)}\right]} \leq b_5 \varepsilon n.
\end{equation*}

    \subsection{Phase transition for the RFIM}
    Let us first prove an auxiliary lemma showing that if $\mu_{\Lambda; \beta, \varepsilon h}^{+, \IS}(\sigma_0 = -1) > c $ has positive probability, then $\mu_{\Lambda; \beta, \varepsilon h}^{+, \IS}\neq \mu_{\Lambda; \beta, \varepsilon h}^{-, \IS}$ $\mathbb{P}$-almost surely.

\begin{lemma}
    If there exists a constant $0\leq c<\frac{1}{2}$ such that, for all $\Lambda\Subset\Z^d$, 
    \begin{equation*}
        \mathbb{P}\left(  \mu_{\Lambda; \beta, \varepsilon h}^{+, \IS}(\sigma_0 = -1) < c \right) > 1-c,
    \end{equation*}
    then $ \mu_{\beta, \varepsilon h}^{+, \IS}\neq  \mu_{\beta, \varepsilon h}^{-, \IS}$ $\mathbb{P}$-almost surely.
\end{lemma}

 \begin{proof}
     Notice first that, as the function $\mathbbm{1}_{\{\sigma_0 = -1\}}$ is non-increasing, by FKG $\mu_{\Lambda; \beta, \varepsilon h}^{+, \IS}(\sigma_0 = -1) \leq \mu_{\Delta; \beta, \varepsilon h}^{+, \IS}(\sigma_0 = -1)$ for all $\Lambda\subset \Delta \Subset \Z^d$. As $\mu_{\beta, \varepsilon h}^{+, \IS}$ is the weak limit of the local measures $\mu_{\Lambda_N; \beta, \varepsilon h}^{+, \IS}$, by continuity of $\mathbb{P}$ we have
     \begin{equation*}
          \mathbb{P}\left(  \mu_{\beta, \varepsilon h}^{+, \IS}(\sigma_0 = -1) < c \right) = \lim_{N\to\infty} \mathbb{P}\left(  \mu_{\Lambda_N; \beta, \varepsilon h}^{+, \IS}(\sigma_0 = -1) < c \right) > 1-c.
     \end{equation*}
     Like in the usual Ising model, we can write 
     \begin{equation*}
          \mu_{\beta, \varepsilon h}^{+, \IS}(\sigma_0 ) =  1 - 2\mu_{\beta, \varepsilon h}^{+, \IS}(\sigma_0 = -1), 
     \end{equation*}
     so $\{h: \mu_{\beta, \varepsilon h}^{+, \IS}(\sigma_0) > 0\} \supset \{h: \mu_{\beta, \varepsilon h}^{+, \IS}(\sigma_0 = -1) < c\}$ and therefore $\mathbb{P}\left(  \mu_{\beta, \varepsilon h}^{+, \IS}(\sigma_0)>0 \right) > 1-c$. Moreover,
     \begin{multline*}
        \mathbb{P}\left( \mu_{\beta, \varepsilon h}^{+, \IS}(\sigma_x) \neq \mu_{\beta, \varepsilon h}^{-, \IS}(\sigma_x) \text{ for some }x\in\Z^d\right) \geq  \mathbb{P}\left( \mu_{\beta, \varepsilon h}^{+, \IS}(\sigma_0)>0,  \mu_{\beta, \varepsilon h}^{-, \IS}(\sigma_0)<0\right) \\
        \geq \mathbb{P}\left( \mu_{\beta, \varepsilon h}^{+, \IS}(\sigma_0)>0\right) + \mathbb{P}\left( \mu_{\beta, \varepsilon h}^{-, \IS}(\sigma_0)<0\right) - 1 \geq 1-2c>0,
     \end{multline*}
  where in the second inequality we used that $\mathbb{P}(A\cap B)\geq \mathbb{P}(A) + \mathbb{P}(B) - 1$ for any events $A$ and $B$. 
  Considering the translation map $T(h)=(h_{x + e_1})_{x\in\Z^d}$ and $\mathcal{B}(\mathbb{R}^{\Z^d})$ the Borel $\sigma$-algebra of $\mathbb{R}^{\Z^d}$, $\mathbb{P}$ is mixing in the dynamical system $(R^{\Z^d}, \mathcal{B}, \mathbb{P}, T)$ and $\{h: \mu_{\beta, \varepsilon h}^{+, \IS}(\sigma_x) \neq \mu_{\beta, \varepsilon h}^{-, \IS}(\sigma_x) \text{ for some }x\in\Z^d\}$ is a translation invariant event with positive probability, therefore $\mathbb{P}\left( \mu_{\beta, \varepsilon h}^{+, \IS}(\sigma_x) \neq \mu_{\beta, \varepsilon h}^{-, \IS}(\sigma_x) \text{ for some }x\in\Z^d\right) =1$. We conclude the proof by noticing that $ \mathbb{P}\left( \mu_{\beta, \varepsilon h}^{+, \IS} \neq \mu_{\beta, \varepsilon h}^{-, \IS}\right) \geq  \mathbb{P}\left( \mu_{\beta, \varepsilon h}^{+, \IS}(\sigma_x) \neq \mu_{\beta, \varepsilon h}^{-, \IS}(\sigma_x) \text{ for some }x\in\Z^d\right)$.
 \end{proof}
\begin{theorem}
For $d\geq 3$, there exists a constant $C\coloneqq C(d,\alpha)$ such that, for all $\beta>0$, $e\leq C$ and $\Lambda\Subset\Z^d$, the event 
    \begin{equation}\label{Eq: PTLR_SR}
        \mu_{\Lambda; \beta, \varepsilon h}^{+, \IS}(\sigma_0 = -1) \leq e^{-C\beta} + e^{-C/\varepsilon^2} 
    \end{equation}
    has $\mathbb{P}$-probability bigger then $1 - e^{-C\beta} - e^{-C/\varepsilon^2}$.\\
    
In particular, for $\beta>\beta_c$ and $\varepsilon$ small enough, there is a phase transition for the long-range Ising model.  
\end{theorem}

\begin{proof}
        The proof is an application of the Peierls' argument, but now on the joint measure $\mathbb{Q}$. By Proposition \ref{Prop: Bound.bad.event_SR}, we have
        \begin{align}\label{Eq: Upper.bound.on.Q.1_SR}
            \mathbb{Q}_{\Lambda; \beta, \varepsilon}^{+, \IS}(\sigma_0 = -1) &=  \mathbb{Q}_{\Lambda; \beta, \varepsilon}^{+, \IS}(\{\sigma_0 = -1\} \cap \mathcal{E}) + \mathbb{Q}_{\Lambda; \beta, \varepsilon}^{+, \IS}(\{\sigma_0 = -1\}\cap \mathcal{E}^c) \nonumber \\
            & \leq \mathbb{Q}_{\Lambda; \beta, \varepsilon}^{+, \IS}(\{\sigma_0 = -1\} \cap \mathcal{E}) +  e^{-C_1/\varepsilon^2} 
        \end{align}
since $\mathbb{Q}_{\Lambda; \beta, \varepsilon}^{+, \IS}(\{\sigma_0 = -1\}\cap \mathcal{E}^c) \leq \mathbb{Q}_{\Lambda; \beta, \varepsilon}^{+, \IS}(\mathcal{E}^c) = \mathbb{P}(\mathcal{E}^c)$.  When $\sigma_0 = -1$, there must exist a contour $\gamma$ with $0\in V(\gamma)$, hence
\begin{equation*}
    \mu_{\Lambda; \beta, \varepsilon h}^{+, \IS}(\sigma_0 = -1) \leq \sum_{\gamma \in \mathcal{C}_0}\mu_{\Lambda; \beta, \varepsilon h}^{+, \IS}(\Omega(\gamma)),
\end{equation*}
where $\Omega(\gamma) \coloneqq \{\sigma\in\Omega : \gamma \subset \Gamma(\sigma)\}$. So we can write

\begin{align}\label{Eq: Upper.bound.on.Q.2_SR}
    \mathbb{Q}_{\Lambda; \beta, \varepsilon}^{+, \IS}(\{\sigma_0 = -1\} \cap \mathcal{E}) &= \int_{\mathcal{E}}\sum_{\sigma : \sigma_0 = -1}g_{\Lambda; \beta, \varepsilon}^{+, \IS}(\sigma, h)dh \nonumber \\
    &\leq  \sum_{\gamma\in\mathcal{C}_0} \int_{\mathcal{E}}\sum_{ \sigma\in\Omega(\gamma)}g_{\Lambda; \beta, \varepsilon}^{+, \IS}(\sigma, h)dh \nonumber \\
    &\leq  \sum_{\gamma \in \mathcal{C}_0} \frac{\int_{\mathcal{E}}\sum_{\sigma\in\Omega(\gamma)}g_{\Lambda; \beta, \varepsilon}^{+, \IS}(\sigma, h)dh}{\int_{\mathcal{E}}\sum_{\sigma\in\Omega(\gamma)}g_{\Lambda; \beta, \varepsilon}^{+, \IS}(\tau_{\gamma}(\sigma), \tau_{\I(\gamma)}(h))dh} \nonumber \\
    & \leq \sum_{\gamma\in\mathcal{C}_0}\sup_{\substack{h\in\mathcal{E}\\ \sigma\in\Omega(\gamma)}}\frac{g_{\Lambda; \beta, \varepsilon}^{+, \IS}(\sigma, h)}{g_{\Lambda; \beta, \varepsilon}^{+, \IS}(\tau_{\gamma}(\sigma), \tau_{\I(\gamma)}(h))}.
\end{align}

In the second inequality we used that $\int_{\mathcal{E}}\sum_{\sigma\in\Omega(\gamma)}g_{\Lambda; \beta, \varepsilon}^{+, \IS}(\tau_{\gamma}(\sigma), \tau_{\I(\gamma)}(h))dh =1$. By \eqref{Def: Delta_A_SR} and the definition of the event $\mathcal{E}$, 
\begin{align}\label{Eq: Upper.bound.on.Q.3_SR}
    \sup_{\substack{h\in\mathcal{E}\\ \sigma\in\Omega(\gamma)}}\frac{g_{\Lambda; \beta, \varepsilon}^{+, \IS}(\sigma, h)}{g_{\Lambda; \beta, \varepsilon}^{+, \IS}(\tau_{\gamma}(\sigma), \tau_{\I(\gamma)}(h))} &\leq \sup_{\substack{h\in\mathcal{E}\\ \sigma\in\Omega(\gamma)}}  \exp{\{{- \beta c_1 |\gamma|}\}}\frac{Z_{\Lambda; \beta, \varepsilon}^{+, \IS}(\tau_{\I(\gamma)}(h))}{Z_{\Lambda; \beta, \varepsilon}^{+, \IS}(h)} \nonumber\\
    &= \sup_{\substack{h\in\mathcal{E}\\ \sigma\in\Omega(\gamma)}}  \exp{\{{- \beta c_1 |\gamma| + \beta \Delta_{\gamma}(h)}\}} \nonumber\\
    &\leq  \exp{\{{- \beta \frac{c_1}{2} |\gamma| }\}},
\end{align}
since $\Delta_{\gamma}(h) \leq \frac{1}{2}(c_1|\gamma|)$, for all $h\in\mathcal{E}$. Equations \eqref{Eq: Upper.bound.on.Q.1_SR}, \eqref{Eq: Upper.bound.on.Q.2_SR} and \eqref{Eq: Upper.bound.on.Q.3_SR} yields
\begin{align*}
     \mathbb{Q}_{\Lambda; \beta, \varepsilon}^{+, \IS}(\sigma_0 = -1) &\leq  \sum_{\substack{\gamma\in \mathcal{E}_\Lambda^{+, \IS}\\ 0\in V(\gamma)}} 2^{|\gamma|}\exp{\{{- \beta \frac{c_1}{2} |\gamma| }\}} + e^{-c_0/\varepsilon^2}\\
     &\leq \sum_{n\geq 1}\sum_{\substack{\gamma\in \mathcal{E}_\Lambda^+, |\gamma|=n \\ 0\in V(\gamma)}} \exp{\{{(-\beta \frac{c_1}{2} + \ln2)n}\}} + e^{-c_0/\varepsilon^2}\\
     &\leq \sum_{n\geq 1}|\mathcal{C}_0(n)| \exp{\{{(-\beta \frac{c_1}{2} +\ln2)n}\}} + e^{-c_1/\varepsilon^2} \leq \sum_{n\geq 1} e^{(c_2 -\beta \frac{c_1}{2} +\ln2)n} + e^{-c_0/\varepsilon^2}. \\
\end{align*}
When $\beta$ is large enough, the sum above converges and there exists a constant $C$ such that   
\begin{equation*}
    \mathbb{Q}_{\Lambda; \beta, \varepsilon}^{+, \IS}(\sigma_0 = -1) \leq e^{-\beta 2C} + e^{-2C / \varepsilon^2}.
\end{equation*}
The Markov Inequality finally yields
\begin{align*}
    \mathbb{P}\left( \mu_{\Lambda; \beta, \varepsilon h}^{+, \IS}(\sigma_0 = -1) \geq e^{-C\beta} + e^{-C/\varepsilon^2}\right) &\leq \frac{\mathbb{Q}_{\Lambda; \beta, \varepsilon}^{+, \IS}(\sigma_0 = -1)}{e^{-C\beta} - e^{-C/\varepsilon^2}} \\
    &\leq \frac{e^{-\beta 2C} + e^{-2C / \varepsilon^2}}{e^{-C\beta} + e^{-C/\varepsilon^2}} \leq e^{-C\beta} + e^{-C/\varepsilon^2},
\end{align*}
what proves our claim.
\end{proof}


\chapter{Long-range Random Field Ising Model}

In this chapter, we prove phase transition in the long-range random field Ising model. First, we use an argument by Ginibre, Grossman, and Ruelle \cite{Ginibre.Grossmann.Ruelle.66} to extend the results from the previous chapter in the region $\alpha>d+1$. When $d<\alpha \leq d+1$, we need to work with a different contour system.  By adapting the contour argument presented in \cite{Affonso.2021},  we extend the arguments of \cite{Ding2021} to the long-range model, proving the main result:
\begin{theorem*}Given $d\geq 3$, $\alpha>d$, there exists $\beta_c\coloneqq\beta(d, \alpha)$ and $\varepsilon_c\coloneqq\varepsilon(d, \alpha)$ such that, for $\beta> \beta_c$ and $\varepsilon\leq \varepsilon_c$, the extremal Gibbs measures $\mu_{\beta, \varepsilon}^+$ and $\mu_{\beta, \varepsilon}^-$ are distinct, that is, $\mu_{\beta, \varepsilon}^+ \neq \mu_{\beta, \varepsilon}^-$ $\mathbb{P}$-almost surely. Therefore the long-range random field Ising model presents phase transition.
\end{theorem*}
We now present a sketch of the main steps of the proof. \\

\textit{Ideas of the proof:} We first introduce a suitable notion of contour, for which we can control both the energy cost of erasing a contour (Proposition \ref{Prop: Cost_erasing_contour}) and the number of contours of a fixed size that surrounds the origin (Corollary \ref{Cor: Bound_on_C_0_n}). Our contours, as in the Pirogov-Sinai theory, are composed of a support, that represents the incorrect points of a configuration, and the plus and minus interior, which are the regions inside the contours. We denote $\mathcal{C}_0$ the set of all contours surrounding the origin and $\mathcal{C}_0(n)$ the set of contours in $\mathcal{C}_0$ with $n$ points in the support. We also denote $\I_-(n)$ the set of all subsets of $\Z^d$ that are the minus interior of a contour in $\mathcal{C}_0(n)$. 

By the Ding-Zhuang method, we can show that phase transition follows from controlling the probability of the bad event
$$\mathcal{E}^c\coloneqq \left\{\sup_{\substack{\gamma\in\mathcal{C}_0}} \frac{|\Delta_{\I_-(\gamma)}(h)|}{c_2|\gamma|} > \frac{1}{4}\right\},$$
where $|\gamma|$ is the size of the support of a contour $\gamma$ and $(\Delta_A)_{A\Subset \Z^d}$ is a family of functions that, by Lemma \ref{Lemma: Concentration.for.Delta.General}, have the same tail of $\sum_{x\in A}h_x$, and the distribution of $\Delta_A(h) - \Delta_{A^\prime}(h)$ is the same as $\Delta_{A\Delta A^\prime}(h)$, for all $A, A^\prime\in\Z^d$ finite, see \cite{Ding2021}. 

When the contours are connected, all the interiors $\I_-(\gamma)$ are connected, and the two properties in Lemma \ref{Lemma: Concentration.for.Delta.General} together with the coarse-graining procedure introduced by Fisher, Fr\"ohlich, and Spencer in \cite{FFS84} is enough to control $\mathbb{P}(\mathcal{E}^c)$. Given a family of scales $(2^{r\ell})_{\ell\geq 0}$, with $r$ being a suitable constant, we can partition $\Z^d$ into disjoint fitting cubes with sides $2^{r\ell}$, see Figure \ref{Fig: Cubes}. Each such cube is called an $r\ell$-cube, and all cubes throughout our analysis will be of this form unless stated otherwise. The strategy of the coarse-graining argument is to, at each scale, approximate each interior $\I_-(\gamma)$ by a simpler region $B_\ell(\gamma)$, formed by the union of disjoint cubes with side length $2^{r\ell}$, see Figure \ref{Fig: Figura7}. The argument follows once you have two estimations: on the error of this approximation and on the size of the set $B_\ell (\mathcal{C}_0(n)) \coloneqq \{B\subset\Z^d: B=B_\ell(\gamma), \text{ for some }\gamma \in\mathcal{C}_0(n)\}$ containing all regions that are approximations of contours in $\mathcal{C}_0(n)$. 

\begin{figure}[ht]
     \centering
     \input{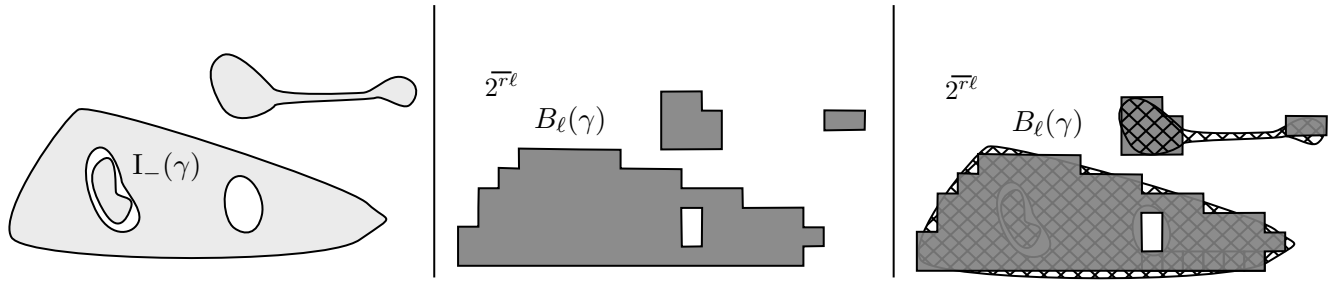}
        \caption{The figure on the left represents the minus interior of a contour $\gamma$, while the central figure represents its approximation. The picture on the right depicts the error of the approximation, that is both the cross-hatched regions not covered by the gray area and the gray areas not intersected by the cross-hatch region.}
    \label{Fig: Figura7}
\end{figure}

In Corollary \ref{Cor: Bound_diam_B_ell}, we show that $|B_\ell(\gamma)\Delta \I_-(\gamma)|\leq c2^{r\ell}|\gamma|$, so the error in the approximation is not too large. This bound follows \cite{FFS84} closely since we approximate the interiors in the same way. Then, we need to estimate $|B_\ell (\mathcal{C}_0(n))|$, which is done by a fairly distinct argument. The difficulty of the proof comes from the fact that our contours may be disconnected. As the regions $B_\ell(\gamma)$ are the union of disjoint cubes, they are, up to a constant, determined by $\fint \mathfrak{C}_{\ell}(\gamma)$, the collection of cubes needed to cover $B_{\ell}(\gamma)$ that share a face with a cube that does not intersect $B_{\ell}(\gamma)$, see Figure \ref{Fig: Figura8}. 

\begin{figure}[h]
    \centering
    \input{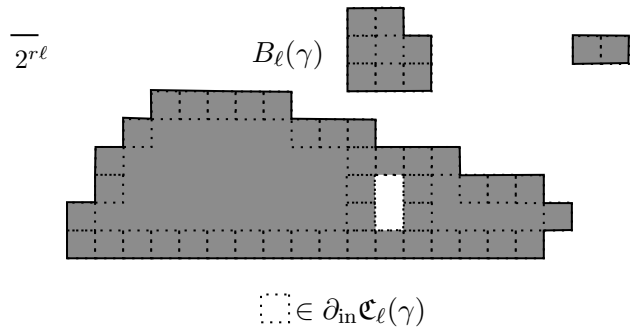}
    \caption{The region $B_\ell(\gamma)$ is the approximation of the interior in Figure \ref{Fig: Figura7}, and $\fint \mathfrak{C}_{r\ell}(\gamma)$ is the collection containing all the doted cubes.}
    \label{Fig: Figura8}
\end{figure}

Following the argument of Fisher, Fr\"ohlich, and Spencer we show that $|\fint \mathfrak{C}_{r\ell}(\gamma)|\leq M_{|\gamma|,\ell}$ (see Proposition \ref{Proposition1}), so we can bound $|B_{\ell}(\mathcal{C}_0(n))|$ by counting all possible choices of $M_{n,\ell}$ non-intersecting cubes with side length $2^{r\ell}$ in $\Z^d$, with a suitable restriction. When the contours are connected, this restriction is that all cubes must be close to a surface with size $n$, and the proof follows once you can use that the minimal path connecting all the cubes has length at most $cn$. This is not true for our contours, so we need a different strategy. 

Let $\C_{rL}(\gamma)$ be the smallest collection of cubes, in the $rL$ scale, needed to cover $\gamma$. The property we will use is that, for all $L\geq \ell$, every cube in $\fint \mathfrak{C}_\ell(\gamma)$ is covered or is next to a cube in  $\C_{rL}(\gamma)$, see Figure \ref{Fig: Figura9}. As every cube with side length $2^{rL}$ contains $2^{rd(L-\ell)}$ cubes with side length $2^{r\ell}$, any fixed collection $\C_{rL}(\gamma)$ covers $2^{rd(L-\ell)}|\C_{rL}(\gamma)|$ cubes in the $r\ell$ scale.

\begin{figure}[ht]
    \centering
    \input{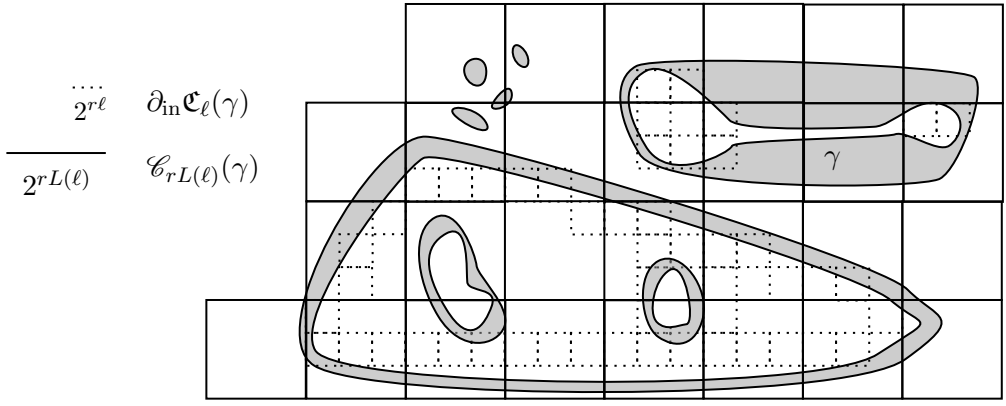}
    \caption{The contour $\gamma$ is the one the originates the interior $\I_-(\gamma)$ in Figure \ref{Fig: Figura7}. The dotted cubes are the cubes in $\fint \mathfrak{C}_\ell(\gamma)$ and the larger cubes are the ones needed to cover the contour $\gamma$.}
    \label{Fig: Figura9}
\end{figure}

So, roughly, we can bound $|B_{\ell}(\mathcal{C}_0(n))|$ by counting all the possible choices of $M_{n,\ell}$ cubes in the $2^{r\ell}$ scale that are covered by a collection in $\C_{rL}(\mathcal{C}_0(n))\coloneqq \{\C_{rL} : \C_{rL} = \C_{rL}(\gamma) \text{ for some }\gamma\in\mathcal{C}_0(n)\}$, that is, we can bound 

\begin{equation*}
    |B_{\ell}(\mathcal{C}_0(n))|\leq \sum_{\C_{rL}\in \C_{rL}(\mathcal{C}_0(n))}\binom{2^{rd(L-\ell)}|\C_{rL}|}{M_{n,\ell}}.
\end{equation*}

The construction of our contours allow us to upper bound $|\C_{rL}(\gamma)|$  and $|\C_{rL}(\mathcal{C}_0(n))|$, see Proposition \ref{Prop. Bound.on.C_rl(gamma)} and Proposition \ref{Prop: Bound_on_rl_coverings}, and we are able to get the same bound for $ |B_{\ell}(\mathcal{C}_0(n))|$ as in \cite{FFS84} by making an appropriate choice of the large scale $L = L(\ell)$ depending on the smaller one. From this, the control on the probability of the bad event follows as an application of Dudley's entropy bound.

In Section 1, we define the model of interest and the notion of contour we will use. Then, in Section 2, we present the Ding and Zhuang method \cite{Ding2021} of proofing phase transition and the bad events for the  RIFM with long-range interaction. After that, we control the probability of the bad events by extending the arguments of \cite{FFS84} to our case. Finally, in the last section, we prove the phase transition for the long-range RFIM. 

\section{The Long-range RFIM}  
The set of configurations of the long-range Ising model is, as usual, $\Omega \coloneqq \{-1,1\}^{\Z^d}$. However, each spin interacts with all others, not only its nearest neighbors, so the interaction $\{J_{xy}\}_{x,y\in\Z^d}$ is defined as

\begin{equation}\label{Long-Range Interaction}
    J_{xy} = \begin{cases}
                   \frac{J}{|x-y|^\alpha} &\text{ if }x\neq y,\\
                   0                &\text{otherwise,} 
              \end{cases}
\end{equation}
where $J >0$, $\alpha>d$ and the distance $|x-y|$ is given by the $\ell_1$-norm. We write $\Lambda\Subset \Z^d$ to denote a finite subset of $\Z^d$
. Fixed such $\Lambda$, the \textit{local configurations} are $\Omega_\Lambda\coloneqq \{-1,1\}^\Lambda$. Moreover, given ${\eta\in\Omega}$, the set of local configurations with $\eta$ boundary conditions is ${\Omega_\Lambda^\eta\coloneqq \{\sigma\in\Omega : \sigma_x=\eta_x, \text{ }\forall x\in\Lambda^c\}}$. The \textit{local Hamiltonian of the random field long-range Ising model} in $\Lambda\Subset\Z^d$ with $\eta$-boundary condition is $H_{\Lambda; \varepsilon h}^{\eta}:  \Omega_\Lambda^\eta \to \mathbb{R}$, given by 

\begin{equation}
    H_{\Lambda; \varepsilon h}^{\eta}(\sigma) \coloneqq -\sum_{x,y\in\Lambda} J_{xy}\sigma_x\sigma_y - \sum_{x\in \Lambda, y\in\Lambda^c} J_{xy}\sigma_x\eta_y - \sum_{x\in\Lambda} \varepsilon h_x\sigma_x,
\end{equation}
where the external field is a family $\{h_x\}_{x\in\Z^d}$ of i.i.d. random variables in $(\widetilde{\Omega}, \mathcal{A}, \mathbb{P})$, and every $h_x$ has a standard normal distribution\footnote{ Our results also hold for more general distributions of $h_x$, see Remarks \ref{Rmk: Bernoulli_external_field} and \ref{Rmk: More.general.h_x}. }. The parameter $\varepsilon >0$ controls the variance of the external field. Given $\Lambda\Subset\Z^d$, consider $\mathscr{F}_\Lambda$ the $\sigma$-algebra generated by the cylinders sets supported in $\Lambda$ and $\mathscr{F}$ the $\sigma$-algebra generated by finite union of cylinders. One of the main objects of study in classical statistical mechanics are the \textit{finite volume Gibbs measures}, which are probability measures in $(\Omega, \mathscr{F})$, given by 
    \begin{equation}
        \mu_{\Lambda;\beta, \varepsilon h}^\eta(\sigma) \coloneqq \mathbbm{1}_{\Omega_\Lambda^\eta}(\sigma)\frac{e^{-\beta H_{\Lambda, \varepsilon h}^{\eta}(\sigma)}}{Z_{\Lambda; \beta, \varepsilon}^{\eta}(h)},
    \end{equation}
where $\beta>0$ is the inverse temperature and $Z_{\Lambda; \beta, \varepsilon}^{\eta}$ is the \textit{partition function}, defined as 

\begin{equation}
    Z_{\Lambda; \beta, \varepsilon}^{\eta}(h)\coloneqq \sum_{\sigma\in\Omega_\Lambda^\eta} e^{-\beta H_{\Lambda, \varepsilon h}^{\eta}(\sigma)}.
\end{equation}
As in the short-range case, since the external field is random, the Gibbs measures are random variables. To make the dependence of $\mu_{\Lambda;\beta, \varepsilon h}^\eta$ on $\widetilde{\Omega}$ explicit, we write $\mu_{\Lambda;\beta, \varepsilon h}^\eta[\omega]$, with $\omega$ being a general element of $\widetilde{\Omega}$. Two particularly important boundary conditions are given by the configurations $\eta_{+} \equiv +1$ and $\eta_{-} \equiv -1$, and are called $+$ and $-$ boundary conditions, respectively. For these boundary conditions, we can $\mathbb{P}$-almost surely define the infinite volume measures by taking the weak*-limit
\begin{equation}
    \mu_{\beta,\varepsilon h}^{\pm}[\omega] \coloneqq \lim_{n\to\infty} \mu_{\Lambda_n;\beta, \varepsilon h}^{\pm}[\omega],
\end{equation}
where $(\Lambda_n)_{n\in\mathbb{N}}$ is any sequence invading $\Z^d$, that is, for any subset $\Lambda\Subset\mathbb{Z}^d$, there exists $N=N(\Lambda)>0$ such that $\Lambda\subset\Lambda_n$ for every $n>N$.
    By Lemma \ref{extremality.of.+.and.-.bc}, for any fixed external field, the measures $\mu_{\Lambda_n;\beta, \varepsilon h}^{\pm}[\omega]$ are monotone, which guarantees the existence of the limits over sequences invading $\Z^d$. 
To have more than one Gibbs measure, it is enough to show that $\mu_{\beta,\varepsilon h}^{+}[\omega]\neq  \mu_{\beta,\varepsilon h}^{-}[\omega]$, with $\mathbb{P}$-probability 1, see \cite[Theorem 7.2.2]{Bovier.06}.

    \subsection{Phase Transition for \texorpdfstring{$\alpha>d+1$}{a>d+1}} 
     When $\alpha>d+1$, Ginibre, Grossman and Ruelle \cite{Ginibre.Grossmann.Ruelle.66} showed that we can use the standard Peierls' argument, with the short-range contours, to prove phase transition for the long-range model. This happens since we can extend the bound \eqref{Eq: Energy_cost_erasing_contour_SR} a class of models with long-range interactions. This is shown in the next proposition

 \begin{proposition}\label{Prop: Erasing_contours_alpha>d+1}
    For the long-range Ising model with $\alpha > d+1$, there is a constant $c_1(\alpha, d)>0$ such that, for any $\sigma\in\Omega$ and $\gamma\in\Gamma(\sigma)$
    \begin{equation}
        H_{\Lambda, 0}^{+}(\tau_{\gamma}(\sigma)) - H_{\Lambda, 0}^{+}(\sigma) \leq - J c_1(\alpha) |\gamma|.
    \end{equation}
    Here, $\Gamma(\sigma)$ is the collection of short-range contours defined in the previous chapter.
\end{proposition}
\begin{proof}
    A straightforward computation shows that
    \begin{align*}
        H_{\Lambda, 0}^{+}(\tau_\gamma(\sigma)) &= -\sum_{\substack{x,y\in \I(\gamma)}} J_{xy}\tau_\gamma(\sigma)_x\tau_\gamma(\sigma)_y - \sum_{\substack{x,y\in \I(\gamma)^c}} J_{xy}\tau_\gamma(\sigma)_x\tau_\gamma(\sigma)_y - \sum_{\substack{x\in \I(\gamma)\\ y\in \I(\gamma)^c}} J_{xy}\tau_\gamma(\sigma)_x\tau_\gamma(\sigma)_y \\
        &= -\sum_{\substack{x,y\in \I(\gamma)}} J_{xy}\sigma_x\sigma_y -\sum_{\substack{x,y\in \I(\gamma)^c}} J_{xy}\sigma_x\sigma_y + \sum_{\substack{x\in \I(\gamma)\\ y\in \I(\gamma)^c}} J_{xy}\sigma_x\sigma_y\\
        & = H_{\Lambda, 0}^{+}(\sigma) + 2\sum_{\substack{x\in \I(\gamma)\\ y\in \I(\gamma)^c}} J_{xy}\sigma_x\sigma_y.
    \end{align*}
    
    Therefore, the difference can be bounded in the following way 
    \begin{align*}
        \frac{1}{2}(H_{\Lambda, 0}^{+}(\tau_{\gamma}(\sigma)) - H_{\Lambda, 0}^{+}(\sigma)) &= \sum_{\substack{x\in \I(\gamma)\\ y\in \I(\gamma)^c}} J_{xy}\sigma_x\sigma_y
        =\sum_{\substack{x\in \fint\I(\gamma), \\   y\in \fext\I(\gamma)^c}} J\sigma_x\sigma_y + \sum_{\substack{x\in \I(\gamma), \\  y\in \I(\gamma)^c \\ |x-y|\geq 2}} J_{xy}\sigma_x\sigma_y\\
        &\leq -J|\gamma| + \sum_{\substack{x\in \I(\gamma) \\  y\in \I(\gamma)^c \\ |x-y|\geq 2}} J_{xy} \\
        &= -J|\gamma| + \sum_{\substack{k\in \Z^d \\ |k|\geq 2}} J_{0,k}|\{\{x,y\} : x\in \I(\gamma), \  y\in \I(\gamma)^c, x-y = k\}|.
    \end{align*}
    For $i=1,\dots, d$, let $\gamma_i$ be the faces of $\gamma$ perpendicular to the direction $e_i$. Using that $$|{\{x,y\} : x\in \I(\gamma), \  y\in \I(\gamma)^c, x-y = k}| \leq \sum_{i=1}^d |k_i||\gamma_i|,$$ 
we get     
\begin{align*}
    H_{\Lambda, 0}^{+}(\tau_{\gamma}(\sigma)) - H_{\Lambda, 0}^{+}(\sigma) &\leq  -2J|\gamma| + 2\sum_{\substack{k\in \Z^d \\ |k|\geq 2}} \frac{J}{|k|^{\alpha-1}}|\gamma|. 
\end{align*}
Taking $c_1(\alpha)=2(1- \sum_{\substack{k\in \Z^d \\ |k|\geq 2}}\frac{1}{|k|^{\alpha-1}})$, we conclude our proof by noticing that $c_1(\alpha)>0$ if and only if $\alpha > d+1$.
\end{proof}

\begin{remark}
    Notice that the proof above holds for any interaction  $\bm{J} = \{J_{xy}\}_{x,y\in\Z^d}$ that is translation invariant and satisfy
\begin{equation*}
    \sum_{\substack{x\in\Z^d \\ |x|>1}} |x_i|J_{0,x} < J_{0,e_i}
\end{equation*}
for every $i=1,\dots, d$, where $e_i$ is a base vector in the $i$-th direction and $x_i$ is the $i$-th coordinate of $x$.
\end{remark}

With Proposition \ref{Prop: Erasing_contours_alpha>d+1}, we can use the exact same argument from the previous section, only replacing the measure $\mu_{\Lambda;\beta, \varepsilon h}^{+, \IS}(\sigma)$ by $ \mu_{\Lambda;\beta, \varepsilon h}^+(\sigma)$, to prove the next theorem. 

\begin{theorem}
For $d\geq 3$ and $\alpha>d+1$, there exists a constant $C\coloneqq C(d,\alpha)$ such that, for all $\beta>0$ and  $\varepsilon\leq C$, the event 
    \begin{equation}
        \mu_{\Lambda; \beta, \varepsilon h}^+(\sigma_0 = -1) \leq e^{-C\beta} + e^{-C/\varepsilon^2} 
    \end{equation}
    has $\mathbb{P}$-probability bigger then $1 - e^{-C\beta} - e^{-C/\varepsilon^2}$.\\
    
In particular, for $\beta>\beta_c$ and $\varepsilon$ small enough, there is phase transition for the long-range Ising model.  
\end{theorem}

The only difference between the two cases is that, for the long-range model, the error function is 

\begin{equation}\label{Def: Delta_A}
\Delta_A(h) \coloneqq -\frac{1}{\beta}\log{\frac{Z_{\Lambda; \beta, \varepsilon}^{+}(h)}{Z_{\Lambda; \beta, \varepsilon}^{+}(\tau_{A}(h))}},
\end{equation}
for any $A\Subset \Z^d$. We are abusing the notation here since $\Delta_A(h)$ denotes two different random variables,  \eqref{Def: Delta_A_SR} and \eqref{Def: Delta_A}. We hope this does not cause any confusion, since the only property we use from $\Delta_A$ is that it satisfies Lemma \ref{Lemma: Concentration.for.Delta.General}. Moreover, all remarks and claims made in Section \ref{Sec: Ding and Zhuang approach} hold for both definitions of $\Delta_A$. 

For $d<\alpha\leq d+1$, we need to use a different contours system, introduced next.

\section{Long-range Contours}
Contours were first defined in the seminal paper of R. Peierls \cite{Peierls.1936}, where he introduced these geometrical objects to prove phase transition in the Ising model for $d\geq 2$. This technique is known nowadays as the \textit{Peierls' Argument}. One of the most successful extensions of this argument was made by S. Pirogov and Y. Sinai \cite{Pirogov.Sinai.75}, and extended by Zahradnik \cite{Zahradnik.84}. This is known as the \textit{Pirogov-Sinai} Theory, which can be used in models with short-range interactions and finite state spaces, even without symmetries. The Pirogov-Sinai Theory was one of the achievements cited when Yakov Sinai received the Abel Prize \cite{Sinai_Abel_Prize}.

For long-range models, using the usual Peierls' contours with plaquettes of dimension $d-1$, Ginibre, Grossman, and Ruelle, in \cite{Ginibre.Grossmann.Ruelle.66}, proved phase transition for $\alpha > d+1$.  Park, in \cite{Park.88.I,Park.88.II}, considered systems with two-body interactions satisfying $|J_{xy}|\leq |x-y|^{-\alpha}$ for $\alpha > 3d+1$, and extended the Pirogov-Sinai theory for this class of models.  Fr{\"o}hlich and Spencer, in \cite{Frohlich.Spencer.82}, proposed a different contour definition for the one-dimensional long-range Ising models. Roughly speaking, collections of intervals are the new contours but arranged in a particular way. When they are sufficiently far apart, the collections of intervals are deemed as different contours, while collections of intervals close enough are considered a single contour. Note that this definition drastically contrasts with the notion of contour in the multidimensional setting since now they are not necessarily connected objects of the lattice. This fact implies that the control of the number of contours for a fixed size could be much more challenging. 

Inspired by such contours, Affonso, Bissacot, Endo, and Handa proposed a definition of contour extending the contours of Fr{\"o}hlich and Spencer to any dimension $d\geq 2$, see \cite{Affonso.2021}. With these contours, they were able to use Peierls' argument to show phase transition in the whole region $\alpha>d$, with $d\geq 2$. Note that such contours can be very sparse, in the sense that its diameter can be much larger than its size. We modify the contour definition of \cite{Affonso.2021} using a similar partition through multiscale methods. We choose to use the new definition of contours for two main reasons: the definition is simpler, and we can improve the control of the number of cubes needed to cover a contour, from a polynomial bound to an exponential bound, see Propositions \ref{Prop. Bound.on.C_rl(gamma)_Lucas} and \ref{Prop. Bound.on.C_rl(gamma)}. We would like to stress that the main results of this paper still hold if we adopt the notion of contour presented in \cite{Affonso.2021}. In Remark \ref{Rmk: Adaptation_for_Mar_partition} we describe the key adaptations that must be made in our arguments.  In this section, we describe our contours.

\begin{definition}\label{def1}
	Given $\sigma \in \Omega$, a point $x \in \Z^d$ is called \emph{+ (or - resp.)} \emph{correct} if $\sigma_y = +1$, (or $-1$, resp.) for all points $y$ such that $|x-y|\leq 1$. The \emph{boundary} of $\sigma$, denoted by $\partial \sigma$, is the set of all points in $\Z^d$ that are neither $+$ nor $-$ correct.
\end{definition}

The boundary of a configuration is not finite in general, it can even be the whole lattice $\Z^d$. To avoid this problem, we will restrict our attention to configurations with finite boundaries. Such configurations, by definition of incorrectness, satisfy $\sigma \in \Omega^+_\Lambda$ or $\sigma \in \Omega^-_\Lambda$ for some $\Lambda\Subset \Z^d$. We also defined, for each $\Lambda\Subset \Z^d$, $\Lambda^{(0)}$ as the unique unbounded connected component of $\Lambda^c$. The \textit{volume} of $\Lambda$ is defined as $V(\Lambda)\coloneqq \Z^d\setminus \Lambda^{(0)}$. The \textit{interior} of $\Lambda$ is $\I(\Lambda)\coloneqq \Lambda^c\setminus \Lambda^{(0)}$.

    The usual definition of contours in Pirogov-Sinai theory considers only the connected subsets of the boundary $\partial \sigma$. We have to proceed differently for long-range models since every point in the lattice interacts with all the others. The definition below, which was presented in \cite{Affonso.2021}, allows contours to be disconnected, and in return, we can control the interaction between two contours.
    
\begin{definition}\label{def:d_condition}
	Fix real numbers $M,a,r>0$. For each $A\Subset\Z^d$, a set $\overline{\Gamma}(A) \coloneqq \{\overline{\gamma} : \overline{\gamma} \subset A\}$ is called an $(M,a,r)$-\emph{partition} when the following conditions are satisfied:
	\begin{enumerate}[label=\textbf{(\Alph*)}, series=l_after] 
		\item They form a partition of $A$, i.e.,  $\bigcup_{\overline{\gamma} \in \Gamma^1(A)}\overline{\gamma}=A$ and $\overline{\gamma} \cap \overline{\gamma}' = \emptyset$ for distinct elements of $\Gamma^1(A)$. Moreover, each $\overline{\gamma}'$ is contained in only one connected component of $(\overline{\gamma})^c$. 
		
		\item For all $\overline{\gamma} \in \Gamma^1(A)$ there exist $1\leq n \leq 2^r-1$ and a family of subsets $(\overline{\gamma}_{k})_{1\leq k \leq n}$ satisfying 
		
		\begin{enumerate}[label=\textbf{(B\arabic*)}]
			\item  $\overline{\gamma} = \bigcup_{1\leq k \leq n}\overline{\gamma}_{k}$,
			\item For all distinct $\overline{\gamma},\overline{\gamma}' \in \Gamma^1(A)$,
			\be\label{B_distance}
			\dis(\overline{\gamma},\overline{\gamma}') > M \min\left \{\underset{1\leq k \leq n}{\max}\diam(\overline{\gamma}_k),\underset{1\leq j\leq n'}{\max}\diam(\overline{\gamma}'_{j})\right\}^a,
			\ee
			where $(\overline{\gamma}'_j)_{1\leq j \leq n'}$ is the family given by item $\textbf{(B1)}$ for $\overline{\gamma}'$.
		\end{enumerate} 
	\end{enumerate}
\end{definition}

This partition was used to define the contours in \cite{Affonso.2021}, and they are enough to control the energy of erasing a contour. Moreover, by suitably constructing these partitions, with a proper choice of the constants $M,a$, and $r$, the authors can control the number of contours with a fixed size surrounding the origin. These are the two main ingredients of the Peierls' argument. We introduce a new partition, that has a much simpler definition. This partition will also be constructed by a multiscale procedure, described next.
    
\begin{definition}\label{Def: delta-partiton}
    Let $M>0$ and $a,\delta >d$. For each $A\Subset\Z^d$, a set $\Gamma(A) \coloneqq \{\overline{\gamma} : \overline{\gamma} \subset A\}$ is called a $(M,a,\delta)$-\emph{partition} when the following two conditions are satisfied.
	\begin{enumerate}[label=\textbf{(\Alph*)}, series=l_after] 
		\item They form a partition of $A$, i.e.,  $\bigcup_{\overline{\gamma} \in \Gamma(A)}\overline{\gamma}=A$ and $\overline{\gamma} \cap \overline{\gamma}' = \emptyset$ for distinct elements of $\Gamma(A)$.  
		
		\item For all $\overline{\gamma}, \overline{\gamma}^\prime \in \Gamma(A)$ 
			\be\label{B_distance_2}
			\d(\overline{\gamma},\overline{\gamma}') > M\min\left \{|V(\overline{\gamma})|,|V(\overline{\gamma}')|\right\}^\frac{a}{\delta}.
			\ee
	\end{enumerate}
\end{definition}

\begin{remark}\label{Rmk: choice_of_a_and_delta}
    Our construction and the control of the energy works for any $d<\delta<\frac{a(\alpha - d)}{2}$ and  $a>\frac{2(d+1)}{(\alpha-d) \wedge 1}$. To simplify the calculations, we will take $\delta = d+1$ and  $a \coloneqq a(\alpha,d) = \frac{3(d+1)}{(\alpha-d) \wedge 1}$ from now on, so $\frac{a}{\delta} = \frac{3}{(\alpha-d) \wedge 1}$ and a $(M,a,\delta)$-partition will be called $(M,a)$-partition.
\end{remark}

In Figure \ref{Fig: Figura10} we give an example of a region $A\Subset \Z^d$ that is only one contour using the $(M,a,r)$-partition of \cite{Affonso.2021}, but can be partitioned into multiple components to form a $(M,a)$-partition. 

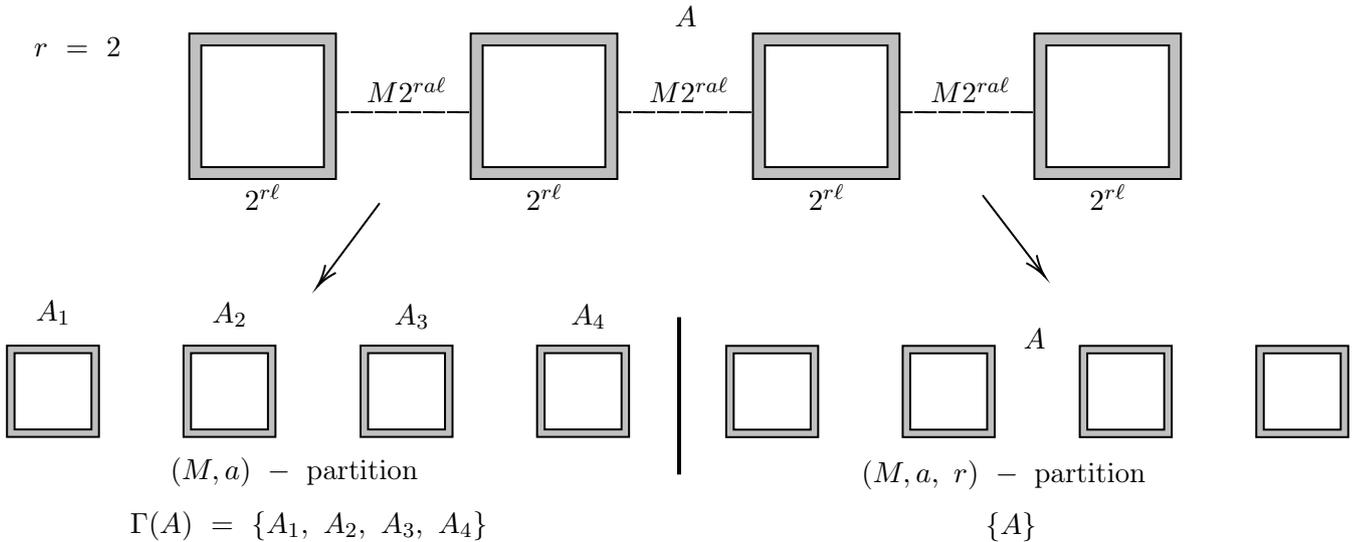
\begin{figure}[H]
    \centering
    \tikzset{every picture/.style={line width=0.75pt}} 

\begin{tikzpicture}[x=0.75pt,y=0.75pt,yscale=-1,xscale=1]

\draw  [fill={rgb, 255:red, 155; green, 155; blue, 155 }  ,fill opacity=0.63 ] (386.36,19.67) -- (459.53,19.67) -- (459.53,92.84) -- (386.36,92.84) -- cycle ;
\draw  [fill={rgb, 255:red, 255; green, 255; blue, 255 }  ,fill opacity=1 ] (392.34,25.65) -- (453.55,25.65) -- (453.55,86.86) -- (392.34,86.86) -- cycle ;
\draw  [dash pattern={on 4.5pt off 4.5pt}] (319.59,59.27) -- (386.01,59.27) -- (386.01,59.22) -- (319.59,59.22) -- cycle ;
\draw  [fill={rgb, 255:red, 155; green, 155; blue, 155 }  ,fill opacity=0.63 ] (526.82,19.67) -- (600,19.67) -- (600,92.84) -- (526.82,92.84) -- cycle ;
\draw  [fill={rgb, 255:red, 255; green, 255; blue, 255 }  ,fill opacity=1 ] (532.8,25.65) -- (594.02,25.65) -- (594.02,86.86) -- (532.8,86.86) -- cycle ;
\draw  [dash pattern={on 4.5pt off 4.5pt}] (460.05,59.27) -- (526.48,59.27) -- (526.48,59.22) -- (460.05,59.22) -- cycle ;
\draw  [fill={rgb, 255:red, 155; green, 155; blue, 155 }  ,fill opacity=0.63 ] (105,19.67) -- (178.18,19.67) -- (178.18,92.84) -- (105,92.84) -- cycle ;
\draw  [fill={rgb, 255:red, 255; green, 255; blue, 255 }  ,fill opacity=1 ] (110.98,25.65) -- (172.2,25.65) -- (172.2,86.86) -- (110.98,86.86) -- cycle ;
\draw  [fill={rgb, 255:red, 155; green, 155; blue, 155 }  ,fill opacity=0.63 ] (245.47,19.67) -- (318.64,19.67) -- (318.64,92.84) -- (245.47,92.84) -- cycle ;
\draw  [fill={rgb, 255:red, 255; green, 255; blue, 255 }  ,fill opacity=1 ] (251.45,25.65) -- (312.66,25.65) -- (312.66,86.86) -- (251.45,86.86) -- cycle ;
\draw  [dash pattern={on 4.5pt off 4.5pt}] (178.7,59.27) -- (245.13,59.27) -- (245.13,59.22) -- (178.7,59.22) -- cycle ;
\draw  [fill={rgb, 255:red, 155; green, 155; blue, 155 }  ,fill opacity=0.63 ] (190.49,176.67) -- (236.39,176.67) -- (236.39,222.57) -- (190.49,222.57) -- cycle ;
\draw  [fill={rgb, 255:red, 255; green, 255; blue, 255 }  ,fill opacity=1 ] (194.24,180.42) -- (232.64,180.42) -- (232.64,218.82) -- (194.24,218.82) -- cycle ;
\draw  [fill={rgb, 255:red, 155; green, 155; blue, 155 }  ,fill opacity=0.63 ] (278.6,176.67) -- (324.5,176.67) -- (324.5,222.57) -- (278.6,222.57) -- cycle ;
\draw  [fill={rgb, 255:red, 255; green, 255; blue, 255 }  ,fill opacity=1 ] (282.35,180.42) -- (320.75,180.42) -- (320.75,218.82) -- (282.35,218.82) -- cycle ;
\draw  [fill={rgb, 255:red, 155; green, 155; blue, 155 }  ,fill opacity=0.63 ] (14,176.67) -- (59.9,176.67) -- (59.9,222.57) -- (14,222.57) -- cycle ;
\draw  [fill={rgb, 255:red, 255; green, 255; blue, 255 }  ,fill opacity=1 ] (17.75,180.42) -- (56.15,180.42) -- (56.15,218.82) -- (17.75,218.82) -- cycle ;
\draw  [fill={rgb, 255:red, 155; green, 155; blue, 155 }  ,fill opacity=0.63 ] (102.11,176.67) -- (148.01,176.67) -- (148.01,222.57) -- (102.11,222.57) -- cycle ;
\draw  [fill={rgb, 255:red, 255; green, 255; blue, 255 }  ,fill opacity=1 ] (105.86,180.42) -- (144.26,180.42) -- (144.26,218.82) -- (105.86,218.82) -- cycle ;
\draw  [fill={rgb, 255:red, 155; green, 155; blue, 155 }  ,fill opacity=0.63 ] (549.49,176.77) -- (595.39,176.77) -- (595.39,222.67) -- (549.49,222.67) -- cycle ;
\draw  [fill={rgb, 255:red, 255; green, 255; blue, 255 }  ,fill opacity=1 ] (553.24,180.52) -- (591.64,180.52) -- (591.64,218.92) -- (553.24,218.92) -- cycle ;
\draw  [fill={rgb, 255:red, 155; green, 155; blue, 155 }  ,fill opacity=0.63 ] (637.6,176.77) -- (683.5,176.77) -- (683.5,222.67) -- (637.6,222.67) -- cycle ;
\draw  [fill={rgb, 255:red, 255; green, 255; blue, 255 }  ,fill opacity=1 ] (641.35,180.52) -- (679.75,180.52) -- (679.75,218.92) -- (641.35,218.92) -- cycle ;
\draw  [fill={rgb, 255:red, 155; green, 155; blue, 155 }  ,fill opacity=0.63 ] (373,176.77) -- (418.9,176.77) -- (418.9,222.67) -- (373,222.67) -- cycle ;
\draw  [fill={rgb, 255:red, 255; green, 255; blue, 255 }  ,fill opacity=1 ] (376.75,180.52) -- (415.15,180.52) -- (415.15,218.92) -- (376.75,218.92) -- cycle ;
\draw  [fill={rgb, 255:red, 155; green, 155; blue, 155 }  ,fill opacity=0.63 ] (461.11,176.77) -- (507.01,176.77) -- (507.01,222.67) -- (461.11,222.67) -- cycle ;
\draw  [fill={rgb, 255:red, 255; green, 255; blue, 255 }  ,fill opacity=1 ] (464.86,180.52) -- (503.26,180.52) -- (503.26,218.92) -- (464.86,218.92) -- cycle ;
\draw   (349,162.84) -- (350,162.84) -- (350,241) -- (349,241) -- cycle ;
\draw    (200,105) -- (170.19,145.24) ;
\draw [shift={(169,146.84)}, rotate = 306.53] [color={rgb, 255:red, 0; green, 0; blue, 0 }  ][line width=0.75]    (10.93,-3.29) .. controls (6.95,-1.4) and (3.31,-0.3) .. (0,0) .. controls (3.31,0.3) and (6.95,1.4) .. (10.93,3.29)   ;
\draw    (501,101) -- (530.79,140.25) ;
\draw [shift={(532,141.84)}, rotate = 232.8] [color={rgb, 255:red, 0; green, 0; blue, 0 }  ][line width=0.75]    (10.93,-3.29) .. controls (6.95,-1.4) and (3.31,-0.3) .. (0,0) .. controls (3.31,0.3) and (6.95,1.4) .. (10.93,3.29)   ;

\draw (26,20.4) node [anchor=north west][inner sep=0.75pt]    {$r\ =\ 2$};
\draw (333.04,40.2) node [anchor=north west][inner sep=0.75pt]    {$M2^{ra\ell }$};
\draw (412.57,94.1) node [anchor=north west][inner sep=0.75pt]    {$2^{r\ell }$};
\draw (473.5,40.2) node [anchor=north west][inner sep=0.75pt]    {$M2^{ra\ell }$};
\draw (553.03,94.1) node [anchor=north west][inner sep=0.75pt]    {$2^{r\ell }$};
\draw (131.21,94.1) node [anchor=north west][inner sep=0.75pt]    {$2^{r\ell }$};
\draw (192.15,40.2) node [anchor=north west][inner sep=0.75pt]    {$M2^{ra\ell }$};
\draw (271.68,94.1) node [anchor=north west][inner sep=0.75pt]    {$2^{r\ell }$};
\draw (94,231.4) node [anchor=north west][inner sep=0.75pt]    {$( M,a) \ -\ \text{partition}$};
\draw (74,258.4) node [anchor=north west][inner sep=0.75pt]    {$\Gamma ( A) \ =\ \{A_{1} ,\ A_{2} ,\ A_{3} ,\ A_{4}\}$};
\draw (439,233.06) node [anchor=north west][inner sep=0.75pt]    {$( M,a,\ r) \ -\ \text{partition}$};
\draw (501,259.06) node [anchor=north west][inner sep=0.75pt]    {$\{A\}$};
\draw (27,152.4) node [anchor=north west][inner sep=0.75pt]    {$A_{1}$};
\draw (115,153.4) node [anchor=north west][inner sep=0.75pt]    {$A_{2}$};
\draw (206,154.4) node [anchor=north west][inner sep=0.75pt]    {$A_{3}$};
\draw (294,154.4) node [anchor=north west][inner sep=0.75pt]    {$A_{4}$};
\draw (346,4.4) node [anchor=north west][inner sep=0.75pt]    {$A$};
\draw (520,166.4) node [anchor=north west][inner sep=0.75pt]    {$A$};

\end{tikzpicture}
    \caption{We wish to partition the gray area $A$. Each cube has side $2^{r\ell}$ and the distance between each other is $M2^{ra\ell}$. With $r=2$, $2^{r}-1 = 3$ and therefore no partition into smaller parts is a $(M,a,r)$-partition. However, the partition into connected components $\{A_1, A_2, A_3, A_4\}$ is an $(M,a)$-partition, since $V(A_i)^{\frac{a}{\delta}} = 2^{ra\frac{d}{\delta}\ell}<2^{ra\ell}$ whenever $\delta>d$.}
    \label{Fig: Figura10}
\end{figure}

The existence of a $(M,a)$-partition for any $A\Subset\Z^d$ does not depend on the choice of $M,a>0$. However, to guarantee the existence of phase transition, we have to choose particular values for these parameters, see Remark \ref{Rmk: choice_of_a_and_delta}. Later on, in Proposition \ref{Prop: Cost_erasing_contour}, $M$ will be taken large enough. 

We write $\Gamma(\sigma) \coloneqq \Gamma(\partial\sigma)$ for a $(M,a)$-partition of $\partial\sigma$. In general, there is more than one $(M,a)$-partition for each region $A\in\Z^d$. Given two partitions $\Gamma$ and $\Gamma^\prime$ of a set $A$, we say that $\Gamma$ \emph{is finer than} $\Gamma'$, and denote $\Gamma\preceq\Gamma^{\prime}$, if for every $\overline{\gamma} \in \Gamma$ there is $\overline{\gamma}' \in \Gamma'$ with $\overline{\gamma} \subseteq \overline{\gamma}'$. The next proposition shows that the finest $(M,a)$-partition exists.

 \begin{proposition}
     For every $A\Subset\Z^d$, there is a finest $(M,a)$-partition.
 \end{proposition}
\begin{proof}
    Given any two $(M,a)$-partitions $\Gamma(A)$ and $\Gamma^\prime(A)$, consider
\begin{equation*}
    \Gamma\cap\Gamma^\prime\coloneqq \{\overline{\gamma}\cap\overline{\gamma}^\prime : \overline{\gamma} \in \Gamma(A), \ \overline{\gamma}^\prime\in\Gamma^\prime(A), \ \overline{\gamma}\cap\overline{\gamma}^\prime\neq \emptyset\}.
\end{equation*}

Then, $\Gamma\cap\Gamma^\prime$ is a $(M,a)$-partition of $A$ finer than $\Gamma(A)$ and $\Gamma^\prime(A)$. Indeed, given $\overline{\gamma}_1\cap\overline{\gamma}_1^\prime, \overline{\gamma}_2\cap\overline{\gamma}_2^\prime\in \Gamma\cap\Gamma^\prime$, 
\begin{align*}
    \d(\overline{\gamma}_1\cap\overline{\gamma}_1^\prime, \overline{\gamma}_2\cap\overline{\gamma}_2^\prime)\geq \d(\overline{\gamma}_1, \overline{\gamma}_2) &\geq M\min{\{|V(\overline{\gamma}_1)|, |V(\overline{\gamma}_2)|\}}^{\frac{a}{d+1}}\\
    &\geq  M\min{\{|V(\overline{\gamma}_1\cap\overline{\gamma}_1^\prime)|, |V(\overline{\gamma}_2\cap \overline{\gamma}_2^\prime)|\}}^{\frac{a}{d+1}}.
\end{align*}
As the number of $(M,a)$-partitions is finite, we construct the finest one by intersecting all of them. 
\end{proof}

 From now on, when taking a $(M,a)$-partition $\Gamma(A)$, we will always assume it is the finest. It is easy to see that the finest $(M,a)$-partition $\Gamma(A)$ satisfies the following property:

\begin{itemize}
    \item[\textbf{(A1)}] For any $\overline{\gamma},\overline{\gamma}^\prime\in \Gamma(A)$, $\overline{\gamma}'$ is contained in only one connected component of $(\overline{\gamma})^c$.
\end{itemize}

Property \textbf{(A1)} is essential to define labels as in \cite{Affonso.2021}. It implies in particular that $\overline{\gamma}^\prime$ is contained in the unbounded component of $\overline{\gamma}^c$ if and only if $V(\overline{\gamma})\cap V(\overline{\gamma}^\prime) = \emptyset$. See Figure \ref{Fig. Exemple_A1} for an example of partition not satisfying \textbf{(A1)}.

 
\tikzset{
pattern size/.store in=\mcSize, 
pattern size = 5pt,
pattern thickness/.store in=\mcThickness, 
pattern thickness = 0.3pt,
pattern radius/.store in=\mcRadius, 
pattern radius = 1pt}
\makeatletter
\pgfutil@ifundefined{pgf@pattern@name@_2ysuecq82}{
\makeatletter
\pgfdeclarepatternformonly[\mcRadius,\mcThickness,\mcSize]{_2ysuecq82}
{\pgfpoint{-0.5*\mcSize}{-0.5*\mcSize}}
{\pgfpoint{0.5*\mcSize}{0.5*\mcSize}}
{\pgfpoint{\mcSize}{\mcSize}}
{
\pgfsetcolor{\tikz@pattern@color}
\pgfsetlinewidth{\mcThickness}
\pgfpathcircle\pgfpointorigin{\mcRadius}
\pgfusepath{stroke}
}}
\makeatother

 
\tikzset{
pattern size/.store in=\mcSize, 
pattern size = 5pt,
pattern thickness/.store in=\mcThickness, 
pattern thickness = 0.3pt,
pattern radius/.store in=\mcRadius, 
pattern radius = 1pt}
\makeatletter
\pgfutil@ifundefined{pgf@pattern@name@_p46h6u1xc}{
\makeatletter
\pgfdeclarepatternformonly[\mcRadius,\mcThickness,\mcSize]{_p46h6u1xc}
{\pgfpoint{-0.5*\mcSize}{-0.5*\mcSize}}
{\pgfpoint{0.5*\mcSize}{0.5*\mcSize}}
{\pgfpoint{\mcSize}{\mcSize}}
{
\pgfsetcolor{\tikz@pattern@color}
\pgfsetlinewidth{\mcThickness}
\pgfpathcircle\pgfpointorigin{\mcRadius}
\pgfusepath{stroke}
}}
\makeatother

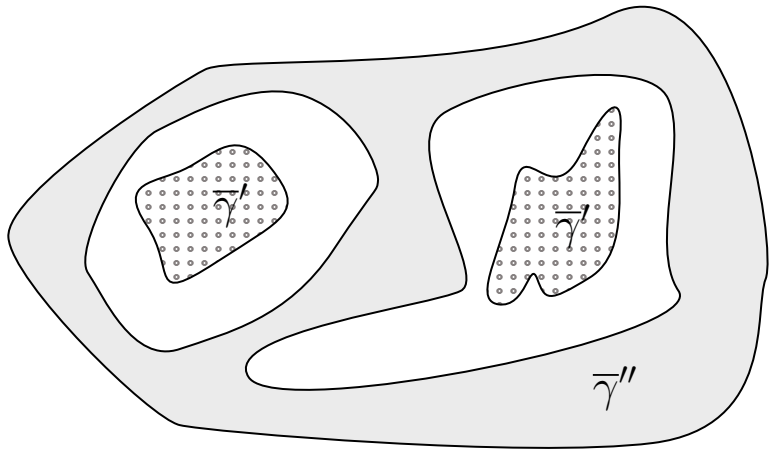
\begin{figure}[H]
	\centering
	
    \tikzset{every picture/.style={line width=0.75pt}} 

\begin{tikzpicture}[x=0.75pt,y=0.75pt,yscale=-0.75,xscale=0.75]

\draw  [fill={rgb, 255:red, 234; green, 234; blue, 234 }  ,fill opacity=0.94 ] (186,64) .. controls (206,54) and (372,68.5) .. (455,27.5) .. controls (538,-13.5) and (567,184) .. (559,204.5) .. controls (551,225) and (566.42,303.92) .. (484,314.5) .. controls (401.58,325.08) and (189.03,307.2) .. (169,302.5) .. controls (148.97,297.8) and (57,206.5) .. (55,176.5) .. controls (53,146.5) and (166,74) .. (186,64) -- cycle ;
\draw  [fill={rgb, 255:red, 255; green, 255; blue, 255 }  ,fill opacity=1 ] (356,212.5) .. controls (361.71,209.65) and (360.41,200.35) .. (356.46,188.03) .. controls (346.6,157.21) and (320.24,107.48) .. (346,92.5) .. controls (382.06,71.54) and (480,53.5) .. (494,83.5) .. controls (508,113.5) and (481,182.5) .. (501,212.5) .. controls (521,242.5) and (235,300.5) .. (215,270.5) .. controls (195,240.5) and (336,222.5) .. (356,212.5) -- cycle ;
\draw  [fill={rgb, 255:red, 255; green, 255; blue, 255 }  ,fill opacity=1 ] (154,105) .. controls (174,95) and (220,69) .. (254,82.5) .. controls (288,96) and (301.65,130.51) .. (301,138.5) .. controls (300.35,146.49) and (301,140) .. (271,185.5) .. controls (241,231) and (200.19,240.89) .. (169,251.5) .. controls (137.81,262.11) and (118.4,216.1) .. (108,200.5) .. controls (97.6,184.9) and (134,115) .. (154,105) -- cycle ;
\draw  [pattern=_2ysuecq82,pattern size=5.25pt,pattern thickness=0.75pt,pattern radius=0.75pt, pattern color={rgb, 255:red, 133; green, 128; blue, 128}] (173,130.5) .. controls (193,120.5) and (207,101) .. (230,130.5) .. controls (253,160) and (239,163.5) .. (193,193.5) .. controls (147,223.5) and (167,198.5) .. (147,168.5) .. controls (127,138.5) and (153,140.5) .. (173,130.5) -- cycle ;
\draw  [pattern=_p46h6u1xc,pattern size=5.25pt,pattern thickness=0.75pt,pattern radius=0.75pt, pattern color={rgb, 255:red, 133; green, 128; blue, 128}] (436,120.5) .. controls (456,80.5) and (465,81.5) .. (462,109.5) .. controls (459,137.5) and (470.24,185.85) .. (442,203.5) .. controls (413.76,221.15) and (412.21,218.61) .. (407,204.5) .. controls (401.79,190.39) and (400,224.5) .. (381,221.5) .. controls (362,218.5) and (386,181.5) .. (392,143.5) .. controls (398,105.5) and (416,160.5) .. (436,120.5) -- cycle ;

\draw (442,261.4) node [anchor=north west][inner sep=0.75pt]  [font=\LARGE]  {$\overline{\gamma}^{\prime\prime} $};
\draw (189,138.4) node [anchor=north west][inner sep=0.75pt]  [font=\LARGE]  {$\overline{\gamma} '$};
\draw (417,152.4) node [anchor=north west][inner sep=0.75pt]  [font=\LARGE]  {$\overline{\gamma} '$};

\end{tikzpicture}
\caption{An example of how Condition \textbf{(A1)} works: considering $\overline{\gamma}'$ the dotted region and $\overline{\gamma}^{\prime \prime}$ the grey region, one can readily see that $\overline{\gamma}'$ intersects two different connected components of $(\overline{\gamma}^{\prime\prime})^c$. To turn this into a partition satisfying condition \textbf{(A1)}, one should separate $\overline{\gamma}'$ in two different sets of $\Gamma(A)$.}
\label{Fig. Exemple_A1}
\end{figure}

Counting the number of contours surrounding zero using the finest $(M,a)$-partition may be troublesome since the definition provides very little information on these objects. To extract good properties of these contours, we establish a multiscale procedure, depending on a parameter $r$, that creates a $(M,a)$-partition of any given set. To define this procedure, we introduce some notation. 

 For any $x\in\Z^d$ and $m\geq 0$,
\begin{equation}
    C_{m}(x) \coloneqq \left(\prod_{i=1}^d{\left[2^{m}x_i , \ 2^{m}(x_i+1) \right)}\right)\cap \Z^d,
\end{equation}
is the cube of $\mathbb{Z}^d$ centered at $2^{m}x + 2^{m-1} - \frac{1}{2}$ with side length $2^{m} -1$, see Figure \ref{Fig: Cubes}. Any such cube is called an $m$-cube. As all cubes in this paper are of this form, with centers $2^{m}x + 2^{m-1} - \frac{1}{2}$ and $x \in \mathbb{Z}^d$, we will often omit the point $x$ in what follows, writing $C_m$ for an $m$-cube instead of $C_m(x)$. An arbitrary collection of $m$-cubes will be denoted $\mathscr{C}_m$ and $B_{\mathscr{C}_m}\coloneqq \cup_{C\in\mathscr{C}_m}C$ is the region covered by $\mathscr{C}_m$. We denote by $\mathscr{C}_m(\Lambda)$ the covering of $\Lambda\Subset\Z^d$ with the smallest possible number of $m$-cubes. 

For each $n\geq 0$, define the graph $G_n(\Lambda) = (V_n(\Lambda), E_n(\Lambda))$ with vertex set $V_n(\Lambda) = \mathscr{C}_n(\Lambda)$ and $E_n(\Lambda) = \{ (C_n, C_n^\prime) : d(C_n,C_n^\prime) \leq M2^{an}\}$. Let $\mathscr{G}_n(\Lambda)$ be the connected components of $G_n(\Lambda)$. Given ${G = (V,E) \in \mathscr{G}_n(\Lambda)}$, we denote $\Lambda^G \coloneqq \Lambda \cap B_V$ the area of $\Lambda$ covered by $G$. In the next proposition, we will introduce a procedure that can generate a $(M,a)$-partition that can be nontrivial.

\begin{proposition}\label{Prop:Construction_(M,a,delta)_partition}
    For any $r> 0$ and $A\Subset\Z^d$, there is a possibly non-trivial $(M,a)$-partition $\Gamma^r(A)$.
\end{proposition}

\begin{proof}
  Given $r> 0$ and $A\Subset\Z^d$, $\Gamma^r(A)$ is the partition of $A$ created by the following procedure. In the first step we consider $A_1 \coloneqq A$ and we take the connected components of $G\in \mathscr{G}_r(A_1)$ such that $A_1^G$ have small density, that is, consider
\begin{equation*}
   \mathscr{P}_1 \coloneqq \{G \in  \mathscr{G}_r(A_1) : |V(A_1^G)|\leq {2^{r(d+1)}}\}.
\end{equation*}
    Then, the subsets to be removed in the first step are $\Gamma_1^r(A) \coloneqq \{A_1^G : G\in \mathscr{P}_1\}$ and the set left to partition is $A_2 \coloneqq A_1\setminus \bigcup\limits_{\gamma \in \Gamma_1^r(A)}\gamma$.
 We can repeat this procedure inductively by taking 
\begin{equation*}
   \mathscr{P}_n \coloneqq \{G \in  \mathscr{G}_{rn}(A_{n}) : |V(A_n^G)|\leq {2^{rn{(d+1)}}}\},
\end{equation*}
then define $\Gamma_n^r(A) \coloneqq \{A_n^G : G \in \mathscr{P}_n\}$ and $A_{n+1} \coloneqq A_n \setminus \bigcup\limits_{ \gamma\in \Gamma_n^r(A)}\gamma$. As the cubes invade the lattice, this procedure stops, in the sense that for some $N$ large enough, $\mathscr{P}_n=\emptyset$ for all $n\geq N$. We then define $\Gamma^r(A) \coloneqq \cup_{n\geq 0} \Gamma_n^r(A)$.
 By this construction, $\Gamma^r(A)$ is clearly a partition of $A$, so condition \textbf{(A)} follows. To show condition \textbf{(B)}, take $\overline{\gamma},\overline{\gamma}^\prime\in \Gamma^r(A)$. Let $m\geq n\geq 1$ be such that $\overline{\gamma}\in\Gamma_n^r(A)$ and $\overline{\gamma}^\prime\in\Gamma_m^r(A)$. Then, 
\begin{equation*}
\d(\overline{\gamma},\overline{\gamma}^\prime)\geq M2^{rna}\geq M\left(2^{rn(d+1)}\right)^{\frac{a}{d+1}} \geq M|V(\overline{\gamma})|^{\frac{a}{d+1}}.
\end{equation*}
If $m=n$, the same inequality holds for $|V(\overline{\gamma}^\prime)|$ and condition \textbf{(B)} holds. When $m>n$, $\overline{\gamma}^\prime$ was not removed at step $n$, so $|V(\overline{\gamma}^\prime)|> 2^{rn(d+1)} \geq |V(\overline{\gamma})|$, so $|V(\overline{\gamma})| = \min\{|V(\overline{\gamma})|,|V(\overline{\gamma}^\prime)|\}$ and again we get condition \textbf{(B)}. 
\end{proof}

The construction in Proposition \ref{Prop:Construction_(M,a,delta)_partition} works for any $r>0$, but we need to take $r$ large enough for the computations in Section 3 to work. So we fix $r\coloneqq 4\lceil\log_2(a+1) \rceil + d +1$, where $\lceil x \rceil$ is the smallest integer greater than or equal to $x$. This $r$ is taken larger than the one in \cite{Affonso.2021} to simplify some calculations. All our computations should work with the previous choice of $r$, with some adaptation. Next, we define the label of a contour. 
 
\begin{definition}
    For $\Lambda\subset\Z^d$, the \textit{edge boundary} of $\Lambda$ is $\partial\Lambda = \{\{x,y\} \subset \Z^d: |x-y|=1, x \in \Lambda, y \in \Lambda^c\}$. The \textit{inner boundary} of $\Lambda$ is $\fint\Lambda\coloneqq\{x\in\Lambda : \exists y\in \Lambda^c \text{ such that }|x-y|=1\}$ and the \textit{external boundary} is $\fext\Lambda\coloneqq\{x\in\Lambda^c : \exists y\in \Lambda \text{ such that }|x-y|=1\}$
\end{definition}

\begin{remark}
    The usual isoperimetric inequality states that $2d|\Lambda|^{\frac{d-1}{d}}\leq |\partial \Lambda|$. The inner boundary and the edge are related by $|\fint \Lambda|\leq |\partial \Lambda|\leq 2d|\fint \Lambda|$, so we can write the inequality as $|\Lambda|^{\frac{d-1}{d}}\leq |\fint \Lambda|$.
\end{remark}

To define the label of a contour, the naive definition would be to take the sign of the inner boundary of the set $\overline{\gamma}$. However, this cannot be done since this inner boundary may have different signs, see Figure \ref{fig: Figura2}.

\input{Figures/Figura.2}

For any $\Lambda\Subset\Z^d$, its connected components are denoted $\Lambda^{(1)}, \dots, \Lambda^{(n)}$. Given $\overline{\gamma} \in\Gamma(\sigma)$, a connected component $\overline{\gamma}^{(k)}$ is \textit{external} if $V(\overline{\gamma}^{(j)})\subset V(\overline{\gamma}^{(k)})$, for all other connected components $\overline{\gamma}^{(j)}$ satisfying $V(\overline{\gamma}^{(j)})\cap V(\overline{\gamma}^{(k)}) \neq \emptyset$. Denoting 
\begin{equation*}
    \overline{\gamma}_\mathrm{ext} = \hspace{-0.5cm}\bigcup_{\substack{k\geq 1 \\ \overline{\gamma}^{(k)} \text{ is external}}}\hspace{-0.5cm}\overline{\gamma}^{(k)},
\end{equation*} 
it is shown in \cite[Lemma 3.8]{Affonso.2021} that the sign of $\sigma$ is constant in $\fint V(\overline{\gamma}_{\mathrm{ext}})$. The \textit{label} of $\overline{\gamma}$ is the function $\lab_{\overline{\gamma}} :\{(\overline{\gamma})^{(0)}, \I(\overline{\gamma})^{(1)}\dots, \I(\overline{\gamma})^{(n)}\} \rightarrow \{-1,+1\}$ defined as: $\lab_{\overline{\gamma}}(\I(\overline{\gamma})^{(k)})$ is the sign of the configuration $\sigma$ in $\fint V(\I(\overline{\gamma})^{(k)})$, for $k\geq 1$, and $\lab_{\overline{\gamma}}((\overline{\gamma})^{(0)})$ is the sign of $\sigma$ in $\fint V(\overline{\gamma}_\mathrm{ext})$.  We then define the contours.

\begin{definition}
Given a configuration $\sigma$ with finite boundary, its \emph{contours} $\gamma$ are pairs $(\overline{\gamma},\lab_{\overline{\gamma}})$,  where $\overline{\gamma} \in \Gamma(\sigma)$ and $\lab_{\overline{\gamma}}$ is the label of $\overline{\gamma}$ as defined above. The \emph{support of the contour}  $\gamma$ is defined as $\Sp(\gamma)\coloneqq \overline{\gamma}$ and its \emph{size} is given by $|\gamma| \coloneqq |\Sp(\gamma)|$.
\end{definition}

Another important definition is of the \textit{interior} of a contour $\gamma$, given by $\I(\gamma) \coloneqq \I(\Sp(\gamma))$. This notion of interior and volume works as expected since $\I(\gamma) = V(\Sp(\gamma)) \setminus \Sp(\gamma)$. We also split the interior according to its labels as

\begin{equation*}
    \I_\pm(\gamma) = \hspace{-0.7cm}\bigcup_{\substack{k \geq 1, \\ \lab_{\overline{\gamma}}(\I(\gamma)^{(k)})=\pm 1}}\hspace{-0.7cm}\I(\gamma)^{(k)}.
\end{equation*}
To simplify the notation, we write $V(\gamma)\coloneqq V(\Sp(\gamma))$. Different from Pirogov-Sinai theory, where the interiors of contours are a union of simply connected sets, the interior $\I(\gamma)$ is at most the union of connected sets, that is, they may have holes. 

Moreover, there is no bijection between families of contours $\Gamma = \{\gamma_1,\dots,\gamma_n\}$ and configurations. Usually, more than one configuration can have the same boundary. First, $\Gamma$ may not even form a $(M,a)$-partition. Even so, their labels may not be compatible. We say that $\Gamma$ is \textit{compatible} when there exists a configuration $\sigma$ with contours precisely $\Gamma$.

\input{Figures/Figura.3}

A contour $\gamma$ in $\Gamma$ is \textit{external} if its external connected components are not contained in any other $V(\gamma')$, for $\gamma' \in \Gamma\setminus\{\gamma\}$. Taking $\I_\pm(\Gamma) \coloneqq \cup_{\gamma\in\Gamma}\I_\pm(\gamma)$ and $V(\Gamma)\coloneqq\cup_{\gamma\in\Gamma}V(\gamma)$, for each $\Lambda\Subset\Z^d$ we consider the sets\\
\begin{equation*}
\mathcal{E}^\pm_\Lambda \coloneqq\{\Gamma= \{\gamma_1, \ldots, \gamma_n\}: \Gamma \text{ is compatible,} \gamma_i \text{ is external}, \lab_{\gamma_i}((\gamma_i)^{(0)})=\pm1, V(\Gamma) \subset \Lambda\}, \vspace{0.2cm}
\end{equation*}
of all external compatible families of contours with external label $\pm$ contained in $\Lambda$.  When we write $\gamma \in \mathcal{E}^\pm_\Lambda$ we mean $\{\gamma\} \in \mathcal{E}^\pm_\Lambda$. Most of the time the set $\Lambda$ will play no role, so we will often omit the subscript. 

The first step for a Peierls-type argument to hold is to control the number of contours with a fixed size. Consider $\mathcal{C}_0(n) \coloneqq \{\gamma \in \mathcal{E}^+_\Lambda: 0 \in V(\gamma), |\gamma|=n\}$, the set of contours with fixed size with the origin in its volume, and  $\mathcal{C}_0 \coloneqq \cup_{n\geq 1}\mathcal{C}_0(n)$. 

We will later show in Corollary \ref{Cor: Bound_on_C_0_n} that the size of the set $\mathcal{C}_0(n)$ is exponentially bounded depending on $n$. 
Before we start the estimations of the Peierls arguments, let us present a property of the new contours that, despite not being necessary in our results, highlights a major difference between our construction and the contours of \cite{Affonso.2021}.

\begin{proposition}
     There exists a constant $c\coloneqq c(d,a)$ such that all $\sigma\in\Omega$ and $\gamma\in\Gamma(\sigma)$,
    \begin{equation*}
        \diam(\gamma)\leq c|\gamma|^{1 + \frac{a}{d+1}(1 + \frac{1}{d-1})}
    \end{equation*}     
\end{proposition}

\begin{proof}
    Let $j\geq 1$ be such that $\gamma\in \Gamma^r_j(\sigma)$, that is, $\gamma$ is removed in the $j$-th step of the construction. Then, $|V(\gamma)|\geq 2^{r(d+1)(j-1)}$. For any $\Lambda,\Lambda^\prime\Subset \Z^d$, 
    \begin{equation*}
        \diam(\Lambda\cup \Lambda) \leq \diam(\Lambda) + \diam(\Lambda^\prime) + \dis(\Lambda,\Lambda^\prime),
    \end{equation*}
    and we can always extract a vertex from a connected graph in a way that the induced sub-graph is still connected, by removing a leaf of a spanning tree. Using this and the fact that $G_{rj}=(\C_{rj}(\gamma), E_{rj}(\gamma))$ is connected, we can bound
    \begin{align*}
        \diam(\gamma)&\leq \diam(B_{\C_{rj}}(\gamma)) \leq |\C_{rj}(\gamma)|(d2^{rj} + M2^{raj}) \\ 
        &\leq 2Md2^{raj}|\gamma| = 2Md2^{ra}2^{ra(j-1)}|\gamma| \leq 2Md2^{ra}|V(\gamma)|^{\frac{a}{d+1}}|\gamma| \\
        &\leq 2Md2^{ra}|\gamma|^{1 + \frac{a}{d+1}(1 + \frac{1}{d-1})},
    \end{align*}
    what concludes our proof for $c=2Md2^{ra}$.
\end{proof}

Repeating this argument for the contours defined by $(M,a,r)$-partition, we would get a sub-exponential bound on the diameter. The only difference in the argument is that we must replace $|V(\gamma)|\geq 2^{r(d+1)(j-1)}$ by $|V(\gamma)|\geq 2^{r}$. We also need to use that $j\leq n^{\log_2(a)/r-d-1}$, otherwise, by Proposition \ref{Prop. Bound.on.C_rl(gamma)_Lucas}, we would have only one $r(j-1)$-cube covering $\gamma$, and therefore $j$ would not be the step the contour was removed.

A key step of a Peierls-type argument is to control the energy cost of erasing a contour. Given $\Gamma\in \mathcal{E}^+$, the configurations compatible with $\Gamma$ are $\Omega(\Gamma)\coloneqq \{\sigma\in\Omega^+_\Lambda : \Gamma\subset \Gamma(\sigma)\}$. The map $\tau_\Gamma:\Omega(\Gamma) \rightarrow \Omega_{\Lambda}^+$ defined as 
\be
\tau_\Gamma(\sigma)_x = 
\begin{cases}
	\;\;\;\sigma_x &\text{ if } x \in \I_+(\Gamma)\cup V(\Gamma)^c, \\
	-\sigma_x &\text{ if } x \in \I_-(\Gamma),\\
	+1 &\text{ if } x \in \Sp(\Gamma),
\end{cases}
\ee
erases a family of compatible contours, since the spin-flip preserves incorrect points but transforms $-$-correct points into $+$-correct points. Define, for $B\Subset\Z^d$, the interaction
\begin{equation*}
    F_B\coloneqq \sum_{\substack{x\in B \\ y\in  B^c}}J_{xy}.
\end{equation*}

In \cite{Affonso.2021} it was proved the following proposition:
\begin{proposition}\label{prop2}
	For $M$ large enough, there exists a  constant $c_2\coloneqq c_2(\alpha,d)>0$, such that for any configuration $\sigma\in\Omega$ and $\overline{\gamma} \in \Gamma(\sigma)$, with $\overline{\gamma}$ external and $0\in V(\overline{\gamma})$, it holds that 
	\be
	H_{\Lambda; 0}^+(\sigma)- H_{\Lambda;0}^+(\tau_{\overline{\gamma}}(\sigma))\geq c_2|\overline{\gamma}|.
	\ee
\end{proposition}
We now need to extend this proposition to our contours. To do so, we need some auxiliary lemmas. 

Given $B\subset\Z^d$ and $\sigma\in\Omega$ with $\partial\sigma$ finite, let $\Gamma_{\Int}(\sigma, B)$ be the contours $\gamma^\prime$ with $\Sp(\gamma^\prime) \in \Gamma(\sigma)$ and enclosed by $B$, that is, $\Sp(\gamma^\prime)\subset B$. Define also $\Gamma_{\Ext}(\sigma, B)$ as the contours $\gamma^\prime$ with $\Sp(\gamma^\prime) \in \Gamma(\sigma)$ outside $B$, that is, $\Sp(\gamma^\prime)\subset B^c$.

\begin{lemma}\label{Lemma: Aux_1}
Given $\sigma\in\Omega$ with $\partial \sigma$ finite and a contour $\gamma$ with $\Sp(\gamma)\in \Gamma(\sigma)$ , there is a constant $\kappa^{(1)}_\alpha\coloneqq \kappa^{(1)}_\alpha(\alpha, d)$, such that, for  $B = \Sp(\gamma)$ or $B=\I_-(\gamma)$ we have 

\begin{equation}\label{Eq: Lemma_aux_2}
    \sum_{\substack{x\in B \\ y\in V(\Gamma_{\Ext}(\sigma, B)\setminus\{\gamma\})}} J_{xy} \leq \kappa^{(1)}_\alpha \left[ \frac{|B|}{M^{\alpha - d}}|V(\gamma)|^{\frac{a}{d+1}(d-\alpha)} + \frac{F_B}{M} \right],  
\end{equation}

\end{lemma}

\begin{proof}

Fixed $\sigma$ and $B$, we drop them from the notation, so $\Gamma_{\Ext} \coloneqq \Gamma_{\Ext}(\sigma, B)$.  Splitting $\Gamma_{\Ext}\setminus\{\gamma\}$ into $\Upsilon_1 \coloneqq \{\gamma^\prime \in \Gamma_{\Ext}\setminus\{\gamma\} : |V(\gamma^\prime)| \geq |V(\gamma)|\}$ and $\Upsilon_2 = \Gamma_{\Ext} \setminus (\Upsilon_1\cup \gamma)$ we get 
\begin{equation*}
     \sum_{\substack{x\in B \\ y\in V(\Gamma_{\Ext}\setminus\{\gamma\})}} J_{xy} \leq  \sum_{\substack{x\in B \\ y\in V(\Upsilon_1)}} J_{xy} +  \sum_{\substack{x\in B \\ y\in V(\Upsilon_2)}} J_{xy}.
\end{equation*}

    For any  $\gamma^\prime \in \Upsilon_1$, $\d(\gamma,\gamma^\prime)> M |V(\gamma)|^{\frac{a}{d+1}}$. Moreover, $V(\gamma^\prime)\subset V(\gamma)^c$, otherwise we would have a region contained in $V(\gamma)$ with volume bigger than $|V(\gamma)|$. Hence, for any $x\in B$ and $y\in V(\gamma^\prime)$, with $\gamma^\prime\in \Upsilon_1$, 
    \begin{equation*}
        d(x,y)\geq d(\gamma,\gamma^\prime)\geq M |V(\gamma)|^{\frac{a}{d+1}}.
    \end{equation*}
    We can then bound
\begin{equation}\label{Eq: Lemma_aux_2_1}
    \sum_{\substack{x\in B \\ y\in V(\Upsilon_1)}} J_{xy} \leq \sum_{\substack{x\in B \\ y : |y-x| > R}} J_{xy} = |B|\sum_{y : |y|>R}J_{0y},
\end{equation}
with $R\coloneqq M |V(\gamma)|^{\frac{a}{d+1}}$. 

Defining $s_d(n) \coloneqq |\{x\in \Z^d : |x|=n \}|$, it is known that $s_d(n)\leq 2^{2d - 1}e^{d-1}n^{d-1}$, see for example \cite[Lemma 4.2]{Affonso.2021}. Using the integral bound of the sum, we can show that

    \begin{equation}\label{Eq: Interaction_outside_ball}
 \begin{split}
 \sum_{y : |y|>R}J_{0y} &= J \sum_{n>R} \frac{s_d(n)}{n^\alpha} \leq J2^{2d - 1}e^{d-1} \sum_{n>R} \frac{1}{n^{\alpha-d+1}} \\
 &\leq J2^{2d - 1}e^{d-1} \int_{\floor{R}}^{\infty}\frac{1}{x^{\alpha - d +1}}dx \leq \frac{J2^{2d - 1}e^{d-1}}{(\alpha - d)} \floor{R}^{d-\alpha}\leq \frac{J2^{d - 1 + \alpha}e^{d-1}}{(\alpha - d)} {R}^{d-\alpha}.
 \end{split}
\end{equation}

Together with \eqref{Eq: Lemma_aux_2_1}, this yields 
    \begin{align}\label{Eq: Lemma_aux_2.Part1}
    \sum_{\substack{x\in B \\ y\in V(\Upsilon_1)}} J_{xy} &\leq |B| \frac{J2^{d-1+\alpha}e^{d-1}}{(\alpha - d)}\left[ M |V(\gamma)|^{\frac{a}{d+1}} \right]^{d-\alpha} \nonumber \\
    &\leq \frac{J2^{d-1 + \alpha }e^{d-1}}{(\alpha - d)}\frac{|B|}{M^{\alpha - d}}|V(\gamma)|^{\frac{a}{d+1}(d-\alpha)}.
\end{align}

To bound the other term, split $\Upsilon_2$ into layers $\Upsilon_{2,m} \coloneqq \{ \gamma^\prime \in \Upsilon_2 : |V(\gamma^\prime)|=m\}$, for $1\leq m\leq |V(\gamma)|-1$. Denoting  $y_{\gamma^\prime,x} \in \Sp(\gamma^\prime)$ the point satisfying $\d(x, \Sp(\gamma^\prime)) = \d(x, y_{\gamma^\prime, x})$, we can bound 

\begin{align*}
   \sum_{\substack{x\in B \\ y\in V(\Upsilon_{2,m})}} J_{xy} &\leq  \sum_{\substack{x\in B \\ \gamma^\prime \in \Upsilon_{2,m}}} |V(\gamma^\prime)| J_{x,y_{\gamma^\prime, x}} \leq m \sum_{\substack{x\in B \\ \gamma^\prime \in \Upsilon_{2,m}}} J_{x,y_{\gamma^\prime, x}}.
\end{align*}

    Consider the minimal paths $\lambda_{x,y_{\gamma^\prime, x}}$ that connects $x$ and $y_{\gamma^\prime, x}$. Denoting $\lambda_{x,y_{\gamma^\prime, x}}^\prime$ the path restricted to the last $\frac{M}{3}m^{\frac{a}{d+1}}$ steps, we have
    \begin{equation}\label{Eq: Lemma_aux_2.Tese_1}
         \frac{M}{3}m^{\frac{a}{d+1}}J_{x,y_{\gamma^\prime, x}}\leq \sum_{y\in \lambda_{x,y_{\gamma^\prime, x}}^\prime} J_{xy}.
    \end{equation}
    Since $\d(x,y_{\gamma^\prime, x}) \geq \d(\gamma, \gamma^\prime)>Mm^{\frac{a}{d+1}}$, and $\d(y_{\gamma^\prime, x}, y_{\gamma^{\prime\prime}, x}) > \d(\gamma^\prime, \gamma^{\prime\prime})>Mm^{\frac{a}{d+1}}$ for any $\gamma^\prime, \gamma^{\prime\prime}\in\Upsilon_{2,m}$, the balls with radius $\frac{M}{3}m^{\frac{a}{d+1}}$ centered in $y_{\gamma^\prime, x}$, for all $\gamma^\prime\in \Upsilon_{2,m}$ are disjoint and are contained in $B^c$. Hence, the paths $\lambda_{x,y_{\gamma^\prime, x}}^\prime$ and $\lambda_{x,y_{\gamma^{\prime\prime}, x}}^\prime$ are disjoint, for all $\gamma^\prime, \gamma^{\prime\prime}\in\Upsilon_{2,m}$, and are in $B^c$. Therefore, summing equation \eqref{Eq: Lemma_aux_2.Tese_1} over all points contours in $\Upsilon_{2,m}$ we get
    \begin{equation*}
       \sum_{\gamma^\prime\in\Upsilon_{2,m}}J_{x,y_{\gamma^\prime, x}}\leq  \frac{3}{Mm^{\frac{a}{d+1}}}\sum_{\gamma^\prime\in\Upsilon_{2,m}}\sum_{y\in \lambda_{x,y_{\gamma^\prime, x}}^\prime} J_{xy} \leq \frac{3}{Mm^{\frac{a}{d+1}}}\sum_{y\in B^c}J_{xy},
    \end{equation*}
    and we recuperate the desired inequality by summing over $x\in B$ we conclude that 
    \begin{equation*}
    \sum_{\substack{x\in B \\ \gamma^\prime \in\Upsilon_{2,m}}} J_{x,y_{\gamma^\prime, x}} \leq \frac{3}{Mm^{\frac{a}{d+1}}}F_B.
    \end{equation*}

That gives us
\begin{align}\label{Eq: Lemma_aux_2.Part2}
     \sum_{\substack{x\in B \\ y\in V(\Upsilon_{2})}} J_{xy} \leq \sum_{m=1}^{|V(\gamma)|-1} \frac{3}{Mm^{\frac{a}{d+1}-1}}F_B\leq \frac{3\zeta({\frac{a}{d+1}-1})}{M}F_B,
\end{align}
what concludes the proof for $\kappa_\alpha^{(1)}\coloneqq \frac{J2^{d-1 + \alpha}e^{d-1}}{(\alpha - d)} + {3\zeta({\frac{a}{d+1}-1})}$.
\end{proof}

\begin{corollary}\label{Cor: Corrolary_of_Lemma_aux}
For any configuration $\sigma\in\Omega$ and $\gamma\in\Gamma(\sigma)$, 

\begin{equation}
  \sum_{\substack{x\in \Sp(\gamma) \\ y\in V(\Gamma(\sigma)\setminus\{\gamma\})}} J_{xy} \leq  2\kappa^{(1)}_\alpha\frac{F_{\Sp(\gamma)}}{M^{(\alpha - d)\wedge 1}},  \\
\end{equation}
 
\begin{equation}
    \sum_{\substack{x\in \I_-(\gamma) \\ y\in V(\Gamma_{\Ext}(\sigma, \I_-(\gamma))\setminus\{\gamma\})}} J_{xy}  \leq 2\kappa^{(1)}_\alpha\frac{F_{\I_-(\gamma)}}{M^{(\alpha - d)\wedge 1}},
\end{equation}
and 
\begin{equation}
    \sum_{\substack{x\in \I_-(\gamma)^c \\ y\in V(\Gamma_{\Int}(\sigma, \I_-(\gamma)))}} J_{xy} \leq  \kappa^{(1)}_\alpha\frac{F_{\I_-(\gamma)}}{M}
\end{equation}
\end{corollary}

\begin{proof}
    The first inequality is a direct application of the Lemma \ref{Lemma: Aux_1} for $B=\Sp(\gamma)$, once we note that $\Gamma_{\Ext}(\sigma, B) = \Gamma(\sigma)\setminus\{\gamma\}$ and, our choice of $a$,  $\frac{|\gamma|}{|V(\gamma)|^{\frac{a}{d+1}(\alpha-d)} } \leq \frac{|\gamma|}{|V(\gamma)|}\leq 1$. The second inequality is likewise a direct application of Lemma \ref{Lemma: Aux_1} for $B=V(\gamma)$, since  $\frac{|V(\gamma)|}{|V(\gamma)|^{\frac{a}{d+1}(\alpha-d)} } \leq \frac{|V(\gamma)|}{|V(\gamma)|}\leq 1$.
For the last inequality, we cannot apply Lemma \ref{Lemma: Aux_1} directly. However, the proof works in the similar steps when we take $B = \I_-(\gamma)^c$. Moreover, notice that $V(\Gamma_{\Int}(\sigma, \I_-(\gamma))) = V(\Gamma_{\Ext}(\sigma, \I_-(\gamma)^c))$ and, for all $\gamma^\prime\in \Gamma_{\Int}(\sigma, \I_-(\gamma))$, $|V(\gamma^\prime)| < |V(\gamma)|$.  In the notation of the proof of Lemma \ref{Lemma: Aux_1}, this means that $\Upsilon_{2} = \Gamma_{\Ext}(\sigma, \I_-(\gamma)^c)$, so equation \eqref{Eq: Lemma_aux_2.Part2} yields
\begin{align*}
      \sum_{\substack{x\in \I_-(\gamma)^c \\ y\in V(\Gamma_{\Ext}(\sigma, \I_-(\gamma)^c))}} J_{xy} \leq \zeta\Big({\frac{a}{d+1}-1}\Big)\frac{F_{\I_-(\gamma)^c}}{M}.
\end{align*}
Since $F_{\I_-(\gamma)^c} = F_{\I_-(\gamma)}$ and $\zeta({\frac{a}{d+1}-1})\leq \kappa^{(1)}_\alpha  $, we get the desired bound.
\end{proof}

We are ready to prove the main proposition of this section:

\begin{proposition}\label{Prop: Cost_erasing_contour}
	For $M$ large enough, there exists a constant $c_2\coloneqq c_2(\alpha,d)>0$, such that for  any $\Lambda \Subset \Z^d$, $\gamma\in \mathcal{E}^+_\Lambda$, and $\sigma \in \Omega(\gamma)$ it holds that
	\be
	H_{\Lambda; 0}^+(\sigma)- H_{\Lambda;0}^+(\tau_{\gamma}(\sigma))\geq c_2\left(|\gamma| + F_{\I_-(\gamma)} + F_{\Sp(\gamma)} \right).
	\ee
\end{proposition}

\begin{proof}
    Fix some $\sigma \in \Omega(\gamma)$. We will denote $\tau_\gamma(\sigma)\coloneqq \tau$ and $\Gamma(\sigma) \coloneqq \Gamma$ throughout this proposition. The difference between the Hamiltonians is
  \begin{equation}\label{prop2:eq1}
  \begin{split}
  H_{\Lambda}^+(\sigma)-H_{\Lambda}^+(\tau) &= \sum_{\{x,y\}\subset \Lambda}J_{xy}(\tau_x\tau_y-\sigma_x\sigma_y) +\sum_{\substack{x\in \Lambda \\ y \in \Lambda^c}}J_{xy}(\tau_x - \sigma_x) \\
  &= \sum_{\{x,y\}\subset V(\Gamma)}J_{xy}(\tau_x\tau_y-\sigma_x\sigma_y) +\sum_{\substack{x\in V(\Gamma) \\ y \in \Lambda\setminus V(\Gamma)}}J_{xy}(\tau_x - \sigma_x) +\sum_{\substack{x\in V(\Gamma) \\ y \in \Lambda^c}}J_{xy}(\tau_x - \sigma_x) \\
  &= \sum_{\{x,y\}\subset V(\Gamma)}J_{xy}(\tau_x\tau_y - \sigma_x\sigma_y) + \sum_{\substack{x \in V(\Gamma)\\ y \in V(\Gamma)^c}}J_{xy}(\tau_x - \sigma_x),
  \end{split}
  \end{equation}
  where the second equality is due the fact that outside $V(\Gamma)$ the configurations $\sigma$ and $\tau$ are equal to $+1$.  Since $\tau_x=\sigma_x$ for $x\in V(\Gamma\setminus\{\gamma\})$ we have
  \begin{equation*}
  \begin{split}
  \sum_{\{x,y\}\subset V(\Gamma)}J_{xy}(\tau_x\tau_y - \sigma_x\sigma_y) + \sum_{\substack{x \in V(\Gamma)\\ y \in V(\Gamma)^c}}J_{xy}(\sigma_x - \tau_x)= \sum_{\{x,y\}\subset V(\gamma)}J_{xy}(\tau_x\tau_y-\sigma_x\sigma_y)+\sum_{\substack{x\in V(\gamma) \\ y \in V(\gamma)^c}}J_{xy}(\tau_x\sigma_y-\sigma_x\sigma_y).
  \end{split}
  \end{equation*}
  We focus now on the sum involving the terms $\{x,y\} \subset V(\gamma)$. We can split it accordingly with $V(\gamma)= \Sp(\gamma) \cup \I(\gamma)$. Then,
  \begin{equation*}
  \begin{split}
  \sum_{\{x,y\}\subset V(\gamma)}J_{xy}(\tau_x\tau_y-\sigma_x\sigma_y) &= \hspace{-0.3cm}\sum_{\{x,y\}\subset \Sp(\gamma)}\hspace{-0.3cm}J_{xy}(\tau_x\tau_y-\sigma_x\sigma_y)+\hspace{-0.3cm}\sum_{\{x,y\}\subset \I(\gamma)}\hspace{-0.3cm}J_{xy}(\tau_x\tau_y-\sigma_x\sigma_y)+\sum_{\substack{x\in \Sp(\gamma) \\ y \in \I(\gamma)}}\hspace{-0.1cm}J_{xy}(\tau_x\tau_y-\sigma_x\sigma_y) \\
  &=\hspace{-0.3cm}\sum_{\{x,y\}\subset \Sp(\gamma)}\hspace{-0.3cm}J_{xy}(1-\sigma_x\sigma_y)-2\sum_{\substack{x\in \I_-(\gamma) \\ y \in \I_+(\gamma)}}J_{xy}\sigma_x\sigma_y +\sum_{\substack{x\in \Sp(\gamma) \\ y \in \I(\gamma)}}J_{xy}((-1)^{\mathbbm{1}_{\I_-(\gamma)}(y)}\sigma_y-\sigma_x\sigma_y),
  \end{split}
  \end{equation*}
  where for the second equality we used the definition of the map $\tau_\gamma$ and $\mathbbm{1}_{\I_-(\gamma)}(y)= 1$ if $y \in \I_-(\gamma)$ and $\mathbbm{1}_{\I_-(\gamma)}(y)=0$ if $y \in \I_+(\gamma)$. For the same reason, we have
  \[
  \sum_{\substack{x\in V(\gamma) \\ y \in V(\gamma)^c}}J_{xy}(\tau_x\sigma_y-\sigma_x\sigma_y) = \sum_{\substack{x\in \Sp(\gamma) \\ y \in V(\gamma)^c}}J_{xy}(\sigma_y-\sigma_x\sigma_y)-2\sum_{\substack{x\in \I_-(\gamma) \\ y \in V(\gamma)^c}}J_{xy}\sigma_x\sigma_y
  \]
  Putting everything together and using that $\pm\sigma_y - \sigma_x\sigma_y = 2\mathbbm{1}_{\{\sigma_x\neq \sigma_y\}}-2\mathbbm{1}_{\{\sigma_y=\mp 1\}}$ we get
		\begin{align}\label{prop2:eq2}
			H_{\Lambda}^+(\sigma)-H_{\Lambda}^+(\tau) &= \sum_{\substack{x \in \Sp(\gamma)\\ y \in \Z^d}}J_{xy}\mathbbm{1}_{\{\sigma_x \neq \sigma_y\}}+\sum_{\substack{x \in \Sp(\gamma)\\ y \in \Sp(\gamma)^c}}J_{xy}\mathbbm{1}_{\{\sigma_x \neq \sigma_y\}}-2\sum_{\substack{x \in \I_-(\gamma) \\ y \in B(\gamma)}}J_{xy}\sigma_x\sigma_y \nonumber\\
			&- 2\sum_{\substack{x \in \Sp(\gamma) \\ y \in B(\gamma)}}J_{xy}\mathbbm{1}_{\{\sigma_y = -1\}}-2\sum_{\substack{x \in \Sp(\gamma) \\ y \in \I_-(\gamma)}}J_{xy}\mathbbm{1}_{\{\sigma_y = +1\}},
		\end{align}
  with $B(\gamma)  \coloneqq \I_+(\gamma)\cup V(\gamma)^c$.

We first analyze the last two negative terms. It holds that 
\begin{equation*}
    \sum_{\substack{x \in \Sp(\gamma) \\ y \in B(\gamma)}} J_{xy}\mathbbm{1}_{ \{ \sigma_y = -1 \}} + \sum_{\substack{x \in \Sp(\gamma) \\ y \in \I_-(\gamma)}} J_{xy}\mathbbm{1}_{ \{ \sigma_y = +1\}} \leq \sum_{\substack{x \in \Sp(\gamma) \\ y \in V(\Gamma(\sigma)\setminus\{\gamma\})}} J_{xy},
\end{equation*}
since the characteristic function can only be non-zero in the volume of other contours. By Corollary \ref{Cor: Corrolary_of_Lemma_aux},
\begin{equation}\label{Eq: Aux_1_Diferenca_de_Hamiltonianos}
    \sum_{\substack{x \in \Sp(\gamma) \\ y \in V(\Gamma(\sigma)\setminus\{\gamma\})}} J_{xy}  \leq 2{\kappa^{(1)}_\alpha}\frac{F_{\Sp(\gamma)}}{M^{(\alpha - d)\wedge 1}}.
\end{equation}

    For the negative term left, taking $\Gamma^\prime$ the contours inside $\I_-(\gamma)$ and $\Gamma^{\prime\prime} = \Gamma(\sigma)\setminus {(\Gamma^\prime\cup\gamma)}$ we can write
    \begin{align*}
    \sum_{\substack{x \in \I_-(\gamma) \\ y \in B(\gamma)}} J_{xy}\sigma_x\sigma_y &= \sum_{\substack{x \in V(\Gamma^{\prime}) \\ y \in  V(\Gamma^{\prime\prime}) }} J_{xy}\sigma_x\sigma_y + \sum_{\substack{x \in V(\Gamma^{\prime}) \\ y \in  B(\gamma)\setminus V(\Gamma^{\prime\prime}) }} J_{xy}\sigma_x \nonumber  \\
    &  - \sum_{\substack{x \in \I_-(\gamma)\setminus V(\Gamma^{\prime}) \\ y \in  V(\Gamma^{\prime\prime}) }} J_{xy}\sigma_y - \sum_{\substack{x \in \I_-(\gamma)\setminus V(\Gamma^{\prime}) \\ y \in  B(\gamma)\setminus V(\Gamma^{\prime\prime}) }} J_{xy}.
\end{align*}
We can write the first term as
\begin{equation*}
    \sum_{\substack{x \in V(\Gamma^{\prime}) \\ y \in  V(\Gamma^{\prime\prime}) }} J_{xy}\sigma_x\sigma_y = \sum_{\substack{x \in V(\Gamma^{\prime}) \\ y \in  V(\Gamma^{\prime\prime}) }} J_{xy} - \sum_{\substack{x \in V(\Gamma^{\prime}) \\ y \in  V(\Gamma^{\prime\prime}) }} 2J_{xy}\mathbbm{1}_{\{\sigma_x\neq \sigma_y\}}. 
\end{equation*}
Similarly, we have 
\begin{multline*}
    \sum_{\substack{x \in V(\Gamma^{\prime}) \\ y \in  B(\gamma)\setminus V(\Gamma^{\prime\prime}) }} J_{xy}\sigma_x \nonumber  - \sum_{\substack{x \in \I_-(\gamma)\setminus V(\Gamma^{\prime}) \\ y \in  V(\Gamma^{\prime\prime}) }} J_{xy}\sigma_y =  \sum_{\substack{x \in V(\Gamma^{\prime}) \\ y \in  B(\gamma)\setminus V(\Gamma^{\prime\prime}) }} 2J_{xy}\mathbbm{1}_{\{\sigma_x=+1\}} -  \sum_{\substack{x \in V(\Gamma^{\prime}) \\ y \in  B(\gamma)\setminus V(\Gamma^{\prime\prime}) }}J_{xy}\\
     + \sum_{\substack{x \in \I_-(\gamma)\setminus V(\Gamma^{\prime}) \\ y \in  V(\Gamma^{\prime\prime}) }} 2J_{xy}\mathbbm{1}_{\{\sigma_y=-1\}} - \sum_{\substack{x \in \I_-(\gamma)\setminus V(\Gamma^{\prime}) \\ y \in  V(\Gamma^{\prime\prime}) }}  J_{xy}.
\end{multline*}
Putting the equations together we have
\begin{multline}\label{Eq: Interaction_between_interior_and_B}
     \sum_{\substack{x \in \I_-(\gamma) \\ y \in B(\gamma)}} J_{xy}\sigma_x\sigma_y = \sum_{\substack{x \in V(\Gamma^{\prime}) \\ y \in  V(\Gamma^{\prime\prime}) }} J_{xy} + \sum_{\substack{x \in \I_-(\gamma)\setminus V(\Gamma^{\prime}) \\ y \in  V(\Gamma^{\prime\prime}) }} 2J_{xy}\mathbbm{1}_{\{\sigma_y=-1\}} + \sum_{\substack{x \in V(\Gamma^{\prime}) \\ y \in  B(\gamma)\setminus V(\Gamma^{\prime\prime}) }} 2J_{xy}\mathbbm{1}_{\{\sigma_x=+1\}} \\
       - \sum_{\substack{x \in \I_-(\gamma)\setminus V(\Gamma^{\prime}) \\ y \in  V(\Gamma^{\prime\prime}) }}  J_{xy}  - \sum_{\substack{x \in V(\Gamma^{\prime}) \\ y \in  V(\Gamma^{\prime\prime}) }} 2J_{xy}\mathbbm{1}_{\{\sigma_x\neq \sigma_y\}} -  \sum_{\substack{x \in \I_-(\gamma) \\ y \in  B(\gamma)\setminus V(\Gamma^{\prime\prime}) }}J_{xy}.
\end{multline}

We can bound the first two terms by
\begin{equation}\label{Eq: Aux_1_Interaction_between_interior_and_B}
 \sum_{\substack{x \in V(\Gamma^{\prime}) \\ y \in  V(\Gamma^{\prime\prime}) }} J_{xy} +  \sum_{\substack{x \in \I_-(\gamma)\setminus V(\Gamma^{\prime}) \\ y \in  V(\Gamma^{\prime\prime}) }} 2J_{xy}\mathbbm{1}_{\{\sigma_y=-1\}}  \leq  2\sum_{\substack{x \in \I_-(\gamma) \\ y \in V(\Gamma^{\prime\prime}) }} J_{xy} \leq 4\kappa^{(1)}_\alpha\frac{F_{\I_-(\gamma)}}{M^{(\alpha - d)\wedge 1}}.
 \end{equation}
In the second inequality, we are applying Corollary \ref{Cor: Corrolary_of_Lemma_aux}. For the next term, since $B(\gamma)\setminus V(\Gamma^{\prime\prime})\subset \I_-(\gamma)^c$, we can bound
\begin{equation}\label{Eq: Aux_2_Interaction_between_interior_and_B}
    \sum_{\substack{x \in V(\Gamma^{\prime}) \\ y \in  B(\gamma)\setminus V(\Gamma^{\prime\prime}) }} 2J_{xy}\mathbbm{1}_{\{\sigma_y=+1\}}\leq \sum_{\substack{x \in V(\Gamma^{\prime}) \\ y \in  \I_-(\gamma)^c}} 2J_{xy}\leq 2\kappa_\alpha^{(1)}\frac{F_{\I_-(\gamma)}}{M}.
\end{equation}
 In the last inequality, we are again applying Corollary \ref{Cor: Corrolary_of_Lemma_aux}.

For the negative terms in \eqref{Eq: Interaction_between_interior_and_B}, we bound the term containing $\mathbbm{1}_{\{\sigma_x\neq\sigma_x\}}$ by $0$ and multiply the remaining terms by $\frac{1}{(2d+1)2^{\alpha+2}}$, getting 
\begin{equation*}
      \sum_{\substack{x \in \I_-(\gamma)\setminus V(\Gamma^{\prime}) \\ y \in  V(\Gamma^{\prime\prime}) }}  J_{xy} +  \sum_{\substack{x \in V(\Gamma^{\prime}) \\ y \in  V(\Gamma^{\prime\prime}) }} 2J_{xy}\mathbbm{1}_{\{\sigma_x\neq \sigma_y\}} +  \sum_{\substack{x \in \I_-(\gamma) \\ y \in  B(\gamma)\setminus V(\Gamma^{\prime\prime}) }}J_{xy} \geq \frac{1}{(2d+1)2^{\alpha + 2}}\left[ \sum_{\substack{x \in \I_-(\gamma)\setminus V(\Gamma^{\prime}) \\ y \in  V(\Gamma^{\prime\prime}) }}  J_{xy}  +  \sum_{\substack{x \in \I_-(\gamma) \\ y \in  B(\gamma)\setminus V(\Gamma^{\prime\prime}) }}J_{xy} \right].
\end{equation*}
Using the second inequality of \eqref{Eq: Aux_1_Interaction_between_interior_and_B}, we have 
\begin{equation}\label{Eq: Aux_4_Interaction_between_interior_and_B}
\begin{split}
  \sum_{\substack{x \in \I_-(\gamma)\setminus V(\Gamma^{\prime}) \\ y \in  V(\Gamma^{\prime\prime}) }} J_{xy} + \sum_{\substack{x \in \I_-(\gamma) \\ y \in  B(\gamma)\setminus V(\Gamma^{\prime\prime}) }} J_{xy} 
    &= F_{\I_-(\gamma)} - \sum_{\substack{x \in V(\Gamma^{\prime}) \\ y \in  V(\Gamma^{\prime\prime})}} J_{xy} -  \sum_{\substack{x \in \I_-(\gamma)\\ y \in  \Sp(\gamma)}} J_{xy}  \\
    &\geq \Big(1 -  \frac{2\kappa^{(1)}_\alpha}{M^{(\alpha - d)\wedge 1}}\Big)F_{\I_-(\gamma)} - F_{\Sp(\gamma)}.
    \end{split}
\end{equation}
 Plugging inequalities \eqref{Eq: Aux_1_Interaction_between_interior_and_B}, \eqref{Eq: Aux_2_Interaction_between_interior_and_B} and \eqref{Eq: Aux_4_Interaction_between_interior_and_B} back in \eqref{Eq: Interaction_between_interior_and_B} we get
\begin{align}\label{Eq: Aux_2_Diferenca_de_Hamiltonianos}
    \sum_{\substack{x \in \I_-(\gamma) \\ y \in B(\gamma)}} J_{xy}\sigma_x\sigma_y &\leq 6\kappa^{(1)}_\alpha\frac{F_{\I_-(\gamma)}}{M^{(\alpha - d)\wedge 1}} +  \frac{2\kappa_\alpha^{(1)}}{(2d+1)2^{\alpha + 1}}\frac{F_{\I_-(\gamma)}}{M^{(\alpha - d)\wedge 1}} - \frac{1}{(2d+1)2^{\alpha+1}}F_{\I_-(\gamma)} + \frac{1}{(2d+1)2^{\alpha+1}}F_{\Sp(\gamma)}\nonumber \\
    &\leq \left(6\kappa^{(1)}_\alpha +  \frac{2\kappa_\alpha^{(1)}}{(2d+1)2^{\alpha + 1}}\right)\frac{F_{\I_-(\gamma)}}{M^{(\alpha - d)\wedge 1}} - \frac{F_{\I_-(\gamma)}}{(2d+1)2^{\alpha+1}} + \frac{F_{\Sp(\gamma)}}{(2d+1)2^{\alpha+1}} \nonumber \\
    &\leq 8\kappa^{(1)}_\alpha\frac{F_{\I_-(\gamma)}}{M^{(\alpha - d)\wedge 1}} + \frac{F_{\Sp(\gamma)} - F_{\I_-(\gamma)}}{(2d+1)2^{\alpha+1}}
\end{align}

For the positive terms in \eqref{prop2:eq2}, we can use the triangular inequality to bound, for every $x,y\in\Z^d$,
\begin{align*}
    \sum_{|x-x^\prime|\leq 1}\frac{J_{x^\prime y}}{J_{xy}} &= \sum_{|x-x^\prime|\leq 1} \frac{|x-y|^\alpha}{|x^\prime -y|^\alpha} \leq \sum_{|x-x^\prime|\leq 1} \left(\frac{|x-x^\prime| + |x^\prime - y|}{|x^\prime -y|}\right)^\alpha \\
    &\leq \sum_{|x-x^\prime|\leq 1} \left(\frac{1}{|x^\prime -y|} + 1\right)^\alpha \leq (2d+1)2^\alpha.
\end{align*}
This shows that 
\begin{equation*}
    J_{xy} \geq \frac{1}{(2d+1)2^\alpha}\sum_{|x-x^\prime|\leq 1}J_{x^\prime y},
\end{equation*}
and therefore 
\begin{equation}\label{Eq: Aux_3_Diferenca_de_Hamiltonianos}
    \sum_{\substack{x \in \Sp(\gamma) \\ y \in \Z^d}} J_{xy}\mathbbm{1}_{ \{ \sigma_x \neq \sigma_y \}} +       \sum_{\substack{x \in \Sp(\gamma) \\ y \in \Sp(\gamma)^c}} J_{xy}\mathbbm{1}_{ \{ \sigma_x \neq \sigma_y \}} \geq \frac{1}{(2d+1)2^\alpha}\left(Jc_\alpha |\gamma| + F_{\Sp(\gamma)}\right),
\end{equation}
with $c_\alpha = \sum_{\substack{y\in\Z^d\setminus{0}}}|y|^{-\alpha}$. Plugging \eqref{Eq: Aux_1_Diferenca_de_Hamiltonianos}, \eqref{Eq: Aux_2_Diferenca_de_Hamiltonianos} and \eqref{Eq: Aux_3_Diferenca_de_Hamiltonianos} back in \eqref{prop2:eq2} we get
\begin{align*}
     H_\Lambda^+(\sigma) - H_\Lambda^+(\tau) \geq \frac{Jc_\alpha }{(2d+1)2^\alpha} |\gamma| + \left(\frac{1}{(2d+1)2^{\alpha + 1}} -  \frac{16\kappa^{(1)}_\alpha}{M^{(\alpha - d)\wedge 1}}\right)F_{\I_-(\gamma)} + \left( \frac{1 }{(2d+1)2^{\alpha+1}} -  \frac{4\kappa^{(1)}_\alpha}{M^{(\alpha - d)\wedge 1}}\right)F_{\Sp(\gamma)},
\end{align*}
what proves the proposition for $M^{(\alpha - d)\wedge 1}>16\kappa^{(1)}_\alpha2^{\alpha + 1}(2d+1)$. 
\end{proof}

    \section{Joint measure and bad events}       
    In the short-range case, the spins that need to be flipped to erase a contour are precisely the ones in the interior of it. This is not the case for the long-range model, so we make a slight modification in the argument, and instead of performing the same flips in both spaces, we flip the external field only on $\I_-(\gamma)$. Doing this, not only does the partition function change, but we also get an extra cost when comparing the original energy with the energy after performing such a transformation. This extra term depends only on the external field in $\Sp{(\gamma)}$.  

In this section, we define the measure in the joint space and show that, with high probability, the change of partition function resulting from such flipping is upper-bounded by the size of the support $|\gamma|$, with high probability. 

Given $\Lambda\subset\Z^d$, a contour associated with a configuration in $\Omega_\Lambda^+$ is not always inside $\Lambda$. To avoid this, we consider the event $\Theta_{\Lambda} \coloneqq \{ \sigma : \sigma_x \text{ is } +\text{-correct for all }x\in\fint \Lambda \}$ and the conditional measure 
\begin{equation}
    \nu_{\Lambda; \beta, \varepsilon h}^{+}(A) \coloneqq \mu_{\Lambda; \beta,\varepsilon h}^+(A | \Theta_{\Lambda})
\end{equation}
for any $A\subset \Omega$ measurable. First, we show that $\nu_{\Lambda; \beta, \varepsilon h}^{+}$ is also a local Gibbs measure. To do this, we need the so-called Markov property of the local Gibbs measures, which states that, for any $\Lambda_1 \subset \Lambda_2 \Subset \Z^d$, $\eta\in\Omega$ and $\omega\in \Omega_{\Lambda_2}^\eta$,  
\begin{equation*}
    \mu_{\Lambda_2; \beta,\varepsilon h}^\eta( \ \cdot \  | \  \sigma_x=\omega_x, \ \forall x\in \Lambda_2\setminus\Lambda_1) = \mu_{\Lambda_1; \beta,\varepsilon h}^\omega( \cdot ).
\end{equation*}

Consider the set $\Lambda^\prime = \{x\in \Lambda : d(x, \Lambda^c)>2\}$, that is the set 
$\Lambda$ after we remove the sites at the inner boundary and all of its neighbors. Then, $\Theta_{\Lambda} = \{\sigma\in \Omega_\Lambda^+ : \sigma_x = +1, \ \forall x\in \Lambda \setminus \Lambda^\prime\}$. Using the Markov property, for every $A\subset \Omega$ measurable we have
\begin{align*}
      \nu_{\Lambda; \beta, \varepsilon h}^{+}(A) = \mu_{\Lambda; \beta,\varepsilon h}^+(A |  \sigma_x = +1, \ \forall x\in \Lambda \setminus \Lambda^\prime) = \mu_{\Lambda^\prime; \beta,\varepsilon h}^+(A).
\end{align*}
This not only shows that  $\nu_{\Lambda; \beta, \varepsilon h}^{+}$ is a local Gibbs measure, but it also that $\lim_{n\to\infty} \nu_{\Lambda_n; \beta, \varepsilon h}^{+} =  \mu_{\beta,\varepsilon h}^+$, with the limit being taken over any sequence $(\Lambda_n)_{n\geq 0}$ invading $\Z^d$.
Define the joint measure for $(\sigma, h)$ as

\begin{equation*}
    \mathbb{Q}_{\Lambda; \beta, \varepsilon}^+(\sigma \in A, h\in B) \coloneqq \int_{B} \nu_{\Lambda;\beta, \varepsilon h}^+(A) d\mathbb{P}(h),
\end{equation*}
for $A\subset\Omega$ measurable and $B\subset \mathbb{R}^{\Lambda}$ a Borel set. This measure $\mathbb{Q}_{\Lambda;\beta,\varepsilon}$ has density
\begin{equation*}
    g_{\Lambda; \beta, \varepsilon}^+(\sigma, h) \coloneqq \prod_{x\in\Lambda}\frac{1}{\sqrt{2\pi}}e^{-\frac{1}{2}h_x^2} \times \nu_{\Lambda;\beta, \varepsilon h}^+(\sigma).
\end{equation*}

The operation $\tau_{\gamma}$ used to remove a contour $\gamma\in\Gamma(\sigma)$ can be written as a particular case of the following one: given $A\subset\Z^d$, take $\tau_A:\mathbb{R}^{\Z^d} \xrightarrow{} \mathbb{R}^{\Z^d}$ as 
\begin{equation}
    (\tau_A(\sigma))_x \coloneqq \begin{cases}
                        -\sigma_x &\text{if }x\in A,\\
                        \sigma_x   &\text{otherwise},
                      \end{cases}
\end{equation}
for every $x\in\Z^d$. Defining $\s(\gamma, \sigma)^\pm\coloneqq \{ x\in \s(\gamma): \sigma_x = \pm 1\}$, the transformation that erases a contour $\gamma$ is $\tau_\gamma(\sigma) = \tau_{\I_-(\gamma)\cup \s^-(\gamma,\sigma)}(\sigma)$.

The main idea used in the proof of phase transition in \cite{Ding2021} is to make the Peierls' argument on the measure $\mathbb{Q}_{\Lambda;\beta,\varepsilon}$, and perform in the external field the same flips one does in the configuration when erasing a contour. Formally, in \cite{Ding2021} they compare the density $g_{\Lambda; \beta, \varepsilon}^+(\sigma, h)$ with the density after erasing a contour $\gamma\in\Gamma(\sigma)$, and performing the same flips on the external field. For the short-range model, the spins that need to be flipped to erase a contour are precisely the ones in the interior of it. This is not the case for the long-range setting, so we compare $g_{\Lambda; \beta, \varepsilon}^+(\sigma, h)$ with the density after erasing $\gamma$ and flipping the external field only in $\I_-(\gamma)$, getting

\begin{align}\label{Eq: quotient.of.gs}
    \frac{g_{\Lambda; \beta, \varepsilon}^+(\sigma, h)}{g_{\Lambda; \beta, \varepsilon}^+(\tau_{\gamma}(\sigma),\tau_{\I_-(\gamma)}(h))} 
    &= \exp{\{\beta H_{\Lambda, \varepsilon \tau_{\I_-(\gamma)}(h)}^{+}(\tau_{\gamma}(\sigma)) - \beta H_{\Lambda, \varepsilon h}^{+}(\sigma)\}}\frac{Z_{\Lambda; \beta, \varepsilon}^{+}(\tau_{\I_-(\gamma)}(h))}{Z_{\Lambda; \beta, \varepsilon}^{+}(h)}  \nonumber \\ 
    &\leq \exp{\{- \beta c_2 |\gamma| -2\beta\sum_{x\in \Sp^-(\gamma,\sigma)}\varepsilon h_x\}}\frac{Z_{\Lambda; \beta, \varepsilon}^{+}(\tau_{\I_-(\gamma)}(h))}{Z_{\Lambda; \beta, \varepsilon}^{+}(h)}.
\end{align}
where the constant $c_2$ is the one given by Proposition \ref{Prop: Cost_erasing_contour}.

The sum of the external field in $\Sp^-(\gamma,\sigma)$ can be shown to be of order $|\Sp^-(\gamma,\sigma)|$, and do not influence the Peierls' argument. However, the quotient of the partition functions can be bigger than the exponential term. Again, the bad event is

$$\mathcal{E}^c\coloneqq \left\{\sup_{\substack{\gamma\in\mathcal{C}_0}} \frac{|\Delta_{\I_-(\gamma)}(h)|}{c_2|\gamma|} > \frac{1}{4}\right\},$$

where $\Delta_A(h)$ is the function previously defined in \eqref{Def: Delta_A}. Let us remind again that we are abusing the notation once $\Delta_A(h)$ denotes two different random variables,  \eqref{Def: Delta_A_SR} and \eqref{Def: Delta_A}. We do so since the only property we use from $\Delta_A$ is that it satisfies Lemma \ref{Lemma: Concentration.for.Delta.General}. Moreover, all remarks and claims made in Section \ref{Sec: Ding and Zhuang approach} hold for both definitions of $\Delta_A$. 

To control the probability of the bad event, we the concentration inequalities presented in Section \ref{Sec: Probability results}. Using then, the bound on the bad event $\mathcal{E}^c$ follows from the next proposition.
\begin{proposition}\label{Prop: Bound.gamma_2}
    Given $n\geq 0$, $d\geq 3$ and $\alpha > d$, there is a constant $L_1 \coloneqq L_1(d,\alpha)>0$  such that $$\gamma_2(\I_-(n),\d_2) \leq \varepsilon L_1 n.$$
\end{proposition}
As a direct consequence of this Proposition, we can control the probability of the bad event.
\begin{proposition}\label{Prop: Bound.bad.event.1}     
    There exists $C_1\coloneqq C_1(\alpha, d)$ such that $\mathbb{P}(\mathcal{E}^c)\leq e^{-\frac{C_1}{\varepsilon^2}}$ for any $\varepsilon^2<C_1$. 
\end{proposition}

\begin{proof}   
 By the union bound,
\begin{align}\label{Eq: Union_bound_bad_event}
    \mathbb{P}\left({\sup_{\substack{\gamma\in\mathcal{C}_0}} \frac{|\Delta_{\I_-(\gamma)}(h)|}{c_2|\gamma|} > \frac{1}{4}}\right) \leq \sum_{n=2}^\infty \mathbb{P}\left({\sup_{\substack{\gamma\in\mathcal{C}_0(n)}} |\Delta_{\I_-(\gamma)}(h)| > \frac{c_2}{4}}|\gamma|\right). 
\end{align}
Let $\gamma,\gamma^\prime\in \mathcal{C}_0(n)$ be two contours satisfying $\diam(\I_-(n)) = \d_2(\I_-(\gamma),\I_-(\gamma^\prime))$, where the diameter is in the $\d_2$ distance. By the isoperimetric inequality,
\begin{equation*}
    \diam(\I_-(n))= 2\varepsilon{|\I_-(\gamma)\Delta \I_-(\gamma^\prime)|}^{\frac{1}{2}} \leq 2\sqrt{2}\varepsilon n^{(\frac{d}{d-1})\frac{1}{2}} = 2\sqrt{2}\varepsilon n^{(\frac{1}{2} + \frac{1}{2(d-1)})}.
\end{equation*}
Together with Proposition \ref{Prop: Bound.gamma_2}, this yields
\begin{align*}
  \frac{c_2}{4}|\gamma| &= L\left[\varepsilon L_1 n + \varepsilon L_1 \left(\frac{c_2}{4\varepsilon L_1 L} - 1\right)n\right]\\
    &\geq  L\left[\gamma_2(\I_-(n),\d_2) +  \frac{L_1}{2\sqrt{2}} \left(\frac{c_2}{4\varepsilon L_1 L} - 1\right)n^{\frac{1}{2} - \frac{1}{2(d-1)}}\diam(\I_-(n))\right]\\
    &\geq   L\left[\gamma_2(\I_-(n),\d_2) +  \frac{C_1^\prime}{\varepsilon}n^{\frac{1}{2} - \frac{1}{2(d-1)}}\diam(\I_-(n))\right],
\end{align*}
with $C_1^\prime = \frac{c_2}{16\sqrt{2}L}$ and $\varepsilon<\frac{c_2}{8L_1L}$. Applying Theorem \ref{Theo: Theo_2.2.27_Talagrand} with $u = \frac{C_1^\prime}{\varepsilon}n^{\frac{1}{2} - \frac{1}{2(d-1)}}$, we have
\begin{align*}
    \mathbb{P}\left({\sup_{\substack{\gamma\in\mathcal{C}_0(n)}} |\Delta_{\I_-(\gamma)}(h)| > \frac{c_2}{4}}|\gamma|\right) &=  \mathbb{P}\left({\sup_{\substack{\I\in\I_-(n)}} |\Delta_{\I}(h)| > \frac{c_2}{4}}n\right) \\
    &\leq \mathbb{P}\left({\sup_{\substack{\I\in\I_-(n)}} |\Delta_{\I}(h)| >    L\left[\gamma_2(\I_-(n),\d_2) +  \frac{C_1^\prime}{\varepsilon}n^{\frac{1}{2} - \frac{1}{2(d-1)}}\diam(\I_-(n))\right]}\right) \\ 
    &\leq \exp{\left\{ - \frac{C_1^{\prime2}n^{1 - \frac{1}{(d-1)}}}{\varepsilon^2}\right\}}. 
\end{align*}
Using this back in equation \eqref{Eq: Union_bound_bad_event}, we get
\begin{align*}
       \mathbb{P}\left({\sup_{\substack{\gamma\in\mathcal{C}_0}} \frac{|\Delta_{\I_-(\gamma)}(h)|}{c_2|\gamma|} > \frac{1}{4}}\right) &\leq \sum_{n=2}^\infty \exp{\left\{ - \frac{C_1^{\prime2}n^{1 - \frac{1}{(d-1)}}}{\varepsilon^2}\right\}} \\
       & \leq \sum_{n=2}^\infty \exp{\left\{ - \frac{C_1^{\prime2}n^{ \frac{1}{2}}}{\varepsilon^2}\right\}}.\\
\end{align*}
The integral bound gives us
\begin{align*}
     \sum_{n=2}^\infty \exp{\left\{ - \frac{C_1^{\prime2}n^{ \frac{1}{2}}}{2\varepsilon^2}\right\}} \leq \int_{1}^\infty \exp{\left\{ - \frac{C_1^{\prime2}x^{ \frac{1}{2}}}{2\varepsilon^2}\right\}} dx \leq 8\frac{\varepsilon^2}{C_1^{\prime2}}\exp{\left\{ - \frac{C_1^{\prime2}}{2\varepsilon^2}\right\}}.
\end{align*}
We conclude that 
\begin{align*}
           \mathbb{P}\left({\sup_{\substack{\gamma\in\mathcal{C}_0}} \frac{|\Delta_{\I_-(\gamma)}(h)|}{c_2|\gamma|} > \frac{1}{4}}\right) &\leq  \frac{8}{C_1^{\prime2}} \varepsilon^2 \exp{\left\{ - \frac{C_1^{\prime2}}{2\varepsilon^2}\right\}}\leq \exp{\left\{ - \frac{C_1}{\varepsilon^2}\right\}},
\end{align*}
for $\varepsilon^2 < C_1$ and $C_1 \coloneqq C_1(\alpha, d) = \min\{\left(\frac{c_2}{8L_1L}\right)^2, \frac{C_1^{\prime2}}{8} \}$.  The dependency on $\alpha$ is due to the dependency on $c_2(\alpha, d)$.
\end{proof}

The next two subsections are dedicated to proving Proposition \ref{Prop: Bound.gamma_2}.

    \section{Controling \texorpdfstring{$\gamma_2(\I_-(n), \d_2)$}{I-(n),d2}}\label{Sec: Controling_gamma_2} The next section is dedicated to proving Proposition \ref{Prop: Bound.gamma_2}. To construct the cover by balls in Dudley's entropy bound, we use the coarse-graining idea introduced in \cite{FFS84}.  For each $\ell>0$ and each contour ${\gamma\in\mathcal{C}_0(n)}$, we will associate a region $B_\ell(\gamma)$ that approximates the interior $\I_-(\gamma)$ in a scaled lattice, with the scale growing with $\ell$. This is done in a way that two interiors that are approximated by the same region are in a ball in distance $\d_2$ with a fixed radius, depending on $\ell$.

An $r\ell$-cube $C_{r\ell}$ is \textit{admissible} if more than a  half of its points are inside $\I_-(\gamma)$. Thus, the set of admissible cubes is
\begin{equation*}
    \mathfrak{C}_\ell(\gamma) \coloneqq \left\{C_{r\ell} : |C_{r\ell}\cap \I_-(\gamma)| \geq \frac{1}{2}|C_{r\ell}|\right\}.
\end{equation*}
With this notion of admissibility, two contours with the same admissible cubes should be close in distance $\d_2$. Consider functions $B_\ell:\mathcal{E}^+_\Lambda \xrightarrow[]{} \mathcal{P}(\Z^d)$, with $\mathcal{P}(\Z^d) \coloneqq \{A:A\Subset\Z^d\}$, that takes contours $\gamma$ to $B_\ell(\gamma) \coloneqq B_{\mathfrak{C}_{\ell}(\gamma)}$, the region covered by the admissible cubes. We will be interested in counting  the image of $B_\ell$ by $\mathcal{C}_0(n)$, that is, $|B_\ell(\mathcal{C}_{0}(n))| = |\{B:B=B_\ell(\gamma)\text{ for some }\gamma \in \mathcal{C}_0(n)\}|$. Notice that $B_\ell(\gamma)$ is uniquely determined by $\partial B_\ell(\gamma)$. Given any collection $\C_{m}$, we define the \textit{edge boundary of } $\C_m$ as 
$$
\partial \C_m \coloneqq \{ \{C_{m}, C^\prime_{m}\} : C_{m} \in \C_m, \ C_m^\prime \notin \C_m \textrm{ and} \  C_m^\prime \text{ shares a face with }C_m\}.
$$ 
We also define the \textit{inner boundary of }$\C_m$ as
$$
\fint \C_m\coloneqq \{ C_{m}\in \C_m : \exists C_m^\prime \notin \C_m \textrm{ such that }  \{C_m,C_m^\prime\}\in\partial \C_m\}.
$$ 
With this definition, it is clear that $\partial B_\ell(\gamma)$ is uniquely determined by $\partial \mathfrak{C}_\ell(\gamma)$. Hence, defining $\partial \mathfrak{C}_{r\ell}(\mathcal{C}_0(n)) \coloneqq \{\partial\C_{r\ell} : \C_{r\ell}=\mathfrak{C}_\ell(\gamma )\text{ for some }\gamma \in \mathcal{C}_0(n)\}$, we have $|B_\ell(\mathcal{C}_0(n))| = |\partial\mathfrak{C}_\ell(\mathcal{C}_0(n))|$. In a similar fashion we define $\fint \mathfrak{C}_{r\ell}(\mathcal{C}_0(n)) \coloneqq \{\fint\C_{r\ell} : \C_{r\ell}=\mathfrak{C}_\ell(\gamma )\text{ for some }\gamma \in \mathcal{C}_0(n)\}$. We will now control the number of cubes in $\mathfrak{C}_\ell(\gamma)$. 

\begin{proposition}\label{Proposition1}For the functions $(B_\ell)_{\ell\geq 0}$ defined above, there exists constants $b_1,b_2$ depending only on $d$ and $r$ such that 
\begin{equation}\label{Eq: Prop.1.FFS.i}
    |\fint\mathfrak{C}_\ell(\gamma)| \leq b_1\frac{|\fext \I_-(\gamma)|}{2^{r\ell(d-1)}} \leq b_1 \frac{|\gamma|}{2^{r\ell(d-1)}},
\end{equation}
    and 
\begin{equation}\label{Eq: Prop.1.FFS.ii}
    |B_\ell(\gamma)\Delta B_{\ell+1}(\gamma)| \leq b_2 2^{r\ell} |\gamma|,
\end{equation}
for every $\ell\geq 0$ and $\gamma\in\mathcal{C}_0(n)$.
\end{proposition}

\begin{remark}\label{rmk: Upper_bound_on_ell}
    This proposition shows that when $\frac{b_1|\gamma|}{2^{r\ell(d-1)}}<1$ there are no admissible cubes. Therefore, in some propositions, we assume $\ell\leq \frac{\log_{2^r}(b_1 |\gamma|)}{d-1}$, since the relevant bounds on the complementary case follow trivially.
\end{remark}

The proof of Proposition \ref{Proposition1} follows the same steps of \ref{Proposition1_SR}, since both lemmas \ref{Lemma: Geo.discreta.1} and \ref{Lemma: Proposicao1.Aux1} can be used for the new admissible regions as they do not require $A\Subset\Z^d$ to be connected.  

The next Corollary estimates the difference between the minus interior of a contour and its approximation, see Figure \ref{Fig: Figura7}.
\begin{corollary}\label{Cor: Bound_diam_B_ell}
     There exists a constant $b_3>0$ such that, for any $\ell>0$ and any two contours $\gamma_1,\gamma_2 \in \mathcal{C}_0(n)$ with $B_\ell(\gamma_1)=B_{\ell}(\gamma_2)$, 
    \begin{equation*}
        \d_2(\I_-(\gamma_1),\I_-(\gamma_2))\leq 4 \varepsilon b_3 2^{\frac{r\ell}{2}} n^{\frac{1}{2}}. 
    \end{equation*} 
\end{corollary}

\begin{proof}
    This is a simple application of the triangular inequality, since $\d_2(\I_-(\gamma_1),\I_-(\gamma_2)) \leq \d_2(\I_-(\gamma_1),B_\ell(\gamma_1)) + \d_2(\I_-(\gamma_2),B_\ell(\gamma_2))$ and 
    \begin{align*}
        \d_2(\I_-(\gamma_1),B_\ell(\gamma_1)) &\leq \sum_{i=1}^\ell \d_2(B_i(\gamma_1),B_{i-1}(\gamma_1)) = \sum_{i=1}^\ell 2\varepsilon\sqrt{|B_i(\gamma_1)\Delta B_{i-1}(\gamma_1)|} \\
        & \leq 2\varepsilon\sqrt{b_2}\sqrt{n} \sum_{i=1}^\ell  2^{\frac{ir}{2}}   \leq 4\varepsilon\sqrt{b_2}2^{\frac{r\ell}{2}} \sqrt{n} 
    \end{align*}
    where in the second to last equation used \eqref{Eq: Prop.1.FFS.ii}. As the same bound holds for $d_2(\I_-(\gamma_2),B_\ell(\gamma_2))$, the corollary is proved by taking $b_3 = 2\sqrt{b_2}$.
\end{proof}

\begin{remark}\label{Rmk: Bounding_N_by_B_ell}
    Corollary \ref{Cor: Bound_diam_B_ell} shows that we can create a cover of $\I_-(n)$, indexed by $B_\ell(\mathcal{C}_0(n))$, of balls with radius $4 \varepsilon b_3 2^{\frac{r\ell}{2}} n^{\frac{1}{2}}$. Therefore $N(\I_-(n), \d_2, 4\varepsilon b_3 2^{\frac{r\ell}{2}} n^{\frac{1}{2}}) \leq |B_\ell(\mathcal{C}_0(n))|$. 
\end{remark}
We now proceed to bounding $|B_\ell(\mathcal{C}_0(n))|$. As we discussed before, in the definition of admissibility at the beginning of this subsection, $|B_\ell(\mathcal{C}_0(n))| = |\partial \mathfrak{C}_{r\ell}(\mathcal{C}_0(n))|$. In the short-range case, a key ingredient to count the admissible cubes is that despite $B_\ell(\gamma)$ not being connected, all cubes are close to a connected region with size $|\gamma|$. As the contours now may not connected, we need to change the strategy: we choose a suitable scale $L(\ell)$ and count how many $rL(\ell)$-coverings of contours there are. That is, we first control $|\C_{rL(\ell)}(\mathcal{C}_0(n))|$. Once a $rL(\ell)$-covering is fixed, we choose which $r\ell$-cubes inside this covering will be admissible. At last, we choose the scale $L(\ell)$ in a suitable way.

The first step is to bound $|\C_{rL}(\mathcal{C}_0(n))|$, for $L>0$. For $n,m\geq 0$, we say that $\mathscr{C}_n$ is \textit{subordinated} to $\C_m$, denoted by $\C_n\preceq \C_m$, if for all $C_n\in\C_n$, there exists $C_m\in \C_m$ such that $C_n\subset C_m$. Moreover, define 
\begin{equation*}
    N(\C_m, n, V) \coloneqq |\{\C_n : \C_n\preceq \C_m, |\C_n|=V\}|,
\end{equation*}
the number of collections of $n$-cubes $\C_n$ subordinated to a fixed collection $\C_m$ and with $|\C_n|=V$. Notice that every $m$-cube contains $2^d$ $(m-1)$-cubes, all of them being disjoint. Therefore, the number of $n$-cubes inside a $m$-cube is $2^{(m-n)d}$ and we have $N(\C_m, n, V) = \binom{2^{(m-n)d} |\C_{m}|}{V}$.
In particular, the bound on the binomial $\binom{n}{k}\leq \left(\frac{en}{k}\right)^k$ yields
\begin{equation}\label{Eq: Bound.on.N}
    N(\C_{r(\ell+1)}, r\ell, V) = \binom{2^{rd}|\C_{r(\ell+1)}|}{V} \leq \left(\frac{2^{rd}e|\C_{r(\ell+1)}|}{V}\right)^{V}.
\end{equation}
For any subset $\Lambda \Subset \Z^d$, define
\begin{equation*}
    V_r^\ell(\Lambda)\coloneqq \sum_{n=\ell}^{n_r(\Lambda)} |\C_{rn}(\Lambda)|,
\end{equation*}
where $n_r(\Lambda)\coloneqq \ceil{\log_{2^r}(\diam (\Lambda))}$. To control $V_r^\ell(\Lambda)$ we bound the number of coverings at a fixed step $L>0$.

\begin{proposition}\label{Prop. partition.a.graph}
Let $k\geq 1$ and $G$ be a finite, non-empty, connected simple graph with vertex set $v(G)$. Then, $G$ can be covered by $\ceil*{|v(G)|/k}$ connected sub-graphs of size at most $2k$.
\end{proposition}
We omit the proof since it is the same as in \cite{Affonso.2021}. Remember that, given ${G = (V,E) \in \mathscr{G}_n(\Lambda)}$, $\Lambda^G \coloneqq \Lambda \cap B_V$ denotes the area of $\Lambda$ covered by $G$. Remember also that, for $A\Subset \Z^d$ and $j\geq 1$, $\Gamma^r_j(A)$ are the partition elements removed at step $j$, in the construction presented in Section 2.  Using this construction we can prove the following lemma.

\begin{lemma}\label{Lemma: Big.clusters_2}
    Let $A\Subset \Z^d$, $\gamma\in\Gamma^r(A)$ and $j \geq 1$ be such that $\gamma\in \Gamma^r_j(A)$. Then, for any $\ell < j$ and $G_{r\ell}\in \mathscr{G}_{r\ell}(\gamma)$,
    \begin{equation}\label{Eq: Lower_bound_on_the_covering_of_gamma_G}
        2^{r(1-\frac{1}{d})\ell} \leq |\C_{r\ell}(\gamma^{G_{r\ell}})| 
    \end{equation}
\end{lemma}
\begin{proof}
        Given $G_{r\ell}\in \mathscr{G}_{r\ell}(\gamma)$, by our construction of the contour, $2^{r(d+1)\ell} < |V(\gamma^{G_{r\ell}})|$. A trivial bound gives us $|V(\gamma^{G_{r\ell}})| \leq 2^{r\ell d}|\C_{r\ell}(V(\gamma^{G_{r\ell}}))|$. Associating each cube $C_m(x)$ to $x$, we get a one-to-one correspondence between $m$-cubes and lattice points that preserves neighbors, that is, two m-cubes $C_m(x)$ and $C_m(y)$ share a face if and only if $|x-y|=1$. We can therefore apply the isoperimetric inequality to get $|\C_{r\ell}(V(\gamma^{G_{r\ell}}))| \leq |\fint \C_{r\ell}(V(\gamma^{G_{r\ell}}))|^{\frac{d}{d-1}}\leq |\C_{r\ell}(\gamma^{G_{r\ell}})|^{\frac{d}{d-1}}$, where in the last equation we are using that every cube in the boundary of cubes must cover at least one point of $\gamma^{G_{r\ell}}$. We conclude that $2^{r(d+1)\ell} \leq 2^{r\ell d}|\C_{r\ell}(\gamma^{G_{r\ell}})|^{\frac{d}{d-1}}$, and \eqref{Eq: Lower_bound_on_the_covering_of_gamma_G} follows.
\end{proof}

As a corollary, we can recuperate a key lemma of \cite{Affonso.2021}, which is the following.
\begin{lemma}\label{Lemma: Big.clusters_1}
    Given $A\Subset \Z^d$, $n> 1$ and $\gamma\in\Gamma^r(A)$, if $|\mathscr{G}_{rn}(\gamma)|\geq 2$ then $|v(G_{rn}(\gamma))| \geq 2^r$ for every $G_{rn}(\gamma)\in \mathscr{G}_{rn}(\gamma)$ 
\end{lemma}

   The next proposition bounds the partial volume.
\begin{proposition}\label{Prop. Bound.on.V_r^l(gamma)}
    There exists a constant $b_3 \coloneqq b_3(d, M, r)$ such that, for any $A\Subset \Z^d$, $\gamma\in\Gamma^r(A)$ and $\ell \geq 0$,
    
     \begin{equation*}
        V_r^\ell(\gamma)\leq b_3 (\ell\vee 1)|\mathscr{C}_{r\ell}(\gamma)|. 
    \end{equation*}

\end{proposition}

\begin{proof}
Start by noticing that $\gamma \in \Gamma^r(A)$ implies that $\Gamma^r(\gamma) = \{\gamma\}$. Let's assume first that $\ell\geq 2$. Define $g : \mathbb{N} \xrightarrow{} \Z$ by
\begin{equation}
    g(n)\coloneqq \floor*{\frac{n - 2 - \log_{2^r}(2M)}{a}}.
\end{equation}
It was proved in \cite[Proposition 3.13]{Affonso.2021} that 
\begin{equation}\label{Eq: Bound_c_n_by_C_g(n)}
    |\C_{rn}(\gamma)| \leq \frac{1}{2^{r-d-1}}|\C_{rg(n)}(\gamma)|,
\end{equation}
whenever $g(n)>0$, and every connected component of $G_{rg(n)}(\gamma)$ has more than $2^r -1$ vertices. This is equivalent,  by Lemma \ref{Lemma: Big.clusters_1}, to $|\mathscr{G}_{rg(n)}(\gamma)|\geq 2$ or $|\mathscr{G}_{rg(n)}(\gamma)|=1$ with $|v(G_{rg(n)}(\gamma))| \geq 2^r$. 
Consider then the auxiliary quantities
\begin{align*}
    &l_1(n)\coloneqq\max\{m : g^m(n)\geq \ell\} &\text{and} &&l_2(n)\coloneqq\max\{ m : |\mathscr{G}_{rg^m(n)}(\gamma)| = 1 \text{ and } |v(G_{rg^m(n)})|\leq 2^r-1\}.
\end{align*}

We first show that $l_2(n)$ is not zero for only a constant number of scales $n$. For any $m\leq l_2(n)$, as $\Sp{\gamma}\subset B_{\C_{rg^m(n)(\gamma)}}$, $\diam(\gamma)\leq \diam(B_{\C_{rg^m(n)(\gamma)}})$. Moreover, since the graph $G_{rg^m(n)}\in \mathscr{G}_{rg^m(n)}(\gamma)$ has $v(G_{rg^m(n)}) = \C_{rg^m(n)}(\gamma)$, and $|\C_{rg^m(n)(\gamma)}|\leq 2^{r}-1$.  For any $\Lambda,\Lambda^\prime\Subset \Z^d$, 
    \begin{equation*}
        \diam(\Lambda\cup \Lambda) \leq \diam(\Lambda) + \diam(\Lambda^\prime) + \dis(\Lambda,\Lambda^\prime),
    \end{equation*}
    and we can always extract a vertex from a connected graph in a way that the induced sub-graph is still connected, by removing a leaf of a spanning tree. Using this we can bound 
\begin{align}\label{Eq: Bound_diam_removing_trees}
    \diam(\gamma)\leq \diam(B_{\C_{rg^m(n)(\gamma)}}) & \leq \sum_{C_{rg^m(n)}\in v(G_{rg^m(n)})} \diam(C_{rg^{m}(n)}) + |v(G_i)|M2^{arg^m(n)} \nonumber\\
    &\leq(d2^{rg^m(n)} + Md^a2^{arg^m(n)})|\C_{rg^m(n)(\gamma)}|\leq 2Md^a2^{arg^m(n)+r}.
\end{align}
Applying the logarithm with respect to base $2^{r}$ we get
\begin{equation*}
    \log_{2^r}(\diam(\gamma)) \leq \log_{2^r}(2Md^a) + ag^m(n)+1 \leq \log_{2^r}(2Md^a) + \frac{n}{a^{m-1}} + 1
\end{equation*}
Assuming $\diam(\gamma)>2^{2r + 1}Md^a$, we can isolate the term depending on $m$ in the equation above and take the logarithm on both sides to get
\begin{equation*}
    m \leq 1 + \frac{\log_2(n) - \log_2(\log_{2^r}(\diam(\gamma)) - \log_{2^r}(2Md^a) - 1)}{\log_2(a)}.
\end{equation*}
Equation above holds for any element of $\{m : |\mathscr{G}_{rg^m(n)}(A)| = 1, |v(G_{rg^m(n)})|\leq 2^r-1\}$ thus it also holds for $l_2(n)$. This shows in particular that $l_2(n)=0$ for $n<\log_{2^r}(\diam(\gamma)) - \log_{2^r}(2Md^a) - 1$. Taking $N_0 = n_r(\gamma) - \log_{2^r}(2Md^a) - 2$, as $N_0\leq \log_{2^r}(\diam(\gamma)) - \log_{2^r}(2Md^a) - 1$ we can bound 
\begin{equation}\label{Eq: Bound_last_terms_of_V_r_l}
    \sum_{n=N_0}^{n_r(\gamma)}|\C_{rn}(\gamma)| \leq (\log_{2^r}(2Md^a)+2)|\C_{r\ell}(\gamma)|.
\end{equation}

We consider now $n<N_0$. Knowing that $l_2(n)=0$ and $|\C_{k}(\gamma)|\leq |\C_j(\gamma)|$, for all $j\leq k$, we get 
\begin{equation}\label{Eq: bound.on.rn.covering}
    |\C_{rn}(\gamma)|\leq \frac{1}{2^{(r-d-1)l_1(n)}}|\C_{r\ell}(\gamma)|.
\end{equation}
We claim that
\begin{equation}\label{Eq: lower.bound.on.l1}
    l_1(n) \geq \begin{cases}
                        0, &\text{ if }n\leq \overline{b}+\ell\\ 
                        \left\lfloor\frac{\log_2(n) - \log_2(\overline{b} + \ell)}{\log_2(a)}\right\rfloor, & \text{ if }n > \overline{b} + \ell, 
                \end{cases}
\end{equation}
where $\overline{b} = (a+2 + \log_{2^r}(2M))(a-1)^{-1}$. Given $n > \overline{b} + \ell$, consider
\begin{equation*}
    \Tilde{g}(n) = \frac{n - 2 - \log_{2^r}(2M)}{a} - 1.
\end{equation*}
It is clear that $g(n)\geq \Tilde{g}(n)$ and both functions are increasing, therefore $g^m(n)\geq \Tilde{g}^m(n)$ for every $m\geq 0$. As
\begin{equation*}
    \Tilde{g}^m(n) = \frac{n}{a^m} - b^\prime\frac{a^m - 1}{a^{m-1}(a-1)},
\end{equation*}
with $b^\prime = (a+2 + \log_{2^r}(2M))a^{-1}$, it is sufficient to have
\begin{equation*}
    \frac{n}{a^m} -\frac{a b^\prime}{(a-1)}\geq \ell.
\end{equation*}
We get the desired bound by applying the logarithm with base two in the equation above. The bounds \eqref{Eq: bound.on.rn.covering} and \eqref{Eq: lower.bound.on.l1} yields

\begin{align*}
    V_r^\ell(\gamma) &=  \sum_{n=\ell}^{\overline{b} + \ell-1}|\C_{rn}(\gamma)| +   \sum_{n=\overline{b} + \ell}^{N_0-1}|\C_{rn}(\gamma)| +  \sum_{n=N_0}^{n_r(\gamma)}|\C_{rn}(\gamma)| \\
    &\leq 
     \overline{b}|\C_{r\ell}(\gamma)| + |\C_{r\ell}(\gamma)|2^{r-d-1}(\overline{b}+\ell)^{\frac{r-d-1}{\log_2(a)}}\sum_{n=\overline{b} + \ell}\frac{1}{n^{\frac{r-d-1}{\log_2(a)}}} + (\log_{2^r}(2Md^a)+2)|\C_{r\ell}(\gamma)|\\
    &\leq (\overline{b} + \log_{2^r}(2Md^a) +2)|\C_{r\ell}(\gamma)| +|\C_{r\ell}(\gamma)| 2^{r-d-1}(\overline{b}+1)^{\frac{r-d-1}{\log_2(a)}}\ell^{\frac{r-d-1}{\log_2(a)}}\sum_{n=\ell + 1}^{\infty}\frac{1}{n^{\frac{r-d-1}{\log_2(a)}}}\\
    &\leq \left(\overline{b} + \log_{2^r}(2Md^a) +2+ 2^{r-d-1}(\overline{b}+1)^{\frac{r-d-1}{\log_2(a)}}\frac{\log_2(a)}{r-d-1+\log_2(a)}\right)\ell|\C_{r\ell}(\gamma)|,
\end{align*}
where in the last inequality we used the integral bound $$\sum_{n=\ell + 1}^{\infty}{n^{-\frac{r-d-1}{\log_2(a)}}}\leq \int_{\ell}^\infty {x^{-\frac{r-d-1}{\log_2(a)}}} dx = \frac{\log_2(a)}{r-d-1+\log_2(a)}\ell^{1 - \frac{r-d-1}{\log_2(a)}}.$$

If $\diam(\gamma)\leq 2^{2r + 1}M$, we have
\begin{equation*}
     V_r^\ell(\gamma) \leq (n_r(\gamma) - \ell+1)|\C_{r\ell}(\gamma)| \leq (3 + \log_{2^r}(2M))|\C_{r\ell}(\gamma)|.
\end{equation*}
Taking $b_3^\prime\coloneqq \max\{{2^{r-d+2}(2 + \frac{a}{d -1})(\overline{b} + \log_{2^r}(2Md^a) +3)^{\frac{r-d-1}{\log_2(a)}}}, 3 + \log_{2^r}(2M)\}$ we get the desired bound when $\ell\geq 2$. For $\ell=0$, a trivial bound yields  $V_r^0(\gamma) = 2|\gamma| + V_r^2(\gamma)\leq (2 + b_3^\prime 2)|\gamma|$. Similarly, for $\ell =1$,  $V_r^1(\gamma) = |\C_{r}(\gamma)| + V_r^2(\gamma)\leq (1+ b_3^\prime 2)|\mathscr{C}_{r}(\gamma)|$ and we conclude the proof by taking $b_3 \coloneqq 2(b_3^\prime +1)$.
\end{proof}

We then need to bound the minimal number of $r\ell$-cubes necessary to cover a contour. Using only Lemma \ref{Lemma: Big.clusters_1}, it is possible to prove the next proposition, in the same steps as in \cite[Proposition 3.13]{Affonso.2021}.

\begin{proposition}\label{Prop. Bound.on.C_rl(gamma)_Lucas}
    There exists a constant $b_4^{\prime\prime}\coloneqq b_4^{\prime\prime}(\alpha, d)$ such that for any $A\Subset \Z^d$, $\gamma\in\Gamma(A)$ and $1 \leq \ell\leq n_r(A)$, 
     \begin{equation*}
       |\C_{r\ell}(\gamma)|\leq b_4^{\prime\prime}\frac{|\gamma|}{\ell^{\frac{r-d-1}{\log_2(a)}}}.
    \end{equation*}
\end{proposition}

Next, we can improve this upper bound using our construction. This is the most relevant property of the new contours. 
\begin{proposition}\label{Prop. Bound.on.C_rl(gamma)}
    There exists constants $b_4\coloneqq b_4(\alpha, d)$ and $b_4^\prime\coloneqq b_4^\prime(\alpha, d)$  such that for any $A\Subset \Z^d$, $\gamma\in\Gamma^r_j(A)$ and $0 \leq \ell<j$,
    
     \begin{equation}\label{Eq: Bound.on.C_rl(gamma)_small_l}
       |\C_{r\ell}(\gamma)|\leq b_4\frac{(\ell \vee 1)^{\kappa}}{2^{ra^\prime\ell}}|\gamma|,
    \end{equation}
    with $a^\prime \coloneqq \frac{(1-\frac{1}{d})}{a -\frac{1}{d}}$ and $\kappa \coloneqq \frac{d+1 + r(1-\frac{1}{d}) (a+2 - d^{-1}+ \log_{2^r}(2M))(a-d^{-1})^{-1}}{\log_2(a + 1 -d^{-1})} $.
    Moreover, for $\ell\geq j$
    \begin{equation}\label{Eq: Bound.on.C_rl(gamma)_large_l}
        |\C_{r\ell}(\gamma)|\leq b_4^\prime \ell^{\kappa} \left(\frac{|\gamma|}{2^{r\frac{a^\prime}{a} \ell}} \vee 1\right).
    \end{equation}
\end{proposition}

\begin{proof}
     Lets first consider $\ell < j$. Define $f : \mathbb{N} \xrightarrow{} \Z$ by
\begin{equation}
    f(\ell)\coloneqq \floor*{\frac{\ell - \log_{2^r}(2M) - 1}{a + (1-\frac{1}{d})}}.
\end{equation}
Following the proof of \eqref{Eq: Bound_c_n_by_C_g(n)} in \cite[Proposition 3.13]{Affonso.2021}, we can show that 
\begin{equation}\label{Eq: Bound_C_l_by_C_f(l)}
    |\C_{r\ell}(\gamma)| \leq \frac{2^{d+1}}{2^{r(1-\frac{1}{d})f(\ell)}}|\C_{rf(\ell)}(\gamma)|.
\end{equation}
 By definition, $\mathscr{G}_{rf(\ell)}(\gamma)$ is the set of all connected components of $G_{rf(\ell)}(\gamma)$, hence
    \begin{equation}\label{Eq: 3.16.Lucas}
        |\C_{rf(\ell)}(\gamma)| = 2^{r(1-\frac{1}{d})f(\ell)}\sum_{G\in \mathscr{G}_{rf(\ell)}(\gamma)}\frac{|v(G)|}{2^{r(1-\frac{1}{d})f(\ell)}}.
    \end{equation}
    Proposition \ref{Prop. partition.a.graph} guarantees that we can split $G$ into sub-graphs $G_i$, with $1\leq i\leq \ceil{v(G)/2^{r(1-\frac{1}{d})f(\ell)}}$ and $|v(G_i)|\leq 2^{r(1-\frac{1}{d})f(\ell)+1}$. Proceeding as in \eqref{Eq: Bound_diam_removing_trees}, we can bound 
    \begin{align*}
        \diam(B_{v(G_i)}) &\leq \sum_{C_{rf(\ell)}\in v(G_i)} \diam(C_{rf(\ell)}) + |v(G_i)|M2^{arf(\ell)}\\
        &\leq |v(G_i)|(d2^{rf(\ell)} + M2^{arf(\ell)}) \leq 2M2^{r[f(\ell)(1-\frac{1}{d}) + a] + 1}\\
        &\leq 2^{r\ell}.
    \end{align*}
    
    The last inequality holds since $M,a,r\geq 1$. This shows that every $G_i$ can be covered by a cube with center in $\Z^d$ and side length $2^{r\ell}$. Every such cube can be covered by at most $2^d$ $r\ell$-cubes. Indeed, it is enough to consider the simpler case when the cube is of the form
    \begin{equation}\label{Cube.Q}
        \prod_{i=1}^d[q_i , q_i + 2^{r\ell})\cap\Z^d,
    \end{equation}
    with $q_i\in\{0, 1, \dots, 2^{r\ell}-1\}$, for $1\leq i \leq d$. It is easy to see that \begin{equation*}
        [q_i , q_i + 2^{r\ell }]\subset [0, 2^{r\ell})\cup [2^{r\ell}, 2^{r\ell +1}). 
    \end{equation*} 
    Taking the products for all $1\leq i\leq d$, we get $2^d$ $r\ell$-cubes that covers \eqref{Cube.Q}. 
    We conclude that, to cover a connected component $G\in \mathscr{G}_{rf(\ell)}$, we need at most $2^d\ceil{|v(G)|/2^{r(1-\frac{1}{d})f(\ell)}}$ $rf(\ell)$-cubes, yielding us
    \begin{equation}\label{Eq: 3.18.Lucas}
        |\C_{r\ell}(\gamma)|\leq |\C_{r\ell}( B_{\C_{rf(\ell)}(\gamma)})| \leq \sum_{G\in \mathscr{G}_{rf(\ell)}} |\C_{r\ell}(v(G))| \leq \sum_{G\in \mathscr{G}_{rf(\ell)}} 2^d\left\lceil{\frac{|v(G)|}{2^{r(1-\frac{1}{d})f(\ell)}}}\right\rceil. 
    \end{equation}
    When every connected component of $G_{rf(\ell)}(\gamma)$ has more than $2^{r(1-\frac{1}{d})f(\ell)}$ vertices, we can bound 
    \begin{equation*}
        \frac{1}{2}\left\lceil{\frac{|v(G)|}{2^{r(1-\frac{1}{d})f(\ell)}}}\right\rceil \leq {\frac{|v(G)|}{2^{r(1-\frac{1}{d})f(\ell)}}}.
    \end{equation*}
    Together with Inequalities \eqref{Eq: 3.16.Lucas}  and \eqref{Eq: 3.18.Lucas}, this yields
    \begin{equation}\label{Eq: Bound_C_rl_by_gamma_with_f(l)}
         |\C_{r\ell}(\gamma)| \leq \sum_{G\in \mathscr{G}_{rf(\ell)}} 2^{d+1} {\frac{|v(G)|}{2^{r(1-\frac{1}{d})f(\ell)}}}= \frac{2^{d+1}}{2^{r(1-\frac{1}{d})f(\ell)}}|\C_{rf(\ell)}(\gamma)|.
    \end{equation}

    Equation \eqref{Eq: Bound_C_l_by_C_f(l)} can be iterated as long as $f(\ell)$ is positive. Considering then the auxiliary quantity
    \begin{equation*}
        m(\ell)\coloneqq\max\{m : f^m(\ell)\geq 0\},  
    \end{equation*}
we have 
\begin{equation}\label{Eq: Bound_C_rl_by_gamma_with_f(l)_2}
    |\C_{r\ell}(\gamma)| \leq \frac{2^{(d+1)m(\ell)}}{2^{r\left(1-\frac{1}{d}\right)\left(\sum_{i=1}^{m(\ell)}f^i(\ell)\right)}}|\gamma|,
\end{equation}
so we need upper and lower estimates for $m(\ell)$.
We claim that
\begin{equation}\label{Eq: lower.bound.on.m}
    m(\ell) \geq \begin{cases}
                        0, &\text{ if }\ell\leq \overline{b}\\ 
                        \left\lfloor\frac{\log_2(\ell) - \log_2(\overline{b})}{\log_2(a + (1-\frac{1}{d}))}\right\rfloor, & \text{ if }\ell > \overline{b}, 
                \end{cases}
\end{equation}
where $\overline{b} = (\overline{a}+1 + \log_{2^r}(2M))(\overline{a}-1)^{-1}$ and $\overline{a}\coloneqq a + (1-\frac{1}{d})$. Given $\ell > \overline{b}$, consider
\begin{equation*}
     \overline{f}(\ell) = \frac{\ell - 1 - \log_{2^r}(2M)}{a + (1 - \frac{1}{d})} - 1.
\end{equation*}
It is clear that $f(\ell)\geq \overline{f}(\ell)$ and both functions are increasing, therefore $f^m(\ell)\geq \overline{f}^m(\ell)$ for every $m\geq 0$.  As
\begin{equation*}
       \overline{f}^m(\ell) = \frac{\ell}{\overline{a}^m} - b^\prime\frac{\overline{a}^m - 1}{\overline{a}^{m-1}(\overline{a}-1)},
\end{equation*}
with $b^\prime = (\overline{a}+1 + \log_{2^r}(2M))\overline{a}^{-1}$, it is sufficient to have
\begin{equation*}
    \frac{\ell}{\overline{a}^m} -\frac{\overline{a} b^\prime}{(\overline{a}-1)}\geq 0.
\end{equation*}
We get the desired bound by applying the logarithm with base two in the equation above. Moreover, we can bound
\begin{align*}
    \sum_{i=1}^{m(\ell)}f^{i}(\ell) & \geq  \sum_{i=1}^{m(\ell)}\frac{\ell}{\overline{a}^i} - m(\ell)\frac{\overline{a}b^\prime}{\overline{a} - 1} = \frac{1}{\overline{a}}(\frac{1-\frac{1}{\overline{a}^{m(\ell)}}}{{1-\frac{1}{\overline{a}}}})\ell - m(\ell)\overline{b} \\
    &\geq \frac{1}{\overline{a}-1}(1-\frac{1}{\overline{a}^{m(\ell)}})\ell - m(\ell)\overline{b} \geq \frac{1}{\overline{a}-1}(\ell-{\overline{a}\overline{b}}) - m(\ell)\overline{b}
\end{align*}

For the upper bound on $m(\ell)$, take $\Tilde{f}(\ell) \coloneqq \frac{\ell}{a + (1-\frac{1}{d})}$. As $f(\ell)\leq \Tilde{f}(\ell)$ and $\Tilde{f}$ is increasing, for every $m\geq 0$,  $f^m(\ell)\leq \Tilde{f}^m(\ell)$. Notice that, if $\Tilde{f}^m(\ell)\leq 1$, $f^{m+1}(\ell)<0$, and therefore $m+1>m(\ell)$. As $\Tilde{f}^m(\ell)\leq 1$ if and only if $\ell \leq [a + (1-\frac{1}{d})]^m$, taking $m=\left\lceil\frac{\log_2(\ell)}{\log_2\left(a + (1-\frac{1}{d})\right)}\right\rceil$ we get $\left\lceil\frac{\log_2(\ell)}{\log_2\left(a + (1-\frac{1}{d})\right)}\right\rceil + 1 > m(\ell)$.
Applying this bound on \eqref{Eq: Bound_C_rl_by_gamma_with_f(l)_2} we conclude that 
\begin{equation}\label{Eq: bound_on_C_ell_covering_l_geq_ab}
     |\C_{r\ell}(\gamma)| \leq \frac{2^{d+1 + r(1-\frac{1}{d})\left(\frac{\overline{a}}{\overline{a}-1} + 1\right)\overline{b}}\ell^{\frac{d+1 + r(1-\frac{1}{d})\overline{b}}{\log_2(\overline{a})}}}{2^{r(1-\frac{1}{d})\frac{1}{\overline{a}-1}\ell}}|\gamma|,
\end{equation}
for $\ell>\overline{b}$. When $\ell\leq\overline{b}$, we can take $\overline{b}_4\coloneqq \min\{{ (j\vee 1)^{\frac{d+1 + r(1-\frac{1}{d})\overline{b}}{\log_2(\overline{a})}}{2^{-r(1-\frac{1}{d})\frac{1}{\overline{a}-1}j}}} : 0\leq j \leq \overline{b}\}$ and then 
\begin{equation*}
       |\C_{r\ell}(\gamma)| \leq |\gamma|\leq \frac{1}{\overline{b}_4}\frac{(\ell\vee 1)^{\frac{d+1 + r(1-\frac{1}{d})\overline{b}}{\log_2(\overline{a})}}}{2^{r(1-\frac{1}{d})\frac{1}{\overline{a}-1}\ell}}|\gamma|.
\end{equation*}
This, together with equation \eqref{Eq: bound_on_C_ell_covering_l_geq_ab}, yields inequality \eqref{Eq: Bound.on.C_rl(gamma)_small_l} with $b_4 \coloneqq \max\{ 2^{d+1 + r(1-\frac{1}{d})\left(\frac{\overline{a}}{\overline{a}-1} + 1\right)\overline{b}}, \overline{b}_4^{-1}\}$.

To prove inequality \eqref{Eq: Bound.on.C_rl(gamma)_large_l}, we first notice that for any $\ell\geq j$,
\begin{align}\label{Eq:  Bound.on.C_rl(gamma)_large_l_aux_1}
    |\C_{r\ell}(\gamma)|\leq |\C_{r(j-1)}(\gamma)| \leq  b_42^{ra^\prime}\frac{j^\kappa}{2^{ra^\prime j}}|\gamma|.
\end{align}
When $\ell \leq aj$, this already gives us \eqref{Eq:  Bound.on.C_rl(gamma)_large_l}. For $\ell > aj$, we can give a better bound once we notice that, by the construction of the contour, the graph $G_{rj}(\gamma)$ is connected and its vertices are the covering $\C_{rj}(\gamma)$. In a similar fashion as done previously, by Proposition \ref{Prop. partition.a.graph} we can split $G_{rj}(\gamma)$ into $\ceil{|v(G_{rj}(\gamma))|/k}$ connected sub-graphs $G_1,\dots, G_k$, with $k\coloneqq { \left\lceil \frac{2^{r\ell}}{2^{raj}} \right\rceil}$ and $|v(G_i)|\leq 2 \left\lceil \frac{2^{r\ell}}{2^{raj}} \right\rceil$ for all $i=1,\dots, k$. Assuming $\ell > aj$, we have $|v(G_i)|\leq 2^{r(\ell - aj) + 2}$. As
  \begin{align*}
        \diam(B_{v(G_i)}) &\leq |v(G_i)|(d2^{rj} + M2^{arj}) \leq  2Md2^{raj}|v(G_i)|\\
        &\leq 8Md2^{r\ell},
    \end{align*}
$B_{v(G_i)}$ can be covered by $(8Md)^d$ cubes centered in $\Z^d$ with side length $2^{r\ell}$. As we seen before, every such cube can be covered by at most $2^d$ $r\ell$-cubes, therefore $|\C_{r\ell}(B_{v(G_i)})| \leq (16Md)^d$ and we conclude that 
\begin{align}\label{Eq:  Bound.on.C_rl(gamma)_large_l_aux_2}
    |\C_{r\ell}(\gamma)| \leq \sum_{i=1}^{\ceil{|v(G_{rj}(\gamma))|/k}} |\C_{r\ell}(B_{v(G_i)})| &\leq (16Md)^d \left\lceil \frac{|\C_{rj}(\gamma)|}{\frac{2^{r\ell}}{2^{raj}}} \right\rceil \nonumber \\ 
    &\leq 2(16Md)^d b_42^{ra^\prime} j^\kappa\left(\frac{2^{r(aj - \ell)}}{2^{ra^\prime j}}|\gamma| \vee 1\right).
\end{align}
  As we are assuming $\ell > aj$, 
\begin{align*}
\frac{2^{r(aj - \ell)}}{2^{ra^\prime j}} &= \frac{2^{r(a - a^\prime)j}}{2^{r\ell}} \\
&\leq \frac{2^{r(1 - \frac{a^\prime}{a})\ell}}{2^{r\ell}}  = \frac{1}{2^{r\frac{a^\prime}{a}\ell}},
\end{align*}
what concludes the proof of \eqref{Eq:  Bound.on.C_rl(gamma)_large_l} with $b_4^\prime \coloneqq 2(16Md)^d b_42^{ra^\prime}$.
\end{proof}

For any non-negative $V, M, a, r$, define
\begin{equation*}
    \mathcal{F}_{V}^\ell\coloneqq \{ \C_{r\ell} : V_r^\ell(B_{\C_{r\ell}}) = V, B_{\C_{r\ell}}\subset [-\diam(B_{\C_{r\ell}}), \diam(B_{\C_{r\ell}})]^d\}.
\end{equation*}

Using equation \eqref{Eq: Bound.on.N}, in the same steps as \cite[Proposition 3.11]{Affonso.2021}, we can show that the number of collections in $\mathcal{F}_V$ is exponentially bounded by $V$.

\begin{proposition}\label{Prop. Bound.on.Fv}
    There exists $b_5\coloneqq b_5(d,r)$ such that
    \begin{equation}\label{Eq: Bound.on.F_V}
        |\mathcal{F}^\ell_V| \leq e^{b_5V}.
    \end{equation}
\end{proposition}

\begin{proof}
We start by splitting $\mathcal{F}^\ell_V$ into $\mathcal{F}^\ell_{V,m} \coloneqq \{ \C_{r\ell}\in \mathcal{F}^\ell_V : n_r(B_{\C_{r\ell}})=m \}$. Since $\ell\leq n_r(B_{\C_{r\ell}}) \leq V_r^\ell(B_{\C_{r\ell}}) +\ell$, we get

\begin{equation}
    |\mathcal{F}^\ell_V| \leq \sum_{m=\ell}^{V+\ell} |\mathcal{F}^\ell_{V,m}|. 
\end{equation}
Denoting $(V_{rn})_{n=\ell}^{m}$ an arbitrary family of natural numbers satisfying 
\begin{equation}\label{Eq: sum.of.V_rn}
    \sum_{n=\ell}^{m} V_{rn}\leq  V,
\end{equation}
with $V_{rn}\leq V_{r(n-1)}$, we can bound
\begin{align}\label{Eq: bound.FVL}
    |\mathcal{F}^\ell_{V,m}| \leq \sum_{(V_{rn})_{n=\ell}^{m}} |\{\C_{r\ell} : B_{\C_{r\ell}} \subset [-2^{rm},2^{rm}]^d, |\C_{rn}(B_{\C_{r\ell}})| = V_{rn}, \text{ for every } \ell \leq n\leq m, n_r(B_{\C_{r\ell}}) = m\}|.
\end{align}
    As  $[-2^{rm},2^{rm}]^d$ is a cube centered in $\Z^d$ with side length $2^{rm+1}+1$, it can be covered by $3^d$ $(rm+1)$-cubes. Indeed, it is enough to consider the simpler case when the cube is of the form
    \begin{equation}\label{Cube.Q_2}
        \prod_{i=1}^d[q_i , q_i + 2^{rm + 1}]\cap\Z^d,
    \end{equation}
    with $q_i\in\{0, 1, \dots, 2^{rm+1}\}$, for $1\leq i \leq d$. It is easy to see that \begin{equation*}
        [q_i , q_i + 2^{rm +1}]\subset [0, 2^{rm+1})\cup [2^{rm+1},2^{rm +2})\cup [2^{rm +2}, 2^{rm +3}). 
    \end{equation*} 
    Taking the products for all $1\leq i\leq d$, we get $3^d$ $(rm+1)$-cubes that covers \eqref{Cube.Q_2}. 
We can give an upper bound to the right-hand side of equation \eqref{Eq: bound.FVL} by counting the number of families $(\C_{rn})_{n=\ell}^m$ such that $\C_{rn}\preceq \C_{r(n+1)}$, for $n<m$, and $\C_{rm}\preceq \C^0_{rm+1}$, yielding us 
\begin{align*}
    |\mathcal{F}^\ell_{V,m}| 
    &\leq \sum_{(V_{rn})_{n=\ell}^{m-1}} |\{ (\C_{rn})_{n=\ell}^{m} : |\C_{rn}|=V_{rn}, \C_{rn}\preceq \C_{r(n+1)},\C_{rm}\preceq \C^0_{rm+1}\}| \\
    & \leq  \sum_{(V_{rn})_{n=\ell}^{m-1}} \sum_{\substack{\C_{rm}\preceq \C^0_{rm+1}\\ |\C_{rm}|=V_{rm}}}\sum_{\substack{\C_{r(m-1)} \\ |\C_{r(m-1)}|=V_{r(m-1)}\\ \C_{r(m-1)\preceq \C_{rm}}}} \cdots \sum_{\substack{\C_{r(\ell+1)} \\ |\C_{r(\ell+1)}|=V_{r(\ell+1)}\\ \C_{r(\ell+1)\preceq \C_{r(\ell+2)}}}} N(\C_{r(\ell+1)}, r\ell, V_{r\ell}).           
\end{align*}
Iterating equation \eqref{Eq: Bound.on.N} we get that
\begin{align*}
    |\mathcal{F}^\ell_{V,m}| &\leq \sum_{(V_{rn})_{n=\ell}^{m}}\left( \frac{2^{d}e |\C^0_{rm+1}|}{V_{rm}}\right)^{V_{rm}}\prod_{n=\ell}^{m-1}\left( \frac{2^{rd}e V_{r(n+1)}}{V_{rn}}\right)^{V_{rn}}\\
    &\leq  \sum_{(V_{rn})_{n=\ell}^{m-1}}\left( \frac{2^{d}e3^d}{V_{rm}}\right)^{V_{rm}}\prod_{n=\ell}^{m-1}e^{(rd\log(2) +1)V_{rn}} \\
    &\leq \sum_{(V_{rn})_{n=\ell}^{m-1}}e^{(d\log(2) + 1 + d\log(3))V_{rm}}\prod_{n=\ell}^{m-1}e^{(rd\log(2) +1)V_{rn}} \leq  \sum_{(V_{rn})_{n=\ell}^{m-1}}e^{(rd\log(2) + 1 + d\log(3))V}
\end{align*}

As the number of solutions of \eqref{Eq: sum.of.V_rn} is bounded by $2^V$, we conclude that 
\begin{equation*}
     |\mathcal{F}^\ell_{V}| \leq \sum_{m=\ell}^{V+\ell}|\mathcal{F}_{V,m}^\ell| \leq  V2^Ve^{(rd\log(2) + 1 + d\log(3))V},
\end{equation*}
therefore equation \eqref{Eq: Bound.on.F_V} holds for $b_5\coloneqq [rd +1]\log(2) + 2 + d\log(3)$.
\end{proof}

With these propositions we can control the number of coverings of contours at a given scale, that its, we can give an upper bound to $\left|\C_{r\ell}\left( \mathcal{C}_0(n) \right) \right| = \left|\{\C_{r\ell} : \C_{r\ell}=\C_{r\ell}(\gamma) \ \text{for some }\gamma\in\mathcal{C}_0(n)\} \right|$.

\begin{proposition}\label{Prop: Bound_on_rl_coverings}
    Let $n\geq 0$, $\Lambda\Subset\Z^d$. There exists a constant $b_6\coloneqq b_6(a,d)>0$ such that,
    \begin{equation*}
        |\C_{r\ell}(\mathcal{C}_0(n))|\leq \exp{\left\{ b_6 (\ell\vee 1)^{\kappa+1}\left(\frac{n}{2^{r\frac{a^\prime}{a}\ell}}\vee 1\right)    \right\}}.
    \end{equation*}
\end{proposition}
\begin{proof}
    Proposition \ref{Prop. Bound.on.V_r^l(gamma)} together with Proposition \ref{Prop. Bound.on.C_rl(gamma)} yields, 
    \begin{equation}\label{Eq: Bound_partial_volume_l_between_1_and_j}
        V_r^\ell(\gamma) = V_r^\ell(B_{\C_{r\ell}(\gamma)}) \leq b_3(b_4+b_4^\prime)(\ell\vee 1)^{\kappa+1}\left(\frac{n}{2^{r\frac{a^\prime}{a}\ell}}\vee 1\right) =: R_{n,\ell}.
    \end{equation}
    Therefore, 
    \begin{multline*}
        \left\{\C_{r\ell} :\C_{r\ell}=\C_{r\ell}(\gamma) \ \textrm{for some }\gamma\in\mathcal{C}_0(n)\right\} \\ \subset \left\{ \C_{r\ell} : V_r^\ell(B_{\C_{r\ell}}) \leq R_{n,\ell} \ , B_{\C_{r\ell}}\subset  [-\diam(B_{\C_{r\ell}}), \diam(B_{\C_{r\ell}})]^d\right\} \\ \subset \bigcup_{V=1}^{\left\lceil R_{n,\ell}\right\rceil}\mathcal{F}^\ell_V
    \end{multline*}
    and Proposition \ref{Prop. Bound.on.Fv} yields 
    \begin{align*}
        |\left\{\C_{r\ell} :\C_{r\ell}=\C_{r\ell}(\gamma) \ \textrm{for some }\gamma\in\mathcal{C}_0(n)\right\}| &\leq \sum_{V=1}^{\left\lceil R_{n,\ell}\right\rceil}|\mathcal{F}^\ell_V|  
        \leq \exp{\left\{ 2b_5b_3(b_4+b_4^\prime)(\ell\vee 1)^{\kappa+1}\left(\frac{n}{2^{r\frac{a^\prime}{a}\ell}}\vee 1\right)\right\}}.
    \end{align*}
    This concludes the proof for $b_6\coloneqq 2b_5b_3(b_4 + b_4^\prime)$.
\end{proof}

A consequence of Proposition \ref{Prop: Bound_on_rl_coverings} is that we get an exponential bound on the number of contours with a fixed size.

\begin{corollary}\label{Cor: Bound_on_C_0_n}
	Let $d\ge 2$, and $\Lambda\Subset \mathbb{Z}^d$. For all $n\geq 1$, $|\mathcal{C}_0(n)| \leq e^{b_6 n}$.  
\end{corollary}

\begin{remark}
    This result was also proved in \cite{Affonso.2021} when the contours are the finest partition. However, it is only required that the contours are the partition given by the construction in \cite[Proposition 3.5]{Affonso.2021}. Our Lemma \ref{Lemma: Big.clusters_2} shows that our contours are a subset of such contours, thus the same counting holds. 
\end{remark}

By the construction of all contours, not all the graphs $(\C_{r\ell}(\gamma), E_{r\ell}(\gamma))$ are connected. However, the use of Propositions \ref{Prop. Bound.on.V_r^l(gamma)} and \ref{Prop. Bound.on.C_rl(gamma)_Lucas} together, recuperates a bound on $|\C_{r\ell}(\mathcal{C}_0(n))|$ as if it was the case. This is shown by the next proposition. Although such a Proposition is not necessary to prove our results,  we provide a proof for completeness. 
\begin{lemma}\label{Prop: Counting_spamming_trees}
    Given $\ell>0$, consider the graph $G=(\C_{r\ell}(\Z^d), E)$, with two vertices $C,C^\prime$ being connected if and only if $d(C,C^\prime)\leq M2^{ra\ell}$. There exists a constant $b_5^\prime \coloneqq b_5^\prime(d,\alpha)$ such that
    \begin{equation}
        |\{{\C_{r\ell}}: C_{r\ell}(0)\in \C_{r\ell}, \ \C_{r\ell} \emph{ is connected }, |\C_{r\ell}|=N\}| \leq e^{b_5^\prime\ell N}.
    \end{equation}
\end{lemma}

\begin{proof}
    To count $|\{{\C_{r\ell}}: C_{r\ell}(0)\in \C_{r\ell}, \ \C_{r\ell} \emph{ is connected }, |\C_{r\ell}|=N\}|$, it is enough to count the number spanning trees containing $C_{r\ell}(0)$ with $N$ vertices. Let $\mathcal{T}_0$ be the set of all such trees. Fixed $T\in\mathcal{T}_0$, for each $C_{r\ell}\in v(T)$, let $\d_T(C_{r\ell})$ be the degree of $C_{r\ell}$. As $T$ is a tree, $\sum_{C_{r\ell}\in v(T)}\d_T(C_{r\ell}) = 2(N-1)$. Moreover, as there are at most $2^{rd(a\ell + \log_{2^r}M - \ell)}$ $r\ell$-cubes inside a $r(a\ell + \log_{2^r}M)$-cube, each cube $C_{r\ell}\in T$ has at most $2^{rd(a +\log_{2^r}M -1)\ell}$ neighbours. Let $(d_i)_{i=1}^N$ denote a general solution to 
    \begin{equation}\label{Eq: sum_d_i}
        \sum_{i=1}^{N}d_i = 2(N-1),
    \end{equation}
    with $d_i\leq 2^{rd(a +\log_{2^r}M -1)\ell}$ for all $i=1,\dots, N$. Then 
    \begin{align*}
        |\{{\C_{r\ell}}: C_{r\ell}(0)\in \C_{r\ell}, \ \C_{r\ell} \emph{ is connected }, |\C_{r\ell}|=N\}| \leq \sum_{(d_i)_{i=1}^N} |\{T \in\mathcal{T}_0: d_T(C^i) = d_i\}|.
    \end{align*}
    In the set above, $\{C^1, C^2, \dots, C^N\}$ is any ordering of $v(T)$ with $C^1 = C_{r\ell}(0)$. Therefore, 
    \begin{equation*}
        |\{T \in\mathcal{T}_0: d_T(C^i) = d_i\}| \leq \prod_{i=1}^N\binom{2^{rd(a +\log_{2^r}M -1)\ell}}{d_i}\leq (e2^{rd(a +\log_{2^r}M -1)\ell})^{N}. 
    \end{equation*}
    As the number of solutions to \eqref{Eq: sum_d_i} is bounded by $2^N$, we conclude that
    \begin{align*}
        \{{\C_{r\ell}}: C_{r\ell}(0)\in \C_{r\ell}, \ \C_{r\ell} \emph{ is connected }, |\C_{r\ell}|=N\}| \leq  2^Ne^N2^{rd(a +\log_{2^r}M -1)\ell N} \leq e^{b_5^\prime \ell N}, 
    \end{align*}
    with $b_5^\prime \coloneqq \log{2} + 1 + {rd(a +\log_{2^r}M -1)\log{2}}$.
\end{proof}

We are finally ready to upper bound the number of admissible regions $ |B_\ell(\mathcal{C}_0(n))|$ at scale $r\ell$.
\begin{proposition}\label{Prop: Bound_on_boundary_of_admissible_sets}
    Let $n\geq 0$, $\Lambda\Subset\Z^d$ and $1\leq\ell\leq \log_{2^r}(b_1 n)(d-1)^{-1}$. There exists a constant $c_4\coloneqq c_4(\alpha, d)$ such that,
    \begin{equation}\label{Eq: Bound_on_boundary_of_admissible_sets}
        |B_\ell(\mathcal{C}_0(n))|\leq \exp{\left\{c_4 \frac{\ell^{\kappa + 1} n}{2^{r\ell(d-1)}}   \right\}}.
    \end{equation}
\end{proposition}

\begin{proof}
    The upper bound on $\ell$ may seem artificial, but Remark \ref{rmk: Upper_bound_on_ell} shows that this is not the case. Remember that $|B_\ell(\mathcal{C}_0(n))|=|\partial\mathfrak{C}_\ell(\mathcal{C}_0(n))|$. Moreover, given $\{C_{r\ell},C_{r\ell}^\prime\}\in\partial\mathfrak{C}_\ell(\gamma)$, either $C_{r\ell}\in\fint\mathfrak{C}_{\ell}$ or $C_{r\ell}^\prime\in\fint\mathfrak{C}_{\ell}$. Using that $\sum_{k=0}^p\binom{p}{k} = 2^{p}$, we have
    \begin{equation}\label{Eq: Replacing_edge_by_inner_boundary}
       \begin{split} |\partial\mathfrak{C}_\ell(\mathcal{C}_0(n))| &= \sum_{\fint\C_{r\ell}\in \fint\mathfrak{C}_{\ell}(\mathcal{C}_0(n))} |\{\partial\C_{r\ell}^\prime : \fint\C_{r\ell}^\prime = \fint\C_{r\ell}\}| \\
        &\leq \sum_{\fint\C_{r\ell}\in \fint\mathfrak{C}_{\ell}(\mathcal{C}_0(n))} \sum_{k=1}^{2d|\fint\C_{r\ell}|}\binom{2d|\fint\C_{r\ell}|}{k}  \\
        &\leq\sum_{\fint\C_{r\ell}\in \fint\mathfrak{C}_{\ell}(\mathcal{C}_0(n))} 2^{2d|\fint\C_{r\ell}|}\leq |\fint\mathfrak{C}_\ell(\mathcal{C}_0(n))|e^{\log(2)2db_1\frac{n}{2^{r\ell(d-1)}}},
        \end{split}
    \end{equation}
     where in the last inequality we applied Proposition \ref{Proposition1}.  For every $L\geq \ell$ and an arbitrary collection $\C_{rL}$, define $\overline{\C_{rL}} = \C_{rL}\cup \{C_{rL}^\prime : \exists C_{rL}\in\C_{rL} \text{ such that } C_{rL}^\prime \text{ shares a face with } C_{rL}\}$. 
     
     Given $C_{r\ell}\in \fint \mathfrak{C}_\ell(\gamma)$, either $C_{r\ell}$ or one of its neighbouring cubes intersects $\Sp(\gamma)$. Hence, for any $L\geq \ell$, $\fint \mathfrak{C}_{\ell}(\gamma)\preceq \overline{\C_{rL}(\gamma)}$. Moreover, the number of $r\ell$-cubes inside a collection $\overline{\C_{rL}(\gamma)}$ of $rL$-cubes is bound by $|\overline{\C_{rL}(\gamma)}|2^{rd(L-\ell)} \leq 2d|\C_{rL}(\gamma)|2^{rd(L-\ell)}$. Using again Proposition \ref{Proposition1}, we can bound 
    \begin{equation}\label{Eq: Bound_on_internal_boundary}
    \begin{split}
           |\fint \mathfrak{C}_{\ell}(\mathcal{C}_0(n))| &\leq  \sum_{\substack{\C_{rL} \in \C_{rL}(\mathcal{C}_0(n))}}\sum_{k=0}^{\left\lceil{\frac{b_1n}{2^{r\ell(d-1)}}}\right\rceil}\binom{2d|\C_{rL}|2^{rd(L-\ell)}}{k} \\
           &\leq \sum_{\substack{\C_{rL} \in \C_{rL}(\mathcal{C}_0(n))}}\left(\frac{e2d|\C_{rL}|2^{rd(L-\ell)}}{{b_1n}{2^{-r\ell(d-1)}}}\right)^{\frac{b_1n}{2^{r\ell(d-1)}}} \\
           &\leq   \sum_{\substack{\C_{rL} \in \C_{rL}(\mathcal{C}_0(n))}}\left(\frac{e2d|\C_{rL}|2^{rdL}}{{b_1n}{2^{r\ell}}}\right)^{\frac{2b_1n}{2^{r\ell(d-1)}}},
           \end{split}
    \end{equation}
    where in the last equation we used that, for any $0<M\leq N$, $\sum_{p=0}^{M}\binom{N}{p}\leq \left(\frac{eN}{M}\right)^{M}$. Moreover, the restriction $\ell\leq \log_{2^r}(b_1 n)(d-1)^{-1}$ gives us $1\leq \frac{b_1 n}{2^{r\ell(d-1)}}$, so we bounded $\left\lceil \frac{b_1n}{2^{r\ell(d-1)}} \right\rceil \leq \frac{2b_1n}{2^{r\ell(d-1)}}$. Given a scale $\ell$, we choose $L(\ell) \coloneqq \left\lfloor \frac{a(d-1)\ell}{a^\prime} \right\rfloor$. The restriction $\ell\leq \log_{2^r}(b_1 n)(d-1)^{-1}$ allow us to bound $\left(\frac{n}{2^{r\frac{a^\prime}{a} L(\ell)}} \vee 1\right)\leq \left(\frac{b_1 n}{2^{r\frac{a^\prime}{a} L(\ell)}} \vee 1\right)\leq \left(b_12^{r\frac{a^\prime}{a}}\frac{n}{2^{r(d-1)\ell}} \vee 1\right) \leq b_1(2^{r\frac{a^\prime}{a}}\vee 1)\frac{n}{2^{r(d-1)\ell}}$ assuming w.l.o.g. that $b_1\geq1$, so, for any $ \C_{rL(\ell)} \in \C_{rL(\ell)}(\mathcal{C}_0(n))$, Proposition \ref{Prop. Bound.on.C_rl(gamma)} yields
    \begin{align*}
        |\C_{rL(\ell)}|2^{rdL(\ell)} &\leq (b_4+b_4^\prime)L(\ell)^{\kappa}\left(\frac{n}{2^{r\frac{a^\prime}{a} L(\ell)}}\vee 1\right) 2^{rd\frac{a(d-1)}{a^\prime}\ell} \\
        &\leq b_1 (b_4+b_4^\prime)\left(\frac{(d-1)a}{a^\prime}\right)^{\kappa}\left\lceil{2^{r\frac{a^\prime}{a}}}\right\rceil{\ell^{\kappa}}{2^{(d-1)r(\frac{ad}{a^\prime} -1)\ell}}n,
    \end{align*}
    hence 
    \begin{align*}
        \left(\frac{e2d|\C_{rL(\ell)}|2^{rdL(\ell)}}{{b_1n}{2^{r\ell}}}\right)^{\frac{2b_1n}{2^{r\ell(d-1)}}} &\leq \left(\frac{e2db_1 (b_4+b_4^\prime)\left(\frac{(d-1)a}{a^\prime}\right)^{\kappa}\left\lceil{2^{r\frac{a^\prime}{a}}}\right\rceil{\ell^{\kappa}}{2^{(d-1)r(\frac{ad}{a^\prime} -1)\ell}}n}{{b_1n}{2^{r\ell}}}\right)^{\frac{2b_1n}{2^{r\ell(d-1)}}}\\
        &\leq \left({e2d (b_4+b_4^\prime)\left(\frac{a(d-1)}{a^\prime}\right)^{\kappa}\left\lceil{2^{r\frac{a^\prime}{a}}}\right\rceil{\ell^{\kappa}}{2^{[(d-1)(\frac{ad}{a^\prime} -1) - 1]r\ell}}}\right)^{\frac{2b_1n}{2^{r\ell(d-1)}}}\\
        &\leq \exp\left\{ c_4^\prime \frac{\ell n}{2^{r\ell(d-1)}} \right\},
    \end{align*}
    with $c_4^\prime = [1  + \log(2d (b_4+b_4^\prime)\left(\frac{(d-1)a}{a^\prime}\right)^{\kappa}\left\lceil{2^{r\frac{a^\prime}{a}}}\right\rceil) + \kappa  + ((d-1)(\frac{ad}{a^\prime} -1) -1)\log(2)r]2b_1$. Moreover, by Proposition \ref{Prop: Bound_on_rl_coverings},
    \begin{align*}
         |\C_{rL(\ell)}(\mathcal{C}_0(n))| &\leq  \exp{\left\{ b_6L(\ell)^{\kappa+1} \left(\frac{n}{2^{r\frac{a^\prime}{a} L(\ell)}} \vee 1 \right) \right\}} \leq  \exp{\left\{ b_6 b_1\left(\frac{a(d-1)}{a^\prime}\right)^{\kappa +1}\left\lceil2^{r\frac{a^\prime}{a}}\right\rceil\frac{\ell^{\kappa+1}n}{2^{(d-1)r\ell}}\right\}}   \\
    \end{align*}
    so equations \eqref{Eq: Replacing_edge_by_inner_boundary} and \eqref{Eq: Bound_on_internal_boundary} yield
    \begin{align*}
         |\partial\mathfrak{C}_\ell(\mathcal{C}_0(n))| &\leq  \exp{\left\{b_6 b_1\left(\frac{a(d-1)}{a^\prime}\right)^{\kappa +1}\left\lceil2^{r\frac{a^\prime}{a}}\right\rceil\frac{\ell^{\kappa+1}n}{2^{r\ell(d-1)}} + c_4^\prime \frac{\ell n}{2^{r\ell(d-1)}} + \log(2)2db_1\frac{n}{2^{r\ell(d-1)}} \right\}}.
    \end{align*}

  that concludes our proof taking $c_4\coloneqq b_6 b_1\left(\frac{a(d-1)}{a^\prime}\right)^{\kappa +1}\left\lceil2^{r\frac{a^\prime}{a}}\right\rceil + c_4^{\prime} + \log(2)2db_1$.
\end{proof}

\begin{remark}\label{Rmk: Adaptation_for_Mar_partition}
    Using the notion of long-range contours of \cite{Affonso.2021}, we can get a worse upper bound on $|B_\ell(\mathcal{C}_0(n))|$ that is still good enough to prove phase transition in $d\geq 3$. Using Proposition \ref{Prop. Bound.on.C_rl(gamma)_Lucas}, we can prove in the same steps as Proposition \ref{Prop: Bound_on_rl_coverings} that $|\C_{r\ell}(\mathcal{C}_0(n))|\leq b_4^{\prime\prime}n \ell^{-\frac{r - d- 1 - \log_2(a)}{\log_2(a)}}$. With this, we can proceed similarly as in the proof of Proposition \ref{Prop: Bound_on_boundary_of_admissible_sets} but now choosing $L(\ell) = 2^{2r\left\lfloor \frac{\log_2(a)\ell}{r - d -1 - \log_2(a)} \right\rfloor}$, which gives us the bound 
    \begin{equation*}
        |B_\ell(\mathcal{C}_0(n))| \leq \exp{\left\{ c_4^{\prime}\frac{n}{2^{r\ell(d-1-\frac{2\log_2(a)}{r - d -1 - \log_2(a)})}} \right\}}.
    \end{equation*}
    For $r$ large enough, $d-1-\frac{2\log_2(a)}{r - d -1 - \log_2(a)}>1$ and the proof of Proposition \ref{Prop: Bound.gamma_2} follows with small adaptations. 
\end{remark}


\subsubsection*{Proof of Proposition \ref{Prop: Bound.gamma_2}}
At last, we prove the main proposition of this section. 

\begin{proof}
 As $N(\I_-(n), \d_2, \epsilon)$ is decreasing in $\epsilon$, we can use Dudley's entropy bound to get
    \begin{align*}
        {\mathbb{E}\left[\sup_{\gamma\in\mathcal{C}_0(n)}{\Delta_{\I_-(\gamma)}(h)}\right]} &=  {\mathbb{E}\left[\sup_{\I\in\I_-(n)}{\Delta_{\I}(h)}\right]}\\
        &\leq \int_{0}^\infty \sqrt{\log N(\I_-(n), \d_2, \epsilon)}d\epsilon \nonumber\\
        &\leq 2\varepsilon b_3 n^{\frac{1}{2}}\sum_{\ell=1}^\infty (2^{\frac{r\ell}{2}} - 2^{\frac{r(\ell-1)}{2}})\sqrt{\log N(\I_-(n), \d_2,\varepsilon b_3 2^{\frac{r\ell}{2}}n^{\frac{1}{2}})}.
    \end{align*}
We can bound the first term by noticing that $N(\I_-(n), d_2, 0) = |\I_-(n)|\leq 2^n|\mathcal{C}_0(n)|$, $2^n$ being an upper bound on the number of labels given a fixed support. By Corollary \ref{Cor: Bound_on_C_0_n}, $|\mathcal{C}_0(n)| \leq e^{c_1 n}$, and hence
\begin{equation*}
    4\varepsilon b_3 \overline{L} n^{\frac{1}{2}}\sqrt{\log{N(\I_-(n), d_2, 0)}} \leq 4\varepsilon b_3 \overline{L} (c_1 + \log 2)^{\frac{1}{2}} n.
\end{equation*}

Since $\d_2(\I_-(\gamma_1),\I_-(\gamma_2))\leq 2\varepsilon\sqrt{|\I_-(\gamma_1)| + |\I_-(\gamma_2)|}\leq 2\sqrt{2}\varepsilon n^{\frac{1}{2} + \frac{1}{2(d-1)}}$ for any $\gamma_1,\gamma_2\in\mathcal{C}_0(n)$, when $4\varepsilon b_3 \overline{L} 2^{\frac{r\ell}{2}}n^{\frac{1}{2}}\geq 2\sqrt{2}\varepsilon n^{\frac{1}{2} + \frac{1}{2(d-1)}}$, only one ball covers all interiors, hence all the terms in the sum above with $\ell > k(n)\coloneqq \floor{\frac{\log_{2^r}(n)}{(d-1)}}$ are zero. As $N(\I_-(n), \d_2,\varepsilon b_3 2^{\frac{r\ell}{2}}n^{\frac{1}{2}})\leq |B_{\ell}(\mathcal{C}_0(n))|$, see Remark \ref{Rmk: Bounding_N_by_B_ell}, using Proposition \ref{Prop: Bound_on_boundary_of_admissible_sets} we get
\begin{align*}
        {\mathbb{E}\left[\sup_{\gamma\in\mathcal{C}_0(n)}{\Delta_{\I(\gamma)}(h)}\right]} &\leq 4\varepsilon b_3 \overline{L} 2^{\frac{r}{2}}\sqrt{c_4} n^{\frac{1}{2}}\sum_{\ell=1}^{k(n)}2^{\frac{r\ell}{2}}\sqrt{\frac{\ell^{\kappa + 1} n }{2^{r(d-1)\ell}}} +  4\varepsilon b_3 \overline{L} (c_1 + \log 2)^{\frac{1}{2}} n\nonumber\\
        &\leq  4\varepsilon b_3 \overline{L} 2^{\frac{r}{2}} \sqrt{c_4}\left[ (c_1 + \log 2)^{\frac{1}{2}} + \sum_{\ell=1}^{\infty}\left(\frac{\ell^\frac{\kappa+1}{2}}{2^{\frac{r\ell(d-2)}{2}}} \right)\right]n.
\end{align*}

The series above converges for any $d\geq 3$, and we conclude that 
\begin{equation*}
       {\mathbb{E}\left[\sup_{\gamma\in\mathcal{C}_0(n)}{\Delta_{\I_-(\gamma)}(h)}\right]} \leq \varepsilon L_1^\prime n,
\end{equation*}
with $L_1^\prime\coloneqq   4 b_3 \overline{L} 2^{\frac{r}{2}}\sqrt{c_4}\left[ (c_1 + \log 2)^{\frac{1}{2}} + \sum_{\ell=1}^{\infty}\left(\frac{\ell^\frac{\kappa+1}{2}}{2^{\frac{r\ell(d-2)}{2}}} \right)\right]$. The desired result follows from Theorem \ref{MMT} taking the constant $L_1 \coloneqq L L_1^\prime$.
\end{proof}

\section{Phase transition}      
\begin{theorem}\label{Theo: Transicao_de_fase}
For $d\geq 3$ and $\alpha>d$, there exists a constant $C\coloneqq C(d,\alpha)$ such that, for all $\beta>0$ and $e\leq C$, the event 
    \begin{equation}\label{Eq: PTLR}
        \nu_{\Lambda; \beta, \varepsilon h}^+(\sigma_0 = -1) \leq e^{-C\beta} + e^{-C/\varepsilon^2} 
    \end{equation}
    has $\mathbb{P}$-probability bigger then $1 - e^{-C\beta} - e^{-C/\varepsilon^2}$.\\
    
In particular, for $\beta>\beta_c$ and $\varepsilon$ small enough, there is phase transition for the long-range Ising model.  
\end{theorem}

\begin{proof}
        The proof is an application of the Peierls' argument, but now on the joint measure $\mathbb{Q}$. By Proposition \ref{Prop: Bound.bad.event.1}, we have
        \begin{align}\label{Eq: Upper.bound.on.Q.1}
            \mathbb{Q}_{\Lambda; \beta, \varepsilon}^+(\sigma_0 = -1) &=  \mathbb{Q}_{\Lambda; \beta, \varepsilon}^+(\{\sigma_0 = -1\} \cap \mathcal{E}) + \mathbb{Q}_{\Lambda; \beta, \varepsilon}^+(\{\sigma_0 = -1\}\cap \mathcal{E}^c) \nonumber \\
            & \leq \mathbb{Q}_{\Lambda; \beta, \varepsilon}^+(\{\sigma_0 = -1\} \cap \mathcal{E}) +  e^{-C_1/\varepsilon^2}.
            \end{align}
 When $\sigma_0 = -1$, there must exist a contour $\gamma$ with $0\in V(\gamma)$, hence
\begin{equation*}
    \nu_{\Lambda; \beta, \varepsilon h}^+(\sigma_0 = -1) \leq \sum_{\gamma \in \mathcal{C}_0}\nu_{\Lambda; \beta, \varepsilon h}^+(\Omega(\gamma)),
\end{equation*}
where $\Omega(\gamma) \coloneqq \{\sigma\in\Omega : \Sp(\gamma) \subset \Gamma(\sigma)\}$. So we can write

\begin{align}\label{Eq: Upper.bound.on.Q.2}
    \mathbb{Q}_{\Lambda; \beta, \varepsilon}^+(\{\sigma_0 = -1\} \cap \mathcal{E}) &= \int_{\mathcal{E}}\sum_{\sigma : \sigma_0 = -1}g_{\Lambda; \beta, \varepsilon}^+(\sigma, h)dh \nonumber \\
    &\leq  \sum_{\gamma\in\mathcal{C}_0} \int_{\mathcal{E}}\sum_{\sigma\in\Omega(\gamma)}g_{\Lambda; \beta, \varepsilon}^+(\sigma, h)dh \nonumber \\
    &\leq  \sum_{\gamma \in \mathcal{C}_0} 2^{|\gamma|}\int_{\mathcal{E}}\sup_{\sigma\in\Omega(\gamma)}\frac{g_{\Lambda; \beta, \varepsilon}^+(\sigma, h)}{g_{\Lambda; \beta, \varepsilon}^+(\tau_{\gamma}(\sigma), \tau_{\I_-(\gamma)}(h))} \prod_{x\in\Lambda}\frac{1}{\sqrt{2\pi}}e^{-\frac{1}{2}h_x^2}dh_x 
\end{align}

In the third equation, we used that $\sum_{\sigma\in\Omega(\gamma)}g_{\Lambda; \beta, \varepsilon}^+(\tau_{\gamma}(\sigma), \tau_{\I_-(\gamma)}(h)) \leq 2^{|\gamma|}\times \prod_{x\in\Lambda}\frac{1}{\sqrt{2\pi}}e^{-\frac{1}{2}h_x^2}$, since the number of configurations that are incorrect in $\Sp(\gamma)$ are bounded by $2^{|\gamma|}$. Equation \eqref{Eq: quotient.of.gs} implies, 
\begin{align}\label{Eq: Upper.bound.on.Q.3}
    \sup_{ \sigma\in\Omega(\gamma)}\frac{g_{\Lambda; \beta, \varepsilon}^+(\sigma, h)}{g_{\Lambda; \beta, \varepsilon}^+(\tau_{\gamma}(\sigma), \tau_{\I_-(\gamma)}(h))} &\leq    \exp{\{{- \beta c_2 |\gamma| + \beta \Delta_{\gamma}(h)}\}} \sup_{ \sigma\in\Omega(\gamma)}\exp{\{ -2\beta\sum_{x\in \Sp^-(\gamma,\sigma)}\varepsilon h_x\}}\nonumber\\
    &\leq  \exp{\{- \beta \frac{3c_2}{4} |\gamma|\}} \sup_{ \sigma\in\Omega(\gamma)}\exp{\{ -2\beta\sum_{x\in \Sp^-(\gamma,\sigma)}\varepsilon h_x\}},
\end{align}
since $\Delta_{\gamma}(h) \leq \frac{c_2}{4}|\gamma|$, for all $h\in\mathcal{E}$. Moreover, notice that
\begin{align*}
\int_{\mathcal{E}} \sup_{ \sigma\in\Omega(\gamma)}\exp{\{ -2\beta\sum_{x\in \Sp^-(\gamma,\sigma)}\varepsilon h_x-\frac{1}{2}\sum_{x\in \Lambda}h_x^2\}}dh_x&\leq  \sup_{ \sigma\in\Omega(\gamma)}\int_{\mathbb{R}^\Lambda}\exp{\{ -2\beta\sum_{x\in \Sp^-(\gamma,\sigma)}\varepsilon h_x-\frac{1}{2}\sum_{x\in \Lambda}h_x^2\}}dh_x \\
&= (2\pi)^{\frac{|\Lambda|}{2}} \sup_{ \sigma\in\Omega(\gamma)}\exp{\{ 2(\beta \varepsilon)^2|\Sp^-(\gamma,\sigma)|\}},
\end{align*}
where in the equality we used the Gaussian integral formula. Notice that, since in this problem we need to take $\beta$ large and $\varepsilon$ small, we can make this choice in a way that $8\beta \varepsilon^2 \leq c_2$. This observation, together with Equations \eqref{Eq: Upper.bound.on.Q.1}, \eqref{Eq: Upper.bound.on.Q.2} and \eqref{Eq: Upper.bound.on.Q.3} and the inequality above yields
\begin{align*}
     \mathbb{Q}_{\Lambda; \beta, \varepsilon}^+(\sigma_0 = -1) &\leq  \sum_{\substack{\gamma\in \mathcal{E}_\Lambda^+\\ 0\in V(\gamma)}} 2^{|\gamma|}\exp{\{{- \beta( \frac{3c_2}{4} - 2\beta\varepsilon^2)|\gamma|}\}} + e^{-C_1/\varepsilon^2} \\
     &\leq \sum_{n\geq 1}\sum_{\substack{\gamma\in \mathcal{E}_\Lambda^+, |\gamma|=n \\ 0\in V(\gamma)}} \exp{\{{(-\beta \frac{c_2}{2} + \log2)n}\}} + e^{-C_1/\varepsilon^2}\\
     &\leq \sum_{n\geq 1}|\mathcal{C}_0(n)| \exp{\{{(-\beta \frac{c_2}{2} +\log2)n}\}} + e^{-C_1/\varepsilon^2} \\
    &\leq \sum_{n\geq 1} e^{(b_6 -\beta \frac{c_2}{2} +\log2)n} + e^{-C_1/\varepsilon^2}.
\end{align*}
When $\beta$ is large enough, the sum above converges and there exists a constant $C$ such that   
\begin{equation*}
    \mathbb{Q}_{\Lambda; \beta, \varepsilon}^+(\sigma_0 = -1) \leq e^{-\beta 2C} + e^{-2C / \varepsilon^2}.
\end{equation*}
The Markov Inequality finally yields
\begin{align*}
    \mathbb{P}\left( \nu_{\Lambda; \beta, \varepsilon h}^+(\sigma_0 = -1) \geq e^{-C\beta} + e^{-C/\varepsilon^2}\right) &\leq \frac{\mathbb{Q}_{\Lambda; \beta, \varepsilon}^+(\sigma_0 = -1)}{e^{-C\beta} - e^{-C/\varepsilon^2}} \\
    &\leq \frac{e^{-\beta 2C} + e^{-2C / \varepsilon^2}}{e^{-C\beta} + e^{-C/\varepsilon^2}} \leq e^{-C\beta} + e^{-C/\varepsilon^2},
\end{align*}
what proves our claim.
\end{proof}      

\chapter*{Concluding Remarks}
\addcontentsline{toc}{chapter}{Concluding Remarks}
\markboth{Concluding Remarks}{}
In this thesis, we studied two Ising models: the semi-infinite Ising model with a decaying field and the long-range random field Ising model. 

For the semi-infinite Ising model with an external field of the form $h_i = \lambda|i_d|^{-\delta}$, where $\lambda$ is the wall influence, we proved the existence of a critical value $\overline{\lambda}_c$ for which there is phase transition for $0\leq \lambda<\overline{\lambda}_c$ and there is uniqueness for $\lambda>\overline{\lambda}_c$. Moreover, we proved that when $\delta<1$, $\overline{\lambda}_c=0$ hence there is uniqueness for all temperatures and all wall parameters.

In the semi-infinite model, the external field $\overline{\lambda}_c$ resembles the long-range Ising model interaction. One natural question is if the phase transition results can be extended to a long-range semi-infinite Ising model. The case $\delta>1$ should be expandable with little effort. However, the proof of uniqueness when $\delta<1$ should be far from trivial. 

For the critical exponent $\delta=1$, the Ising model with external field considered in \cite{Bissacot_Cass_Cio_Pres_15} presents phase transition for a region of the parameters. This indicates that the semi-infinite model with a decaying field should also present a phase transition. However, our methods are not well suited to treat this case since the wall-free energy may be ill-defined. 

One last thing regarding this model is that, on \cite{FP-II}, they show that the macroscopic behavior of the system is as expected: there is a macroscopic layer on the wall when $\lambda>\lambda_c$, and the same do not happen when $\lambda<\lambda_c$. This macroscopic behavior is described using the Peierls' contours configuration for "defect" boundary conditions: if there is a contour surrounding the origin, the wall layer is thick, otherwise, it is thin. However, the Peierls' contours are not well suited for dealing with long-range interactions, so an interesting question is how to define the macroscopic behavior of the system in terms of the long-range contours considered in \cite{Affonso.2021, Affonso.Bissacot.Maia.23}. \\

In the second part of the thesis, we proved the existence of phase transitions for the long-range random field Ising model in $d\geq 3$ and $\alpha >d$, by following a new method of proving phase transition introduced by Ding and Zhuang \cite{Ding2021}, and using a modification of multidimensional contours defined in \cite{Affonso.2021}. The key part of the argument was to extend the results of \cite{FFS84} to contours that are not necessarily connected. This proof can be extended to other models with a contour system, as long as the probability of a bad event $\mathcal{E}^c$ decreases to zero for large $\varepsilon$.

For the long-range random field model, there is a wide range of very interesting problems to study next. The most interesting one is the presence of phase transition in dimensions $d=1,2$. In $d=1$, it is expected that the PT occurs in the entire region $\alpha\in (1, \frac{3}{2})$, without the restriction $J(1)>>1$ presented in \cite{Cassandro.Picco.09}. In two dimensions, as shown in \cite{Ding.Wirth.20}, we are unable to control the greedy animal lattice normalized by the boundary, so the PT results do not extend to this region. Also, it is expected that the long-range model random field model preserves the critical temperature as the model without an external field, in the same sense as the RFIM \cite{Ding.Liu.Xia.22}.

Regarding the decay of correlations, not much is known about the long-range model. The most prominent results are in one dimension, where Imbrie and Newman \cite{Imbrie.82, Imbrie.Newman.88} showed a polynomial decay rate, that matches the interaction at high temperatures. In \cite{Iag77}, the authors use the GHS inequality to show that, when there is no external field,  the truncated correlation of two points $x,y$ is always larger than $c(d,\alpha)J_{xy}$. The only known result for the long-range random field model is \cite{Klein.Masooman.97}, where they show a polynomial decay of correlations at high external fields or high temperatures.

A first interesting question is if this lower bound of \cite{Iag77} still holds when we have a random field. For an upper bound, a new correlation inequality \cite{Ding_Song_Sun_23} shows that the addition of any field only increases the difference of the magnetization at plus and minus boundary conditions. As an application, the bounds of Imbrie and Newman can be used to control the difference of the magnetization in the random field case. However, their results are limited to the region $\alpha\geq 2$, so it is still left to study the behavior of the correlations in the region $\alpha\in[\frac{3}{2},2)$. Similarly, in $d=2$, the RFIM presents exponential decay of correlations \cite{Aizenman_Harel_Peled_20, Ding_Xia_21}. For the long-range model the decay should be slower (polynomial) in the uniqueness region  $\alpha \geq 3$.


\backmatter \singlespacing   
\bibliographystyle{hacm}
\bibliography{bib}  

\end{document}